\documentclass[acmsmall,nonacm]{acmart}
\setcopyright{none}              
\authorsaddresses{}              
\renewcommand\footnotetextcopyrightpermission[1]{} 

\usepackage{graphicx}
 
\usepackage{amssymb}
\usepackage{amsmath, amsthm, amsfonts, enumitem, float, bm}
\usepackage{algpseudocode}
\usepackage{algorithm}
\usepackage{graphicx}
\usepackage{xcolor}
\usepackage{tikz-cd}
\usepackage{epsf}
\usepackage{hyperref}
\usepackage{balance}
\usepackage{lipsum}
\usepackage{mdframed}
\usepackage{wrapfig}
\usepackage{enumitem}
\usepackage{enumitem}
\usepackage{microtype}
\usepackage{mathtools}

\usepackage{tikz}
\usetikzlibrary{shapes,arrows.meta,positioning}
\usepackage{caption}

\usepackage{mdframed}
\usepackage{wrapfig}

\makeatletter 
\newcommand\mynobreakpar{\par\nobreak\@afterheading} 
\makeatother
\newenvironment{myitemize}{\mynobreakpar\begin{itemize}}{\end{itemize}}

\makeatother

\newtheorem{theorem}{Theorem}
\numberwithin{theorem}{section}
\newtheorem{lemma}[theorem]{Lemma}
\newtheorem{corollary}[theorem]{Corollary}
\newtheorem{definition}[theorem]{Definition}

\newtheorem{claim}[theorem]{Claim}
\newtheorem{remark}[theorem]{Remark}
\newtheorem{notation}[theorem]{Notation}
\newtheorem{observation}[theorem]{Observation}
\newtheorem{proposition}[theorem]{Proposition}
\newtheorem*{theorem*}{Theorem}
\newtheorem*{lemma*}{Lemma}

\newcommand{\View}{\mathrm{View}}	
\newcommand{\Lap}{\mathrm{Lap}}
\newcommand{\R}{\mathbb{R}}
\newcommand{\N}{\mathbb{N}}

\newcommand{\init}{\text{init}}
\newcommand{\halt}{\text{halt}}
\newcommand{\varnull}{\text{null}}
\newcommand{\partext}{\text{par}}
\newcommand{\params}{\texttt{params}}

\newcommand{\E}{\mathbb{E}}

\newcommand{\eps}{\epsilon}

\newcommand{\circstar}{\!\circ^*\!}

\newcommand{\ttL}{{\texttt{L}}}
\newcommand{\ttR}{{\texttt{R}}}
\newcommand{\filt}{\mathrm{Filt}}
\newcommand{\comp}{{\textrm{Comp}}}
\newcommand{\concomp}{{\textrm{ConComp}}}
\newcommand{\extconcomp}{{\textrm{ExtConComp}}}
\newcommand{\iecc}{{\textrm{IECC}}}
\newcommand{\supp}{\mathrm{supp}}	

\newcommand{\rml}{\mathrm{L}}

\newcommand{\V}{\mathcal{V}}
\newcommand{\I}{\mathcal{I}}
\newcommand{\RR}{\mathrm{RR}}
\newcommand{\irr}{\mathrm{IRR}}
\newcommand{\tv}{\mathrm{TV}}

\newcommand{\svt}{\textrm{SVT}}

\newcommand{\zo}{\{0,1\}}
\newcommand{\cM}{\mathcal{M}}
\newcommand{\cN}{\mathcal{N}}
\newcommand{\cP}{\mathcal{P}}
\newcommand{\cT}{\mathcal{T}}
\newcommand{\cA}{\mathcal{A}}
\newcommand{\cX}{\mathcal{X}}
\newcommand{\cY}{\mathcal{Y}}
\newcommand{\calF}{\mathcal{F}}

\newcommand{\calY}{\mathcal{Y}}
\newcommand{\calX}{\mathcal{X}}
\newcommand{\calS}{\mathcal{S}}
\newcommand{\rmlower}{\mathrm{Lower}}

\newcommand{\prop}{\mathrm{prop}}
\newcommand{\myomit}[1]{{}}

\newcommand{\sinit}{\texttt{SVT-init}}
\newcommand{\countinit}{\texttt{d-counter-init}}
\newcommand{\countupdate}{\texttt{d-counter-update}}
\newcommand{\histinit}{\texttt{Hist-init}}
\newcommand{\histupdate}{\texttt{Hist-update}}
\newcommand{\call}{\texttt{SVT-update}}
\newcommand{\thr}{\texttt{Thresh}}
\newcommand{\tth}{\texttt{d-counter}}
\newcommand{\tthss}{\texttt{HSS}}

\newcommand{\negspace}{\!\!\!\!}
\newcommand{\nospaceeq}{\!=\!}

\allowdisplaybreaks
\title{Concurrent Composition for Differentially Private Continual Mechanisms}

\begin{document}
\title{Concurrent Composition for Differentially Private Continual Mechanisms}

\author{Monika Henzinger}
\affiliation{%
  \institution{Institute of Science and Technology}
  \city{Klosterneuburg}
  \country{Austria}
}

\author{Roodabeh Safavi}
\affiliation{%
  \institution{Institute of Science and Technology}
  \city{Klosterneuburg}
  \country{Austria}
}

\author{Salil Vadhan}
\affiliation{%
  \institution{Harvard University}
  \city{Cambridge}
  \country{United States}
}

\begin{abstract}
    Many intended uses of differential privacy involve a \textit{continual mechanism} that is set up to run continuously over a long period of time, making more statistical releases as either queries come in or the dataset is updated.  In this paper, we give the first general treatment of privacy against {\em adaptive} adversaries for mechanisms that support dataset updates and a variety of queries, all arbitrarily interleaved. It also models a very general notion of neighboring, that includes both event-level and user-level privacy. We prove several \textit{concurrent} composition theorems for continual mechanisms, which ensure privacy even when an adversary can interleave its queries and dataset updates to the different composed mechanisms.  Previous concurrent composition theorems for differential privacy were only for the case when the dataset is static, with no adaptive updates. We also give the first interactive and continual generalizations of the ``parallel composition theorem'' for noninteractive differential privacy.  Specifically, we show that the analogue of the noninteractive parallel composition theorem holds if either there are no adaptive dataset updates or each of the composed mechanisms satisfies pure differential privacy, but it fails to hold for composing approximately differentially private mechanisms with dataset updates. Thus, we prove a tight new composition theorem for this case. In addition, we prove concurrent filter compositions theorems for the scenarios in which the privacy parameters are adaptively chosen. We extend these results to other measures of differential privacy, including R\'enyi DP and $f$-DP.
    
    We then formalize a set of general conditions on a continual mechanism $\mathcal{M}$ that runs multiple continual sub-mechanisms such that the privacy guarantees of $\mathcal{M}$ follow directly using the above concurrent composition theorems on the sub-mechanisms, without further privacy loss. This enables us to give a simpler and modular privacy analysis of a recent continual histogram mechanism of Henzinger, Sricharan, and Steiner. In the case of approximate DP, ours is the first proof that shows that its privacy holds against adaptive adversaries. We also provide a framework that simplifies the analysis of local differential privacy when the protocol includes multi-round server-user interactions. Using this result, we simplify the privacy analysis of the core decomposition protocol of Dhulipala, Henzinger, Li, Liu, Sricharan, and Zhu~\cite{dhulipala2025near}.
\end{abstract}

\maketitle

\section{Introduction}
Differential privacy~\cite{dwork2006calibrating} is the now-standard framework for protecting privacy when performing statistical analysis or machine learning on sensitive datasets about individuals.  In addition to having a rich and rapidly growing scientific literature, differential privacy has also seen large-scale adoption by technology companies and government agencies.

The familiar formulation of differential privacy is for a noninteractive mechanism $\cM : \cX\rightarrow \cY$, which takes a dataset $x\in \cX$ and produces a statistical release (e.g. a collection of statistics or the parameters of a machine learning model) $y\in \cY$.
Differential privacy requires that for every two datasets $x$ and $x'$ that ``differ on one individual's data,'' the output distributions $\cM(x)$ and $\cM(x')$ are ``close'' to each other.  Formally, we have a {\em neighboring relation} $\sim$ on $\cX$, which specifies which datasets ``differ on one individual's data,'' or more generally what information requires privacy protection.  We say that $\cM$ is {\em $(\eps,\delta)$-DP with respect to $\sim$} if for all $x\sim x'$, we have:
\begin{equation} \label{ineq:indisting}
\forall S\subseteq \cY\qquad \Pr[\cM(x)\in S]\leq e^\eps\cdot \Pr[\cM(x')\in S]+\delta.
\end{equation}
Typically we think of $\eps$ as a small constant, e.g. $\eps=.1$, while $\delta$ is cryptographically negligible.  The case that $\delta=0$ is referred to as {\em pure} differential privacy, and otherwise Inequality~\ref{ineq:indisting} is known as {\em approximate} differential privacy.
This simple definition is very useful and easy to work with, but does not capture many scenarios in which we want to apply differential privacy.

\medskip\noindent\textbf{Differential Privacy Over Time.}
Many intended uses of differential privacy involve a system that is set up to run continuously over a long period of time, making more statistical releases as either queries come in or the dataset is updated, for example with people leaving and entering.  (The latter case, with dataset updates, is known as differential privacy under {\em continual observation}, and was introduced in the concurrent works of Dwork, Naor, Pitassi, and Rothblum~\cite{Dwork2010} and Chan, Shi, and Song~\cite{DBLP:journals/tissec/ChanSS11} -- see~\cite{DBLP:conf/esa/HenzingerS25} for a survey.)
These settings introduce a new attack surface, where an adversary may be {\em adaptive}, issuing queries or influencing dataset updates based on the responses it receives.  

When $\delta>0$, there are mechanisms that are differentially private against oblivious adversaries but not against adaptive ones.\footnote{A simple artificial example is the following: the mechanism's first release can include a uniformly random number $r$ from $\{1,2,\ldots,\lceil 1/\delta\rceil\}$, and then if the adversary's next query is $r$, the mechanism publishes the entire dataset in the clear.}  One natural example is the ``advanced composition theorem'' for differential privacy~\cite{dwork2010boosting}, which is known to fail if the adversary can select the privacy-loss parameters $(\eps_i,\delta_i)$ of the composed mechanisms adaptively~\cite{rogers2016privacy}.  

Until recently, adaptivity has been treated in an ad hoc manner in the differential privacy literature.  For example, the aforementioned advanced composition theorem of \cite{dwork2010boosting} was formalized in a way that allows an adversary to adaptively choose the mechanisms and pairs of adjacent datasets, just not the privacy-loss parameters, and new composition theorems that allow for adaptive choices of privacy-loss parameters (known as privacy {\em filters} and privacy {\em odometers}) were given by \cite{rogers2016privacy,whitehouse2023fully}. 
Other early examples where adaptive queries were analyzed include the Sparse Vector Technique (SVT)~\cite{dwork2009complexity,HardtRo10,lyu2016understanding} and the Private Multiplicative Weights (PMW) algorithm~\cite{HardtRo10}.  

A few years ago, Vadhan and Wang~\cite{vadhan2021concurrent} initiated a systematic study of differentially private mechanisms $\cM$ that support adaptive queries.  Specifically, they formalized the notion of an {\em interactive mechanism} $\cM$, which begins with a dataset as input and maintains a potentially secret state as it answers queries that come in over time.  The use of a secret state, as utilized in both SVT and PMW, allows for correlating the answers to queries over time, and often allows for the privacy-loss budget to degrade much more slowly than composing noninteractive mechanisms would allow. 
They defined $\cM$ to be an {\em $(\eps,\delta)$-DP interactive mechanism} if for every pair of adjacent datasets $x\sim x'$ and every adaptive adversary $\cA$, the {\em view} of $\cA$ when interacting with $\cM(x)$ is $(\eps,\delta)$-indistinguishable (in the sense of Inequality~(\ref{ineq:indisting})) from its view when interacting with $\cM(x')$, where the {\em view} of $\cA$ consists of the random coin tosses of $\cA$ and the transcript of the interaction.

Vadhan and Wang~\cite{vadhan2021concurrent} asked whether our existing composition theorems for differential privacy, where the mechanisms being composed are noninteractive mechanisms, extend to the {\em concurrent composition} of interactive mechanisms. In concurrent composition of interactive mechanisms, an adaptive adversary can interleave its queries to the different interactive mechanisms being composed.  This is a realistic attack scenario when multiple interactive differentially private mechanisms are deployed in practice, and is also a useful primitive for analyzing differentially private algorithms that use several interactive mechanisms as building blocks.
Vadhan and Wang showed that the basic composition theorem for pure differential privacy~\cite{dwork2006calibrating} extends to concurrent composition. This result was later generalized by Lyu~\cite{lyu2022composition} to all composition theorems for approximate differential privacy. Subsequent works gave concurrent composition theorems for other flavors of differential privacy, such as R\'enyi DP~\cite{lyu2022composition} and $f$-DP~\cite{vadhan2022concurrent}.

Note that the extension of composition theorems for noninteractive mechanisms to concurrent composition theorems for interactive mechanisms is non-trivial: Let us call an interactive mechanism $M$ $M$-adaptive private if it is differentially private against an adversary that was only interactive with $M$. Consider an adversary that first interacts adaptively with $M_1$, then $M_2$ and then with $M_1$ again. Note that for the second sequence of interactions with $M_1$ the adversary takes into account its prior interaction with $M_1$ as well as with $M_2$. Thus, the adversary has more information about the data than if it had only interacted with $M_1$ before and the fact that $M_1$ is $M_1$-adaptive private does not mean that it is private against this interleaved composition of $M_1$ and $M_2$.

Haney, Shoemate, Tian, Vadhan, Vyrros, Wang, and Xu~\cite{haney2023concurrent} study the concurrent extension of composition theorems for noninteractive mechanisms where the adversary instantiates arbitrary mechanisms with adaptive privacy parameters. This setting is referred to as \emph{filter composition}, where a \emph{filter }is a function $\filt$ that maps the privacy parameters $(\eps_i, \delta_i)$ of the created mechanisms to an element of $\{\top, \bot\}$. Suppose an adversary adaptively creates noninteractive mechanisms such that $\filt$ applied to the privacy parameters of the created mechanisms equals $\top$ at every point in time. The $\filt$-filter composition of noninteractive mechanisms is said to satisfy $(\eps, \delta)$-DP if the view of every such adversary is $(\eps, \delta)$-indistinguishable when all created mechanisms are run on dataset $x$ or its neighboring dataset $x'$. Haney et al.~extend this definition to the concurrent $\filt$-filter composition of interactive mechanisms, where the adversary may not only create new mechanisms but also interact with existing interactive ones. They prove that if the $\filt$-filter composition of noninteractive mechanisms is $(\eps, \delta)$-DP, then the concurrent $\filt$-filter composition of interactive mechanisms is also $(\eps, \delta)$-DP.

All of the above works are for the case where the {\em queries} to the mechanisms are adaptive, but the datasets are static, given as input to the mechanisms at the start of the interaction.\footnote{One can extend the concurrent composition theorems to the case where, similarly to the advanced composition theorem~\cite{dwork2010boosting}, an adversary can adaptively select a pair of adjacent datasets $(x_i,x_i')$ for the $i$'th mechanism $\cM_i$, but once the interactive mechanism $\cM_i$ is started, its dataset is fixed and no further updates are allowed. Thus, whenever $(x_{i+1},x_{i+1}')$ are given, a new interactive mechanism $\cM_{i+1}$ has to be started that does not know  the secret state of $\cM_i$.}
However, in some real-world applications, both queries and dataset updates may be issued adaptively, i.e., the dynamic updates to the dataset might be chosen adaptively. 
For instance, Ghazi et al.~\cite{ghazi2025privacy} study an advertising problem in which the server receives both private data and queries over time. They design algorithms that accept dataset updates while applying noninteractive mechanisms—such as the discrete Laplace mechanism—to answer queries and analyze the privacy guarantees for their algorithm using existing results on privacy filters for composing noninteractive mechanisms with adaptive privacy-loss parameters. This independent work\footnote{The first version of this work was available on arXiv prior to \cite{ghazi2025privacy}.} shows that if a mechanism receiving both queries and dataset updates produces a differentially private output at each time step $i$ (with adaptive parameters $\eps_i,\delta_i$ at step $i$) such that $\sum_i \eps_i \leq \eps$ and $\sum_i \delta_i \leq \delta$, then the overall mechanism is $(\eps,\delta)$-DP.  
This guarantee can be seen as a special case of our filter composition theorem, where the composition rule is simple summation and each composed mechanism is restricted to a noninteractive mechanism corresponding to a single step. 
While some continual mechanisms do indeed employ independent noninteractive mechanisms at each step, many important examples—such as the Sparse Vector Technique (SVT)—generate outputs that are correlated across time. For these mechanisms, privacy must be analyzed over the entire output sequence rather than on a step-by-step basis to get the desired guarantees. (If the $\eps$-DP SVT mechanism were analyzed on a step-by-step basis, then the privacy guarantee would degrade to $\Theta(q\cdot \eps)$-DP, where $q$ is the number of queries issued, defeating the entire point of SVT, which is to have a privacy guarantee and noise scale that is independent of the number of queries.)  
As a consequence, in the concurrent composition of continual mechanisms (or even interactive mechanisms), one cannot simply reduce the analysis to the composition of per-step noninteractive mechanisms. Indeed, this is why concurrent composition theorems for interactive mechanisms took several years to establish \cite{vadhan2021concurrent,lyu2022composition,vadhan2022concurrent,haney2023concurrent}. Thus, the results of Ghazi et al.~\cite{ghazi2025privacy} do not apply to the more general setting of (filter) concurrent composition of continual mechanisms.\footnote{It is Ghazi et al.'s definition of Individual DP (IDP, their Definition 4.1) and their composition theorem that is framed in terms of IDP that restricts to step-by-step privacy analyses corresponding to composition of noninteractive mechanisms. Their definition of $(\eps,\delta)$-DP interactive mechanisms (their Definition 3.6) is somewhat more general than their definition of IDP, but it still does not allow for exploiting correlated randomness across queries to reduce privacy loss, as in SVT. (This is because their Definition 3.4 requires the next state of the mechanism to be a function of only the dataset and the query, and not of the randomness used to generate the answer.)}

Jain, Raskhodnikova, Sivakumar, and Smith~\cite{pricedpco} recently gave a formalization of adaptive adversaries under continual observation (where there are dataset updates but the queries are fixed)  
for the special case of {\em event-level} privacy, where two streams of updates are considered adjacent if they differ on only one update. Denisov, McMahan, Rush, Smith and Thakurta~\cite{denisov2022improved} extended this definition to arbitrary neighbor relations on update sequences. The lack of a general toolkit for reasoning about adaptivity in dataset updates has led to complicated and ad hoc proofs of privacy of continual mechanisms (see e.g.~\cite{henzingerSS2023}).  Rectifying this situation is what motivated our work.

\medskip\noindent\textbf{Comparison with algorithms and data structures.}  We note that the different flavors of  mechanisms discussed above are differential privacy analogues of familiar variants of randomized algorithms and data structures.  A noninteractive mechanism is simply a {\em static} (or {\em batch}) {\em algorithm}.  An interactive mechanism corresponds to a {\em static data structure}.  Our focus in this paper is on continual mechanisms, which allow both queries and updates, and these can be viewed as {\em dynamic data structures}.  Adaptivity is a common concern for the correctness properties of randomized data structures; here we are concerned instead with its effect on privacy properties.  

\medskip\noindent\textbf{Our contributions}
\textbf{(A)} We give a formalism that captures privacy against adaptive adversaries for {\em continual mechanisms}, which support both queries and dataset updates (Definition~\ref{def:dp-cm}). This generalizes the previously studied cases of interactive mechanisms (which have static datasets) and continual observation (where mechanisms answer fixed queries). See Figure~\ref{fig:im-co-cm} in Appendix~\ref{app:figures} for an illustration.
We note that specific mechanisms, like variants of SVT, have been proposed to handle both dataset updates and queries~\cite{JainSmWa24}, but without offering a definition of privacy against adversaries that can adaptively select both the updates and queries.

\noindent\textbf{(B)} We prove concurrent composition theorems for continual mechanisms, where an adaptive adversary may interleave queries and updates across mechanisms based on the entire history of observed outputs. As discussed above for the concurrent composition of interactive mechanisms, the concurrent setting involves additional adaptivity that is not supported by composition theorems for noninteractive mechanisms.\footnote{Existing results on the concurrent composition of {\em interactive} mechanisms do not apply to the concurrent composition of {\em continual} mechanisms, since in our model the adversary adaptively selects both (a) the mechanisms receiving neighboring inputs {(including their privacy parameters and neighboring definition)} and (b) the neighboring update sequences they receive over time, which is not the case for interactive mechanisms.} We show tight privacy bounds for the following useful forms of concurrent composition of continual mechanisms. Further explanations are provided in Appendix~\ref{app:variants-concomp}.

\begin{enumerate}
    \item We show that when the privacy-loss parameters are fixed, the standard composition theorems for noninteractive approximate differential privacy extend to the \emph{concurrent composition of continual mechanisms} (Theorem~\ref{thm:comp-fixed-param-cm}). This result also holds for R\'enyi differential privacy (Corollary~\ref{cor:comp-fixed-param-rdp}) and for $f$-DP\footnote{The $f$-DP results in this paper assume \emph{finite communication}: for every created continual mechanism, there exists a constant $C$ bounding the sizes of its query and answer spaces and the number of queries it can answer before halting. This assumption was also made for interactive mechanisms in~\cite{vadhan2022concurrent}.} (Corollary~\ref{cor:comp-fixed-param-f-dp}).
\end{enumerate}
\begin{enumerate}\setcounter{enumi}{1}
    \item We extend the classical ``parallel composition theorem'' for noninteractive mechanisms to the concurrent parallel composition of an \emph{unbounded}\footnote{Throughout this paper, we use the term “unbounded number” to refer to a quantity of unknown (potentially arbitrary) size.} number of adaptively chosen \emph{interactive mechanisms} with an adaptively chosen pair of neighboring datasets (Corollary~\ref{cor:parallel-fixed-param-comp-im}). We show similar extensions for R\'enyi differential privacy (Corollary~\ref{cor:parallel-fixed-param-comp-im-rdp}) and for $f$-DP (Corollary~\ref{cor:parallel-comp-f-dp-im}).
    \item We extend the above concurrent parallel composition of interactive mechanisms to continual mechanisms when the composed mechanisms satisfy pure differential privacy (Corollary~\ref{cor:parallel-comp-pure}). For approximate DP, we construct a counterexample demonstrating that \emph{the concurrent parallel composition of an unbounded number of continual mechanisms} with adaptively chosen dataset updates is not private (Theorem~\ref{thm:counter-example}). Since $(\eps, \delta)$-DP is a special case of $f$-DP, this lower bound also applies to $f$-DP. We further construct an analogous counterexample for the case where the composed mechanisms satisfy R\'enyi DP instead of approximate DP (Theorem~\ref{thm:counter-example-rdp}).
    \item By imposing an upper bound on $\prod_i (1 - \delta_i)$, where $\delta_i$ is the approximate privacy parameter of the $i$-th created mechanism, we derive a new \emph{tight} composition theorem for concurrent parallel composition of continual mechanisms (Theorem~\ref{thm:parallel-comp-approx}).
    \item As an immediate consequence, our new theorem extends the classical parallel composition theorem for noninteractive mechanisms to the concurrent composition of an \emph{unbounded number of adaptively chosen continual mechanisms} satisfying \emph{pure} differential privacy (Corollary~\ref{cor:parallel-comp-pure}).
    \item We show that, for every filter $\filt$, the concurrent $\filt$-filter composition of \emph{continual mechanisms} enjoys the same privacy guarantees as the $\filt$-filter composition of noninteractive mechanisms (Theorem~\ref{thm:filter-comp-cm}). This result is also proved for R\'enyi DP (Corollary~\ref{cor:filter-comp-rdp}) and $f$-DP (Corollary~\ref{cor:filter-comp-f-dp}).
\end{enumerate}
A summary of the positive results for the aforementioned types of concurrent composition of $(\eps,\delta)$-DP interactive and continual mechanisms is provided in Table~\ref{tab:summary}.

\begin{table}[ht]
    \centering
    \small
    \setlength{\tabcolsep}{6pt}
    \renewcommand{\arraystretch}{1.4}
    \begin{tabular}{|c|c|c|c|}
    \hline
     & \shortstack{\rule{0pt}{2ex}Concurrent Composition\\with Fixed Parameters}
     & \shortstack{Concurrent\\Parallel Composition}
     & \shortstack{Concurrent\\Filter Composition} \\
    \hline
    CMs
    &
    \begin{tabular}[c]{@{}c@{}}
    Theorem~\ref{thm:comp-fixed-param-cm} for $(\eps,\delta)$-DP
    \end{tabular}
    &
    \begin{tabular}[c]{@{}c@{}}
    Corollary~\ref{cor:parallel-comp-pure} for $\eps$-DP\\
    Theorem~\ref{thm:parallel-comp-approx} for $(\eps,\delta)$-DP
    \end{tabular}
    &
    \begin{tabular}[c]{@{}c@{}}
    Theorem~\ref{thm:filter-comp-cm} for $(\eps,\delta)$-DP
    \end{tabular}
    \\
    \hline
    IMs
    &
    \begin{tabular}[c]{@{}c@{}}
    \cite{vadhan2021concurrent} for $\eps$-DP\\
    \cite{vadhan2022concurrent,lyu2022composition} for $(\eps,\delta)$-DP
    \end{tabular}
    &
    \begin{tabular}[c]{@{}c@{}}
    Corollary~\ref{cor:parallel-fixed-param-comp-im} for $(\eps,\delta)$-DP
    \end{tabular}
    &
    \begin{tabular}[c]{@{}c@{}}
    \cite{haney2023concurrent} for $(\eps,\delta)$-DP
    \end{tabular}
    \\
    \hline
    \end{tabular}
    \caption{Concurrent composition results for $(\eps,\delta)$-DP interactive mechanisms (IMs) and continual mechanisms (CMs).}\label{tab:summary}
    \Description{Concurrent composition results for $(\eps,\delta)$-DP interactive mechanisms (IMs) and continual mechanisms (CMs).}
\end{table}
\noindent\textbf{(C)} We formalize a set of general conditions on a continual mechanism $\cM$ that runs multiple continual sub-mechanisms such that the privacy guarantees of $\cM$ follow directly by using the above concurrent composition theorems on the sub-mechanisms, \emph{without further privacy loss} (Theorem~\ref{thm:dp-complex-mech}). These conditions are also reformulated as a set of properties for pseudocode describing continual mechanisms (Corollary~\ref{cor:pseudocode}).

\noindent\textbf{(D)} We use our concurrent composition theorems for continual mechanisms, with the aforementioned general conditions, to give a simpler and more modular privacy analysis of a recent continual histogram mechanism~\cite{henzingerSS2023}.  In the case of approximate DP, ours is the first proof that the privacy of the mechanism holds against adaptive adversaries.

\noindent\textbf{(E)} We model \emph{local protocols} in the \emph{local differential privacy model} as continual mechanisms with a modular structure, and demonstrate how our concurrent composition theorems for continual mechanisms—particularly the concurrent parallel composition theorem for interactive mechanisms—can simplify their privacy analysis (Theorem~\ref{thm:ldp-to-concomp} and Corollary~\ref{cor:ldp-parallel-concomp}). As an example, we give a simple privacy analysis for the core decomposition local protocol proposed by Dhulipala, Henzinger, Li, Liu, Sricharan, and Zhu~\cite{dhulipala2025near}.

We elaborate on all of these contributions below.

\medskip 
\paragraph{\bf A general formalism for privacy against adaptive adversaries.}
Our basic object of study is a {\em continual mechanism} $\cM$, which is an interactive mechanism that is not given any input at the start, but simply receives and sends messages in a potentially randomized and stateful manner.  These messages can represent, for instance, queries or dataset updates.  To model static datasets (as in standard interactive differential privacy), the entire dataset can be given to $\cM$ as its first message.

Towards defining privacy, we consider adaptive adversaries $\cA$ that  send messages $m_i$ for $i=1, 2, 3, \dots$ that are always parsed as pairs $m_i=(m_i[0],m_i[1])$.  Privacy is defined with respect to a {\em verification function} $f$ that takes a sequence of such pairs $(m_1,m_2,\ldots,m_t)$ and outputs either $\top$ or $\bot$.  When $f(m_1,m_2,\ldots,m_t)=\top$, we consider the two sequences $m[0]=(m_1[0],m_2[0],\dots,m_t[0])$ and $m[1]=(m_1[1],m_2[1],\dots,m_t[1])$ to be {\em neighboring}.

Now, by the choice of the verification function $f$, we can easily capture previously studied notions of privacy and more.   For example, for a standard interactive mechanism, $f(m_1,m_2,\ldots,m_t)$ will check that $m_1[0]\sim m_1[1]$ are a pair of neighboring datasets, and that $m_i[0]=m_i[1]$ for all $i>1$ (these represent the adaptive queries).  For continual observation under event-level privacy, we interpret all of the $m_i[b]$'s as dataset updates, and $f$ will require that $m[0]$ and $m[1]$ differ in at most one coordinate $i\in [k]$. 
For continual observation under ``user-level'' privacy, $f$ will require that the entries in which $m[0]$ and $m[1]$ differ correspond to only one user.
For continual mechanisms,
we can interleave adaptive queries among the dataset updates by having $f$ require that all $m_i[0]=m_i[1]$ whenever either one is a query (rather than a dataset update). 

To go from a verification function $f$ to a definition of privacy against adaptive adversaries, we
utilize the notion of an {\em interactive post-processing mechanism} (IPM)~\cite{haney2023concurrent}, which is a stateful procedure that can stand between an interactive mechanism and an adversary, or two interactive or continual mechanisms, or even two IPMs, and interact with both of them adaptively, sending and receiving messages to the mechanisms on its left and right.  Specifically, as illustrated in Figure~\ref{fig:v-i-m}, we place two IPMs between the adversary $\cA$ and the continual mechanism $\cM$.  The first is a {\em verifier} $\V[f]$ that, every time it receives a new message pair $m_k$ from $\cA$, applies $f$ to the entire history of message pairs $(m_1,\ldots,m_k)$ sent by $\cA$,  and forwards $m_k$ onward only if $f$ returns $\top$.  Otherwise $\V[f]$ halts the interaction.  The second IPM is the {\em identifier} $\I$, which starts out with a secret bit $b\in \zo$, and every time it receives a message pair $(m_i[0],m_i[1])$ from $\V[f]$, it forwards $m_i[b]$ on to the mechanism $\cM$. The identifier $\I$ with bit $b$ is shown by $\I(b)$.  When $\cM$ returns an answer $a_i$, $\I$ and $\V[f]$ just forward it back to $\cA$.  We denote the interactive mechanism representing the combination of $\cM$, $I$, and $\V[f]$ by $\V[f]\circstar \I\circstar \cM$.

\begin{figure}[ht]
    \centering
    \scalebox{0.84}{
    \begin{tikzpicture}[
    node distance=1.5cm and 2.5cm,
    rect/.style={rectangle, draw, minimum width=2cm, minimum height=3.5cm, align=center, font=\large},
    arrow/.style={-{Stealth}, thick}
]

\node[rect] (adv1) {\large $\mathcal{A}$};
\node[rect, right=of adv1] (vf) {\large$\mathcal{V}[f]$};
\node[rect, right=of vf] (ib) {\large$\mathcal{I}(b)$};
\node[rect, right=of ib] (m) {\large$\mathcal{M}$};

\draw[arrow] ([yshift=1.25cm]adv1.east) -- node[above,yshift=-3.5pt] { $m_1\!=\!(m_1^0, m_1^1)$} ([yshift=1.25cm]vf.west);
\draw[arrow]  ([yshift=0.75cm]vf.west) -- node[above,yshift=-3.5pt] { $m_1'$} ([yshift=0.75cm]adv1.east);
\draw[arrow] ([yshift=0.25cm]adv1.east) -- node[above,yshift=-3.5pt] { $m_2\!=\!(m_2^0, m_2^1)$} ([yshift=0.25cm]vf.west);
\draw[arrow] ([yshift=-0.25cm]vf.west) -- node[above,yshift=-3.5pt] {$m_2'$} ([yshift=-0.25cm]adv1.east);
\draw[arrow] ([yshift=-0.75cm]adv1.east) -- node[above,yshift=-3.5pt] {$m_3$ (invalid)} ([yshift=-0.75cm]vf.west);
\draw[arrow] ([yshift=-1.25cm]vf.west) -- node[above,yshift=-3.5pt] {$\halt$} ([yshift=-1.25cm]adv1.east);

\draw[arrow] ([yshift=1.25cm]vf.east) -- node[above,yshift=-3.5pt] { $m_1\!=\!(m_1^0, m_1^1)$} ([yshift=1.25cm]ib.west);
\draw[arrow]  ([yshift=0.75cm]ib.west) -- node[above,yshift=-3.5pt] { $m_1'$} ([yshift=0.75cm]vf.east);
\draw[arrow] ([yshift=0.25cm]vf.east) -- node[above,yshift=-3.5pt] { $m_2\!=\!(m_2^0, m_2^1)$} ([yshift=0.25cm]ib.west);
\draw[arrow] ([yshift=-0.25cm]ib.west) -- node[above,yshift=-3.5pt] {$m_2'$} ([yshift=-0.25cm]vf.east);

\draw[arrow] ([yshift=1.25cm]ib.east) -- node[above,yshift=-3.5pt] { $m_1^b$} ([yshift=1.25cm]m.west);
\draw[arrow]  ([yshift=0.75cm]m.west) -- node[above,yshift=-3.5pt] { $m_1'$} ([yshift=0.75cm]ib.east);
\draw[arrow] ([yshift=0.25cm]ib.east) -- node[above,yshift=-3.5pt] { $m_2^b$} ([yshift=0.25cm]m.west);
\draw[arrow] ([yshift=-0.25cm]m.west) -- node[above,yshift=-3.5pt] {$m_2'$} ([yshift=-0.25cm]ib.east);

\end{tikzpicture}
    }
    \caption{Illustration of the interactions between $\cA$, $\V[f]$, $\I(b)$, and $\cM$}
    \label{fig:v-i-m}
    \Description{Illustration of the interactions between $\cA$, $\V[f]$, $\I(b)$, and $\cM$}
\end{figure}
This leads to the following definition of differential privacy for continual mechanisms:
\begin{definition}[DP for continual mechanisms]
Let $f$ be a verification function. A continual mechanism $\cM$ is
{\em $(\eps,\delta)$-DP with respect to $f$} if for every adaptive adversary $\cA$, the view of $\cA$ when interacting with $\V[f]\circstar \I(0)\circstar \cM$ is $(\eps,\delta)$-indistinguishable from its view when interacting with $\V[f]\circstar \I(1)\circstar \cM$.
\end{definition}

The idea of defining security against adaptive adversaries by having the adversary send pairs of messages and trying to distinguish whether the first message in each pair or the second message in each pair is being used goes back to the notion of {\em left-or-right indistinguishability} in cryptography~\cite{BellareDJR97}.  What is different here is its combination with a  verification function $f$ to enforce the neighboring condition, which is necessary to capture differential privacy. 

As far as we know, this definition captures all of the previous definitions of privacy against adaptive adversaries as special cases (by appropriate choices of $f$), as well as modeling adaptive adversaries in even more scenarios (such as continual observation with user-level privacy, and mixtures of adaptively-chosen dataset updates and queries).

Furthermore, it allows us to reduce the analysis of general continual mechanisms to that of the previously studied notion of interactive mechanisms.  Indeed, by inspection, we observe that a \emph{continual} mechanism $\cM$ is {$(\eps,\delta)$-DP with respect to $f$} if and only if the \emph{interactive} mechanism $\cM'(\cdot) = \V[f]\circstar \I(\cdot)\circstar \cM$ is $(\eps,\delta)$-DP on the dataset space $\zo$ (where $0\sim 1$).

\medskip

\paragraph{\bf Concurrent Composition Theorems for Continual Mechanisms.}  With the above definitions in hand,  we can formulate and prove new concurrent composition theorems for continual mechanisms. 
Recall that with concurrent composition of interactive mechanisms, an adaptive adversary (or honest analyst) can interleave its queries to the different interactive mechanisms being composed; for concurrent composition of {\em continual} mechanisms, we also allow the dataset {\em updates} to be interleaved with the queries and each other. See Figure~\ref{fig:concurrent-k-ms} for illustration.
\begin{figure}
    \centering
    \begin{tikzpicture}[
    node distance=1.5cm and 2cm,
    rect/.style={rectangle, draw, minimum width=1.2cm, minimum height=2.5cm, align=center},
    subrect/.style={rectangle, draw, minimum width=1cm, minimum height=0.8cm, align=center},
    arrow/.style={-{Stealth}, thick}
]

\node[rect] (adv) {\large $\mathcal{A}$};
\node[draw, rectangle, right=2.6cm of adv, minimum width=2.8cm, minimum height=4.1cm] (bigm) {};

\node[subrect] (m1) at ([xshift=0.8cm,yshift=1.5cm]bigm.center) {$\mathcal{M}_1$};
\node (dots1) at ([xshift=0.8cm,yshift=1.0cm]bigm.center) {$\vdots$};
\node[subrect] (mj) at ([xshift=0.8cm,yshift=0.3cm]bigm.center) {$\mathcal{M}_j$};
\node (dots2) at ([xshift=0.8cm,yshift=-0.5cm]bigm.center) {$\vdots$};
\node[subrect] (mk) at ([xshift=0.8cm,yshift=-1.5cm]bigm.center) {$\mathcal{M}_k$};

\draw[arrow] ([yshift=1.0cm]adv.east) -- node[above,yshift=-3.5pt] { $m_1\!=\!(m_1^*, j)$} ([yshift=1.0cm]bigm.west);
\draw[arrow]  ([yshift=0.5cm]bigm.west) -- node[above,yshift=-3.5pt] { $m_1'$} ([yshift=0.5cm]adv.east);
\draw[arrow] ([yshift=0cm]adv.east) -- node[above,yshift=-3.5pt] { $m_2\!=\!(m_2^*, k)$} ([yshift=0cm]bigm.west);
\draw[arrow] ([yshift=-0.5cm]bigm.west) -- node[above,yshift=-3.5pt] {$m_2'$} ([yshift=-0.5cm]adv.east);
\node (dots1) at ([xshift=1.3cm,yshift=-0.8cm]adv.east) {$\vdots$};

\draw[arrow] ([yshift=1.0cm]bigm.west) -- ([yshift=0.25cm]mj.west) node[midway,sloped,above,yshift=-3.5pt] {$m_1^*$};
\draw[arrow] ([yshift=-0.25cm]mj.west) -- ([yshift=0.5cm]bigm.west) node[midway,sloped,above,yshift=-3.5pt] {$m_1'$};

\draw[arrow] ([yshift=0cm]bigm.west) -- ([yshift=0.25cm]mk.west) node[midway,sloped,above,yshift=-3.5pt] {$m_2^*$};
\draw[arrow] ([yshift=-0.25cm]mk.west) -- ([yshift=-0.5cm]bigm.west) node[midway,sloped,above,yshift=-3.5pt] {$m_2'$};

\end{tikzpicture}
    \caption{Illustration of the interactions between an honest adversary $\cA$ and the concurrent composition of $k$ continual mechanisms}
    \label{fig:concurrent-k-ms}
    \Description{Illustration of the interactions between an honest adversary $\cA$ and the concurrent composition of $k$ continual mechanisms}
\end{figure}

Our first concurrent composition theorem is as follows:
\begin{theorem}[concurrent composition of continual mechanisms, informally stated] \label{thm:extconcomp-intro}
    For every fixed sequence $(\eps_1,\delta_1),\ldots,(\eps_k,\delta_k)$ of privacy-loss parameters and verification functions $f_1,\ldots,f_k$, the concurrent composition of continual mechanisms $\cM_i$ that are $(\eps_i,\delta_i)$-DP w.r.t. $f_i$ is $(\eps,\delta)$-DP for the same $(\eps,\delta)$ that holds for composing noninteractive mechanisms that are $(\eps_i,\delta_i)$-DP.  
\end{theorem}

Notably, this result allows the adversary to adaptively and concurrently choose the neighboring sequences fed to each of the $\cM_i$'s based on the responses received from the other mechanisms. Analogous results are shown for R\'enyi DP in Corollary~\ref{cor:comp-fixed-param-rdp} and, under the finite-communication assumption, for $f$-DP in Corollary~\ref{cor:comp-fixed-param-f-dp}.

A key step in the proof of Theorem~\ref{thm:extconcomp-intro} is to show that for any $(\eps_i, \delta_i)$-DP continual mechanism $\cM$, the interactive mechanism $\cM'(\cdot) = \V[f]\circstar \I(\cdot)\circstar \cM$ can be simulated by post-processing the randomized response mechanism $\RR_{\eps_i, \delta_i}(\cdot)$. Lyu~\cite{lyu2022composition} proves such a result under the assumptions that the answer space of $\cM$ is finite and that there exists a predetermined upper bound $T$ on the number of interactions. Vadhan and Zhang~\cite{vadhan2022concurrent} additionally assume that the query set is finite. We remove the latter assumption and relax the former to having an infinite (but discrete) answer space.

Next, we explore the concurrent and continual analogue of {\em parallel composition} of differential privacy.  To review parallel composition, suppose that $\cM_1,\ldots,\cM_\ell$ are each $(\eps_0,\delta_0)$-DP noninteractive mechanisms with respect to a neighbor relation $\sim^*$, where each $\cM_i$  expects a dataset $x_i$ as input, and for a parameter $k\in\{1,\ldots,\ell\}$, we define a neighboring relation $\sim^{(k)}$ on $\ell$-tuples as follows:
$(x_1,\ldots,x_\ell)\sim^{(k)} (x_1',\ldots,x_\ell')$ if and only if $x_i=x_i'$ for all but at most $k$ indices $i\in[\ell]$, for which $x_i \sim^* x_i'$.
We refer to the mechanism 
$\cM'(\cdot)=(\cM_1(\cdot),\ldots,\cM_\ell(\cdot))$ with neighbor relation $\sim^{(k)}$ as the {\em $k$-sparse parallel composition} of $\cM_1,\ldots,\cM_\ell$. 
If $\cM_1,\ldots,\cM_\ell$ are noninteractive mechanisms, then it is known that $\cM'$ satisfies the same level of $(\eps,\delta)$-DP as the (standard) composition of $k$ (rather than $\ell$) $(\eps_0,\delta_0)$-DP mechanisms.  

We extend this notion to the \emph{concurrent $k$-sparse parallel composition of continual mechanisms}, where an adversary (or an honest analyst) may adaptively create continual mechanisms (with no bound $\ell$ on the number of mechanism), where each mechanism satisfies $(\eps_0, \delta_0)$-DP with respect to a verification function $f^*$, and interact with them concurrently. Formally, instead of fixing one dataset pair $(x_i, x_i')$ for each mechanism in advance, the adversary issues pairs of messages $m = (m^0, m^1)$ to mechanisms over time. All but at most $k$ mechanisms receive identical message pairs throughout the interaction, while for each of the $k$ exceptional mechanisms, the sequence of message pairs $\sigma$ satisfies $f^*(\sigma) = \top$. If messages are data records to be inserted in the dataset, then this concurrent composition captures adaptive partitioning of data among mechanisms over time. Note that an adversary can decide to designate a mechanism $\cM_i$ as one of the $k$ “exceptional” mechanisms, i.e., issue at least one pair of non-identical messages to $\cM_i$, after potentially interacting with other mechanisms and sending multiple identical message pairs to $\cM_i$ itself.

First, we extend the parallel composition of noninteractive mechanisms to concurrent composition of interactive mechanisms (where there are adaptive queries but datasets $x_i$ are static, rather than dynamically updated).\footnote{As an intermediate step, Qiu et al.\cite{qiu2022differential} analyzes the composition of infinitely many SVT instances (sharing a secret dataset) that are executed sequentially, where each instance begins only after the previous one halts, and where a differing element between two neighboring datasets can affect the output of at most one instance. Their result is weaker than ours as it does not allow for concurrent interaction of the adversary with the interactive mechanisms.}
We recall that an interactive mechanism is modeled as a continual mechanism receiving a dataset as its first message and queries as subsequent messages. (The verification function $f^*$ for this mechanism ensures that the first message is a pair of neighbor datasets, while the subsequent pairs are identical queries.) A formal statement of the following theorem is given in Corollary~\ref{cor:parallel-fixed-param-comp-im}, which more generally composes interactive mechanisms $\cM_i$ with differing $(\eps_i,\delta_i)$ parameters. 

\begin{theorem}[concurrent parallel composition of interactive mechanisms]\label{thm:concurrent-parall-comp-IM-intro}
    For every $\eps_0>0$ and $0\leq \delta_0\leq 1$ and $k\in \N$, and every $(\eps_0,\delta_0)$-DP interactive mechanism $\cM_0$, the concurrent $k$-sparse parallel composition of an unbounded number of copies of $\cM_0$ satisfies the same $(\eps,\delta)$-DP guarantees as the composition of $k$ noninteractive $(\eps_0,\delta_0)$-DP mechanisms.
\end{theorem}

We prove analogous results for R\'enyi DP in Corollary~\ref{cor:parallel-fixed-param-comp-im-rdp} and, under the finite-communication assumption, for $f$-DP in Corollary~\ref{cor:parallel-comp-f-dp-im}.
Then we turn to the case of {\em continual mechanisms}, where there can also be adaptive dataset updates.  We prove a negative result for approximate DP ($\delta_0>0$):

\begin{theorem}[counterexample for the concurrent parallel composition of approximate DP continual mechanisms, informally stated] \label{thm:counterexp-parallelcomp-approxDP-intro}
    For every $\delta_0>0$, there is an $(0,\delta_0)$-DP continual mechanism $\cM_0$ such that for every $\ell\in \N$, the concurrent 1-sparse parallel composition of $\ell$ copies of $\cM_0$ is not $(0,\delta)$-DP for any $\delta<1-(1-\delta_0)^\ell$.
\end{theorem}

This theorem says that, at least in terms of the $\delta$ parameter, concurrent parallel composition of continual mechanisms does not provide any benefit over non-parallel composition, since the bound of $1-(1-\delta_0)^\ell$ is the bound that we would get by just applying composition over all $\ell$ mechanisms. The high-level idea of this counterexample is to define a continual mechanism $\cM_0$ that, upon receiving its first input message (regardless of the value), outputs $\bot$ with probability $\delta$ and $\top$ otherwise. If the first output is $\bot$, then on its second input the mechanism simply returns that input, revealing it entirely. Consider an adversary $\cA$ that creates $\ell$ instances of $\cM_0$, sends the pair $(0,0)$ to each instance, and then sends the pair $(0,1)$ to the first mechanism whose response was $\bot$. With probability $1-(1-\delta)^\ell$, such a mechanism exists and its second response reveals its input bit, which equals the secret bit used by the identifier to filter the adversary’s message pairs. An analogous counterexample for R\'enyi DP is provided in Theorem~\ref{thm:counter-example-rdp}.

Inspired by this counterexample, we prove a weaker version of Theorem~\ref{thm:concurrent-parall-comp-IM-intro} for the concurrent parallel composition of continual mechanisms, stated below. The formal version of this result is given in Theorem~\ref{thm:parallel-comp-approx} and applies to $(\eps_i, \delta_i)$-DP continual mechanisms with possibly different $(\eps_i,\delta_i)$ parameters.

\begin{theorem}[concurrent parallel composition of approx DP continual mechanisms, informally stated] \label{thm:parallelcomp-approxDP-intro}
    For every $\eps_0> 0$, $0\leq \delta'\leq 1$, and $k\in \N$, the concurrent $k$-sparse parallel composition of $(\eps_0,\delta_i)$-DP continual mechanisms, where the parameters $\delta_i$ are chosen adaptively such that $\prod_i (1-\delta_i) \geq 1-\delta'$, satisfies $(\eps,1-(1-\delta)(1-\delta'))$-DP if the composition of $k$ noninteractive $(\eps_0,0)$-DP mechanisms is $(\eps,\delta)$-DP.
\end{theorem}

The proof of Theorem~\ref{thm:parallelcomp-approxDP-intro} relies on a key property of our new post-processing construction (recall that $\RR_{\eps,\delta}(b)$ denotes the randomized response mechanism with parameters $\eps$ and $\delta$ and input $b\in \{0,1\}$):

\begin{proposition}[post-processing of randomized response mechanisms, informal statement]\label{prop:post-rr-intro}
    For every $(\eps, \delta)$-DP interactive mechanism $\cM$ and every pair of neighboring datasets $x_0, x_1$, there exists an interactive post-processing mechanism $\cP$ that receives the outputs of $\RR_{\eps, 0}(b)$ and $\RR_{0, \delta}(b)$ and simulates $\cM(x_b)$ identically. As long as the answer distributions of $\cM(x_0)$ and $\cM(x_1)$ to the queries asked so far are identical, $\cP$ does not access the outcome of $\RR_{\eps, 0}(b)$.
\end{proposition}

The first sentence of this proposition was proven by Lyu~\cite{lyu2022composition} and Vadhan and Zhang~\cite{vadhan2022concurrent} with finiteness condition on the communication. The novelty in our result is that $\cP$ does not access the outcome of $\RR_{\eps, 0}(b)$ until $\cM(x_0)$ and $\cM(x_1)$ differ.
Our construction builds on Lyu's proof,  
where an IPM is initially given the outcome of $\RR_{\eps,\delta}(b)$. Let $\cM'$ be a continual mechanism that is $(\eps, \delta)$-DP with respect to a verification function $f$. Consider the interaction between $\V[f]\circstar \I \circstar \cM'$ and an adversary. As long as the adversary sends pairs of identical messages, the answer distributions of  $\V[f]\circstar \I(0)\circstar \cM'$ and $\V[f]\circstar \I(1)\circstar \cM'$ will be identical. Thus, by replacing $\cM(x_b)$ in Proposition~\ref{prop:post-rr-intro} with $\V[f]\circstar \I(b)\circstar \cM'$, we conclude that there exists an IPM $\cP$ that simulates $\V[f]\circstar \I(b)\circstar \cM'$ using only the output of $\RR_{0,\delta}(b)$ until the adversary issues a pair of non-identical messages. At that point, $\cP$ requires the output of $\RR_{\eps,0}(b)$ to continue its simulation.  This result is the core idea to show that in the concurrent $k$-sparse parallel composition of \emph{continual} mechanisms, the $\eps$ parameters of the (potentially unboundedly many) composed mechanisms that are always assigned pairs of identical messages do not contribute to the overall privacy loss.

In the concurrent parallel composition of \emph{interactive} mechanisms, the adversary sends to each mechanism a pair of identical or neighboring datasets as first message and identical query pairs as subsequent messages. The key idea in the proof of Theorem~\ref{thm:concurrent-parall-comp-IM-intro} is that we can determine whether a mechanism is one of the $k$ exceptional ones solely from its first pair, before the mechanism produces any output. Recall that all messages of the adversary are filtered by a secret bit $b\in\{0,1\}$. If a mechanism is non-exceptional, whether the bit $b$ is zero or one makes no difference, and the first element of every pair can always be forwarded. In this proof, we simulate a mechanism by post-processing of a randomized-response mechanism with input $b$ only if it is exceptional, and show that we incur no privacy cost—neither in $\eps$ nor in $\delta$—from the non-exceptional mechanisms. 

Setting $\delta'=0$ in Theorem~\ref{thm:parallelcomp-approxDP-intro} yields the following corollary, stated formally (and for differing $\eps_i$) in Corollary~\ref{cor:parallelcomp-pureDP-intro}.

\begin{corollary}[concurrent parallel composition of pure DP continual mechanisms, informally stated] \label{cor:parallelcomp-pureDP-intro}
    For every $\eps_0> 0$ and integer $k\in \N$, the concurrent $k$-sparse parallel composition of $(\eps_0,0)$-DP continual mechanisms satisfies the same $(\eps,\delta)$-DP guarantees as the composition of $k$ noninteractive $(\eps_0,0)$-DP mechanisms.\footnote{In some composition theorems (e.g., advanced composition), composing pure DP mechanisms yields $(\eps,\delta)$-DP with $\delta>0$.}
\end{corollary}

{
\medskip\sloppy
\paragraph{\bf Concurrent filter composition theorems for continual mechanisms.}
Let $\filt:\left([0,\infty)\times [0,1]\right)^*\to \{\top,\bot\}$ be a filter. Consider adversaries that create continual mechanisms with adaptive privacy parameters such that $\filt(\sigma)=\top$ at every step, where $\sigma$ is the sequence of privacy parameters $(\eps_i, \delta_i)$ of the created mechanisms $\cM_i$ up to that point. Suppose each $\cM_i$ is $(\eps_i,\delta_i)$-DP with respect to a verification function $f_i$. In addition to creating new mechanisms, these adversaries may adaptively generate pairs of inputs for existing mechanisms such that $f_i$ applied to the sequence of input pairs for each $\cM_i$ evaluates to $\top$. The concurrent $\filt$-filter composition of continual mechanisms is said to be $(\eps, \delta)$-DP if the view of such adversaries is $(\eps, \delta)$-indistinguishable when all mechanisms receive either the first inputs from their respective pairs or all receive the second ones.

\begin{theorem}[Concurrent filter composition for continual mechanisms]
    Let $\filt:\left([0,\infty)\times [0,1]\right)^*\to \{\top,\bot\}$ be a filter. The concurrent $\filt$-filter composition of continual mechanisms satisfies the same $(\eps, \delta)$-DP guarantees as the $\filt$-filter composition of noninteractive mechanisms.
\end{theorem}

This result is also shown for R\'enyi DP in Corollary~\ref{cor:filter-comp-rdp} and for $f$-DP, under the finite-communication assumption, in Corollary~\ref{cor:filter-comp-f-dp}.
}

\medskip
\paragraph{\bf General Conditions for Privacy Analysis of Continual Mechanisms.}  
Composition theorems for differential privacy have two main uses.  One is to bound the accumulated privacy loss when DP systems are used repeatedly over time.  Another is to facilitate the design of complex DP algorithms from simpler ones, for example in modular software designs.  The concurrent composition theorems for interactive mechanisms have already found an application in the privacy analysis of matching algorithms in local edge differential privacy model~\cite{DinitzLiLiZh25}.  (That work also has a result in the continual observation model, but the privacy is only proven for oblivious adversaries, which allows it to be proven without using a composition theorem for continual mechanisms.) Our concurrent composition theorem allows for a simpler and more modular privacy analysis of continual mechanisms built out of other continual mechanisms.

There are various applications where a continual mechanism executes multiple differentially private mechanisms, such as the Sparse Vector Technique (SVT) and continual counters, as sub-mechanisms and concurrently interacts with them. Let $\cM$ denote such a composite mechanism. Upon receiving an input, $\cM$ may create new sub-mechanisms, interact multiple times with the existing ones, and then return an output. $\cM$ generates inputs for its sub-mechanisms based on its own inputs and the previous outputs of all sub-mechanisms. (See Figure~\ref{fig:complex-mech} for an illustration.)
Analyzing the privacy of such mechanism designs is challenging because the creation and inputs of sub-mechanisms are interdependent and adaptive. Even if the inputs of $\cM$ are non-adaptive, the inputs to its sub-mechanisms may become adaptive: Let $\cM_1$ and $\cM_2$ denote the sub-mechanisms of $\cM$. If the inputs of $\cM_1$ depend on the outputs of $\cM_2$ and the inputs of $\cM_2$ depend on the outputs of $\cM_1$, then $\cM_1$ receives inputs correlated with its own past outputs. Later, we see an example of this situation, where a histogram mechanism creates $d$ private continual counters at the initialization and an adaptive number of  SVTs and Laplace mechanisms over time. In that design, the inputs of the counters depend on the output of the (active instance of the) SVT and vice versa.

\begin{figure}
    \centering
    \scalebox{0.9}{
    \begin{tikzpicture}[
    node distance=1.5cm and 2cm,
    rect/.style={rectangle, draw, minimum width=1.2cm, minimum height=2cm, align=center},
    subrect/.style={rectangle, draw, minimum width=1cm, minimum height=0.8cm, align=center},
    arrow/.style={-{Stealth}, thick}
]

\node[rect] (adv) {\large $\mathcal{A}$};
\node[draw, rectangle, right=2.6cm of adv, minimum width=3.5cm, minimum height=4.1cm] (bigm) {};
\node[draw, rectangle, right=2.6cm of adv, minimum width=0.8cm, minimum height=2.5cm] (p) {$\mathcal{P}$};

\node[subrect] (m1) at ([xshift=-0.6cm,yshift=1.6cm]bigm.east) {$\mathcal{M}_1$};
\node (dots1) at ([xshift=-0.6cm,yshift=1.0cm]bigm.east) {$\vdots$};
\node[draw, rectangle, minimum width=1cm, minimum height=1.4cm] (mj) at ([xshift=-0.6cm,yshift=-0.1cm]bigm.east) {$\mathcal{M}_j$};
\node (dots2) at ([xshift=-0.6cm,yshift=-0.9cm]bigm.east) {$\vdots$};
\node[subrect] (mk) at ([xshift=-0.6cm,yshift=-1.6cm]bigm.east) {$\mathcal{M}_k$};

\draw[arrow]  ([yshift=0.5cm]adv.east) -- node[above,yshift=-3.5pt] { $m_1$} ([yshift=0.5cm]bigm.west);
\draw[arrow] ([yshift=0cm]bigm.west) -- node[above,yshift=-3.5pt] { $m_1'$} ([yshift=0cm]adv.east);
\node (dots1) at ([xshift=1.3cm,yshift=-0.5cm]adv.east) {$\vdots$};

\draw[arrow] ([yshift=1.0cm]p.east) -- ([yshift=0.6cm]mj.west) node[midway,sloped,above,yshift=-4.5pt] {$m_{1, 1}$};
\draw[arrow] ([yshift=0.2cm]mj.west) -- ([yshift=0.6cm]p.east) node[midway,sloped,above,yshift=-4.5pt] {$m_{1, 1}'$};
\draw[arrow] ([yshift=0.2cm]p.east) -- ([yshift=-0.2cm]mj.west) node[midway,sloped,above,yshift=-4.5pt] {$m_{1, 2}$};
\draw[arrow] ([yshift=-0.6cm]mj.west) -- ([yshift=-0.2cm]p.east) node[midway,sloped,above,yshift=-4.5pt] {$m_{1, 2}'$};

\draw[arrow] ([yshift=-0.6cm]p.east) -- ([yshift=0.2cm]mk.west) node[midway,sloped,above,yshift=-4.5pt] {$m_{1,3}$};
\draw[arrow] ([yshift=-0.2cm]mk.west) -- ([yshift=-1cm]p.east) node[midway,sloped,above,yshift=-4.5pt] {$m_{1,3}'$};

\end{tikzpicture}
    }
    \caption{Illustration of the interactions between an honest adversary $\cA$ and a mechanism with $k$ sub-mechanisms}
    \label{fig:complex-mech}
    \Description{Illustration of the interactions between an honest adversary $\cA$ and a mechanism with $k$ sub-mechanisms}
\end{figure}
In such modular designs, we model $\cM$ as the \emph{post-processing} of a specific concurrent composition of its (continual) sub-mechanisms. Since $\cM$ accesses the raw (secret) data itself, it could in principle incur additional privacy loss beyond that of the concurrent composition of the sub-mechanisms. To rule out such privacy loss, we introduce a set of conditions on $\cM$ and prove that, under these conditions, the privacy guarantee of $\cM$ is exactly that of the concurrent composition of its sub-mechanisms.  Intuitively, $\cM$ must restrict its use of raw inputs solely to generating inputs for differentially private sub-mechanisms, while all other actions, such as creating new sub-mechanisms or determining output values, must depend only on the privatized outputs of the sub-mechanisms. This will ensure that $\cM$ does not use any ``privacy budget'' in addition to the ``privacy budget'' used by the sub-mechanisms. 

Continual mechanisms are typically presented using pseudocode (rather than as formal state machines exchanging messages, which are used by our conditions on $\cM$), where the pseudocode receives an input, updates its variables, and returns an output. We therefore reformulate our conditions in this setting. We call a variable an \emph{untainted variable} if it is either the output of a private sub-mechanism or computed solely based on other untainted variables. Our conditions are now as follows:  
\begin{myitemize}
  \item \emph{Destination condition:} Decisions about whether to produce an output, instantiate a new sub-mechanism (and of which type), or interact with an existing one (and which one) must be based only on untainted variables.  
  \item \emph{Response condition:} The values of $\cM$'s outputs must be determined solely by untainted variables.  
  \item \emph{Mapping condition:} Fixing the values of untainted variables, neighboring inputs to $\cM$ must be mapped to input sequences for the sub-mechanisms with a structure for which concurrent composition is private (This condition is referred to as the $f \to f'$ property in Section~\ref{sec:analyze-complex-cm}.) 
\end{myitemize}

We show that if these conditions are satisfied, then the privacy guarantee of $\cM$ is the same as that of the underlying concurrent composition. For example, if $\cM$ executes $k$ continual sub-mechanisms and ensures that its neighboring inputs yield neighboring input sequences for each sub-mechanism, then the privacy guarantee of $\cM$ matches that of the concurrent composition of its sub-mechanisms as given in Theorem~\ref{thm:extconcomp-intro}. The same conclusion holds for the concurrent $k$-sparse parallel composition of sub-mechanisms, as established in Theorem~\ref{thm:concurrent-parall-comp-IM-intro} and Corollary~\ref{cor:parallelcomp-pureDP-intro}.

\medskip
\paragraph{\bf An Algorithmic Application.}
We give a new privacy proof for the {\em continual monotone histogram mechanism} of Henzinger, Sricharan, and Steiner (HSS)~\cite{henzingerSS2023}.  This mechanism involves concurrent executions of a simpler continual histogram mechanism and multiple instantiations of the Sparse Vector Technique mechanism, and the original privacy analysis was rather intricate even for oblivious adversaries (because the mechanism itself is adaptive in its queries and dataset updates to the building-block mechanisms). Combining our concurrent composition theorems (both Theorem~\ref{thm:extconcomp-intro} and Corollary~\ref{cor:parallelcomp-pureDP-intro}), we obtain a much simpler privacy proof.  Moreover, for the case of approximate DP, the original proof was only for oblivious adversaries, whereas our proof also applies to adaptive adversaries.

\medskip
\paragraph{\bf An Application in Local Differential Privacy (LDP)}
Consider an untrusted server that concurrently interacts with multiple users, each holding a private dataset, with the goal of computing a function on the users’ aggregated data. A \emph{local protocol} specifies the interaction rules between the server and the users, where both the server and the users are modeled as interactive mechanisms. Differential privacy for a local protocol is defined in terms of the distribution of the server’s \emph{view}, consisting of its internal randomness and the transcript of all exchanged messages. Specifically, a protocol~$P$ is said to satisfy $(\eps,\delta)$-LDP if the server’s view is $(\eps,\delta)$-indistinguishable for every pair of neighboring sequences of user datasets.

A useful notion of neighboring in local differential privacy (LDP) considers two assignments of datasets to users that are identical for all but at most $k$ users, where the exceptional users datasets are neighboring with respect to a given binary relation on user datasets. An example of this is a setting in which each user represents a hospital and each patient may have records in at most $k$ hospitals.
Analyzing the privacy of local protocols in which the server interacts concurrently with users over multiple rounds can be challenging. We observe that such interactions can be modeled as a \emph{concurrent parallel composition} of interactive mechanisms, and our new privacy results for this concurrent composition directly yield a simplified privacy analysis for local protocols:

\begin{theorem}[informal statement]\label{thm:intro-ldp}
    Let $P$ be a local protocol with $n$ users, and let $\sim^*$ be a neighbor relation on (individual) user datasets. Define two sequences of $n$ datasets to be neighbors if they are identical for all but at most $k\in\N$ users, where each differing user’s datasets are neighbors with respect to~$\sim^*$. Suppose the interactive mechanism corresponding to every user is $(\eps_0,\delta_0)$-DP with respect to~$\sim^*$ against adaptive adversaries. If the composition of $k$ noninteractive $(\eps_0,\delta_0)$-DP mechanisms is $(\eps,\delta)$-DP, then the protocol~$P$ is $(\eps,\delta)$-LDP.
\end{theorem}

In the formal version of Theorem~\ref{thm:intro-ldp} (see Corollary~\ref{cor:ldp-parallel-concomp}), we incorporate the fact that the server in~$P$ restricts the queries sent to each user and enforces that they satisfy certain properties. Consequently, in Theorem~\ref{thm:intro-ldp}, each user’s privacy guarantee only needs to hold for adversaries whose queries obey these properties. Furthermore, we provide a more general statement of Theorem~\ref{thm:intro-ldp} for an arbitrary choice of neighbor relation on sequences of user datasets in Theorem~\ref{thm:ldp-to-concomp}. The concurrent composition theorem used in the general result must be selected based on the given neighbor relation.

We illustrate an application of Theorem~\ref{thm:intro-ldp} by simplifying the privacy analysis of the local protocol proposed by Dhulipala et al.~\cite{dhulipala2025near} for the core decomposition problem. 
In that work, a sequence of datasets $x=(x_1, \dots, x_n)$ represents a graph with $n$ vertices, where $x_i$ denotes the set of vertices adjacent to vertex $i$. (Users are vertices.) Differential privacy is defined with respect to \emph{edge-neighboring}: two dataset sequences $x$ and $x'$ are considered neighbors if their corresponding graphs differ by a single edge. The edge-neighboring relation can be generalized to a relation $\sim$, where $x \sim x'$ if and only if the adjacency sets $x_i$ and $x_i'$ are identical for all vertices $i$ but at most $2$ of them, and for those vertices, $x_i$ and $x_i'$ are \emph{add-remove neighbors}, meaning that they differ in the presence or absence of a single (adjacent) vertex. By defining $\sim^*$ in Theorem~\ref{thm:intro-ldp} to be the add-remove neighbor relation and setting $k=2$, we obtain a simpler privacy analysis for the protocol in~\cite{dhulipala2025near}.

We note that Theorem~\ref{thm:intro-ldp} and its generalizations are stated for protocols in which each user holds a static dataset. In many real-world settings, however, users may join or leave, and their data may evolve over time. In such cases, users must be modeled more generally as continual mechanisms rather than interactive mechanisms. The privacy guarantees of these protocols can still be computed using our concurrent composition theorems for continual mechanisms.

We also note that in the above definition, all components (i.e., the server and the users) are assumed to be honest, meaning that they follow the protocol, and non-colluding, meaning that they do not share their internal states outside the protocol. In the presence of possible dishonesty or collusion, all users except the target users who hold the private data, together with the server, must be treated as a single adversary that interacts concurrently with the target users. In this case, privacy of the local protocol follows from the concurrent composition of the interactive and continual mechanisms, depending on whether the datasets evolve over time.

\smallskip 
The paper is organized as follows.
Section~\ref{sec:preliminaries} introduces all necessary notation, including noninteractive mechanisms, interactive mechanisms, and interactive post-processing mechanism.
Section~\ref{sec:cm} defines continual mechanisms, identifiers $\I$, verification functions $f$, verifiers $\V$ and differential privacy for continual mechanisms.
Concurrent composition of continual mechanisms is defined in Section~\ref{sec:concomp-cm} and our first concurrent composition theorem is proven.
Section~\ref{sec:parallel} introduces concurrent $k$-sparse parallel composition of continual mechanisms and interactive mechanisms, proves the theorems for this type of composition and shows the counterexample. In Section~\ref{sec:analyze-complex-cm}, we state the conditions and the theorem for analyzing the privacy of complex continual mechanisms that use noninteractive, interactive and continual sub-mechanisms. In Section~\ref{sec:application}, we apply our theorem to analyze the privacy of the continual histogram mechanism of~\cite{henzingerSS2023}. In Section~\ref{sec:ldp}, we show how local differential privacy can be analyzed using our concurrent composition theorems. In Section~\ref{sec:app:core-decomposition}, we simplify the privacy analysis of the core decomposition local protocol by~\cite{dhulipala2025near}. In Section~\ref{sec:post-irr} and Section~\ref{sec:filter-comp-rr}, we prove two technical intermediate lemmas, required for the previous results. In Section~\ref{sec:filter}, we introduce concurrent filter composition of continual mechanisms and reduce it to the filter composition for noninteractive mechanisms. Finally, in Section~\ref{sec:rdp} and Section~\ref{sec:f-dp}, we provide analogous results for R\'enyi DP and $f$-DP, respectively.
\section{Preliminaries}\label{sec:preliminaries}
\subsection{Non-interactive Mechanisms (NIM)}\label{subsec:preliminaries-nim}

Let $\calX$ be a family of datasets. A {\em randomized mechanism} $\cM:\calX\to \calY$ is a randomized function that maps a dataset $x\in\calX$  to an element of $\calY$.
We call $\cM$ a {\em non-interactive mechanism (NIM)} since it halts immediately after returning an output.

To define differential privacy, we require the definition of indistinguishability for random variables:

\begin{definition}[$(\eps,\delta)$-indistinguishability \cite{kasiviswanathan2014semantics}]\label{def:indistinguishability}
    For $\eps\geq 0$ and $0\leq \delta\leq 1$, two random variables $X_1$ and $X_2$ over the same domain $\calY$ are \emph{$(\eps, \delta)$-indistinguishable} if for every measurable subset $Y\subseteq \calY$,
    $$\Pr[X_1\in Y] \leq e^\eps \cdot \Pr[X_2\in Y] + \delta,$$
    and
    $$\Pr[X_2\in Y] \leq e^\eps \cdot \Pr[X_1\in Y] + \delta.$$
\end{definition}

Differential privacy is defined relative to a {\em neighbor relation}, which is a binary relation on the family of datasets $\calX$. For example, we often consider two tabular datasets $x,x'$ to be neighboring if they differ only by the presence or absence of a single record. Differential privacy for NIMs is defined as follows:

\begin{definition}[$(\eps,\delta)$-DP for NIMs \cite{dwork2006calibrating}]\label{def:dp-nim}
    For $\eps\geq 0$ and $0\leq \delta\leq 1$, a randomized mechanism $\cM:\calX\to \calY$ is \emph{$(\eps, \delta)$-differentially private} (or \emph{$(\eps, \delta)$-DP} for short) with respect to (w.r.t.) a neighbor relation $\sim$ if for every two datasets $x_0,x_1\in \calX$ satisfying $x_0\sim x_1$, the random variables $\cM(x_0)$ and $\cM(x_1)$ are $(\eps, \delta)$-indistinguishable.
\end{definition}

An example of an $(\eps, \delta)$-DP NIM is the approximate randomized response, defined as follows:
\begin{definition}[$\RR_{\eps, \delta}$]\label{def:rr}
    For $\eps\geq 0$ and $0\leq \delta\leq 1$, the \emph{approx randomized response} $\RR_{\eps, \delta}$ is a non-interactive mechanism that takes a bit $b\in\{0,1\}$ as input and returns a pair according to the following distribution:
    \begin{equation*}
        \RR_{\eps, \delta}(b) = 
        \begin{cases}
            (\top, b) &\quad\text{w.p. $\delta$}\\
            (\bot, b) &\quad\text{w.p. $(1-\delta)\frac{e^\eps}{1+e^\eps}$}\\
            (\bot, 1-b) &\quad\text{w.p. $(1-\delta)\frac{1}{1+e^\eps}$}
        \end{cases}
    \end{equation*}
\end{definition}

\begin{lemma}[\cite{KairouzOV15}]\label{lem:rr-is-dp}
    For every $\eps\geq 0$ and $0\leq \delta\leq 1$, the approx randomized response $\RR_{\eps, \delta}$ is $(\eps, \delta)$-DP.
\end{lemma}

The reason \cite{KairouzOV15} introduced $\RR_{\eps, \delta}$ is that it is "universal" among $(\eps, \delta)$-DP NIMs in the following sense:

\begin{theorem}[\cite{KairouzOV15}]
    A mechanism $\cM:\calX\to \calY$ is $(\eps, \delta)$-DP w.r.t. a neighbor relation $\sim$ if and only if for any $x_0\sim x_1$, there exists a randomized ``post-processor'' $\cP:\calY\to \zo\times \{\top, \bot\}$ such that for both $b\in\zo$, $\cM(x_b)$ is identically distributed to $\cP(\RR_{\eps, \delta}(b))$.
\end{theorem}
\subsection{Interactive Mechanisms (IM)}\label{subsec:preliminaries-im}

Unlike an NIM, an \emph{interactive mechanism (IM)} continues interacting with the analyst and answers multiple adaptively asked queries. An IM is defined formally as follows.

\begin{definition}[Interactive Mechanism (IM)]\label{def:im}
    An \emph{interactive mechanism (IM)} is a randomized state function $\cM:S_\cM \times Q_\cM \rightarrow S_\cM\times A_\cM$, where
    \begin{itemize}\itemindent=-17pt
        \item $S_\cM$ is the state space of $\cM$,
        \item $Q_\cM$ is the set of possible incoming messages or queries for $\cM$,
        \item $A_\cM$ is the set of possible outgoing messages or answers generated by $\cM$.
    \end{itemize}
    Upon receiving a query $m\in Q_\cM$, according to some distribution, $\cM$ updates its current state $s\in S_\cM$ to a new state $s'\in S_\cM$ and generates an answer $m'\in A_\cM$, which is described by $(s', m')\gets\cM(s, m)$.
    The set $S_\cM^\init\subseteq S_\cM$ denotes the set of possible initial states for $\cM$. 
    For $s^\init\in S_\cM^\init$, $\cM(s^\init)$ represents the interactive mechanism $\cM$ with the initial state $s^\init$.
\end{definition}

\begin{remark}
    In the differential privacy literature, a query typically refers to a question about the dataset held by a mechanism. However, in this paper, we use query as a shorthand for query message, meaning any input message to a mechanism that requires a response.
\end{remark}

The state of an IM $\cM$ serves as its memory. The initial state of $\cM$ can include secret information, such as a dataset $x\in\calX$, or it can simply be an empty memory. Thus, in some cases, the size of the initial state space $S_\cM^\init$ is as large as the family of datasets $\calX$, while in some other cases, $S_\cM^\init$ is a singleton consisting of a single state representing the empty memory.

\begin{notation}\label{not:singleton-init-space}
    If $S_\cM^{\init}$ is a singleton $\{s\}$, we denote $\cM(s)$ as $\cM$.
\end{notation}

An example of an IM is the \emph{interactive randomized response} mechanism, denoted $\irr_{\eps,\delta}$, which will repeatedly be used in our proofs. Ghazi et al.~\cite{ghazi2025privacy} define this mechanism to initially flip a coin and determine whether to answer all subsequent queries truthfully or in a privacy-preserving manner. Our definition of the interactive randomized response mechanism is different but shares the same intuition. 

Recall that the standard randomized response mechanism $\RR_{\eps,\delta}$ is an NIM that outputs a pair from the space $\{\top, \bot\}\times\zo$. In contrast, we define the interactive variant $\irr_{\eps,\delta}$ to output the pair in two separate steps: it first produces an element from $\{\top, \bot\}$, and later returns a bit in $\zo$, such that the joint distribution of the two outputs matches that of $\RR_{\eps,\delta}$. Specifically, $\irr_{\eps,\delta}$ is an IM with a single query message $q^*$. Given an initial state $b \in \zo$, $\irr_{\eps, \delta}(b)$ responds to the first query $q^*$ with an element in $\{\top, \bot\}$ and returns a bit in $\zo$ when receiving $q^*$ for the second time. $\irr_{\eps, \delta}(b)$ halts the communication upon receiving $q^*$ for the third time. The formal definition is given below.

\begin{definition}[$\irr_{\eps, \delta}$]
\label{def:irr}
    For any $\eps\geq 0$ and $0\leq \delta\leq 1$, the interactive randomized response mechanism $\irr_{\eps, \delta}$ is an IM with state space $S=\{\varnull\}\cup\zo\cup \{\top, \bot\}\times\zo$, initial state space $S^\init=\zo$, and query space $Q=\{q^*\}$. $\irr_{\eps, \delta}$ is represented by a randomized function that maps a state $s \in S$ and query $q^* \in Q$ to a new state $s' \in S$ and a message $m$ as follows:
    \begin{itemize}[left=10pt]
        \item If $s$ is a bit $b\in\zo$ (i.e., $q^*$ is received for the first time), the output $m$ is set to $\bot$ with probability $\delta$ and to $\top$ with probability $1 - \delta$. The new state is then set to $s' = (m, b)$.
        \item If $s$ is a pair $(\tau, b)$ in $\{\top, \bot\}\times\zo$ (i.e., $q^*$ is received for the first time), $s'$ is set to $\varnull$ and $m$ is determined as follows:
        \begin{itemize}
            \item If $\tau=\bot$ (i.e., the first response was $\bot$), then $m=b$ with probability $1$.
            \item If $\tau=\top$ (i.e., the first response was $\top$), $m=b$ with probability $\frac{e^\eps}{1 + e^\eps}$ and $m=1-b$ with probability $\frac{1}{1 + e^\eps}$.
        \end{itemize}
    \end{itemize}
\end{definition}

\subsubsection{Interaction of IMs.} 
Two interactive mechanisms $\cM_1$ and $\cM_2$ satisfying $Q_{\cM_1}=A_{\cM_2}$ and $Q_{\cM_2}=A_{\cM_1}$ can interact with each other through message passing. For initial states $s_1\in S_{\cM_1}^\init$ and $s_2\in S_{\cM_2}^\init$, the interaction between $\cM_1(s_1)$ and $\cM_2(s_2)$ is defined as follows: It starts with round 1. Mechanism $\cM_1(s_1)$ sends a message to $\cM_2(s_2)$ in the odd rounds, and $\cM_2(s_2)$ responds in the subsequent even rounds. The interaction continues until one of the mechanisms halts the communication. A mechanism halts the communication by sending the specific message $\halt$ at its round.

In Definition~\ref{def:im}, we defined an IM as a stateful randomized algorithm mapping its current state and an input message to a new state and an output message. To initiate the interaction and send the first message $m_1$, $\cM_1(s_1)$ invokes $\cM_1(s_1, \lambda)$, where $\lambda$ is the empty message.

\begin{definition}[View of an IM] \label{def:view}
    Let $\cM_1$ and $\cM_2$ be two interactive mechanisms such that $Q_{\cM_1} = A_{\cM_2}$ and $Q_{\cM_2} = A_{\cM_1}$. Let $s_1$ and $s_2$ be initial states of $\cM_1$ and $\cM_2$, respectively. Consider the interaction between $\cM_1(s_1)$ and $\cM_2(s_2)$, initiated by a message from $\cM_1(s_1)$. The \emph{view} of $\cM_1(s_1)$ in this interaction, denoted $\View(\cM_1(s_1), \cM_2(s_2))$, is defined as
    $$\View(\cM_1(s_1), \cM_2(s_2)) = (r_{\cM_1}, m_1, m_2, m_3, \dots),$$
    where $r_{\cM_1}$ represents the internal randomness of $\cM_1$, and $(m_1, m_2, \dots)$ is the transcript of messages exchanged between $\cM_1$ and $\cM_2$.
    We define $\Pi_{\cM_1, \cM_2}$ as the family of all possible views over all initial states $s_1\in S_{\cM_1}^\init$ and $s_2\in S_{\cM_2}^\init$.
\end{definition}

Throughout the paper, whenever discussing the interaction between two IMs or the view of a mechanism in an interaction, we implicitly assume that the query and answer sets of the involved mechanisms match, i.e., $Q_{\cM_1}=A_{\cM_2}$ and $Q_{\cM_2}=A_{\cM_1}$.

\begin{definition}[Equivalent IMs]\label{def:identical-IMs}
     Let $\cM$ and $\cM'$ be two interactive mechanisms with $Q_{\cM}=Q_{\cM'}$ and $A_{\cM}=A_{\cM'}$. Let $s$ and $s'$ be initial states for $\cM$ and $\cM'$, respectively. The mechanisms $\cM(s)$ and $\cM'(s')$ are \emph{equivalent} if for every adversary $\cA$ with $A_{\cA}=Q_{\cM}$ and $Q_{\cA}=A_{\cM}$, the random variables $\View(\cA, \cM(s))$ and $\View(\cA, \cM'(s'))$ are identically distributed.
\end{definition}

\subsubsection{Differential Privacy for IMs.}
To define differential privacy for an interactive mechanism $\cM$, we compare the views of a so-called \emph{adversary} interacting with $\cM(s)$ and $\cM(s')$ for neighboring initial states $s, s'$.

\begin{definition}[Adversary]\label{def:adversary}
    An (adaptive) adversary $\cA$ is an interactive mechanism with a singleton initial state space. Let $s^\init$ be the only initial state of $\cA$. By Notation~\ref{not:singleton-init-space}, we write $\cA$ instead of $\cA(s^\init)$.
\end{definition}

Throughout this paper, adversaries always send the first message. Moreover, since adversaries are stateful algorithms, they can choose their messages adaptively based on the history of messages exchanged, meaning that all adversaries in this paper are \emph{adaptive}.

Existing works \cite{vadhan2021concurrent, vadhan2022concurrent, lyu2022composition, haney2023concurrent} view the initial state of $\cM$ as a dataset, i.e., $S_\cM^\init = \calX$, and (as is usual) they define neighbor relations as binary relations on $\calX$. However, for us it will be convenient to conceptually separate $S_\cM^\init$ from the data (which is changing continuously over time), but still define DP w.r.t. neighbor relations on $S_\cM^\init$. This leads to the following definition of differential privacy for IMs.

\begin{definition}[$(\eps,\delta)$-Differential Privacy for IMs]\label{def:dp-im}
    Let $\cM$ be an interactive mechanism, and let $\sim$ be a neighbor relation on $S_\cM^\init$. For $\eps\geq 0$ and $0\leq \delta\leq 1$, $\cM$ is $(\eps,\delta)$-differentially private (or $(\eps,\delta)$-DP for short) w.r.t. $\sim$ against adaptive adversaries if for every pair of neighboring initial states $s_0,s_1\in S_\cM^\init$ satisfying $s_0\sim s_1$ and every adaptive adversary $\cA$, the random variables $\View(\cA, \cM(s_0))$ and $\View(\cA, \cM(s_1))$ are $(\eps, \delta)$-indistinguishable.
\end{definition}

\subsection{Interactive Post-Processing (IPM)}\label{subsec:preliminaries-ipm}

For NIMs, a randomized post-processing function $\cT:\calY\to\mathcal{Z}$ maps the output of a mechanism $\cN:\calX\to\calY$ to an element of $\mathcal{Z}$. Differential privacy for NIMs is known to be preserved under post-processing.

Haney et al.~\cite{haney2023concurrent} generalize post-processing functions to \emph{interactive post-processing mechanisms (IPMs)} as follows and show that interactive post-processing preserves privacy guarantees for IMs.

\sloppy
\begin{definition}[Interactive Post-Processing Mechanism (IPM) \cite{haney2023concurrent}]\label{def:ipm}
    An interactive post-processing mechanism (IPM) is a stateful randomized algorithm $\cP : S_\cP \times \left((\{\ttL\} \times Q_\cP^L) \cup (\{\ttR\} \times Q_\cP^R)\right)\rightarrow S_\cP \times \left((\{\ttL\} \times A_\cP^L) \cup (\{\ttR\} \times A_\cP^R)\right)$ that stands between two interactive mechanisms, referred to as the left and right mechanisms, communicating with them via message passing. 
    The state space of $\cP$ is denoted by $S_\cP$, while $Q_\cP^L$ and $Q_\cP^R$ represent the sets of possible messages received by $\cP$ from the left and right mechanisms, respectively. Similarly, $A_\cP^L$ and $A_\cP^R$ denote the sets of possible messages $\cP$ sends to the left and right mechanisms. $\cP$ randomly maps its current state $s \in S_\cP$, an indicator $v\in \{\ttL, \ttR\}$ specifying whether the incoming message is from the left or right mechanism, and a message $m\in Q_\cP^L\cup Q_\cP^R$ to a new state $s'\in S_\cP$, an indicator $v'\in \{\ttL,\ttR\}$ specifying whether the outgoing message is for the left or right mechanism, and a message $m'\in A_\cP^L\cup A_\cP^R$. This is described by $(s', v', m') \gets \cP(s, v, m)$.
\end{definition}

The post-processing of an NIM $\cN$ by a non-interactive post-processing function $\cT$, denoted as $\cT\circ\cN$, is an NIM that first runs $\cN$ on the input dataset and then applies $\cT$ to $\cN$'s output, returning $\cT$'s output. In contrast, the post-processing of an IM $\cM$ by an IPM $\cP$, denoted by $\cP\circstar\cM$, is an IM that runs $\cP$ and $\cM$ internally. In an interaction with another IM $\cA$, $\cP\circstar\cM$ forwards its input messages from $\cA$ to $\cP$ as left messages. Then, $\cP$ and $\cM$ communicate multiple times, with $\cM$ acting as the right mechanism, until $\cP$ returns a response for its left mechanism, at which point $\cP\circstar\cM$ returns this response to $\cA$. Note that $\cP$ may return an answer for its left mechanism immediately after receiving a left message, without interacting with $\cM$. The formal definition of $\cP\circstar\cM$ is given below:

\begin{definition}[Post-Processing of an IM $\cM$ by an IPM $\cP$ ($\cP\circstar\cM$)\cite{haney2023concurrent}]\label{def:p_circ_m}
    Let $\cM$ be an IM and $\cP$ an IPM such that $Q_\cP^R=A_\cM$ and $A_\cP^R=Q_\cM$.
    The \emph{post-processing} of $\cM$ by $\cP$ is an interactive mechanism $\cP\circstar\cM: S_{\cP\circstar\cM} \times Q_{\cP}^L \to S_{\cP\circstar\cM} \times A_{\cP}^L$, defined in Algorithm \ref{alg:p_circ_m}. The initial state space of $\cP\circstar\cM$ equals $S_{\cP\circstar\cM}^\init= S_\cM^\init\times S_\cP^\init$. For readability, when $\cP\circstar\cM$ is initialized with the state $(s^\cM, s^\cP)\in S_{\cP\circstar\cM}^\init$, we denote it as $\cP(s^\cP)\circstar\cM(s^\cM)$ instead of $\left(\cP \circstar \cM\right)\left((s^\cM, s^\cP)\right)$.
\end{definition}

\begin{algorithm}
    \caption{Post-Processing of $\cM$ by $\cP$, denoted by $\cP \circstar \cM$~\cite{haney2023concurrent}.}\label{alg:p_circ_m}
    \begin{algorithmic}
        \Procedure{$\left(\cP \circstar \cM\right)$}{$s, m$}:
        \State $(s^\cM, s^\cP)\gets s$
        \State $(s^\cP, v, m') \gets \cP(s^\cP, \ttL, m)$
        \While{$v = \ttR$}
            \State $(s^\cM, m) \gets \cM(s^\cM, m')$
            \State $(s^\cP, v, m') \gets \cP(s^\cP, \ttR, m)$
        \EndWhile
        \State \Return $((s^\cM, s^\cP), m')$
    \EndProcedure
    \end{algorithmic}
\end{algorithm}

Haney et al.~\cite{haney2023concurrent} shows that if an IM $\cM$ is $(\eps, \delta)$-differentially private, then for any IPM $\cP$ with a singleton initial state space, $\cP\circstar\cM$ is also $(\eps, \delta)$-differentially private. This result is formally stated below.

\begin{lemma}[Post-Processing Lemma for IMs \cite{haney2023concurrent}]\label{lem:post-im}
    Let $\cM$ be an IM with an initial state space $S_\cM^\init$ and $\cP$ be an IPM with a singleton initial state space $S_\cP^\init=\{s\}$. For every pair of initial states $s_0, s_1\in S_\cM^\init$, the following holds: If for every adversary $\cA$, the random variables $\View(\cA, \cM(s_0))$ and $\View(\cA, \cM(s_1))$ are $(\eps, \delta)$-indistinguishable, then for every adversary $\cA$, the random variables $\View(\cA, \cP\circstar\cM(s_0))$ and $\View(\cA, \cP\circstar\cM(s_1))$ are $(\eps, \delta)$-indistinguishable. \footnote{In \cite{haney2023concurrent}, the definitions of IMs and IPMs are more restrictive. Specifically, for any IM $\cM$, $S_\cM=Q_\cM=A_\cM=\zo^*$, and for any IPM $\cP$, $S_\cP=M_\cP=\zo^*$; however, their results do not rely on these restrictions, and the interactive post-processing lemma holds for arbitrary sets $S_\cM$, $Q_\cM$, $A_\cM$, $S_\cP$, and $M_\cP$, which are not necessarily countable.}
\end{lemma}

IPMs are not limited to interacting with IMs. Two IPMs can interact with each other as well, as follows:

\begin{definition}[Chain of IPMs]\label{def:chain-IPMs}
    Let $\cP_1$ and $\cP_2$ be two IPMs satisfying $Q_{\cP_1}^R=A_{\cP_2}^L$ and $A_{\cP_1}^R=Q_{\cP_2}^L$. $\cP_1\circstar\cP_2$ is an IPM with state space $S_{\cP_1\circstar\cP_2}=S_{\cP_1}\times S_{\cP_2}$, left query space $Q_{\cP_1\circstar\cP_2}^L=Q_{\cP_1}^L$, left answer space $A_{\cP_1\circstar\cP_2}^L=A_{\cP_1}^L$, right query space $Q_{\cP_1\circstar\cP_2}^R=Q_{\cP_2}^R$, and right answer space $A_{\cP_1\circstar\cP_2}^R=A_{\cP_2}^R$. This mechanism is formally defined in Algorithm~\ref{alg:p1-circ-p2}. For readability, when $\cP_1\circstar\cP_2$ is initialized with $(s^1, s^1)\in S_{\cP_1\circstar\cP_2}^\init$, we denote it as $\cP_1(s^1)\circstar\cP_2(s^2)$ instead of $\left(\cP_1\circstar\cP_2\right)\left((s^1, s^2)\right)$.
\end{definition}
\begin{algorithm}
    \caption{Interactive Post-Processing Mechanism $\cP_1 \circstar \cP_2$.}\label{alg:p1-circ-p2}
    \begin{algorithmic}
        \Procedure{$\left(\cP_1 \circstar \cP_2\right)$}{$s, v, m$}:
        \State $(s^1, s^2)\gets s$
        \While{True}
            \If{$v=\ttL$}
                \State $(s^1, v, m) \gets \cP_1(s^1, v, m)$
                \If{$v=\ttL$}
                    \State Break the while loop
                \EndIf
            \Else
                \State $(s^2, v, m) \gets \cP_2(s^2, v, m)$
                \If{$v=\ttR$}
                    \State Break the while loop
                \EndIf
            \EndIf
        \EndWhile
        \State \Return $((s^1, s^2), v, m)$
    \EndProcedure
    \end{algorithmic}
\end{algorithm}

\section{Continual Mechanisms}\label{sec:cm}
In this section, we introduce \emph{continual mechanisms (CMs)} and propose a general definition of differential privacy (DP) for them. This definition unifies and extends DP for IMs and DP under continual observation. Its generality and simplicity allow us to prove (filter) concurrent composition theorems for CMs without requiring complicated proofs. We later apply these theorems to simplify existing analyses and strengthen results in the continual observation setting.

This section is based on two crucial observations: (i) A message, as sent or received by IMs and IPMs, is a general and powerful concept capable of encoding various types of instructions, and (ii) the essence of DP is to protect the privacy of a single bit that distinguishes neighboring states, sequences, etc. The second observation has been used many times in the DP literature and was specifically used in \cite{pricedpco} to model adaptivity in continual observation. Using these key observations, we first reformulate the definition of DP for IMs and continual observation. Then, inspired by the new formulations, we introduce continual mechanisms (CMs) and formally define DP in this setting.

\subsection{An Alternative Definition of DP for IMs.}\label{subsec:alternative-im}
By Definition~\ref{def:dp-im}, an interactive mechanism $\cM$ with initial state space $S_{\cM}^\init$ is $(\eps, \delta)$-DP w.r.t. a neighboring relation $\sim$, if for every two initial states $s_0, s_1\in S_{\cM}^\init$ such that $s_0\sim s_1$ and every adversary $\cA$, the views of $\cA$ in the interaction with $\cM(s_0)$ and $\cM(s_1)$ are $(\eps, \delta)$-indistinguishable. This means that for the worst-case (i.e., adversarial) choice of neighboring initial states, an adversary by adaptively asking queries should not be able to determine whether $\cM$ is running on $s_0$ or $s_1$.

In an alternative definition, an adversary sends a pair of neighboring initial states $(s_0, s_1)$ as its first message and adaptively asks a sequence of queries $q_1, q_2, \dots$ in the subsequent messages. In this setting, $\cM$ receives one of $s_0$ and $s_1$ as its first message and responds with an acknowledgment message $\top$. Thereafter, $\cM$ answers queries based on the received state, and the goal of the adversary is to determine whether $\cM$ was initialized with $s_0$ or $s_1$. In this definition, the actual initial state of $\cM$ (not the one received as the first message) carries no secret information and serves merely as an empty memory. Therefore, the initial state space of $\cM$ consists of a single fixed initial state. 

To formalize this, we restrict attention to adversaries that send a pair of neighboring initial states $(s_0, s_1)$ as their first message and pairs of identical queries $(q_i, q_i)$ in its subsequent messages. We also define an IPM $\I$ with the initial state space $S_{\I}^\init=\zo$ to stand between the adversary and $\cM$, filtering the pairs based on its, i.e., $\I$'s, initial state. Specifically, for $b\in\zo$, in the interaction between an adversary $\cA$ and $\I(b)\circstar\cM$, upon receiving a message $(m_0,m_1)$ from $\cA$, $\I(b)$ forwards $m_b$ to $\cM$ and returns $\cM$'s answer unchanged to $\cA$. 
With this formulation, an interactive mechanism $\cM$ is $(\eps, \delta)$-DP if for every adversary of the above form, the random variables $\View(\cA, \I(0)\circstar\cM)$ and $\View(\cA, \I(1)\circstar\cM)$ are $(\eps, \delta)$-indistinguishable.

We will enforce the restrictions on the adversary by an IPM, which we call a \emph{verifier} $\V$. The verifier checks whether the initial message is a pair of neighboring datasets and the subsequent messages of the adversary are of the form $(q, q)$; only if these checks pass are the messages forwarded to $\I$. 

\subsection{An Alternative Definition of DP under Continual Observation.}
Dwork, Naor, Pitassi, and Rothblum~\cite{Dwork2010} introduced \emph{differential privacy under continual observation}, where the dataset is not given at the initialization time and instead evolves over time. In this model, the mechanism receives either an empty record or a data record at each step. If the record is non-empty, the mechanism adds it to its dataset. Then, whether the input record is empty or not, the mechanism answers a predetermined question about its current dataset. 

To capture ``event-level privacy'' (where the privacy unit is a single record), Dwork et al.~\cite{Dwork2010} define \emph{neighboring} as a binary relation on sequences of data and empty records: two sequences are neighboring if they are identical except for the presence or absence of a record in at most one step. They called a mechanism \emph{$(\eps, \delta)$-DP under continual observation} if the distributions of answer sequences are $(\eps, \delta)$-indistinguishable for every two neighboring input sequences. This definition does not allow for data sequences that are adaptively chosen based on the answers released.

Jain, Raskhodnikova, Sivakumar, and Smith~\cite{pricedpco} formalized the definition of DP against adaptive adversaries in the continual observation model as follows: An adaptive adversary $\cA$ chooses a single (data or empty) record at every step, except for one \emph{challenge} step, selected by the adversary, where $\cA$ sends a pair of records. They called a mechanism $\cM$ \emph{$(\eps, \delta)$-differentially private against adaptive adversaries} in the continual observation model if for any such adaptive adversary $\cA$, the views of $\cA$---consisting of $\cA$'s internal randomness and the transcript of messages sent and received by $\cA$---are $(\eps, \delta)$-indistinguishable when the first or the second element of the pair in the challenge step is handed to $\cM$.

Denisov, McMahan, Rush, Smith and Thakurta~\cite{denisov2022improved} generalize this definition by allowing the adversary to send a pair of (data or empty) records at \emph{every} step such that the sequence of first records and the sequence of second records are neighbors under an arbitrary binary relation~$\sim$ on sequences of data and empty records. A mechanism $\cM$ is said to satisfy $(\eps,\delta)$-DP against an adaptive adversary if, for every adversary~$\cA_\sim$ that adaptively sends pairs of records forming $\sim$-neighbor sequences, the adversary’s views are $(\eps,\delta)$-indistinguishable when~$\cM$ receives the first elements of each pair and when it receives the second elements.

A mechanism $\cM$ in the continual observation model can be viewed as an IM with a single possible initial state, representing an empty dataset. The input message set of $\cM$ consists of the empty record and all possible data records. Under this formulation, the above definition of differential privacy can be restated as follows: For every adversary $\cA_\sim$, the random variables $\View(\cA_\sim, \I(0)\circstar\cM)$ and $\View(\cA_\sim, \I(1)\circstar\cM)$ are $(\eps, \delta)$-indistinguishable, where $\I$ is the interactive post-processing mechanism discussed previously. Again, below we will enforce the restriction on the adversary via another IPM $\V$.

In this alternative definition, all adversary messages are pairs of data or empty records, and the sequences of first and second elements are neighbors. Recall that in the alternative definition of DP for IMs, the first adversary message must be a pair of neighboring initial datasets, while all subsequent adversary messages must be pairs of identical queries. Next, we generalize the definition of differential privacy by allowing arbitrary conditions on the sequence of pairs sent by adversaries, without restricting the content of the messages.

\subsection{Introducing CMs.}
As discussed in the alternative formulation of IMs, we can assume that the initial state space of IMs is a singleton and encode its initial state as a message. A continual mechanism is defined as follows:

\begin{definition}[Continual Mechanism (CM)]\label{def:cm}
    A \emph{continual mechanism (CM)} is an interactive mechanism with a singleton initial state space.
\end{definition}
To the best of our knowledge, all mechanisms in DP literature can be represented as CMs:

\begin{itemize}[left=0pt]
    \item Every non-interactive mechanism $\cN$, such as the Laplace mechanism, can be modeled as a CM that receives a dataset $x$ as its first message and returns $\cN(x)$ in response. Upon receiving any subsequent input message, the CM representing $\cN$ halts the communication. 
    \item Every interactive mechanism $\cM$ with an initial state space $S_\cM^\init$ can be modeled as a CM that receives an initial state $s\in S_\cM^\init$ as its first message and responds with an acknowledgment message $\top$. This CM then answers the next input messages the same as $\cM(s)$.
    \item Mechanisms such as the Sparse Vector Technique (SVT)~\cite{dwork2009complexity,HardtRo10,lyu2016understanding} and the continual counter~\cite{dwork2006calibrating,DBLP:journals/tissec/ChanSS11,mcmahan2022private,dvijotham2024efficient} in the streaming settings can also be represented as CMs receiving two types of messages: question queries and data-update queries. Question queries request information about the current dataset, while data-update queries consist of instructions for modifying the dataset.
\end{itemize}

\subsection{DP for CMs.}
To formulate the general definition of DP for CMs, we introduce two IPMs: an identifier $\I$ and a verifier $\V$. As described earlier, an identifier receives a secret bit $b$ as its initial state and maps messages of the form $(m_0, m_1)$ to $m_b$. To avoid defining a new family of adversaries for each problem or model, we introduce \emph{verifiers}. 

Let $\cM$ be a CM and $\I$ be an identifier for $\cM$. A verifier $\V$ with a \emph{verification function} $f$, denoted by $\V[f]$, stands between an adversary $\cA$ and the interactive mechanism $\I\circstar\cM$. Upon receiving a message from $\cA$, $\V[f]$ applies $f$ to verify whether the sequence of pairs sent by the adversary up to that point satisfies a set of properties. Comparing the sequences formed by the first and second elements of the pairs, $f$ can be used to encode a neighboring property. If the sequences satisfy the property, $\V[f]$ forwards $\cA$'s message to $\I\circstar\cM$; otherwise, it halts the communication. The formal definitions of identifiers, verification functions, and verifiers are given below.

\begin{definition}[Identifier $\I$]\label{def:id}
    An \emph{identifier} $\I$ is an IPM with state space $S_\I$ and initial state space $S_\I^\init=\zo$. $\I$ stands between two IMs, referred to as the left and right mechanisms. While there is no restriction on the format of messages received from the right mechanism, $\I$ only accepts messages of the form $(m_0, m_1)$ from its left mechanism. Upon receiving a message $(m_0, m_1)$ from the left mechanism, an instance of $\I$ with an initial state $b\in\zo$ forwards $m_b$ to the right mechanism. Upon receiving a message from the right mechanism, $\I(b)$ forwards the message unchanged to the left mechanism. 
\end{definition}

\begin{definition}[Verification Function]\label{def:ver-func}
    A verification function $f$ takes a sequence of messages as input and returns either $\top$ or $\bot$. A message sequence $m_1, \dots, m_k$ is $f$-valid if $f(m_1, \dots, m_k)=\top$.
\end{definition}

\begin{remark}\label{rem:universal-set}
    Since the set of all possible messages does not exist, we assume that all message spaces considered in this paper are subsets of some universal set $\mathcal{U}$. For instance, in \cite{haney2023concurrent}, the universal set is taken to be $\zo^*$. The domain of verification functions is $\mathcal{U}^*$.
\end{remark}

\begin{definition}[Verifier $\V$]\label{def:verifier}
    For a verification function $f$, a verifier $\V[f]$ is an IPM with the initial state space $S_{\V[f]}^\init=\{()\}$, where $()$ represents an empty sequence. Since the initial state space of $\V[f]$ is a singleton, we denote an instance of $\V[f]$ with the initial state $()$ as $\V[f]$. 
    The IPM $\V[f]$ stands between two IMs, referred to as the left and right mechanisms. The state of $\V[f]$ consists of the sequence of messages it has received from the left mechanism so far. Upon receiving a message $m$ from the left mechanism, $\V[f]$ appends $m$ to its state and applies the verification function $f$. If the sequence is $f$-valid, then $\V[f]$ forwards the message $m$ to the right mechanism; otherwise, it halts the communication. $\V[f]$ forwards any message from the right mechanism unchanged to the left mechanism.
\end{definition}

\begin{remark}\label{rem:ver-func-assumption-type-check}
    By definition, for every continual mechanism $\cM$ and every verification function $f$, $\V[f]\circstar\I\circstar\cM$ is an IM with two possible initial states. The mechanism $\I\circstar\cM$ accepts messages of the form $(m_0, m_1)$, where both $m_0$ and $m_1$ are in the query space $Q_\cM$ of $\cM$. Therefore, when analyzing $\V[f]\circstar\I\circstar\cM$, we assume that the verification function $f$ always returns $\bot$ whenever a message in its input sequence is not in $Q_\cM\times Q_\cM$, even if not explicitly say so. 
\end{remark}

Remark~\ref{rem:ver-func-assumption-type-check} implies that $\V[f]\circstar\I\circstar\cM$ can interact with any other IM. Thus, in the following definition of DP for CMs, we can use the general definition of adversaries given in Definition~\ref{def:adversary}, without making any assumptions about the family of adversaries. Instead of defining DP w.r.t. neighboring relations, we introduce a more general formulation in terms of verification functions, as follows:

\begin{definition}[$(\eps, \delta)$-Differential Privacy for CMs]\label{def:dp-cm}
    Let $\cM$ be a CM and $f$ a verification function. We say that $\cM$ is \emph{$(\eps, \delta)$-differentially private} (or \emph{$(\eps, \delta)$-DP} for short) w.r.t. $f$ against adaptive adversaries if for every adversary $\cA$, the random variables $\View(\cA, \V[f]\circstar\I(0)\circstar\cM)$ and $\View(\cA, \V[f]\circstar\I(1)\circstar\cM)$ are $(\eps, \delta)$-indistinguishable.
\end{definition}

Let $S_\cM^\init=\{s^\init\}$ be the initial state space of the continual mechanism $\cM$ in Definition~\ref{def:dp-cm}. By definition, the initial state spaces of $\I$ and $\V[f]$ are $\zo$ and $\{()\}$, respectively. Thus, the initial state space of $\V[f]\circstar\I\circstar\cM$ is $\left\{((), 0, s^\init), ((), 1, s^\init)\right\}$. Define $\sim_{01}$ as the universal relation on this set, where every two elements are neighboring. Comparing Definition~\ref{def:dp-cm} and Definition~\ref{def:dp-im} leads to the following corollary:

\begin{corollary}\label{cor:eq-def-dp-cm}
    A continual mechanism $\cM$ is $(\eps, \delta)$-DP w.r.t. a verification function $f$ against adaptive adversaries if and only if the interactive mechanism $\V[f]\circstar\I\circstar\cM$ is $(\eps, \delta)$-DP w.r.t. the neighboring relation $\sim_{01}$ against adaptive adversaries.
\end{corollary}

\noindent
\underline{CMs vs. IMs.}
In principle, CMs are a special class of IMs with a singleton initial state space. Conversely, as described earlier, IMs can be represented by CMs. Both CMs and IMs are defined by a randomized function mapping a state and query message to another state and an answer message, where message is a general concept. The key distinction between CMs and IMs lies in how differential privacy is defined for each.

By Definition~\ref{def:im}, the initial state space of IMs contains multiple states, and DP for IMs is defined based on the neighboring initial states. However, by Definition~\ref{def:dp-cm}, the initial state space of CMs is singleton, and a verifier forces an adversary to generate pairs such that the sequences of the first and second elements of the pairs satisfy certain properties. While the adversary chooses these two mechanisms, the identifier selects one of them to forward. Thus, the definition of DP for CMs is based on a bit. 
\section{Concurrent Composition of Continual Mechanisms}\label{sec:concomp-cm}
\noindent
\textbf{Composition of NIMs.} {For $1\leq i\leq k$, let $\cN_i$ be an NIM with input domain $\calX_i$, and let $\sim_i$ be a neighboring relation on $\calX_i$. The \emph{composition of the NIMs} $\cN_1, \dots, \cN_k$ is represented by a non-interactive mechanism $\comp[\cN_1, \dots, \cN_k]$ with the domain $\calX=\calX_1\times\dots\times\calX_k$. This mechanism takes a tuple of datasets $x=(x_1,\dots, x_k)$ as input and returns 
$${\comp[\cN_1, \dots, \cN_k](x)=(\cN_1(x_1), \dots, \cN_k(x_k))}.$$ 

Define $\sim$ as a neighboring relation on $\calX$, where two tuples $x=(x_1,\dots, x_k)$ and $x'=(x_1',\dots, x_k')$ satisfy $x\sim x'$ if and only if $x_i\sim_i x_i'$ for all $1\leq i\leq k$. The composition of NIMs $\cN_1, \dots, \cN_k$ is said to be $(\eps, \delta)$-DP if the non-interactive mechanism $\comp[\cN_1, \dots, \cN_k]$ is $(\eps, \delta)$-DP w.r.t. $\sim$ (see Definition \ref{def:dp-nim}).

In the proofs of this section, we consider $\comp[\cN_1, \dots, \cN_k]$ not as a non-interactive, but instead as an interactive mechanism with initial state space $\calX$ that receives $x=(x_1,\dots, x_k)$ not as input, but instead as its first message. Immediately after receiving this message it executes $\comp[\cN_1, \dots, \cN_k]$ on its initial state and returns its output. Upon receiving a second message, this interactive mechanism halts the communication.
}

\medskip\noindent
\textbf{Concurrent Composition of IMs.}{
Vadhan and Wang~\cite{vadhan2021concurrent} extended the definition of the composition of NIMs to the \emph{concurrent composition of IMs}, where an adversary interacts with $k$ interactive mechanisms $\cM_1, \dots, \cM_k$ at the same time. In concurrent composition, the adversary generates a message and determines the receiver of this message (one of the $\cM_i$s) adaptively based on all the previous messages exchanged between the adversary and the mechanisms $M_1, \dots, M_k$. To formalize this, Vadhan and Wang~\cite{vadhan2021concurrent} introduced an interactive mechanism $\concomp[\cM_1, \dots, \cM_k]$, which internally runs $\cM_1, \dots, \cM_k$ and, upon receiving a message $(q, j)$, forwards the query $q$ to $\cM_j$ and returns its response.

The initial state space of $\concomp[\cM_1, \dots, \cM_k]$ is $S_{\cM_1}^\init\times\dots\times S_{\cM_k}^\init$, where $S_{\cM_i}^\init$ is the initial state space of $\cM_i$. Let $\sim_i$ be a neighboring relation on each $S_{\cM_i}^\init$, and define the neighboring relation $\sim$ on $S_{\cM_1}^\init\times\dots\times S_{\cM_k}^\init$ as follows: for every two initial states $s=(s_1, \dots, s_k)$ and $s'=(s_1', \dots, s_k')$, we have $s\sim s'$ if and only if $s_i\sim_i s_i'$ for all $i\in[k]$.

Vadhan and Wang~\cite{vadhan2021concurrent} show that if each IM $\cM_i$ is $\eps_i$-DP w.r.t $\sim_i$, and if the composition of NIMs $\RR_{\eps_1, 0}, \dots, \RR_{\eps_k, 0}$ is $(\eps, \delta)$-DP, then $\concomp[\cM_1, \dots, \cM_k]$ is also $(\eps, \delta)$-DP against an adaptive adversary. Lyu~\cite{lyu2022composition} and Vadhan and Zhang~\cite{vadhan2022concurrent} extend this result for $(\eps_i, \delta_i)$-DP IMs.

In the concurrent composition of IMs, queries asked to each mechanism are chosen adaptively, not only based on that mechanism’s previous responses but also using the responses of the other mechanisms up to that point. The results of \cite{vadhan2022concurrent} and \cite{lyu2022composition} show that this concurrency and additional adaptivity do \textit{not} compromise privacy guarantees. As a result, all $(\eps, \delta)$-DP composition theorems for NIMs, such as the basic~\cite{dwork2006calibrating}, advanced~\cite{dwork2010boosting}, and optimal~\cite{KairouzOV15} composition theorems, can also be applied to the concurrent composition of IMs.
}

\medskip\noindent
\textbf{Concurrent Composition of CMs.}{
We extend the definition of the concurrent composition for IMs to the concurrent composition of CMs, where an adversary can create new CMs or ask pairs of queries from the existing ones. Depending on the type of composition, the created mechanisms are constrained, and the sequence of query pairs asked from each CM satisfies certain properties. Based on a secret bit $b\in\zo$, each CM either receives the first element of every query pair asked from it or the second element. In Section~\ref{subsec:composition-extconcomp}, we generalize $\concomp$ to $\extconcomp$ and give a formal definition for the concurrent composition of CMs.

To extend the proofs of~\cite{lyu2022composition,vadhan2022concurrent} on the privacy guarantee of the concurrent composition of IMs to the concurrent composition of CMs, we want to use the following intermediate lemma from their paper:
For every $(\eps, \delta)$-DP interactive mechanism $\cM$ and every pair of neighboring initial states $s$ and $s'$ for $\cM$, there exists an IPM $\cP$ such that for each $b\in\zo$, the IMs $\cP\circstar\RR_{\eps, \delta}(b)$ and $\cM(b)$ are equivalent. 
However, as we discuss in Section~\ref{subsec:composition-comp-lyu}, there are hidden assumptions in Lyu's proof of this lemma.

In Section~\ref{subsec:composition-comp-fixed-cm}, we analyze the privacy of the concurrent composition of $k$ fixed CMs $\cM_1, \dots, \cM_k$. In this composition, the adversary first creates the mechanisms $\cM_1, \dots, \cM_k$ and then asks pairs of queries from each mechanism, ensuring that the sequence of query pairs for each mechanism is valid. In Section~\ref{subsec:composition-fixed-param-cm}, we extend this composition by fixing $k$ privacy parameters $(\eps_1, \delta_1), \dots, (\eps_k, \delta_k)$ and allowing the adversary to adaptively choose the $(\eps_i, \delta_i)$-DP mechanism $\cM_i$ over time.

}

\subsection{Formalization.}\label{subsec:composition-extconcomp}
We extend the definition of $\concomp$ in \cite{vadhan2021concurrent} and introduce $\extconcomp$ that allows the creation of new mechanisms in addition to queries to existing ones. A \emph{creation query} $(\cM, s, \eps, \delta, f)$ is a message that requests the creation of a continual mechanism $\cM$ with the initial state $s$ and asserts that $\cM$ is $(\eps, \delta)$-DP w.r.t. the verification function $f$. The mechanism $\extconcomp$ accepts creation queries and messages of the form $(q, i)$, where $q$ is a query for the $i$-th created mechanism. The formal definition of $\extconcomp$ is as follows:

\begin{definition}\label{def:extconcomp}
    \emph{Extended concurrent composition, $\extconcomp$,} is a continual mechanism that receives two types of messages:
    \begin{itemize}
        \item A creation query $(\cM, s, \eps, \delta, f)$, where $\cM$ is a CM with a single initial state $s$ that is $(\eps, \delta)$-DP w.r.t. a verification function $f$, and
        \item a query-index pair $(q, i)$, where $i$ is a positive integer, and $q$ is any message.
    \end{itemize}
    Upon receiving $(\cM, s, \eps, \delta, f)$, $\extconcomp$ creates an instance of $\cM$ with the initial state $s$, and when receiving $(q, i)$, it forwards $q$ to the $i$-th created mechanism and returns its response. If less than $i$ mechanisms have been created, $\extconcomp$ halts the communication. The state of $\extconcomp$ is a sequence of pairs $(\cM_i, s_i')$, where the $i$-th pair represents the $i$-th created mechanism $\cM_i$ and its current state $s_i'$. The single initial state space of $\extconcomp$ is the empty sequence $()$. Both query set and answer set of $\extconcomp$ are the universal message set $\mathcal{U}$ in Remark~\ref{rem:universal-set}.
\end{definition}

\begin{remark}
    By Russell’s paradox, a set of all continual mechanisms cannot exist. Therefore, to ensure that Definition~\ref{def:extconcomp} is mathematically correct and the query and answer sets of $\extconcomp$ are well defined, we restrict creation queries to select mechanisms from a well-defined set and verification functions from a well-defined set. These sets may be infinite. The mechanism $\extconcomp$ itself may belong to the mechanism set and, consequently, may be recursively created by itself.
\end{remark}

To define DP for the extended concurrent composition, we must impose constraints on the created mechanisms and their query pairs to limit the adversary’s power. We force a set of constraints using a verifier $\V[f]$, where $f$ is a verification function returning $\bot$ whenever the adversary misbehaves. We define variants of the concurrent composition of CMs by choosing different verification functions $f$:

\begin{definition}
    Let $f$ be a verification function. The \emph{$f$-concurrent composition of CMs} is the interactive mechanism $\V[f]\circstar\I\circstar\extconcomp$.
\end{definition}

Differential privacy for the $f$-concurrent composition of CMs follows the standard definition for IMs:

\begin{definition}[DP for $f$-Concurrent Composition of CMs]
    For a verification function $f$ and privacy parameter $\eps\geq 0$ and $0\leq \delta\leq 1$, \emph{the $f$-concurrent composition of CMs is $(\eps, \delta)$-DP against adaptive adversaries} if the interactive mechanism $\V[f]\circstar\I\circstar\extconcomp$ is $(\eps, \delta)$-DP w.r.t. $\sim_{01}$. Equivalently, this holds if the continual mechanism $\extconcomp$ is $(\eps, \delta)$-DP w.r.t. $f$ against adaptive adversaries.
\end{definition}

Although verification functions $f$ for $\extconcomp$ vary for different compositions, they must all ensure that the format of the adversary's messages is compatible with $\I\circstar\extconcomp$. An adversary is expected to either create a new CM or send a pair of queries to an existing one. The creation of a new mechanism is independent of the initial state of $\I$. Thus, when requesting the creation of a new mechanism, the adversary must send a pair of identical creation queries to $\V[f]$. Moreover, due to the definition of $\I$ and $\extconcomp$, to request asking one of the queries $q_0$ and $q_1$ from the $i$-th created mechanism according to $\I$'s initial state, the adversary must send the message $\left((q_0, i),(q_1, i)\right)$ to $\V[f]$.

Next, we formalize these shared constraints by defining \emph{suitable} message sequences. Later, when defining specific verification functions $f$ for different concurrent compositions, we set $f(m_1, \dots, m_k)=\bot$ whenever the sequence of messages $m_1, \dots, m_k$ is not suitable. For convenience, in the rest of the paper, we represent a pair of identical creation queries $(\cM, s, \eps, \delta, f)$ by $(\cM, s, \eps, \delta, f)^2$.

\begin{definition}[Suitable Sequence of Messages]\label{def:suitable-seq}
    A sequence of messages $m_1, \dots, m_k$ is said to be \emph{suitable}  if the following conditions hold for every $1\leq j\leq k$:
    \begin{enumerate}[label=(\roman*), left=0pt]
        \item Each message $m_j$ is either
        \begin{enumerate}
            \item a pair of identical creation queries of the form $(\cM, s, \eps, \delta, f)^2$, or
            \item a query pair message of the form $\left((q_0, \ell),(q_1, \ell)\right)$, where both queries $q_0$ and $q_1$ share the same index $\ell$.
        \end{enumerate}
        \item If $m_j=\left((q_0, \ell),(q_1, \ell)\right)$, then there are at least $\ell$ creation-query pairs in the sequence $m_1, \dots, m_{j-1}$.
        \item If $m_j=\left((q_0, \ell),(q_1, \ell)\right)$ and $\cM_\ell$ is the $\ell$-th mechanism requested to be created in the preceding messages $m_1, \dots, m_{j-1}$, then both $q_0$ and $q_1$ belong to the query space of $\cM_\ell$, i.e., $q_0, q_1\in Q_{\cM_\ell}$.
        \item If $m_j=(\cM, s, \eps, \delta, f)^2$, then $\cM$ is a CM with a single initial state $s$, and it is $(\eps, \delta)$-DP w.r.t. $f$.
        \item If $m_j=(\cM, s, \eps, \delta, f)^2$ is the $\ell$-th creation-query pair in the sequence $m_1, \dots, m_{j-1}$, then the corresponding sequence of query pairs $(q_0, q_1)$ from the messages of the form $\left((q_0, \ell),(q_1, \ell)\right)$ in $m_1, \dots, m_k$ is $f$-valid.
    \end{enumerate}
\end{definition}

\subsection{$(\eps, \delta)$-DP IMs and $\RR_{\eps, \delta}$}\label{subsec:composition-comp-lyu}

Lyu~\cite{lyu2022composition} and Vadhan and Zhang~\cite{vadhan2022concurrent} show that for every 
$(\eps,\delta)$-DP IM $\cM$
and every pair of neighboring initial states $s_0$ and $s_1$, 
there exists an IPM $\cT$ (with a single initial state) such that for each $b\in\{0,1\}$, the mechanisms $\cT\circstar\RR_{\eps,\delta}(b)$ and $\cM(s_b)$ are equivalent. 
Although not explicitly stated, the construction of $\cT$ proposed in~\cite{lyu2022composition} requires the answer sets of $\cM(s_0)$ and $\cM(s_1)$ to be finite. Vadhan and Zhang~\cite{vadhan2022concurrent} make an even stronger assumption that the query set must be finite. We instead assume that the answer distribution is discrete, which is formally defined as follows:

\begin{definition}
    An IM $\cM$ has \emph{discrete answer distributions} if $A_\cM$ is countable, and for every initial state $s$, integer $t\in\N$, history $(q_1, a_1, \dots, q_{t-1}, a_{t-1})\in (Q_\cM\times A_\cM)^{t-1}$, and new query $q_t\in Q_\cM$, the distribution of $\cM(s)$'s answer in response to $q_t$ conditioned on having the history $(q_1, \dots, a_{t-1})$ has a probability mass function (PMF).
\end{definition}

This is generally not a restrictive assumption in differential privacy. In practice, the binary nature of computers and the fact that the total number of memory bits in the world is finite ensure that every distribution sampled computationally is inherently discrete. Theoretically, when dealing with real-valued outputs, one can enforce discrete answer distributions by rounding each output to the nearest multiple of a sufficiently small constant $\alpha$. This discretization post-processing preserves the privacy of the mechanism and has a negligible effect on the accuracy. An example of this approach will be provided in Section~\ref{sec:application}.

Lyu~\cite{lyu2022composition} additionally assumes that there exists a pre-fixed upper bound $c$ on the number of queries $\cM$ answers. This assumption stems from their use of a backward-defined control function. By employing the limit control functions which we introduce in Section~\ref{sec:post-irr}, one can modify their proof to remove this additional assumption. Moreover, using the formalization of Section~\ref{sec:post-irr} to define recursive functions, we can relax the requirement that the answer set is finite to allow countably infinite answer sets. 

Let $\cM'$ be a continual mechanism that is $(\eps, \delta)$-DP w.r.t. a verification function $f$. By Corollary~\ref{cor:eq-def-dp-cm}, $\V[f]\circstar\I\circstar\cM'$ is an IM that is $(\eps, \delta)$-DP w.r.t. $\sim_{01}$. We have:

\begin{lemma}[slight extension of \cite{lyu2022composition,vadhan2022concurrent}]\label{lem:cm-post-rr-lyu}
    For $\eps> 0$ and $0\leq \delta\leq 1$, let $\cM$ be a CM that is $(\eps, \delta)$-DP w.r.t. a verification function $f$. Suppose $\V[f]\circstar\I\circstar\cM$ has discrete answer distributions. Then, there exists an IPM $\cT$ with a single initial state such that for each $b\in\zo$, the IMs $\cT\circstar\RR_{\eps, \delta}(b)$ and $\V[f]\circstar\I(b)\circstar\cM$ are equivalent.
\end{lemma}

\begin{remark}\label{rem:at-most-one-interaction}
    Since $\RR_{\eps, \delta}$ is a NIM, in Lemma~\ref{lem:cm-post-rr-lyu}, the IPM $\cT$ interacts with $\RR_{\eps, \delta}$ at most once.
\end{remark}

\subsection{Concurrent Composition of CMs $\cM_1, \dots, \cM_k$}\label{subsec:composition-comp-fixed-cm}
For $1\leq i\leq k$, let $\cM_i$ be a CM with the initial state space $S_i=\{s_i\}$ such that $\cM_i$ is $(\eps_i, \delta_i)$-DP w.r.t. a verification function $f_i$, and $\V[f_i]\circstar\I\circstar\cM_i$ has discrete answer distributions. In the concurrent composition of $\cM_1, \dots, \cM_k$, an adversary is restricted to the following steps: (1) It first creates the fixed mechanisms $\cM_1, \dots, \cM_k$ and (2) then adaptively issues pairs of queries for them such that the sequences of query pairs for each $\cM_i$ is $f_i$-valid. Note that creating the CM $\cM_i$ means creating an instance of $\cM_i(s_i)$, where $s_i$ is the unique initial state of $\cM_i$. Based on a secret bit $b\in\zo$, each CM either receives the first element of every query pair addressed to it or the second element, and the adversary tries to guess $b$ using the answers of all $k$ mechanisms.

To formalize this type of concurrent composition of CMs, we introduce the verification function $f^{\alpha_1, \dots, \alpha_k}$. Intuitively, this verification function enforces that the adversary creates the mechanisms $\cM_1, \dots, \cM_k$ in their first $k$ messages and then only sends messages containing pairs of queries for $\cM_i$s. Additionally, it ensures that the sequence of query pairs for each $\cM_i$ is $f_i$-valid. The formal definition of $f^{\alpha_1, \dots, \alpha_k}$ is as follows:

\begin{definition}[$f^{\alpha_1, \dots, \alpha_k}$]\label{def:ver-func-comp-fixed-cm}
    For $i\in[k]$, let $\alpha_i=(\cM_i, s_i, \eps_i, \delta_i, f_i)$ denote a creation query.
    For every $t\in\N$ and every input sequence of messages $m_1, \dots, m_t$, the verification function $f^{\alpha_1, \dots, \alpha_k}$ returns $\top$ if and only if all of the following conditions hold:
    \begin{enumerate}
        \item [(i)] The sequence $m_1, \dots, m_k$ is suitable. For $j\in [t]$, let $m_j=(m_j^0, m_j^1)$.
        \item [(ii)] For $j\leq \min\{t, k\}$, the messages $m_j^0=m_j^1=\alpha_j$.
    \end{enumerate}
\end{definition}

With the above definition of $f^{\alpha_1, \dots, \alpha_k}$, we refer to the $f^{\alpha_1, \dots, \alpha_k}$-concurrent composition of CMs as the \emph{concurrent composition of the continual mechanisms $\cM_1, \dots, \cM_k$}. Therefore, the concurrent composition of $\cM_1, \dots, \cM_k$ is $(\eps,\delta)$-DP against adaptive adversaries if $\extconcomp$ is $(\eps,\delta)$-DP w.r.t. $f^{\alpha_1, \dots, \alpha_k}$ against adaptive adversaries.

Recall that in Section~\ref{sec:cm}, we generalized the notion of neighboring relations for interactive mechanisms to verification functions for continual mechanisms. The neighboring initial states for an interactive mechanism are chosen non-adaptively at initialization. However, the choice of messages forming a valid sequence w.r.t. a verification function is adaptive and happens over time. That is, an adversary generates two different query sequences for a CM, and the choice of differing sequences is made adaptively throughout the interaction.

In the concurrent composition of IMs, queries to each mechanism are chosen adaptively, not only based on that mechanism’s previous responses but also using the responses of the other mechanisms up to that point. By definition, \textit{the choice of neighboring initial states for the interactive mechanisms is non-adaptive}. However, in the concurrent composition of CMs, \textit{the choice of differences between query sequences, i.e., the neighboring query sequences} is adaptive, based not only on the responses of the corresponding mechanism but also on the responses of other mechanisms. (Note that, unlike IMs, CMs inherently have a single fixed initial state; hence, no choice of initial state is involved.) Despite this additional adaptivity, in the following theorem, we show that the privacy results of~\cite{lyu2022composition,vadhan2021concurrent} for $\concomp[\cM_1, \dots, \cM_k]$ also hold for $\V[f^{\alpha_1, \dots, \alpha_k}]\circstar\I\circstar\extconcomp$. 

\begin{theorem}\label{thm:comp-fixed-cm}
    For $1\leq i\leq k$, $\eps_1, \dots, \eps_k>0$, and $0\leq \delta_1, \dots, \delta_k\leq 1$, let $\cM_i$ be a continual mechanism with initial state space $S_{\cM_i}^\init=\{s_i\}$ such that $\cM_i$ is $(\eps_i, \delta_i)$-DP w.r.t. a verification function $f_i$ against adaptive adversaries, and $\V[f_i]\circstar\I\circstar\cM_i$ has discrete answer distributions. For $\eps\geq 0$ and $0\leq \delta\leq 1$, if the composition of non-interactive mechanisms $\RR_{\eps_1, \delta_1}, \dots, \RR_{\eps_k, \delta_k}$ is $(\eps, \delta)$-DP, then the concurrent composition of $\cM_1, \dots, \cM_k$ is also $(\eps, \delta)$-DP against adaptive adversaries.
\end{theorem}

\begin{proof} 
Let $\alpha_i=(\cM_i, s_i, \eps_i, \delta_i, f_i)$ for each $i\in[k]$, and let $f$ denote $f^{\alpha_1, \dots, \alpha_k}$ for simplicity. Our goal is to show that $\extconcomp$ is $(\eps, \delta)$-DP w.r.t. $f$. To prove so, we construct an IPM $\cP$ with a single initial state and show that for each $b\in\zo$, the mechanisms $\V[f]\circstar\I(b)\circstar\extconcomp$ and $\V[f]\circstar\cP\circstar\comp[\RR_{\eps_1, \delta_1}, \dots, \RR_{\eps_k, \delta_k}]\left((b)_{i=1}^k\right)$ are equivalent. 

\sloppy
Since $\V[f] \circstar\cP$ is an IPM with a single initial state, Lemma~\ref{lem:post-im} implies that if $\comp[\RR_{\eps_1, \delta_1}, \dots, \RR_{\eps_k, \delta_k}]$ is $(\eps, \delta)$-DP, then $\V[f]\circstar\cP\circstar\comp[\RR_{\eps_1, \delta_1}, \dots, \RR_{\eps_k, \delta_k}]$ is also $(\eps, \delta)$-DP. Hence, for every adversary $\cA$, the views of $\cA$ when interacting with the mechanisms $\V[f]\circstar\cP\circstar\comp[\RR_{\eps_1, \delta_1}, \dots, \RR_{\eps_k, \delta_k}]\left((0)_{i=1}^k\right)$ and $\V[f]\circstar\cP\circstar\comp[\RR_{\eps_1, \delta_1}, \dots, \RR_{\eps_k, \delta_k}]\left((1)_{i=1}^k\right)$ are $(\eps, \delta)$-indistinguishable. As we will show below, by construction, these views are identical to those obtained by interacting with the original mechanisms $\V[f]\circstar\I(b)\circstar\extconcomp$, implying that the concurrent composition of $\cM_1, \dots, \cM_k$ is $(\eps, \delta)$-DP.

In $\V[f]\circstar\I(b)\circstar\extconcomp$, when $\I(b)$ receives pairs of identical creation queries from $\V[f]$, it forwards one of them to $\extconcomp$, which then creates an instance of $\cM_i$ and returns $\top$. Subsequently, when $\I(b)$ receives two queries for $\cM_i$, it forwards them to $\cM_i$ through $\extconcomp$ and returns its response. This procedure can be equivalently represented by an IM $\iecc(b)$ as follows: when receiving pairs of identical creation queries from $\V[f]$, $\iecc(b)$ creates an instance of $\V[f_i]\circstar\I(b)\circstar\cM_i$ and returns $\top$. Subsequently, when receiving $\left((q_0, i), (q_1, i)\right)$, it forwards $(q_0, q_1)$ to $\V[f_i]\circstar\I(b)\circstar\cM_i$ and returns its response. By definition of $f$ and $\I$, the mechanisms $\V[f]\circstar\I(b)\circstar\extconcomp$ and $\V[f]\circstar\iecc(b)$ are equivalent. Next, we define $\cP$ explicitly and show $\V[f]\circstar\iecc(b)$ and $\V[f]\circstar\cP\circstar\comp[\RR_{\eps_1, \delta_1}, \dots, \RR_{\eps_k, \delta_k}]\left((b)_{i=1}^k\right)$ are equivalent. 

Let $\cT_i$ be the IPM for $\cM_i$ given by Lemma~\ref{lem:cm-post-rr-lyu}, and let $t_i$ be its single initial state. The initial state of $\cP$ is a pair of empty sequences. Let $(\sigma_1, \sigma_2)$ denote the state of $\cP$. When receiving the first left message, $\cP$ interacts with its right mechanism, $\comp[\RR_{\eps_1, \delta_1}, \dots, \RR_{\eps_k, \delta_k}]\left((b)_{i=1}^k\right)$, and sets $\sigma_1$ to its output sequence $(r_1, \dots, r_k)$. $\cP$ never interacts with its right mechanism again. 

Upon receiving a left message $(\cM_i, s_i, \eps_i, \delta_i, f_i)^2$, $\cP$ adds $(\cT_i, t_i)$ to $\sigma_2$ and returns $\top$. Moreover, upon receiving a left message $\left((q_0, i), (q_1, i)\right)$, $\cP$ sends $(q_0, q_1)$ as a left message to $\cT_i$ and returns its response. If $\cT_i$ requests a sample from $\RR_{\eps_i, \delta_i}$, $\cP$ provides it with the previously stored sample $r_i$. Formally, $\cP$ executes $(t_i, v, m)\gets\cT_i\left(t_i, \ttL, (q_0, q_1)\right)$ and as long as $v=\ttR$, it runs $(t_i, v, m)\gets\cT_i\left(t_i, v, r_i\right)$. Once $v=\ttL$, $\cP$ returns the message $m$ to $\V[f]$. (By Remark~\ref{rem:at-most-one-interaction}, $v=\ttR$ at most once.) By Lemma~\ref{lem:cm-post-rr-lyu}, each $\cT_i$, provided with the sample $r_i$ from $\RR_{\eps_i, \delta_i}$, simulates $\V[f_i]\circstar\I(b)\circstar\cM_i$ identically. Hence, by definition of $\iecc$ and $\cP$, the mechanisms $\iecc(b)$ and $\cP\circstar\comp[\RR_{\eps_1, \delta_1}, \dots, \RR_{\eps_k, \delta_k}]\left((b)_{i=1}^k\right)$ are equivalent, finishing the proof.
\end{proof}

\subsection{Concurrent Composition of CMs with Fixed Parameters}\label{subsec:composition-fixed-param-cm}
The \emph{concurrent composition of CMs with fixed parameters} generalizes the concurrent composition of fixed CMs by allowing the adversary to select the CMs $\cM_1, \dots, \cM_k$ adaptively over time, instead of forcing them to create a fixed set of mechanisms in their first $k$ messages. Each selected mechanism $\cM_i$ must be $(\eps_i, \delta_i)$-DP w.r.t. a verification function $f_i$, where the multiset of privacy parameters $(\eps_1, \delta_1), \dots, (\eps_k, \delta_k)$ is fixed in advance, but the verification functions $f_i$ are chosen by the adversary. Note that this implies that at most $k$ continual mechanism are created by the adversary. Moreover, due to the aforementioned technical reasons, each $\V[f_i]\circstar\I\circstar\cM_i$ must have discrete answer distributions. 

For instance, an adversary may first create a continual mechanism $\cM_3$ that is $(\eps_3, \delta_3)$-DP w.r.t. a verification function $f_3$. Then, after interacting with $\cM_3$ for some time and analyzing its responses, the adversary may then select a mechanism $\cM_1$ from the class of $(\eps_1, \delta_1)$-DP CMs. Similar to the concurrent composition of fixed CMs, the adversary's sequence of query pairs for each $\cM_i$ must be $f_i$-valid.

To formally define this composition, we introduce a verification function, $f^{\eps_1, \delta_1, \dots, \eps_k, \delta_k}$ as follows:

\begin{definition}[$f^{\eps_1, \delta_1, \dots, \eps_k, \delta_k}$]\label{def:ver-func-fixed-param-comp-cm}
    For every $t\in\N$ and every input sequence of messages $m_1, \dots, m_t$, the function $f^{\eps_1, \delta_1, \dots, \eps_k, \delta_k}$ returns $\top$ if and only if all of the following conditions hold:
    \begin{enumerate}[left=0pt]
        \item [(i)] The sequence $m_1, \dots, m_t$ is suitable.
        \item [(ii)] There are at most $k$ pairs of identical creation queries in the sequence $m_1, \dots, m_t$. Let $\cM_i'$ denote the $i$-th created mechanism and $(\eps_i', \delta_i')$ and $f_i'$ be its privacy parameters and verification function, respectively. $\V[f_i']\circstar\I\circstar\cM_i'$ has discrete answer distributions. Moreover, the multiset of the privacy parameters $(\eps_i', \delta_i')$ of the created mechanisms is a sub-multiset of the fixed privacy parameters $(\eps_1, \delta_1), \dots, (\eps_k, \delta_k)$. Note that $(\eps_i', \delta_i')$ is adaptively chosen by the adversary, while $(\eps_i, \delta_i)$ is a predetermined publicly-known parameter.
    \end{enumerate}
\end{definition}

We refer to the $f^{\eps_1, \delta_1, \dots, \eps_k, \delta_k}$-concurrent composition of CMs as the \emph{concurrent composition of CMs with privacy parameters $\left((\eps_i, \delta_i)\right)_{i=1}^k$}. A slight adaptation of the proof of Theorem~\ref{thm:comp-fixed-cm} implies the following theorem. 

\begin{theorem}\label{thm:comp-fixed-param-cm}
    For $\eps_1, \dots, \eps_k> 0$, $\eps\geq 0$ and $0\leq \delta, \delta_1, \dots, \delta_k\leq 1$, if the composition of non-interactive mechanisms $\RR_{\eps_1, \delta_1}, \dots, \RR_{\eps_k, \delta_k}$ is $(\eps, \delta)$-DP, then the concurrent composition of $k$ CMs with the privacy parameters $\left((\eps_i, \delta_i)\right)_{i=1}^k$ is also $(\eps, \delta)$-DP against adaptive adversaries.
\end{theorem}
\begin{proof}
    The proof is identical to the proof of Theorem~\ref{thm:comp-fixed-cm} with a modification in the construction of $\cP$: Recall that $(r_1, \dots, r_k)$ denotes the output of $\comp[\RR_{\eps_1, \delta_1}, \dots, \RR_{\eps_k, \delta_k}]$. In the proof of Theorem~\ref{thm:comp-fixed-cm}, the $i$-th creation query was known in advance to be $(\cM_i, s_i, \eps_i, \delta_i, f_i)$. Consequently, $\cT_i$ expected $\RR_{\eps_i, \delta_i}$ as its right mechanism, and when $\cT_i$ needed a response from its right mechanism, $\cP$ sent $r_i$ to it. However, in the concurrent composition with fixed parameters, the IPM requesting the sample from $\RR_{\eps_i, \delta_i}$ need not necessarily be the $i$-th IPM, $\cT_i$. To handle this, we modify $\cP$ so that it maps each output sample $r_i$ to the appropriate IPM $\cT_j$ needing a sample from $\RR_{\eps_i, \delta_i}$. The implementation details are as follows.

    Let $(\cM_j', s_j', \eps_j', \delta_j', f_j')^2$ denote the $j$-th pair of identical creation queries received by $\cP$ from the verifier, and let $\cT_j$ be the corresponding IPM given by Lemma~\ref{lem:cm-post-rr-lyu}. In the proof of Theorem~\ref{thm:comp-fixed-cm}, the state of $\cP$ consists of a pair of two sequences $\sigma_1=(r_1, \dots, r_k)$ and $\sigma_2=\left((\cT_j, t_j)\right)_{j=1}^k$, where $r_j$ is the randomized response defined above and $t_j$ denotes the current state of $\cT_j$. We extend the state of $\cP$ to additionally store an index $y_i$ with each $r_i$ in $\sigma_1$ and the privacy parameters $\eps_j', \delta_j'$ with each $\cT_j$ in $\sigma_2$. Variable $y_i$ is initialized to $0$ and used to record the index $j$ of the mechanism $\cT_j$ associated with $r_i$

    Upon receiving $(r_1, \dots, r_k)$ from $\comp[\RR_{\eps_1, \delta_1}, \dots, \RR_{\eps_k, \delta_k}]$, $\cP$ sets $\sigma_1=\left((r_1, y_1=0), \dots, (r_k, y_k=0)\right)$. Also, upon receiving a pair of creation queries $(\cM_j', s_j', \eps_j', \delta_j', f_j')^2$, $\cP$ appends the tuple $(\cT_j, t_j, \eps_j', \delta_j')$ to $\sigma_2$. When $\cT_j$ requests a sample from $\RR_{\eps_j', \delta_j'}$, $\cP$ iterates over $\sigma_1$ to find the smallest $i^*\in[k]$ such that $y_{i^*}=0$, $\eps_j'=\eps_{i^*}$, and $\delta_j'=\delta_{i^*}$. It then assigns $y_{i^*}=j$ and sends $r_{i^*}$ to $\cT_j$ as a right response. 
    This adjustment ensures that even when there are duplicate privacy parameters among $(\eps_1, \delta_1), \dots, (\eps_k, \delta_k)$, each sample $r_i$ is sent to at most one $\cT_j$. Note that by the definition of $f^{\eps_1, \delta_1, \dots, \eps_k, \delta_k}$, $i^*$ always exists. 
\end{proof}

\section{Concurrent Parallel Composition of CMs}\label{sec:parallel}
The \emph{concurrent $k$-sparse parallel composition} extends the concurrent composition of CMs with $k$ fixed parameters by permitting the creation of an \textit{arbitrary} number of CMs. Specifically, given a multiset $\params$ of $k$ privacy parameters $(\eps_1, \delta_1), \dots, (\eps_k, \delta_k)$, the adversary can create an unlimited number of CMs with arbitrary privacy parameters. However, the adversary must ask pairs of identical queries to all created mechanisms, except for at most $k$ of them. Let $\cM_1, \dots, \cM_k$ be these exceptional mechanisms. Each $\cM_i$, selected by the adversary, must be $(\eps_i, \delta_i)$-DP w.r.t. a verification function $f_i$ and with $(\eps_i, \delta_i) \in {\params}$, and the sequence of query pairs for $\cM_i$ must be $f_i$-valid.

In Section~\ref{subsec:parallel-def-and-counter}, we formally define this composition and show that the result of Theorem~\ref{thm:comp-fixed-param-cm} cannot be extended to the concurrent $k$-sparse parallel composition of CMs when any $\delta_j > 0$. In Section~\ref{subsec:parallel-approx-dp}, we prove a parallel composition theorem that matches the privacy guarantee of Theorem~\ref{thm:comp-fixed-param-cm} when only purely differentially privacy CMs are created. Finally, in Section~\ref{subsec:parallel-restricted-ver-func}, we consider a restricted class of verification functions $f_i$ for the $k$ mechanisms receiving neighboring inputs. Under certain assumptions about these functions, we show that the privacy guarantee of Theorem~\ref{thm:comp-fixed-param-cm} still holds even if the created mechanism are approximately dp.

\subsection{Definition and Counterexample}\label{subsec:parallel-def-and-counter}
To formally define the concurrent $k$-sparse parallel composition of CMs with parameters $\left((\eps_i, \delta_i)\right)_{i=1}^k$, we introduce the verification function $f^{\eps_1, \delta_1, \dots, \eps_k, \delta_k}_\infty$, where the subscript $\infty$ emphasizes that an unbounded number of CMs can be created.

\begin{definition}[$f^{\eps_1, \delta_1, \dots, \eps_k, \delta_k}_\infty$]\label{def:ver-func-par-fixed-param-comp-cm}
    For every $t\in\N$ and every input sequence of messages $m_1, \dots, m_t$, the function $f^{\eps_1, \delta_1, \dots, \eps_k, \delta_k}_\infty$ returns $\top$ if and only if all of the following conditions hold:
    \begin{enumerate}[left=0pt]
        \item [(i)] The sequence $m_1, \dots, m_t$ is suitable.
        \item [(ii)] Let $\cM_i'$ denote the $i$-th created mechanism and $(\eps_i', \delta_i')$ and $f_i'$ be its privacy parameters and verification function, respectively. Define $\mathcal{J}$ as the set of indices $j$ such that there exists a message $\left((q_0, j), (q_1, j)\right)$ where $q_0\neq q_1$. Let $\params_\mathcal{J}$ denote the multiset consisting of privacy parameters $(\eps_j', \delta_j')$ for each $j\in\mathcal{J}$. This condition requires that $\params_\mathcal{J}$ be a sub-multiset of the fixed privacy parameters $(\eps_1, \delta_1), \dots, (\eps_k, \delta_k)$. Note that this implies that $|\mathcal{J}| \le k.$
    \end{enumerate}
\end{definition}

We refer to the $f^{\eps_1, \delta_1, \dots, \eps_k, \delta_k}_\infty$-concurrent composition of CMs as the \emph{concurrent $k$-sparse parallel composition of CMs with privacy parameters $\left((\eps_i, \delta_i)\right)_{i=1}^k$}. Therefore, this composition is $(\eps,\delta)$-DP against adaptive adversaries if $\extconcomp$ is $(\eps,\delta)$-DP w.r.t. $f^{\eps_1, \delta_1, \dots, \eps_k, \delta_k}_\infty$ against adaptive adversaries. 

The following theorem shows that for $k=1$, $\eps_1=0$, and any $\delta_1\in(0, 1]$, the $f^{0, \delta_1}_\infty$-concurrent composition of CMs is not differentially private:

\begin{theorem}\label{thm:counter-example}
    For every $\delta\in (0, 1]$, there exists an adversary $\cA_\delta$ such that the views of $\cA_\delta$ interacting with $\V[f^{0, \delta}]\circstar\I(0)\circstar\extconcomp$ and $\V[f^{0, \delta}]\circstar\I(1)\circstar\extconcomp$ have disjoint supports.
\end{theorem}

\begin{proof}
    We construct $\cA_\delta$ in three steps: (1) introducing a verification function $g$, (2) constructing a CM $\cM_\delta$ that is $(0, \delta)$-DP w.r.t. $g$, and (3) designing $\cA_\delta$. We then prove that for every $b\in\zo$, with probability $1$, $\cA_\delta$ can determine the secret initial state $b$ by interacting with $\V[f^{0, \delta}]\circstar\I(b)\circstar\extconcomp$ finitely many times.  Specifically, we show that with probability $1$, the view of $\cA_\delta$ interacting with $\V[f^{0, \delta}]\circstar\I(b)\circstar\extconcomp$ is finite and includes $b$ as the last message. This immediately implies that the views of $\cA_\delta$ interacting with $\V[f^{0, \delta}]\circstar\I(0)\circstar\extconcomp$ and $\V[f^{0, \delta}]\circstar\I(1)\circstar\extconcomp$ have disjoint supports.

    \medskip\noindent\underline{Step 1 (Introducing $g$):} 
    The verification function $g$ accepts a message sequence $(m_1, \dots, m_t)$ and returns $\top$ if and only if (i) $t\leq 2$, (ii) each $m_j$ is a pair of binary numbers $(b_j, b_j')\in \zo^2$, and (ii) the sequences $(b_1, \dots, b_t)$ and $(b_1', \dots, b_t')$ differ in at most a single element.
    
    \medskip\noindent\underline{Step 2 (Constructing $\cM_\delta$):} The CM $\cM_\delta:\{\bot, \top\}\times \zo\to\{\bot, \top\}\times \{0,1,\top, \bot\}$ with the initial state space $\{\top\}$ is defined as follows:
    
    \begin{itemize}[left=0pt]
        \item $\cM_\delta$ maps state $\top$ and an input message $b\in\zo$ to state $\top$ and output message $\top$ with probability $1-\delta$ and to state $\bot$ and output message $\bot$ with probability $\delta$.
        \item $\cM_\delta$ maps state $\bot$ and an input message $b\in\zo$ to state $\bot$ and output message $b$ with probability $1$.
    \end{itemize}
    
    Thus, $\cM_\delta$ has two possible states, $\top$ and $\bot$, and accepts binary queries. Initially, $\cM_\delta$ starts in state $\top$. On the first query, independent of the value of $b$, it switches to state $\bot$ with probability $\delta$ and returns its final state as an output message. On the second query, if the state is $\top$, the same as the first query, it switches to state $\bot$ with probability $\delta$ and returns the final state as output. However, if the state is $\bot$, $\cM_\delta$ outputs the exact input message $b$.
    
    Since with probability $1 - \delta$, the first two outputs of $\cM_\delta$ are independent of the input, $\cM_\delta$ is $(0, \delta)$-DP w.r.t. the verification function $g$.
    
    \medskip\noindent\underline{Step 3 (Designing $\cA_\delta$):}
    In the interaction with $\V[f^{0, \delta}]\circstar\I(b)\circstar\extconcomp$, the adversary $\cA_\delta$ must determine the secret bit $b$. $\cA_\delta$ accepts messages in $\{0,1,\top,\bot\}$ and sends the following messages: (a) the message $(\cM_\delta, \top, 0, \delta, g)^2$, which is a pair of identical creation queries requesting the creation of $\cM_\delta$, and (b) messages of the form $\left((b_0, j), (b_1, j)\right)$, where $b_0, b_1\in\zo$ and $j\in \N$. 
    
    $\cA_\delta$ starts the communication by sending $(\cM_\delta, \top, 0, \delta, g)^2$ and setting its state $s=1$. While receiving $\top$ as a response, it increments $s$ and alternates its queries:
    \begin{itemize}[left=0pt]
        \item If $s$ is odd, it sends $(\cM_\delta, \top, 0, \delta, g)^2$.
        \item If $s$ is even, it sends $\left((0, s/2), (0, s/2)\right)$.
    \end{itemize}
    In this iterative procedure, $\cA_\delta$ keeps creating a new instance of $\cM_\delta$ and asking the query pair $(0,0)$ from it until one instance returns $\bot$. Since $\extconcomp$ always returns $\top$ when creating a mechanism, this loop stops upon receiving a respond to $\left((0, s/2), (0, s/2)\right)$. 
    
    Once $\cA_\delta$ receives $\bot$, without increasing $s$, it sends $\left((0, s/2), (1, s/2)\right)$. That is, $\cA_\delta$ asks the pair $(0,1)$ from the last created instance of $\cM_\delta$, which has returned $\bot$. If the initial state of the identifier $\I$ in $\V[f^{0, \delta}]\circstar\I\circstar\extconcomp$ is $b$, then $\cM_\delta$ receives the query $b$, and by definition it will return $b$ as output, revealing the secret bit $b$. Upon receiving the response $b$, $\cA_\delta$ halts the communication.
    
    Therefore, with probability $1-(1-\delta)^t$, the adversary $\cA_\delta$ can detect the secret bit $b$ by sending at most $2t+1$ queries. Since $\delta>0$, as $t$ goes to infinity, this probability tends to $1$. Thus, with probability $1$, the view of $\cA_\delta$ interacting with $\V[f^{0, \delta}]\circstar\I(b)\circstar\extconcomp$ is finite and includes $b$ as the last message.
\end{proof}

\subsection{Composition Theorem}\label{subsec:parallel-approx-dp}
In the counterexample of Section~\ref{subsec:parallel-def-and-counter}, the adversary runs an unbounded number of $(0,\delta)$-DP CMs but sends a pair of non-identical queries to only one of them, leading to the exposure of the secret bit with probability $1$. In that construction, each individual mechanism may fail to preserve privacy with probability~$\delta$, and these failures accumulate, ultimately leading to complete exposure. To prevent such leakage, we restrict the second privacy parameters $\delta_j'$ of all created mechanisms to satisfy a certain condition.

\begin{definition}[$f^{\eps_1, \dots, \eps_k}_{\infty, \delta'}$]\label{def:ver-func-par-approx}
    For every $t\in\N$ and every input sequence of messages $m_1, \dots, m_t$, the function $f^{\eps_1, \dots, \eps_k}_{\infty, \delta}$ returns $\top$ if and only if all of the following conditions hold:
    \begin{enumerate}[left=0pt]
        \item [(i)] The sequence $m_1, \dots, m_t$ is suitable.
        \item [(ii)] Let $\cM_i'$ denote the $i$-th created mechanism and $(\eps_i', \delta_i')$ be its privacy parameters. Define $\mathcal{J}$ as the set of indices $j$ such that there exists a message $\left((q_0, j), (q_1, j)\right)$ where $q_0\neq q_1$. Let $\params_\mathcal{J}$ denote the multiset consisting of the first privacy parameters $\eps_j'$ for each $j\in\mathcal{J}$. This condition requires that $\params_\mathcal{J}$ be a sub-multiset of the fixed privacy parameters $\eps_1, \dots, \eps_k$.
        \item [(iii)] The inequality $1-\prod_j(1-\delta_j')\leq \delta'$ holds, where the product is taken over the indices of the created mechanisms.
    \end{enumerate}
\end{definition}

We refer to the $f^{\eps_1, \dots, \eps_k}_{\infty, \delta}$-concurrent composition of CMs as the \emph{concurrent $k$-sparse parallel composition of CMs with privacy parameters $\left((\eps_i, \delta_i)\right)_{i=1}^k$ and upper bound $\delta$}.
We show next the following theorem. 

\begin{theorem}\label{thm:parallel-comp-approx}
    For $\eps_1, \dots, \eps_k> 0$, $\eps\geq 0$ and $0\leq \delta, \delta'\leq 1$, if the composition of non-interactive mechanisms $\RR_{\eps_1, 0}, \dots, \RR_{\eps_k, 0}$ is $(\eps, \delta)$-DP, then the $f^{\eps_1, \dots, \eps_k}_{\infty, \delta'}$-concurrent  composition of CMs is $(\eps, \delta+\delta')$-DP against adaptive adversaries.
\end{theorem}

It is known that the composition of $\RR_{\eps_1, 0}, \dots, \RR_{\eps_k, 0}$ is $\sum_{i=1}^k\eps_i$-DP. Thus, by converging all parameters $\eps_1, \dots, \eps_k$ to zero, we conclude that the $f^{0, \dots, 0}_{\infty, \delta'}$-concurrent composition of CMs is $(0, \delta')$-DP. The proof of Theorem~\ref{thm:counter-example} shows that this composition cannot be $(0, \delta'')$-DP for $\delta''<\delta'$, and thus the result of Theorem~\ref{thm:parallel-comp-approx} is \emph{tight}.

Setting $\delta'=0$ in Theorem~\ref{thm:parallel-comp-approx} means that all created mechanisms must be purely differentially private:
\begin{corollary}\label{cor:parallel-comp-pure}
    For $\eps_1, \dots, \eps_k> 0$, $\eps\geq 0$ and $0\leq \delta\leq 1$, if the composition of non-interactive mechanisms $\RR_{\eps_1, 0}, \dots, \RR_{\eps_k, 0}$ is $(\eps, \delta)$-DP, then the $f^{\eps_1, \dots, \eps_k}_{\infty, 0}$-concurrent composition of CMs is $(\eps, \delta)$-DP against adaptive adversaries. That is, the concurrent $k$-sparse parallel composition of purely differentially private CMs with privacy parameters $\left((\eps_i, \delta_i)\right)_{i=1}^k$ is $(\eps, \delta)$-DP against adaptive adversaries.
\end{corollary}

In the proofs of the previous composition theorems (namely, Theorems~\ref{thm:comp-fixed-cm} and~\ref{thm:comp-fixed-param-cm}), each created CM $\cM_j'$ with privacy parameters $(\eps_j', \delta_j')$ is simulated by the IPM $\cT_j$ from Lemma~\ref{lem:cm-post-rr-lyu}. For each $b\in\zo$, the mechanism $\cT_j$ interacts with $\RR_{\eps_j',\delta_j'}(b)$ once and simulates $\V[f_j']\circstar\I(b)\circstar\cM_j'$ identically. Next, we present a new construction for $\cT_j$, where it interacts with $\irr_{\eps_j', \delta_j'}$ instead of $\RR_{\eps_j',\delta_j'}$. This construction satisfies additional properties needed in the proof of Theorem~\ref{thm:parallel-comp-approx}. In Section~\ref{sec:post-irr}, we prove the following lemma, which directly implies the desired IPM as stated in Corollary~\ref{cor:dp-cm-post-irr}.

\begin{lemma}\label{lem:new}
    For $\eps > 0$ and $0 \leq \delta \leq 1$, let $\cM:S_\cM\times Q_\cM\to S_\cM\times A_\cM$ be an $(\eps, \delta)$-DP IM w.r.t. a neighbor relation $\sim$. Suppose that $\cM$ has discrete answer distributions. Then, for every two neighboring initial states $s_0$ and $s_1$, there exists an IPM $\cP$ such that for each $b\in \zo$:
    \begin{itemize}
        \item The mechanisms $\cM(s_b)$ and $\cP \circstar \irr_{\eps, \delta}(b)$ are equivalent.
        \item For every $k \in \N$ and every query sequence $(q_1, \dots, q_k) \in Q_{\cM}^k$, if the distributions of answers produced by $\cM(s_0)$ and $\cM(s_1)$ to $(q_1, \dots, q_k)$ are identical, then when $\cP \circstar \irr_{\eps, \delta}(b)$ receives the queries $q_1, \dots, q_k$, the IPM $\cP$ interacts with $\irr_{\eps, \delta}(b)$ only once. Otherwise, $\cP$ interacts with $\irr_{\eps, \delta}(b)$ at most twice.
    \end{itemize}
\end{lemma}

To avoid confusion in later proofs, we replace the IPM $\cP$ in Lemma~\ref{lem:new} with $\cT$ in the following corollary.

\begin{corollary}\label{cor:dp-cm-post-irr}
    For $\eps > 0$ and $0 \leq \delta \leq 1$, let $\cM$ be a CM that is $(\eps, \delta)$-DP w.r.t. a verification function $f$. Suppose $\V[f]\circstar\I\circstar\cM$ has discrete answer distributions. Then there exists an IPM $\cT$ such that for each $b\in \zo$:
    \begin{itemize}
        \item The IMs $\V[f]\circstar\I(b)\circstar\cM$ and $\cT\circstar \irr_{\eps, \delta}(b)$ are equivalent.
        \item For every $k \in \N$ and every message sequence $m_1, \dots, m_k$ such that each message $m_j$ is a pair of identical messages, if $\cT\circstar \irr_{\eps, \delta}(b)$ receives $m_1, \dots, m_k$ as input, then the IPM $\cT$ interacts with $\irr_{\eps, \delta}(b)$ at most once.
        \item $\cT$ sends at most two messages to its right mechanism.
    \end{itemize}
\end{corollary}

\begin{proof}
    By Corollary~\ref{cor:eq-def-dp-cm}, $\V[f]\circstar\I\circstar\cM$ is an IM that is $(\eps, \delta)$-DP w.r.t. $\sim_{01}$. By definition of the identifier $\I$, for every $k \in \N$ and every message sequence $m_1, \dots, m_k$, if each message $m_j$ is a pair of identical messages, then the distributions of answers produced by $\V[f]\circstar\I(0)\circstar\cM$ and $\V[f]\circstar\I(1)\circstar\cM$ in respond to $(m_1, \dots, m_k)$ are identical. Applying Lemma~\ref{lem:new} to the IM $\V[f]\circstar\I\circstar\cM$ then yields the desired result.
\end{proof}

Informally, by Definition~\ref{def:irr}, Corollary~\ref{cor:dp-cm-post-irr} states that each created CM $\cM_j'$, with the privacy parameters $(\eps_j', \delta_j')$, 
can be simulated by post-processing the outcomes of two randomized response mechanisms, namely $\RR_{\eps_j', 0}$ and $\RR_{0, \delta_j'}$. 
While simulating $\cM_j'$ in general requires access to both outputs, by this lemma, as long as the adversary asks pairs of identical queries from $\cM_j'$, the output of $\RR_{\eps_j', 0}$ is unnecessary. 
By definition of $f^{\eps_1, \dots, \eps_k}_{\infty, \delta'}$, the multiset of the parameters $\eps_j'$ of the mechanisms receiving non-identical query pairs is a sub-multiset of $\eps_1, \dots, \eps_k$. The high-level idea behind the proof of Theorem~\ref{thm:parallel-comp-approx} is to separately compose $\RR_{\eps_1, 0}, \dots, \RR_{\eps_k, 0}$ and $\RR_{0, \delta_1'}, \RR_{0, \delta_2'}, \dots$, and then apply basic composition to these two sets of composed mechanisms to conclude the final result. 

Since the parameters $\delta_j'$ are chosen adaptively, the composition of $\RR_{0, \delta_1'}, \RR_{0, \delta_2'}, \dots$ is an instance of the \emph{filter composition} of NIMs~\cite{rogers2016privacy}. We do not state this composition in terms of filters and define it using a verification function $f^{\delta'}_{\RR}$ that only allows the creation of randomized response mechanisms $\RR_{0, \delta_j'}$ and returns $\bot$ if $1-\prod (1-\delta_j')>\delta'$.

\begin{definition}[$f^{\delta}_{\RR}$]\label{def:ver-func-sum-delta-rr}
    For every $t\in\N$ and every input sequence of messages $m_1, \dots, m_t$, the function $f^{\delta}_{\RR}$ returns $\top$ if and only if all of the following conditions hold:
    \begin{enumerate}[left=0pt]
        \item [(i)] The sequence $m_1, \dots, m_t$ is suitable.
        \item [(ii)] Let $\cM_j'$ denote the $j$-th mechanism requested to be created and $(\eps_j', \delta_j')$ and $f_j'$ be its privacy parameters and verification function, respectively. Each $\cM_j'$ is the randomized response mechanism $\RR_{\eps_j', \delta_j'}$ (see Definition~\ref{def:rr}), $f_j'$ is its corresponding verification function, and $\eps_j=0$. By definition of $\RR_{\eps_j', \delta_j'}$, $f_j'$ returns $\top$ if and only if it receives a single message $m$ where $m$ is a pair of bits.
        \item [(iii)] The inequality $1-\prod_j(1-\delta_j')\leq \delta$ holds, where the product is taken over the indices of the created mechanisms.
    \end{enumerate}
\end{definition}

 The following lemma is proved in Section~\ref{sec:filter-comp-rr}.
 \footnote{By setting $\eps_1=\eps_2=\dots = 0$ in Theorem~2 of~\cite{whitehouse2023fully}, it follows that for every $\eps>0$, $\delta'>0$, and $0 \leq\delta''\leq 1$, the composition of randomized response mechanisms $\RR_{0, \delta_1}, \RR_{0, \delta_2}, \dots$ with adaptively chosen privacy parameters $\delta_1, \delta_2, \dots$ satisfying $\sum_i \delta_i \leq \delta''$ guarantees $(\eps, \delta'+\delta'')$-DP. Taking $\eps$ and $\delta'$ to zero, we can conclude that this composition is $(0, \delta'')$-DP. To get tighter bounds, we relax the requirement $\sum_i \delta_i \leq \delta''$ to $1-\prod_i (1-\delta_i)\leq \delta''$ and prove the composition is still $(0, \delta'')$-DP. }

\begin{lemma}\label{lem:filter-comp-eps-eq-zero}
    For every $0\leq\delta\leq 1$, the $f^\delta_{\RR}$-concurrent composition of continual mechanisms is $(0, \delta)$-differentially private.
\end{lemma}

Before proving Theorem~\ref{thm:parallel-comp-approx}, we show one final technical lemma that formally composes the fixed-parameter mechanisms $\RR_{\eps_1, 0}, \dots, \RR_{\eps_k, 0}$ and the adaptively chosen mechanisms $\RR_{0, \delta_1'}, \RR_{0, \delta_2'}, \dots$. To define the composition, we introduce the following verification function:

\begin{definition}[$f^{\eps_1, \dots, \eps_k}_{\infty, \delta, RR}$]
    For every input sequence of messages $m_1, \dots, m_t$, the function $f^{\eps_1, \dots, \eps_k}_{\infty, \delta, RR}$ returns $\top$ if and only if all of the following conditions hold:
    \begin{enumerate}[left=0pt]
        \item [(i)] The sequence $m_1, \dots, m_t$ is suitable.
        \item [(ii)] Let $\cM_j'$ denote the $j$-th created mechanism and $(\eps_j', \delta_j')$ and $f_j'$ be its privacy parameters and verification function, respectively. Each $\cM_j'$ equals the randomized response mechanism $\RR_{\eps_j', \delta_j'}$, and at least one of $\eps_j'$ and $\delta_j'$ is zero.
        \item [(iii)] The multiset of positive $\eps_j'$ is a sub-multiset of the fixed parameters $\eps_1, \dots, \eps_k$.
        \item [(iii)] The inequality $1-\prod_j(1-\delta_j')\leq \delta'$ holds, where the product is taken over the indices of the created mechanisms.
    \end{enumerate}
\end{definition}

\begin{lemma}\label{lem:basic-comp-fixed-and-filter}
    For every $\eps\geq 0$ and $0\leq\delta, \delta'\leq 1$, if $\comp(\RR_{\eps_1, 0}, \dots, \RR_{\eps_k, 0})$ is $(\eps, \delta)$-DP, then the $f^{\eps_1, \dots, \eps_k}_{\infty, \delta', RR}$-concurrent composition of CMs is $(\eps, 1-(1-\delta)(1-\delta'))$-DP. 
\end{lemma}

\begin{proof}
    Let $f^\RR=f^{\eps_1, \dots, \eps_k}_{\infty, \delta', RR}$. We design an IPM $\calF$ with a single initial state and a verification function $f^*$ and show:
    \begin{enumerate}
        \item For each $b\in\zo$, the IMs $\V[f^\RR]\circstar\I(b)\circstar\extconcomp$ and $\V[f^\RR]\circstar\calF\circstar\V[f^*]\circstar\I(b)\circstar\extconcomp$ are equivalent.
        \item $\extconcomp$ is $(\eps, \delta+\delta')$-DP w.r.t. $f^*$. 
    \end{enumerate}
    
    From property (2), it follows that for every adversary $\cA$, the views of $\cA$ interacting with $\V[f^*]\circstar\I(0)\circstar\extconcomp$ and $\V[f^*]\circstar\I(1)\circstar\extconcomp$ are $(\eps, \delta+\delta')$-indistinguishable. Then, by the post-processing lemma, for every adversary $\cA$, the views of $\cA$ interacting with $\V[f^\RR]\circstar\calF\circstar\V[f^*]\circstar\I(0)\circstar\extconcomp$ and $\V[f^\RR]\circstar\calF\circstar\V[f^*]\circstar\I(1)\circstar\extconcomp$ are $(\eps, \delta+\delta')$-indistinguishable. Thus, by property (1), the views of any adversary interacting with $\V[f^\RR]\circstar\I(0)\circstar\extconcomp$ and $\V[f^\RR]\circstar\I(1)\circstar\extconcomp$ are $(\eps, \delta+\delta')$-indistinguishable, implying that the $f^{\eps_1, \dots, \eps_k}_{\infty, \delta', RR}$-concurrent composition of CMs is $(\eps, \delta+\delta')$-DP.

    Intuitively, $\calF$ partitions the randomized responses mechanisms $\RR_{\eps_j', \delta_j'}$ requested to be created into two groups: those with $\eps_j' = 0$ and those with $\delta_j' = 0$. Then, instead of directly requesting $\extconcomp$ to execute the randomized response mechanisms, $\calF$ requests $\extconcomp$, called the {
    \em outer $\extconcomp$}, to run two so-called {\em inner} instances of $\extconcomp$, each responsible for executing one group. Note that $\extconcomp$ can create any continual sub-mechanism including another instance of $\extconcomp$. To implement this, $\calF$ begins by creating two internal instances of $\extconcomp$. It then rewrites incoming messages to ensure they are routed to the appropriate internal $\extconcomp$.

    To formally define $\calF$, we first define the verification function $f^*$: For each $i\in[k]$, let $\beta_i$ denote the creation query for $\RR_{\eps_i, 0}$. Let $\gamma_1=(\extconcomp, (), \eps, \delta, f^{\beta_1, \dots, \beta_k})$ be a creation query for creating an instance of $\extconcomp$ running $\RR_{\eps_1, 0}, \dots, \RR_{\eps_k, 0}$ (see Definition~\ref{def:ver-func-comp-fixed-cm}). Also, let $\gamma_2=(\extconcomp, (), 0, \delta, f^\delta_\RR)$ be a creation query for creating another instance of $\extconcomp$ executing adaptively chosen $\RR_{0, \delta_1'}, \RR_{0, \delta_2'}, \dots$ satisfying $1-\prod (1-\delta_j')\leq \delta'$ (see Definition~\ref{def:ver-func-sum-delta-rr}). We define $f^*=f^{\gamma_1, \gamma_2}$ as in Definition~\ref{def:ver-func-comp-fixed-cm}.
    
    We now define $\calF$. Initially (i.e., upon receiving its first left message), $\calF$ creates the two internal instances of $\extconcomp$ by sending the pairs of identical creation queries $(\gamma_1, \gamma_1)$ and $(\gamma_2, \gamma_2)$ to its right mechanism $\V[f^*]$. $\calF$ ignores the acknowledge responses $\top$. The left query set of $\calF$ equals the set of messages the verifier $\V[f^\RR]$ sends to its right mechanism, which consists of (A) pairs of identical creation queries and (B) messages of the form $\left((c, j), (c',j)\right)$ that requests giving the $j$-th created (randomized response) mechanism one of the bits $c$ and $c'$ as input and returning its response. We denote the $j$-th pair of creation queries by $(\alpha_j, \alpha_j)$ and its corresponding randomized response mechanism by $\RR_{\eps_j',\delta_j'}$. Note that the decision of whether to feed $\RR_{\eps_j',\delta_j'}$ with $c$ or $c'$ is made by the identifier.

    \noindent(A) Let $(\alpha_j, \alpha_j)$ denote the $j$-th pair of identical creation queries received by $\calF$, and let $\eps_j'$ and $\delta_j'$  be the privacy parameters in $\alpha_j$. Upon receiving $(\alpha_j, \alpha_j)$ from the left mechanism $\V[f^\RR]$, $\calF$ takes one of the following actions: 
    \begin{itemize}
        \item If $\eps_j'= 0$, $\calF$ sends the message $\left((\alpha_j, 2), (\alpha_j, 2)\right)$ to its right mechanism. This message passes through an identifier before reaching (the outer) $\extconcomp$. The outer mechanism $\extconcomp$ receives $(\alpha_j, 2)$ and forwards the creation query $\alpha_j$ to its second inner instance of $\extconcomp$.
        \item Otherwise, by definition of $f^\RR$, $\delta_j'=0$. In this case, $\calF$ sends the message $\left((\alpha_j, 1), (\alpha_j, 1)\right)$ to its right mechanism, resulting in the first inner instance of $\extconcomp$ to execute $\alpha_j$.
    \end{itemize}
    $\calF$ maintains an initially empty array in its state, initially empty. When processing $(\alpha_j,\alpha_j)$, it adds the pair $(w_j,t_j)$ to this array, where $w_j\in\{1,2\}$ identifies the inner $\extconcomp$ that has created $\RR_{\eps_j', \delta_j'}$, and $t_j$ indicates the index of the created mechanism within that instance.

    \noindent(B) Upon receiving a left message of the form $\left((c, j), (c',j)\right)$, $\calF$ sends the message $\left(((c, t_j), w_j), ((c', t_j), w_j)\right)$ to the right mechanism, using the stored values $w_j$ and $t_j$.

    Except for the first two acknowledgment messages for creating the internal instances of $\extconcomp$, upon receiving any message from the right, $\calF$ forwards it to the left mechanism unchanged. 

    By the definition of $f^\RR$, the privacy parameters $(\eps_j', \delta_j')$ associated with the creation queries satisfy the following properties: (i) $\eps_j'=0$ for all creation queries $\alpha_j$ except at most $k$ of them, which form a sub-multiset of the predetermined parameters $\eps_1, \dots, \eps_k$, and (ii) the inequality $1-\prod(1-\delta_j')$ holds. Therefore, by the construction of $\calF$, the verifier $\V[f^*]$ in $\V[f^\RR]\circstar\calF\circstar\V[f^*]\circstar\I\circstar\extconcomp$ always receives valid inputs and never halts the communication. Thus, for each $b\in\zo$, the IMs $\V[f^\RR]\circstar\calF\circstar\I(b)\circstar\extconcomp$ and $\V[f^\RR]\circstar\calF\circstar\V[f^*]\circstar\I(b)\circstar\extconcomp$ are equivalent. 

    Furthermore, by the design of $\calF$, for each $b\in\zo$, the IMs $\V[f^\RR]\circstar\calF\circstar\I(b)\circstar\extconcomp$ and $\V[f^\RR]\circstar\I(b)\circstar\extconcomp$ are equivalent: They both use the same verification function to validate messages. Upon receiving a query to create $\RR_{\eps_j', \delta_j'}$, they both create the mechanism and return $\top$. And upon receiving the pair of input datasets $(c_0, c_1)$ for $\RR_{\eps_j', \delta_j'}$, they both give $c_b$ to $\RR_{\eps_j', \delta_j'}$ and return its response. Thus, the IMs $\V[f^\RR]\circstar\I(b)\circstar\extconcomp$ and $\V[f^\RR]\circstar\calF\circstar\V[f^*]\circstar\I(b)\circstar\extconcomp$ are equivalent, and the property (1) is satisfied.

    It remains to prove that the property (2) holds. The assumption that $\comp(\RR_{\eps_1, 0}, \dots, \RR_{\eps_k, 0})$ is $(\eps, \delta)$-DP is equivalent to that $\extconcomp$ is $(\eps, \delta)$-DP w.r.t. the verification function $f^{\beta_1, \dots, \beta_k}$. Moreover, by Lemma~\ref{lem:filter-comp-eps-eq-zero}, $\extconcomp$ is $(0, \delta')$-DP w.r.t. $f^\RR$. By the optimal composition theorem of~\cite{murtagh2015complexity}, we have that the composition of $\RR_{\eps, \delta}$ and $\RR_{0, \delta'}$ is $(\eps, 1-(1-\delta)(1-\delta'))$-DP.\footnote{Replacing $k=2$, $\eps_1=\eps$, $\delta_1=\delta$, $\eps_2=0$, $\delta_2=\delta'$, and $\eps_g=\eps$ in Theorem 1.5 of~\cite{murtagh2015complexity} implies that for every $\delta_g$ such that $0\leq 1- \frac{1-\delta_g}{(1-\delta)(1-\delta')}$, the composition of $\RR_{\eps, \delta}$ and $\RR_{0, \delta'}$ is $(\eps, \delta_g)$-DP. $\delta_g=1-(1-\delta)(1-\delta')$ satisfies this inequality.} Therefore, by Theorem~\ref{thm:comp-fixed-cm}, $\extconcomp$ is $(\eps, 1-(1-\delta)(1-\delta'))$-DP w.r.t. the verification function $f^*=f^{\gamma_1, \gamma_2}$, finishing the proof.
\end{proof}

\begin{proof}[Proof of Theorem~\ref{thm:parallel-comp-approx}]
    Let $f=f^{\eps_1, \dots, \eps_k}_{\infty, \delta'}$ and $f^\RR=f^{\eps_1, \dots, \eps_k}_{\infty, \delta', RR}$. The proof proceeds in two steps: (1) We design an IPM $\cP$ and prove that for each $b\in\zo$, the IMs $\V[f]\circstar\I(b)\circstar\extconcomp$ and $\V[f]\circstar\cP\circstar\V[f^\RR]\circstar\I(b)\circstar\extconcomp$ are equivalent. (2) By Lemma~\ref{lem:basic-comp-fixed-and-filter}, the views of any adversary interacting with $\V[f^\RR]\circstar\I(0)\circstar\extconcomp$ and $\V[f^\RR]\circstar\I(1)\circstar\extconcomp$ is $(\eps, \delta+\delta')$-indistinguishable. Therefore, combining (1) and (2) with the post-processing lemma, it follows that the views of any adversary interacting with $\V[f]\circstar\I(1)\circstar\extconcomp$ and $\V[f]\circstar\I(0)\circstar\extconcomp$ are $(\eps, \delta+\delta')$-indistinguishable. This in turn implies that the $f$-concurrent composition of continual mechanisms is $(\eps, \delta+\delta')$-DP.

    We are left with showing (1).
    In $\V[f]\circstar\cP\circstar\V[f^\RR]\circstar\I\circstar\extconcomp$, the IPM $\cP$ receives from the left either (i) pairs of identical creation queries, or (ii) messages of the form $((q, j), (q', j))$, where $q$ and $q'$ are queries to the $j$-th created mechanism. Let $(\alpha_j, \alpha_j)$ be the $j$-th pair of creation queries, with $\alpha_j = (\cM_j', s_j', \eps_j', \delta_j', f_j')$. Also, for any $\eps''$ and $\delta''$, let $\beta_{\eps'', \delta''}$ denote the creation query for the mechanism $\RR_{\eps'', \delta''}$.

    The state of $\cP$ is a pair $(n, \sigma)$, where $n$ is a counter (initialized to $0$) and $\sigma$ is a sequence (initialized to an empty sequence). Upon receiving $(\alpha_i, \alpha_i)$ from the left, $\cP$ appends the tuple $(T_j, t_j, \eps_j', \delta_j', w_j=\varnull)$ to $\sigma$, where $\cT_j$ is the IPM corresponding to $\cM_j'$ from Corollary~\ref{cor:dp-cm-post-irr}, $t_j$ is $\cT_j$'s initial state, $(\eps_j', \delta_j')$ are the privacy parameters associated with $\cM_j'$, and $w_j$ is an initially null variable that will later be replaced by an outcome of a randomized response mechanism. $\cP$ then returns the acknowledgment message $\top$ as response to its left mechanism. 

    We note that although the privacy parameters $\eps_1, \dots, \eps_k$ in the statement of Theorem~\ref{thm:parallel-comp-approx} are assumed to be strictly positive, $\eps_j'$ might be zero. However, Corollary~\ref{cor:dp-cm-post-irr} is only applicable when $\eps_j'>0$. To fix this technical issue, if the privacy parameter $\eps_j'$ of the $j$-th creation query is zero, we change it to a small positive constant, e.g., $0.1$. Since $\cM_j'$ satisfies $(0, \delta_j')$-DP w.r.t. $f_j'$, it also satisfies $(0.1, \delta_j')$-DP w.r.t. $f_j'$. Thus, by Corollary~\ref{cor:dp-cm-post-irr}, there exists an IPM $\cT_j$ interacting with $\irr_{0.1, \delta_j'}$ and simulating $\V[f_j']\circstar\I\circstar\cM_j'$. By Definition~\ref{def:ver-func-par-approx}~(ii), since $0\not\in\{\eps_1, \dots, \eps_k\}$, the adversary always asks pairs of identical queries from $\cM_j'$, implying that $\cT_j$ interacts with its right mechanism at most once. The IPM $\cT_j$ expects $\irr_{0.1, \delta_j'}$ as its right mechanism, and $\irr_{0.1, \delta_j'}$ generates its first response only based on $\delta_j'$, and thus changing the value of $\eps_j'$ from zero to the arbitrary constant $0.1$ is unimportant.

    Upon receiving a left message of the form $((q, j), (q', j))$, the IPM $\cP$ sends $(q, q')$ as a left message to $\cT_j$. Upon receipt, $\cT_j$ may interact with its right mechanism. As described below, $\cP$ simulates the interaction of $\cT_j$ with its right mechanism until $\cT_j$ sends a left message $m$. $\cP$ then returns $m$ as a response to its left mechanism. By Corollary~\ref{cor:dp-cm-post-irr}, $\cT_j$  sends at most two right messages. 

    When $\cT_j$ sends a right message for the first time (i.e., when $w_j=\varnull$), $\cP$ requests the creation of $\RR_{0, \delta_j'}$ by sending the message $(\beta_{0, \delta_j'}, \beta_{0, \delta_j'})$ to its right mechanism $\V[f^\RR]\circstar\I\circstar\extconcomp$. $\cP$ discards the acknowledgment response $\top$, increments $n$ by one, and sends the message $\left((0, n), (1, n)\right)$ to its right mechanism. We will later show that $\V[f^\RR]$ never halts the communication. Given the initial state $b\in\zo$, upon receiving $\left((0, n), (1, n)\right)$, the identifier $\I(b)$ forwards $(b, n)$ to $\extconcomp$, which in turn sends $b$ to $\RR_{0, \delta_j'}$. The response of $\RR_{0, \delta_j'}$ is then forwarded to $\cP$ unchanged. Let $(\tau_j, z_j)\in \{\top, \bot\}\times\zo$ denote this response. $\cP$ sets $w_j=(\tau_j, z_j)$ and gives $\tau$ as a right response to $\cT_j$. Note that if $\tau_j=\bot$, the bit $b$ is exposed and $z_j=b$.

    When $\cT_j$ sends a second right message (i.e., when $w_j\neq \varnull$), $\cP$ proceeds as follows:
    \begin{itemize}[left=0pt]
        \item If $\tau_j=\bot$, then $\cP$ directly returns $z_j$ as the right response to $\cT_j$, thereby revealing the initial state of $\I$.
        \item Otherwise, $\cP$ requests the creation of $\RR_{\eps_j', 0}$ by sending $(\beta_{\eps_j', 0}, \beta_{\eps_j', 0})$ to its right mechanism $\V[f^\RR]\circstar\I\circstar\extconcomp$. As before, $\cP$ drops the acknowledgment response $\top$, increments the counter $n$, and sends $\left((0, n), (1, n)\right)$ to its right mechanism. In this case, $\RR_{\eps_j', 0}$ receives the initial state of $\I$ as input and returns a response $(\tau_j', z_j')$, which is forwarded unchanged to $\cP$. By definition $\tau_j'=\top$ with probability $1$. Finally, $\cP$ sends $z_j'$ as a right response to $\cT_j$.
    \end{itemize}

    By the definition of $f$, there are at most $k$ mechanisms among the creation queries sent by $\V[f]$ that receive messages of the form $((q, j), (q', j))$ with $q\neq q'$, and the parameter $\eps_j'$ of these mechanisms form a sub-multiset of $\eps_1, \dots, \eps_k$. By Corollary~\ref{cor:dp-cm-post-irr}, only the IPM $\cT_j$ corresponding to these mechanisms may send a second message to their right mechanisms. Moreover, by the definition of $f$, the parameters $\delta_j'$ of all mechanisms satisfy $1- \prod_j (1-\delta_j')\delta'$. Thus, by the design of $\cP$, the verifier $\V[f^\RR]$ never halts the communication and forwards all messages to the opposite side unchanged.

    Comparing Definition~\ref{def:irr} and Definition~\ref{def:rr}, we observe that the left responses generated by each $\cT_j$ are identically distributed to those that would be produced if it were interacting directly with the IM $\irr_{\eps_j', \delta_j'}(b)$. Thus, by Corollary~\ref{cor:dp-cm-post-irr}, each $\cT_j$ simulates $\V[f_j']\circstar\I(b)\circstar\cM_j'$ identically, implying that the IMs $\V[f]\circstar\I(b)\circstar\extconcomp$ and $\V[f]\circstar\cP\circstar\V[f^\RR]\circstar\I(b)\circstar\extconcomp$ are equivalent.
\end{proof}

\subsection{Composition Theorem for Restricted Verification Functions.}\label{subsec:parallel-restricted-ver-func}

In Theorem~\ref{thm:counter-example}, we showed that for any positive $\delta$, the concurrent $1$-sparse parallel composition of CMs with the privacy parameter $(0, \delta)$ fails to preserve privacy. In our counterexample, the adversary leveraged the first outcome of a $(0, \delta)$-DP CM to decide whether to spend its budget on issuing non-identical queries to that mechanism. Intuitively, when $\delta > 0$, there might exist certain ``bad first answers'' that occur with low probability; however, once we condition on their occurrences, the distribution of the second answer ceases to be differentially private.

In this section, we force the adversary to decide whether it will always issue pairs of identical queries to a CM when it sends the first message to that mechanism. Specifically, if the first message sent to a CM includes identical queries, then all subsequent messages for that CM must also have identical queries. This leads to the following definition.

\begin{definition}[First-Pair Consistent]\label{def:fpc}
    A verification function $f$ is \emph{first-pair consistent} if the following holds: for every input message sequence $m_1, \dots, m_t$, $f(m_1, \dots, m_t)=\bot$ whenever $m_1$ is a pair of identical messages, i.e., $m_1=(m', m')$ for some $m'$, and there exists $2\leq j\leq t$ such that $m_j$ is not a pair of identical messages.
\end{definition}

For instance, let $\cM$ be an IM that is $(\eps, \delta)$-DP w.r.t. a neighbor relation $\sim$. As discussed in Section~\ref{subsec:alternative-im}, $\cM$ can be represented by a CM $\cM'$ (with initial state $\varnull$) that receives an initial state $s\in S_\cM^\init$ as the first message and behaves the same as $\cM(s)$ afterwards. Define the verification function $f_\sim$ as follows: given an input message sequence $(m_1, \dots, m_t)$, $f_\sim$ returns $\top$ if and only if $m_1$ is a pair of neighboring initial states in $S_\cM^\init$, and for all $2 \leq j \leq t$, $m_j$ is a pair of identical queries in $Q_\cM$. By definition, $f_\sim$ is first-pair consistent.

To formally define the restricted variant of the concurrent parallel composition of CMs, we introduce the verification function $f^{\eps_1, \delta_1, \dots, \eps_k, \delta_k}_{\infty, \mathit{FPC}}$, where $FPC$ stands for first-pair consistency.

\begin{definition}[$f^{\eps_1, \delta_1, \dots, \eps_k, \delta_k}_{\infty, \mathit{FPC}}$]\label{def:ver-func-par-fixed-param-comp-fpc}
    For any $t\in\N$ and any input message sequence $m_1, \dots, m_t$, the function $f^{\eps_1, \delta_1, \dots, \eps_k, \delta_k}_{\infty, \mathit{FPC}}$ returns $\top$ if and only if both of the following conditions hold:
    \begin{enumerate}[left=0pt]
        \item [(i)] $f^{\eps_1, \delta_1, \dots, \eps_k, \delta_k}_{\infty}(m_1, \dots, m_t)=\top$
        \item [(ii)] For every message $m_j$ of the form $(\cM', s', \eps', \delta', f')^2$, the verification function $f'$ is first-pair consistent.
    \end{enumerate}
\end{definition}

We refer to the $f^{\eps_1, \delta_1, \dots, \eps_k, \delta_k}_{\infty, \mathit{FPC}}$-concurrent composition of CMs as the \emph{FPC concurrent $k$-sparse parallel composition of CMs with privacy parameters $\left((\eps_i, \delta_i)\right)_{i=1}^k$}.

\begin{theorem}\label{thm:parallel-fixed-param-comp-fpc}
    For $\eps_1, \dots, \eps_k> 0$, $\eps\geq 0$ and $0\leq \delta, \delta_1, \dots, \delta_k\leq 1$, if the composition of non-interactive mechanisms $\RR_{\eps_1, \delta_1}, \dots, \RR_{\eps_k, \delta_k}$ is $(\eps, \delta)$-DP, then the FPC concurrent $k$-sparse parallel composition of CMs with the privacy parameters $(\eps_1, \delta_1), \dots,(\eps_k, \delta_k)$ is also $(\eps, \delta)$-DP against an adaptive adversary.
\end{theorem}

\begin{proof}
    We denote $f$ as shorthand for $f^{\eps_1, \delta_1, \dots, \eps_k, \delta_k}_{\infty, \mathit{FPC}}$. The proof is identical to the proof of Theorem~\ref{thm:comp-fixed-param-cm}, except for the definition of IPMs $\cT_j$ that are executed by the IPM $\cP$: In the proof of Theorem~\ref{thm:comp-fixed-param-cm}, each $\cT_j$ was defined according to Lemma~\ref{lem:cm-post-rr-lyu}. In this proof, each $\cT_j$ is defined by Corollary~\ref{cor:dp-cm-post-irr} if the first message of the form $\left((q_0, j), (q_1, j)\right)$ satisfies $q_0\neq q_1$, and it behaves identically to $\V[f_j']\circstar\I(0)\circstar\cM_j'$ otherwise.

    To implement this, the sequence $\sigma_2$ in $\cP$'s state must include an additional boolean variable $u_j$ for each created mechanism $\cM_j'$, indicating whether the mechanism has previously received any messages. Specifically, when $\cP$ receives the creation query pair $(\cM_j', s_j', \eps_j', \delta_j', f_j')^2$ from the verifier, instead of adding $(\cT_j, t_j, \eps_j', \delta_j')$, it adds the tuple $(\cM_j', s_j', \eps_j', \delta_j', u_j=0)$ to $\sigma_2$. 

    Upon receiving a left message of the form $\left((q_0, j), (q_1, j)\right)$, if $u_j=0$, before responding, $\cP$ replaces the $j$-th entry $(\cM_j', s_j', \eps_j', \delta_j', u_j=0)$ in $\sigma_2$ with $(\cT_j, t_j, \eps_j', \delta_j', 1)$, where $T_j$ is an IPM defined as follows and $t_j$ is its initial state:
    \begin{itemize}[left=0pt]
        \item If $q_0=q_1$, then $\cT_j$ is an IPM with the single initial state $\left((), 0, s_j'\right)$ that never interacts with its right mechanism and interacts with its left mechanism exactly as $\V[f_j']\circstar\I(0)\circstar\cM_j'$. That is, upon receiving a right message, $\cT_j$ halts the communication, and upon receiving a left message, it updates its state and responds identically to $\V[f_j']\circstar\I(0)\circstar\cM_j'$.
        \item If $q_0 \ne q_1$, then $\cT_j$ is the IPM defined by Lemma~\ref{lem:cm-post-rr-lyu} corresponding to $\cM_j'$.
    \end{itemize}

    To generate a response to a left message of the form $\left((q_0, j), (q_1, j)\right)$, as before, $\cP$ forwards $(q_0, q_1)$ to $\cT_j$ as a left message and provide a properly chosen $r_i$ as a right response if needed (see the proof of Theorem~\ref{thm:comp-fixed-param-cm}).

    By Definition~\ref{def:ver-func-par-fixed-param-comp-fpc} (ii), we know that for each index $j$, if $q_0=q_1$ in the first message of the form $\left((q_0, j), (q_1, j)\right)$, then $q_0=q_1$ in all subsequent messages of this form. In this case, by the above definition, for each $b\in\zo$, the mechanism $\cT_j$ never interacts with $\RR_{\eps_j', \delta_j'}(b)$ and simulates $\V[f_j']\circstar\I(b)\circstar\cM_j'$ identically. Moreover, when $q_0\neq q_1$ in the first message, by Definition~\ref{def:ver-func-par-fixed-param-comp-fpc} (i) and Definition~\ref{def:ver-func-par-fixed-param-comp-cm} (ii), the privacy parameters $(\eps_j', \delta_j')$ corresponding to $\cT_j$s with the second definition form a sub-multiset of $(\eps_1, \delta_1), \dots, (\eps_k, \delta_k)$. Therefore, whenever an IPM $\cT_j$ requires a sample from $\RR_{\eps_j', \delta_j'}$, $\cP$ can find such sample in $\sigma_1$ that is unused.
\end{proof}

\subsubsection{Concurrent $k$-sparse Parallel Composition of IM}\label{subsubsec:parallel-comp-im}
Consider a setting where an adversary creates (an unlimited number of) IMs over time while adaptively asking queries from the existing ones. When creating an IM $\cM_i$ that is $(\eps_i, \delta_i)$-DP w.r.t. a neighbor relation $\sim_i$, the adversary also selects a pair of neighboring initial states for $\cM_i$. This selection depends on the history of queries asked so far and their responses. Based on a secret bit $b$, either all created mechanisms are initialized with the first initial state in their respective pairs, or all are initialized with the second. The goal of the adversary is to distinguish whether $b=0$ or $b=1$.

This scenario is an special case of FPC concurrent $k$-sparse parallel composition of CMs, where, instead of creating $\cM_i$ and assigning a pair of initial states to it, the adversary sends $(\cM_i', \varnull, \eps_i, \delta_i, f_{\sim_i})^2$, where $\cM_i'$ is the CM simulating $\cM_i$. Moreover, to ask a query $q$ from the $i$-th created IM, the adversary sends $\left((q, i), (q, i)\right)$. Since each $f_{\sim_i}$ is first-pair consistent, Theorem~\ref{thm:parallel-fixed-param-comp-fpc} applies to the \emph{concurrent $k$-sparse parallel composition of IM} (that are differentially private w.r.t. neighbor relations).

\begin{corollary}\label{cor:parallel-fixed-param-comp-im}
    For $\eps_1, \dots, \eps_k> 0$, $\eps\geq 0$ and $0\leq \delta, \delta_1, \dots, \delta_k\leq 1$, if the composition of non-interactive mechanisms $\RR_{\eps_1, \delta_1}, \dots, \RR_{\eps_k, \delta_k}$ is $(\eps, \delta)$-DP, then the concurrent $k$-sparse parallel composition of IM with privacy parameters $\left((\eps_i, \delta_i)\right)_{i=1}^k$ is also $(\eps, \delta)$-DP.
\end{corollary}

NIMs are a special subclass of IMs that halt the communication upon receiving a second input message. Since Theorem~\ref{thm:parallel-fixed-param-comp-fpc} applies to all IMs, it also holds for the \emph{concurrent $k$-sparse parallel composition of NIMs with privacy parameters $\left((\eps_1, \delta_i\right)_{i=1}^k$}. 
In the next section, we compare the $k$-sparse parallel composition of IMs to the classical parallel composition theorem for NIMs.

\subsubsection{Comparison with the Parallel Composition of NIMs.} 
Consider a set of $\ell$ NIMs, $\cN_1, \dots, \cN_\ell$, where each $\cN_i$ has a domain $\calX_i$ and is $(\eps_i, \delta_i)$-DP w.r.t. a neighbor relation $\sim_i$. The \emph{$k$-sparse parallel composition of NIMs} $\cN_1, \dots, \cN_\ell$ is said to be $(\eps, \delta)$-DP if the non-interactive mechanism $\comp(\cN_1, \dots, \cN_\ell)$ with the domain $\calX=\calX_1\times\dots\times \calX_\ell$ is $(\eps, \delta)$-DP w.r.t. the neighbor relation $\sim_{\partext}$, defined as follows: for any two $x=(x_1, \dots x_\ell)$ and $x'=(x_1', \dots x_\ell')$ in $\calX$, we have $x\sim_{\partext} x'$ if and only if there exists a subset $\mathcal{J}\subseteq\{1,\dots,\ell\}$ of size at most $k$ such that $x_j\sim_j x_j'$ for every $j\in \mathcal{J}$, and $x_j=x_j'$ for every $j\in\{1,\dots,\ell\}\setminus\mathcal{J}$.

\medskip
The concurrent parallel composition of IMs generalizes the parallel composition of NIMs in several key ways: Rather than composing a fixed set of $\ell$ NIMs, an adversary adaptively selects an \emph{unbounded number} of IMs with arbitrary privacy parameters. Additionally, instead of non-adaptively selecting worst-case neighboring initial states for (at most) $k$ mechanisms and identical initial states for the other $\ell-k$ mechanisms, the adversary adaptively (1) chooses pairs of initial states based on prior interactions with other mechanisms, while (2)  ensuring that the initial states in at most $k$ pairs are neighboring, and that they are identical in the other pairs. In the concurrent parallel composition of CMs, the adversary is even stronger: instead of choosing a single pair of neighboring initial states, \emph{it adaptively selects a sequence of pairs}, satisfying the constraints enforced by a verification function.

The $1$-sparse parallel composition of NIMs is commonly used when a dataset is partitioned among $\ell$ NIMs such that for any two neighboring datasets differing by a single record, all partitions remain identical except for the presence of the differing record in one partition. Consider the case where an adversary creates CMs receiving data or empty records as query messages. In the concurrent parallel composition of CMs, the adversary adaptively determines the number of partitions, assigns records to them, and selects which partitions contain differing sets of records.

It is known that the $k$-sparse parallel composition of $\ell$ $(\eps_0, \delta_0)$-DP NIMs is $(\eps, \delta)$-DP if the composition of $\ell$ instances of $\RR_{\eps_0, \delta_0}$ is $(\eps, \delta)$-DP. We showed that any NIM that is $(\eps_0, \delta_0)$-DP w.r.t. a neighbor relation $\sim$ can be modeled by a CM that is $(\eps_0, \delta_0)$-DP w.r.t. $f_\sim$, which is a first-pair consistent verification function (see the example below Definition~\ref{def:fpc}). Thus, setting $\eps_1=\dots=\eps_k=\eps_0$, $\delta_1=\dots=\delta_k=\delta_0$ in Theorem~\ref{thm:parallel-fixed-param-comp-fpc} recovers and generalizes the classic parallel composition theorem. Our theorem extends differential privacy guarantees to settings where an adversary can adaptively create an unbounded number of NIMs.

\section{Designing Complex Continual Mechanisms}\label{sec:analyze-complex-cm}
So far we discussed the setting where the messages are issued by an adversary. However,
in the design of  differentially private mechanisms in the streaming setting, it is common to use existing differentially private interactive mechanisms, such as the Sparse Vector Technique (SVT) and continual counters, as sub-mechanisms. 
Thus, it is the main mechanism $\cM$ that issues messages to its sub-mechanisms.
The main mechanism may fix the set of sub-mechanisms at initialization or select these mechanisms adaptively over time. 
For example, in Section~\ref{sec:application}, we present a mechanism that maintains $d$ fixed binary counters while adaptively instantiating an unbounded number of SVT and Laplace mechanisms as required.

The privacy analysis of such mechanisms $\cM$ is challenging as the creation of sub-mechanisms and their inputs is adaptive and dependent. The goal of this section is to give a set of conditions for $\cM$ and prove a theorem that shows that under these conditions the privacy of $\cM$ can directly be bounded by applying our concurrent composition theorems for continual mechanisms to the sub-mechanisms. Specifically, given that the $f'$-concurrent composition of CMs is $(\eps, \delta)$-DP w.r.t. some verification function $f'$, we would like to show that $\cM$ is $(\eps, \delta)$-DP w.r.t. the given verification function $f$.

As discussed in Section~\ref{sec:cm}, these sub-mechanisms can be modeled as CMs. We formalize $\cM$, which runs multiple CMs, as a continual mechanism of the form $\cP\circstar\extconcomp$, where $\cP$ is an IPM and $\extconcomp$ is the CM defined in Section~\ref{subsec:composition-extconcomp}. Since $\cM$ is a CM, its initial state space is a singleton, implying that the initial state space of $\cP$ is also a singleton. Upon receiving a message, $\cP$ may create new differentially private CMs or ask queries from the existing ones by communicating with $\extconcomp$, possibly multiple times. 

To prove that $\cM$ is $(\eps, \delta)$-DP w.r.t. $f$, we need to show that the views of every adversary $\cA$ interacting with $\V[f]\circstar\I(0)\circstar \cM = \V[f]\circstar\I(0)\circstar\cP\circstar\extconcomp$ and $\V[f]\circstar\I(1)\circstar \cM=\V[f]\circstar\I(1)\circstar\cP\circstar\extconcomp$ are $(\eps, \delta)$-indistinguishable.
Intuitively, the messages that $\cP\circstar\extconcomp$ receives from $\I(b)$ contain information about $b$. Therefore, to deduce that $\cP\circstar\extconcomp$ is $(\eps, \delta)$-DP w.r.t. $f$ from the fact that $\extconcomp$ is $(\eps, \delta)$-DP w.r.t. $f'$, the IPM $\cP$ must not directly use the messages from $\I$ to respond. Instead, it should only use these messages to determine which queries it wants to send to its sub-mechanisms via $\extconcomp$. More precisely, $\cP$ must generate an answer for $\I$ only  based on the privatized answers provided by $\extconcomp$ and \emph{independent} of the un-privatized input sequence (Response Property). Moreover, $\cP$ must decide whether to continue interacting with $\extconcomp$ or to return a response based on these privatized answers (Destination Property).

Finally, we need a property for $\cP$ to convert $f$-valid message sequences to $f'$-valid ones: Let $m_1^0, m_2^0, \dots$ and $m_1^1, m_2^1, \dots$ be two sequences such that the sequence $(m_1^0, m_1^1), (m_2^0, m_2^1), \dots$ is $f$-valid. Upon receiving each sequence $m_1^b, m_2^b, \dots$, the mechanism $\cP$ must send the messages $m_1^{b\prime}, m_2^{b\prime}, \dots$ to $\extconcomp$ such that $(m_1^{0\prime}, m_1^{1\prime}), (m_2^{0\prime}, m_2^{1\prime}), \dots$ is $f'$-valid ($f\to f'$ Property).

In Section~\ref{subsec:properties}, we formalize the above properties for $\cP$, and in Section~\ref{subsec:privacy-complex-cm}, we prove the desired privacy guarantees for $\cM$. Combining this result with the concurrent composition theorems provided in the previous sections, we can avoid complicated privacy analysis for many CMs. To obtain rigorous results, our statements and proofs must be presented in a technical and somewhat complex form. In contrast, mechanisms are usually described more simply using pseudocode rather than message-passing state machines. In Section~\ref{subsec:pseudocode}, we reformulate our results in terms of pseudocode, enabling new mechanisms to be analyzed directly without the need for reformulating to the setting in this paper.

\subsection{Properties.}\label{subsec:properties}
To formally state that $\cP$ maps $f$-valid sequences to $f'$-valid ones, \emph{we assume that $\cP$ is deterministic}. This means $\cP$ can be expressed as a deterministic function that maps the sequence of received messages to its next output message. This assumption is not restrictive, as randomness can be modeled by a $0$-DP sub-mechanism $\cM_i$ that accepts a single message $\top$, and upon receiving it, generates a sufficiently long sequence of random bits or samples from a fixed distribution. 

Recall that the goal of this section is to show $\cP\circstar\extconcomp$ is $(\eps, \delta)$-DP w.r.t. $f$ if $\extconcomp$ is $(\eps, \delta)$-DP w.r.t. $f'$. That is, we need to show that the views of every adversary interacting with $\V[f]\circstar\I(0)\circstar\cP\circstar\extconcomp$ and $\V[f]\circstar\I(1)\circstar\cP\circstar\extconcomp$ are $(\eps, \delta)$-indistinguishable.

The messages received by $\cP$ from $\I$ contain information about the bit used by $\I$ to filter the adversary's pairs. Therefore, to directly use the existing composition theorems for $\extconcomp$ to analyze the privacy of $\cP\circstar\extconcomp$, we need to ensure that $\cP$ only utilizes the messages from $\I$ to generate queries for the existing private sub-mechanisms executed by $\extconcomp$. Moreover, we need to guarantee that the sequence $(m_1^{0\prime}, m_1^{1\prime}), (m_2^{0\prime}, m_2^{1\prime}), \dots$ is $f'$-valid, where $m_1^{b\prime}, m_2^{b\prime}, \dots$ is the sequence of messages that $\cP$ sends to $\extconcomp$ when it receives messages from $\I(b)\circstar\V[f]$. The following properties formalize these:

\begin{itemize}[left=0pt]
    \item \textbf{Destination Property.} The destination $\ttL$ or $\ttR$ of messages sent by $\cP$ is determined solely based on messages received from its right mechanism (here, $\extconcomp$). This ensures that if two transcripts contain identical sequences of messages from the right mechanism but different sequences from the left mechanism, $\cP$ will still interact with the same mechanism in the next round.
    \item \textbf{Response Property.} The responses sent by $\cP$ to the left mechanism depend only on the messages received from the right mechanism. That is, if two transcripts contain identical sequences of messages from the right mechanism and $\cP$ sends its next message to the left mechanism, then this message is identical in both cases.
    \item \textbf{$f\to f'$ Property.}
    For $v\in\{\ttL, \ttR\}$, let the pair $(m, v)$ denote a message $m$ sent to or received from the $v$'th mechanism (i.e., left or right) of $\cP$. For $b\in\zo$, let $\sigma^b=\left((m_1^b, v_1^b), \dots, (m_t^b, v_t^b)\right)$ be a sequence of $t$ input messages to $\cP$, and let $\tau^b=\left((m_1^{b\prime}, v_1^{b\prime}), \dots, (m_{t}^{b\prime}, v_{t}^{b\prime})\right)$ be the corresponding sequence of $t$ output messages from $\cP$. Suppose the following conditions hold:
    \begin{enumerate}
        \item[(i)] $v_i^0=v_i^1$ for each $i\in [t]$, i.e., the sender of messages is identical;
        \item[(ii)] $m_i^0=m_i^1$ for each $i\in [t]$ where $v_i^0=v_i^1=\ttR$, meaning that messages from the right mechanism are identical; and
        \item[(iii)] the sequence of $(m_i^0, m_i^1)$s for all $i\in [t]$ where $v_i^0=v_i^1=\ttL$ is $f$-valid.
    \end{enumerate}
    Then, by the destination property of $\cP$, $v_i^{0\prime}=v_i^{1\prime}$ for each $i\in [t]$, and by the response property, $m_i^{0\prime}=m_i^{1\prime}$ for each $i\in [t]$ with $v_i^0=v_i^1=\ttL$. The $f\to f'$ property states that under these conditions, the sequence $\left((m_i^{0\prime}, m_i^{1\prime})\right)$ for all $i\in [t]$ where $v_i^{0\prime}=v_i^{1\prime}=\ttR$ must be $f'$-valid.
\end{itemize}

\subsection{Privacy Result.}\label{subsec:privacy-complex-cm}
In this section, we will prove the following theorem:
\begin{theorem}\label{thm:dp-complex-mech}
    Let $\cP$ be a deterministic interactive post-processing mechanism with a singleton initial state space satisfying the destination, response, and $f\to f'$ properties. If $\extconcomp$ is $(\eps, \delta)$-DP w.r.t. $f'$ against adaptive adversaries, then the continual mechanism $\cP\circstar\extconcomp$ is $(\eps, \delta)$-DP w.r.t. $f$ against adaptive adversaries.
\end{theorem}
In Theorem~\ref{thm:dp-complex-mech}, we assume that the views of every adversary interacting with $\V[f']\circstar\I(0)\circstar\extconcomp$ and $\V[f']\circstar\I(1)\circstar\extconcomp$ are $(\eps, \delta)$-indistinguishable. The goal is to show the same result for $\V[f]\circstar\I(0)\circstar\cP\circstar\extconcomp$ and $\V[f]\circstar\I(1)\circstar\cP\circstar\extconcomp$. The key idea of the proof is to show that the IMs $\V[f]\circstar\I(b)\circstar\cP\circstar\extconcomp$ and $\cT\circstar\V[f']\circstar\I(b)\circstar\extconcomp$ are identical for some IPM $\cT$ with a singleton initial state space. The proof consists of two steps.

\vspace{0.3cm}\noindent
\underline{Step 1 (Swapping $\cP$ and $\I(b)$).} We start by swapping the positions of $\cP$ and $\I(b)$ in $\V[f]\circstar\I(b)\circstar\cP\circstar\extconcomp$. The verifier $\V[f]$ sends pairs of query messages, while $\cP$ expects a single one. Also, the identifier $\I(b)$ expects pairs of messages, while $\cP$ generates a single response. Therefore, when swapping $\cP$ and $\I(b)$, we should replace $\cP\circstar\I(b)$ with $\I(b)\circstar\cP^2$, where $\cP^2$ is an interactive post-processing mechanism that runs two instances of $\cP$ in parallel. 
The right mechanism of $\cP^2$ is $\I\circstar\extconcomp$, and its left mechanism is $\V[f]$. The exact definition of $\cP^2$ is given below.

\begin{definition}[$\cP^2$]\label{def:p2}
    Let $\cP : S_\cP \times \left((\{\ttL\} \times Q_\cP^L) \cup (\{\ttR\} \times Q_\cP^R)\right)\rightarrow S_\cP \times \left((\{\ttL\} \times A_\cP^L) \cup (\{\ttR\} \times A_\cP^R)\right)$ be a deterministic interactive post-processing mechanism with a singleton initial state space $S_\cP^\init=\{s^\init\}$, satisfying the destination and response properties. $\cP^2: S_\cP^2 \times \left((\{\ttL\} \times Q_\cP^L) \cup (\{\ttR\} \times Q_\cP^R \times Q_\cP^R)\right)\rightarrow S_\cP \times \left((\{\ttL\} \times A_\cP^L) \cup (\{\ttR\} \times A_\cP^R)\right)$ is an interactive post-processing mechanism that executes two instances of $\cP$, referred to as $\cP_0$ and $\cP_1$. The state of $\cP^2$ is a pair $(s_0, s_1)$ denoting the current states of $\cP_0$ and $\cP_1$. The initial state space of $\cP^2$ is $\{(s^\init, s^\init)\}$.
    
    When $\cP^2$ receives a message $m\in Q_\cP^L$ from the left mechanism, it forwards two copies of $m$ to both $\cP_0$ and $\cP_1$ as left messages. Specifically, it executes $(s_0, v_0', m_0')\gets \cP_0(s_0, \ttL, m)$ and $(s_1, v_1', m_1')\gets \cP_1(s_1, \ttL, m)$.  When $\cP^2$ receives a pair of messages $(m_0, m_1)\in Q_\cP^R \times Q_\cP^R$ from the right mechanism, it passes each $m_b$ to $\cP_b$ as a right message. Specifically, it executes $(s_0, v_0', m_0')\gets \cP_0(s_0, \ttR, m_0)$ and $(s_1, v_1', m_1')\gets \cP_1(s_1, \ttR, m_0)$. By the destination property of $\cP$ and the fact that both $\cP_0$ and $\cP_1$ receive identical left messages, we have $v_0'=v_1'$. If $v_0'=v_1'=\ttR$, then $\cP^2$ sends $(m_0', m_1')$ to its right mechanism. If $v_0'=v_1'=\ttL$, then by the response property, $m_0'=m_1'$, and $\cP^2$ returns $m_0'$ to its left mechanism.
\end{definition}

The following claim follows directly from the definition of $\cP^2$:

\begin{claim}\label{cla:swap}
    For each $b\in\zo$, the interactive mechanisms $\I(b)\circstar\cP\circstar\extconcomp$ and $\cP^2\circstar\I(b)\circstar\extconcomp$ are identical.
\end{claim}

\noindent\underline{Step 2 (Inserting $\V[f']$).} Due to the $f\to f'$ property, when $\cP^2$ receives an $f$-valid message sequence from its left mechanism, it sends an $f'$-valid message sequence to its right mechanism. We insert $\V[f']$ between $\cP^2$ and $\I$ in $\cP^2\circstar\I\circstar\extconcomp$. By definition, when $\V[f']$ receives an $f'$-valid message sequence from its left mechanism, it acts like it does not exist. Specifically, it forwards every message from its left and right mechanism unchanged to the opposite mechanism. Thus, we have:

\begin{claim}\label{cla:add-verifier}
    For each $b\in\zo$, the interactive mechanisms $\V[f]\circstar\cP^2\circstar\I(b)\circstar\extconcomp$ and $\V[f]\circstar\cP^2\circstar\V[f']\circstar\I(b)\circstar\extconcomp$ are identical.
\end{claim}

\begin{proof}[Proof of Theorem~\ref{thm:dp-complex-mech}]
    $\V[f]\circstar\cP^2$ in Claim~\ref{cla:add-verifier} is an IPM with a singleton initial state space. Thus, by Lemma~\ref{lem:post-im}, the fact that for every adversary $\cA$, the views of $\cA$ interacting with $\V[f']\circstar\I(0)\circstar\extconcomp$ and $\V[f']\circstar\I(1)\circstar\extconcomp$ are $(\eps, \delta)$-indistinguishable implies that for every adversary $\cA$, the views of $\cA$ interacting with $\V[f]\circstar\I(0)\circstar\cP\circstar\extconcomp$ and $\V[f]\circstar\I(1)\circstar\cP\circstar\extconcomp$ is $(\eps, \delta)$-indistinguishable, which finishes the proof.
\end{proof}

\subsection{Simplified Formulation for Pseudocodes}\label{subsec:pseudocode}
Continual mechanisms are typically not described as state machines exchanging messages. Instead, they are presented in the form of pseudocode: at each time step, the pseudocode receives an input, updates its internal variables, and produces an output. For clarity and practicality, we restate the results of this section in the pseudocode framework.

\paragraph{Tainted vs. Untainted Variables.}
To rigorously track sensitive data, we rely on concepts from computer security. \emph{Information flow} characterizes how data transfers between variables during a program's execution \cite{denning1976lattice,sabelfeld2003language}. \emph{Taint analysis} builds upon this to monitor the propagation of sensitive inputs, known as \emph{taint sources} \cite{newsome2005dynamic,schwartz2010all}. Any variable computed using a taint source becomes a \emph{tainted variable}. To safely release data, it must be processed by a \emph{sanitizer}---in our context, a differentially private submechanism. The sanitizer mathematically interrupts the sensitive information flow, rendering its outputs and all subsequent downstream computations as \emph{untainted variables}.

To formalize this, we partition the variables of a mechanism into two categories using a recursive definition based on their inputs:
\begin{itemize}
    \item \textbf{Tainted variables:} The mechanism's inputs are strictly tainted. Furthermore, any intermediate variable computed via a function where at least one input argument is a tainted variable is also considered tainted. These variables retain a direct dependence on the raw inputs.
    \item \textbf{Untainted variables:} Variables computed solely from privatized data or fixed values (such as the privacy parameters or constants) are untainted. Crucially, the output of any differentially private submechanism (the sanitizer) is defined as untainted, regardless of whether its inputs were tainted. Any subsequent variable computed as a function of \emph{only} untainted variables remains untainted. Consequently, untainted variables are purely a post-processing of privatized outputs.
\end{itemize}

\paragraph{Note.}
For simplicity, we assume the pseudocode of $\cM$ is deterministic as any randomness could be modeled as a submechanism call. For instance, sampling from the uniform distribution could be represented as invoking a $0$-DP submechanism that outputs such a sample on a fixed input.

\begin{definition}\label{def:pseudocode-properties}
    For $t \in \N$, given two sequences $\sigma^0=(m_1^0,\dots,m_t^0)$ and $\sigma^1=(m_1^1,\dots,m_t^1)$, their \emph{combination} is defined as $\big((m_1^0,m_1^1),\dots,(m_t^0,m_t^1)\big)$. Recall the three properties from Section~\ref{subsec:properties}. We reformulate them for pseudocodes as follows:
    \begin{enumerate}
        \item \textbf{Destination property.}  
        A pseudocode satisfies this property if (i) the decision of whether to produce an output, create a new submechanism, or interact with an existing one, (ii) the choice of the new submechanism in the second case, and (iii) the selection of the submechanism to interact with in the third case is made solely based on untainted variables.
        \item \textbf{Response property.} 
        A pseudocode satisfies this property if its output is determined solely by untainted variables.   
        \item \textbf{$f \to f'$ property.}  
        Let $f$ and $f'$ be verification functions. A pseudocode satisfies this property if, for every $t\in\N$ and every pair of input sequences $(\sigma^0, \sigma^1)$ of length $t$ forming an $f$-valid combination, the following holds: Assume that all untainted variables of $\cM$ are updated identically through the end of time step $t$. For each $b\in \zo$, define $\sigma^b$ to denote the sequence of creation queries and inputs to submechanisms when processing $\sigma^b$. Here, a creation query is a tuple specifying a created submechanism, its privacy parameters, and associated verification functions. Also, an input $m$ to the $j$-th submechanism is represented by the pair $(m, j)$. Then, the combination of $\sigma^0$ and $\sigma^1$ is required to be $f'$-valid.
    \end{enumerate}
\end{definition}

\begin{corollary}\label{cor:pseudocode}
    Let $f$ and $f'$ be two verification functions. For $\eps\geq 0$ and $0\leq \delta
    \leq 1$, suppose the $f'$-concurrent composition of continual mechanisms is $(\eps, \delta)$-DP. Then, every continual mechanism whose pseudocode satisfies the destination, response, and $f\to f'$ properties, defined above, is $(\eps, \delta)$-DP w.r.t. $f$.
\end{corollary}

\section{Application (Continual Monotone Histogram Query Problem)}\label{sec:application}
We next show an application of our technique to the continual monotone histogram query problem. So far, we defined differential privacy (DP) with respect to (w.r.t.) verification functions. To be consistent with the literature, in this section we use the notion of DP w.r.t. ``neighboring input sequences'' instead:

\begin{remark}\label{rem:neighbor-to-adaptive-ver-func}
    Let $\sim$ be a binary relation on a family of message sequences. Throughout this section, we say that a mechanism is $(\eps, \delta)$-DP w.r.t. the neighbor relation $\sim$ (or simply it is $(\eps, \delta)$-DP) if the mechanism is $(\eps, \delta)$-DP w.r.t. the verification function $f_\sim$ \emph{against adaptive adversaries}, where $f_\sim$ is defined as follows: For every message sequence $m_1, \dots, m_t$, $f_\sim(m_1\dots, m_t)=\top$ if and only if the message sequence is suitable in the sense of Definition~\ref{def:suitable-seq} (implying $m_j=(m_j^0, m_j^1)$ for every $j\in \{1, \dots, t\}$) and the sequences $\sigma_0=(m_1^0, \dots, m_t^0)$ and $\sigma_1=(m_1^1, \dots, m_t^1)$ satisfy $\sigma_0\sim \sigma_1$. 
\end{remark}

\begin{definition}\label{def:histogram-problem}
    Let $d, T\in \N$ and let $q:\N^d \to \N$ be a function that is monotone in its input. In the (differentially private) \emph{continual monotone histogram query problem}, a continual mechanism $\cM$ receives $T$ inputs of row vectors $x=(x^1,\dots, x^T)$, each in $\{0,1\}^d$. Differential privacy for $\cM$ is defined with respect to the neighbor relation $\sim_H$ defined as follows (see Remark~\ref{rem:neighbor-to-adaptive-ver-func}). For $T\in\N$, two input streams in $\left(\zo^d\right)^T$ are neighboring if they are identical except for a single step. The output of $\cM$ at each time step $t$ is (an additive approximation to) $q(t)=q((\sum_{\ell=1}^t x_i^\ell)_{i\in[d]})$.
\end{definition}

\begin{definition}\label{def:1-sensitive}
    Recall function $q:\zo^d\to \R$ and neighbor relation $\sim_H$ from  Definition~\ref{def:histogram-problem}. $q$ is said to be \emph{$1$-sensitive} if $\max_{x^0 \sim_H x^1}|q(x^0)-q(x^1)|$ is at most $1$.
\end{definition}

Henzinger, Sricharan, and Steiner~\cite{henzingerSS2023} presented
the best known differentially private mechanism for the monotone histogram query problem, where the query is $1$-sensitive. For $\eps$-differential privacy it has, with probability at least $1-\beta$, an additive error of $O((d\log^2 (d q^*/\beta) +\log T)/\eps)$, and for $(\eps,\delta)$-differential privacy, it has an additive error of $O((d\log^{1.5} (d q^*/\beta) +\log T)\log\frac{1}{\delta}/\eps)$, where $q^* = \max_{t \in [T]} q(t)$. The privacy of the mechanism is proven ``from scratch'' using the concept of a privacy game~\cite{pricedpco} and the bound was only shown for $(\eps, \delta)$-DP against an oblivious adversary. 

In Section~\ref{subsec:application-algorithm}, we give an overview of the mechanism in \cite{henzingerSS2023}, which we will call $\tthss$, and in Section~\ref{subsec:application-new-privacy}, we give a simple privacy proof for it using our approach. With the new privacy analysis, the result for $(\eps,\delta)$-differential privacy is stronger and holds against an adaptive adversary.

\subsection{Overview of Mechanism.}\label{subsec:application-algorithm}
The mechanism $\tthss$, described in Algorithm~\ref{alg:histogram}, employs a variant of the Sparse Vector Technique mechanism, $\svt$, and a $d$-dimensional counter mechanism, $\tth$, as submechanisms. We start by defining these mechanisms and then explain $\tthss$.

\noindent\medskip \textbf{Sparse vector technique mechanism (SVT).} {
The mechanism $\svt$ is described in Algorithm~\ref{alg:svt} using two operations. The first, $\sinit(\eps, q, h)$, initializes the mechanism with a privacy parameter $\eps$, a query $q$, and an initial $d$-dimensional vector $h$. This procedure indeed describes the creation of a continual mechanism $\svt_{h, \eps, q}$. The second operation, $\call(x, \thr)$, maps inputs $(x, \thr)\in \zo^d\times \R$ to outputs in $\{\top, \bot\}$. The mechanism maintains a running sum of all inputs $x$ and, in a differentially private manner, checks whether applying $q$ to the running sum yields a value at least $\thr$. If so, it outputs $\top$ and halts. Otherwise, it outputs $\bot$ and continues. These two procedures are given in detail in Algorithm~\ref{alg:svt}. The command $\Lap(b)$ denotes a random sample from the Laplace distribution with scale parameter $b$.

To clarify the connection between these pseudocodes and our definition of continual mechanisms (Definition~\ref{def:cm}), we note that the state of $\svt_{h,\eps,q}$ is the tuple of its variables. When $\sinit(\eps,q,h)$ is invoked, the continual mechanism $\svt_{h, \eps, q}$ with the single initial state $(\eps, q, h, \tau=\varnull)$ is created. In the pseudocode of $\sinit$, the variable $\tau$ takes a value. In our setting, such initialization is considered to happen when $\svt_{h, \eps, q}$ receives its first input. The $\call$ operation then describes how $\svt_{h,\eps,q}$ maps its current state (i.e., set of variables) and an input message $m=(x, \thr)\in \zo^d\times \R$ to a new state (with updated variables) and a response $m'\in\{\top, \bot\}$.

\begin{algorithm}
    \caption{Sparse vector technique for histogram queries}\label{alg:svt}
    \begin{algorithmic}
        \Procedure{$\sinit$}{$\eps, q, h$}:
        \State $h_i\gets h_i$ for all $i \in [d]$
        \State $\eps = \eps$
        \State $q = q$
        \State $\tau \gets \Lap(1/\eps)$ 
        \EndProcedure
        \Procedure{$\call$}{$x, \thr$}:
        \State $h_i = h_i +  x_i$ for all $i \in [d]$
        \If{$q(h) + \Lap(2/\eps) > \thr + \tau$}
            \State return $\top$ and halt
        \Else
            \State return $\bot$
        \EndIf
    \EndProcedure
    \end{algorithmic}
\end{algorithm}

As part of their privacy proof, the following is shown in~\cite{henzingerSS2023}:
\begin{lemma}\label{lem:svt-privacy}
    If $q$ is a $1$-sensitive function, then for every $\eps$ and $h$, the mechanism $\svt_{\eps/6, q, h}$ is $\eps/3$-DP w.r.t. the neighbor relation $\sim_S$ defined as follows (see Remark~\ref{rem:neighbor-to-adaptive-ver-func}). For $T\in\N$, two input streams $\left((x_1, \tau_1), \dots, (x_T, \tau_T)\right)$ and $\left((y_1, \nu_1), \dots, (y_T, \nu_T)\right)$ in $\left((\zo^d\times \R)\right)^T$ are neighboring if the (threshold) sequences $(\tau_1, \dots, \tau_T)$ and $(\nu_1, \dots, \nu_T)$ are identical at every entry and the (binary) sequences $(x_1, \dots, x_T)$ and $(y_1, \dots, y_T)$ are identical in all but one entry.
\end{lemma}
}

\noindent\medskip \textbf{$d$-dimensional continual counter.} {
The mechanism $\tth$ is described by two operations. The first, $\countinit(\eps)$, creates $d$ continual counters that are $\eps/d$-DP w.r.t. the neighbor relation $\sim'$, where two real-valued sequences of the same length are neighbors if they differ only on a single step by at most $1$. The second, $\countupdate(c)$, gives the elements of its $d$-dimensional input $c\in \R^d$ to their corresponding counters and returns the aggregated outputs as a $d$-dimensional vector. We denote this continual mechanism by $\tth_\eps$. In another version, $\countinit(\eps, \delta)$ creates $d$ continual counters, each satisfying $(\eps/d, \delta/d)$-DP w.r.t. $\sim'$. We represent this continual mechanism by $\tth_{\eps, \delta}$.

By definition, mechanisms $\tth_\eps$ and $\tth_{\eps, \delta}$ are the concurrent composition of $d$ many $\eps/d$- and $(\eps/d, \delta/d)$-DP fixed private continual counters. We can assume that these continual counters return an integer value at each point of time, and thus they have discrete answer distributions. If the output space is real, we use a counter with integer outputs that post-processes the original counter by rounding its output up to the nearest integer. By the post-processing lemma, the privacy guarantees still hold and the additive error is only negligibly affected by at most $1$. Therefore, Theorem~\ref{thm:comp-fixed-cm} is applicable and by the basic composition, we have:

\begin{corollary}\label{cor:d-dim-counter-privacy}
    Mechanisms $\tth_\eps$ and $\tth_{\eps, \delta}$ are $\eps$-DP and $(\eps, \delta)$-DP w.r.t. the neighbor relation $\sim_C$ defined as follows (see Remark~\ref{rem:neighbor-to-adaptive-ver-func}). For $T\in \N$, two input streams $x=(x_1, \dots, x_T)$ and $y=(y_1, \dots, y_T)$ in $(\R^d)^T$ are neighboring if there exists $i\in [T]$ such that $x_j=y_j$ for all $j\in [T]\setminus \{i\}$ and $|x_i-y_i|_\infty \leq 1$, i.e., each entry of $x_i$ and $y_i$ differs by at most $1$.
\end{corollary}
}

\noindent\medskip \textbf{Continual histogram query mechanism.} {
The mechanism $\tthss$, described in Algorithm~\ref{alg:histogram}, is initialized with a privacy parameter $\eps>0$, a query $q:\N^d\to \N$, and an accuracy parameter $0<\beta\leq 1$ and receives a vector $x \in \{0,1\}^d$ as input at each time step. This mechanism partitions the input stream into subsequences, called \emph{intervals}. At the end of each interval, it forwards the $d$-dimensional sum of all the inputs of the interval to the differentially private $d$-dimensional continual counter $\tth$. (Note that $\tth$ can be used directly to solve the problem; however, $\tthss$ achieves better accuracy bounds in many cases.) The mechanism $\tthss$ maintains two $d$-dimensional vectors: (1) variable $c$, where $c_i$ stores the exact sum of coordinate $i$ within the current interval, and (2) variable $h$, where $h_i$ stores the approximate sum of coordinate $i$ over all \emph{previous} intervals computed by $\tth$. Both vectors are initialized to $0^d$.  

At the beginning of each interval (including the initial time step), $\tthss$ creates a fresh instance of the Sparse Vector Technique mechanism $\svt$ with privacy parameter $\eps/6$, query $q$, and initial vector $h=0^d$ (denoted by $\svt_{\eps/6,q,h}$). The mechanism discards the previous instance. $\tthss$ keeps track of the current time step and the current interval using variables $t$ and $j$, respectively. In Algorithm~\ref{alg:histogram}, $\gamma$ and $\xi$ are some deterministic functions, mapping $t,j,\beta, \eps$ to a real value. 

The mechanism $\tthss$ also stores a threshold variable $\thr$, initially set to $\gamma(1,1,\beta,\eps)$. Upon receiving an input $x$ at step $t$, the mechanism updates $c$ and invokes $\call(x,\thr)$. If the response, denoted by $a$, is $\top$, the current interval ends and $\tthss$ takes the following actions: (threshold update) It first privately evaluates $q(c+h)$ via the Laplace noise and compares the noisy result to $\thr-\xi(t,j,\beta,\eps)$. If the inequality holds, $\thr$ is suitably increased. Then independent of the outcome of this comparison, $\thr$ is updated again to reflect the increase in the number of intervals. (summation update) Subsequently, $\tthss$ passes $c$ to $\tth$, updates $h$ with the output of $\tth$, and resets $c$ to $0^d$. Since $\svt$ halts after returning $\top$, $\tthss$ also creates a new instance $\svt_{\eps/6,q,h}$. 

Finally, regardless of the output of $\call(x,\thr)$, $\tthss$ updates $\thr$ once more to account for the increased number of time steps and outputs $h$. 

\begin{algorithm}
    \caption{Mechanism for answering a monotone histogram query}\label{alg:histogram}
    \begin{algorithmic}
        \Procedure{$\histinit$}{$\eps, q, \beta$}:
            \State $\countinit(\eps/3)$ (or $\countinit(\eps/3,\delta)$)
            \State $h_i\gets 0$ for all $i \in [d]$, $c_i\gets 0$ for all $i \in [d]$
            \State $out \gets q(h)$
            \State $j \gets 1$
            \State $\sinit(\eps/6,q,h)$
            \State $a\gets \bot$
            \State $\thr = \gamma(1,j,\beta,\eps)$ 
            \State $t\gets 0$
        \EndProcedure
        \Procedure{$\histupdate$}{$x$}:
            \State $t\gets t+1$
            \State $c_i\gets c_i + x_i$ for all $i \in [d]$
            \State $a= \call(x, \thr)$
            \If{$a == \top$}
                \If{$q(c + h) +\Lap(3/\eps) > \thr - \xi(t,j,\beta, \eps)$}
                \State $\thr = \thr + \gamma(t,j,\beta,\eps)$
                \EndIf
                \State $h = \countupdate(c)$ 
                \State $out \gets q(h)$
                \State $\sinit(\eps/6,q,h)$
                \State $c_i\gets 0$ for all $i \in [d]$
                \State $j \gets j + 1$
                \State $\thr = \thr - \gamma(t,j-1,\beta,\eps) + \gamma(t,j,\beta,\eps)$
            \EndIf
            \State $\thr = \thr - \gamma(t,j,\beta,\eps) + \gamma(t+1,j,\beta,\eps)$  
            \State return $out$
    \EndProcedure
    \end{algorithmic}
\end{algorithm}
}
\subsection{New Privacy Analysis.}\label{subsec:application-new-privacy}
In Algorithm~\ref{alg:histogram}, the summation $q(c + h) +\Lap(3/\eps)$ can be rewritten as an instance of Laplace mechanism with privacy parameter $\eps/3$ and input $q(c+h)$. The privacy of the Laplace mechanism is defined with respect to the neighbor relation $\sim_L$, where two real values $r_1, r_2\in \R$ satisfy $r_1\sim_L r_2$ if and only if $|r_1-r_2|\leq 1$. Consequently, the mechanism $\tthss$ executes:
\begin{enumerate}
    \item one instance of the $d$-dimensional counter $\tth$, which consists of $d$ instances of a scalar counter,
    \item an arbitrary number of $\svt$ instances, and 
    \item an arbitrary number of Laplace mechanisms.  
\end{enumerate}
By Corollary~\ref{cor:d-dim-counter-privacy}, $\tth$ is $\eps/3$-DP (or $(\eps/3,\delta)$-DP when approximate counters are used) w.r.t. $\sim_C$. By Lemma~\ref{lem:svt-privacy}, each $\svt$ instance is $\eps/3$-DP w.r.t. $\sim_S$. Finally, the Laplace mechanism with parameter $\eps/3$ is a non-interactive mechanism (NIM) that is well known to satisfy $\eps/3$-DP on neighboring real inputs differing by at most $1$ (i.e., w.r.t. to neighbor relation $\sim_L$ defined above). The idea is to show that the concurrent composition of multiple instances of these mechanisms created by $\tthss$ is $\eps$-DP (or $(\eps, \delta)$-DP when $\tth$ is approx DP). Then, we will show that $\tthss$ satisfies the required properties in Corollary~\ref{cor:pseudocode}, and thus it has the same privacy guarantees as the concurrent composition. As concurrent composition theorems in this paper require the composed mechanisms to have discrete answer distributions, we first show that we can assume the output space of $\tth$ and Laplace mechanism is discrete. Since $\svt$ has a finite output space, its answer distributions are also discrete.

\medskip
\noindent\textit{Discretization of outputs.}  
As argued earlier, the counters executed by $\tth$ output integer values, and thus $\tth$ has discrete answer distributions. Likewise, we can assume that the Laplace mechanism in $\tthss$ rounds its noisy real-valued output to the nearest integer. By the post-processing lemma, this rounding does not degrade privacy, and by definition, it increases the additive error by at most $1$, which is unimportant for the accuracy guarantees of~\cite{henzingerSS2023}. Hence, both $\tth$ and the Laplace mechanism can be treated as having discrete output distributions. 

\medskip
\noindent\textit{Pseudocode privacy analysis.}  
Recall from Section~\ref{subsec:pseudocode} that untainted variables of a mechanism are (i) outputs of submechanisms, (ii) fixed values (e.g., privacy and accuracy parameters), or (iii) variables computed solely as functions of other untainted variables. In Algorithm~\ref{alg:histogram}, the variables $\eps$, $\beta$, $q$, $a$, $h$, $out$, $j$, $t$, and $\thr$ are untainted, while $c$ is not: The values of the variables $\eps$, $\beta$, and $q$ are fixed and independent of the sensitive data. Variables $a$ and $h$ denote the privatized outcomes of $\svt$ and $\tth$, respectively. Variable $out$ equals $q(h)$, where $h$ is an untainted variable. Index $j$ is initially $1$ and is increased by $1$ when $a=\top$. Time step $t$ is initially $0$ and is increased by $1$ at each step, independent of the input and other variables. Finally, variable $\thr$ is initialized and updated to values determined by the untainted variables $\eps, \beta, j, t$. It is important to note that the decision of updating variables $h$ (and consequently $out$) and $\thr$ depends only on the privatized output of $\svt$, i.e., the untainted variable $a$.

We first show that the destination property holds. For that we have to show that the decision of whether to produce an output, create a new submechanism (and if so, which one) or interact with an existing one (and if so, which one) is solely based on untainted variables. All submechanisms of $\tthss$ are instantiated either in the initial step (and, thus, whether it receives input is independent of tainted variables and the inputs $x$) or based on outputs of the private submechanism $\svt$. The active $\svt$ instance receives input at every step, and thus whether it receives input is independent of tainted variables and the inputs $x$. Also, whether the mechanism $\tth$ gets updated depends only on the privatized output of $\svt$. Finally, $\tthss$ returns an output after the final update of $\thr$, without any condition on the inputs $x$ or the tainted variables. Thus, the \emph{destination property} in Corollary~\ref{cor:pseudocode} holds. Moreover, the final output of $\tthss$ is $out$, which is an untainted variable, so the \emph{response property} in that corollary also holds. It remains to show $\tthss$ satisfies the $f\to f'$ property, where $f$ is the verification function corresponding to the neighbor relation $\sim_H$ in Definition~\ref{def:histogram-problem} and $f'$ is a verification function satisfying that the $f'$-concurrent composition of the continual mechanisms is $(\eps, \delta)$-DP.

Let $T\in\mathbb{N}$, and let $x=(x_1,\dots,x_T)$ and $y=(y_1,\dots,y_T)$ be two sequences in $\left(\{0,1\}^d\right)^T$ satisfying $x\sim_H y$, i.e., there exists $i\in[T]$ such that $x_j=y_j$ for all $j\in [T]\setminus \{i\}$. Fix arbitrary values for all untainted variables at every time step. Conditioned on these values, we compare the input streams of $\tthss$'s submechanisms when this mechanisms runs on $x$ and $y$:
\begin{itemize}[left=0pt]
    \item \textbf{Inputs of $\tth$:}
    By construction, the input stream of $\tthss$ is partitioned into intervals determined by the time steps at which $\svt$ outputs $\top$. Since these outputs are untainted, they are identical in both executions to the fixed values specified in the condition. Let $1\le t_1<\dots<t_k\le T$ be the steps $\svt$ outputs $\top$, and set $t_0=0$. The submechanism $\tth$ receives the interval sums $\left(\sum_{\ell=t_{j-1}+1}^{t_j}x_\ell\right)_{j=1}^k$ and $\left(\sum_{\ell=t_{j-1}+1}^{t_j}y_\ell\right)_{j=1}^k$ as input sequence when $\tthss$ runs on $x$ and $y$, respectively. All intervals not containing $i$ have identical sums. If $i>t_k$, no (closed) interval includes $i$ and thus all sums coincide. If $i\leq t_k$, let $j^*\in [k]$ denote the index satisfying $t_{j^*-1}+1\le i\le t_{j^*}$. Then
    \[
        \sum_{\ell=t_{j^*-1}+1}^{t_{j^*}} x_\ell
        -
        \sum_{\ell=t_{j^*-1}+1}^{t_{j^*}} y_\ell
        = x_i-y_i .
    \]
    Since $x_i,y_i\in\{0,1\}^d$, we have $x_i-y_i\in\{-1,0,1\}^d$ and thus $\|x_i-y_i\|_\infty\le 1$. Hence the inputs to $\tth$ differ in at most one coordinate by at most one, i.e., they are $\sim_C$-neighbors (see Corollary~\ref{cor:d-dim-counter-privacy}). 
    \item \textbf{Inputs of $\svt$:} The threshold variable $\thr$ is untainted. Let $\thr_\ell$ denote its fixed value at time $\ell$. Let time steps $t_1, \dots, t_k$ be defined as above. Each time $\svt$ outputs $\top$, the mechanism $\tthss$ reinitializes it. Therefore, $\tthss$ initiates $k+1$ instances of $\svt$. Let $t_0=0$ and $t_{k+1}=T$. For each $j\in [k+1]$, the $j$-th instance of $\svt$ receives $\left((x_\ell, \thr_\ell)\right)_{\ell=t_{j-1}+1}^{t_j}$ and $\left((y_\ell, \thr_\ell)\right)_{\ell=t_{j-1}+1}^{t_j}$ as input sequence when $\tthss$ runs on $x$ and $y$, respectively.
    All instances except possibly $j^*$ (the one containing $i$) receive identical inputs. Let $j^*\in [k+1]$ denote the index where $t_{j^*-1}+1\leq i\leq t_{j^*}$. For every $j\in [k+1]\setminus \{j^*\}$, the input sequences of the $j$-th $\svt$ submechanism are identical when $\tthss$ runs on $x$ and $y$. Moreover, the input sequences of the $j^*$-th $\svt$ submechanism are identical except for the inputs $(x_i, \thr_i)$ and $(y_i, \thr_i)$, which have the same threshold $\thr_i$. Hence, these input sequences are $\sim_S$-neighbors (see Lemma~\ref{lem:svt-privacy}).
    \item \textbf{Inputs of the Laplace Mechanism:} The Laplace mechanism is applied to the value $q(h+c)$, where $c$ is the interval sum. The variable $h$ is untainted and therefore fixed and identical in both executions of $\tthss$. As shown above, all interval sums coincide except possibly one, and that one differs by at most $1$ in $\ell_\infty$-norm. Hence the corresponding values $h+c$ also differ by at most $1$ in $\ell_\infty$-norm. Since $q$ is $1$-sensitive (see Definition~\ref{def:1-sensitive}), the real-valued outputs of $q$ on these inputs differ by at most $1$. Therefore, when $\tthss$ runs on $x$ and $y$, all corresponding Laplace submechanisms receive identical inputs, except possibly one pair whose inputs differ by at most $1$, that is, they are $\sim_L$-neighbors.
\end{itemize}
Hence, the concurrent composition of all submechanisms of $\tthss$ can be decomposed into three concurrent compositions: (1) all Laplace mechanisms, (2) all $\svt$ instances, and (3) the $d$ continual counters forming $\tth$. 
The first two are the parallel compositions of purely differentially private mechanisms. By Corollary~\ref{cor:parallel-comp-pure}, both are $\eps/3$-DP. The third one is, by construction, the concurrent composition of $d$ fixed continual counters. In Corollary~\ref{cor:d-dim-counter-privacy}, we showed that this composition, i.e., $\tth$, is $\eps/3$-DP (or $(\eps/3,\delta)$-DP when approx DP counters are used).  

By basic composition and Theorem~\ref{thm:comp-fixed-cm}, the overall concurrent composition of all three groups of mechanisms is $\eps$-DP when $\tth$ uses pure counters, and $(\eps,\delta)$-DP when approximate counters are used. (More precisely, it is DP w.r.t. a verification function $f'$ that ensures all submechanisms but one $\svt$ instance, one Laplace mechanism, and the $d$ counters of $\tth$ are assigned pairs of identical inputs, while the first and second inputs for each of these mechanisms form neighboring sequences.)

By Corollary~\ref{cor:pseudocode}, since $\tthss$ satisfies the destination, response, and neighboring-mapping properties, and since the concurrent composition of its submechanisms is $\eps$-DP (or $(\eps,\delta)$-DP), the mechanism $\tthss$ satisfies the same privacy guarantee. We recall that all privacy guarantees in this section hold against adaptive adversaries.

\section{Analyzing Local Differential Privacy}\label{sec:ldp}
So far, we have discussed the application of concurrent composition theorems in designing continual mechanisms with modular structures. Another natural application of these theorems arises in the setting of \emph{local differential privacy} (LDP), where multiple users interact with an untrusted server over (possibly) multiple rounds. Each user holds a private data record (or dataset) and, upon receiving a query from the server, produces a response that satisfies differential privacy. After several rounds of interaction, the server aggregates the received messages and outputs an estimate of a function over the users' data. Formally, a \emph{local protocol} consists of an interactive mechanism for the server and an interactive mechanism for each user. Local differential privacy is defined with respect to a neighbor relation on sequences of datasets held by the users.

\begin{definition}\label{def:ldp}
	Let $P$ be a local protocol with $n$ users. Let $\calX$ denote the domain of datasets held by the users, and let $\sim$ be a neighbor relation on $\calX^n$. For example, two dataset sequences $x=(x_1,\dots,x_n)$ and $x'=(x_1',\dots,x_n')$ may satisfy $x\sim x'$ if they differ in exactly one coordinate $i$, where $x_i\sim^* x_i'$ for a neighbor relation $\sim^*$ on $\calX$.  The protocol $P$ is said to satisfy $(\eps,\delta)$-\emph{local differential privacy} (or $(\eps,\delta)$-LDP) w.r.t. $\sim$ if, for every pair of neighboring dataset sequences $x\sim x'$, the distributions of the server’s view---consisting of its random coins and the transcript of all exchanged messages---are $(\eps,\delta)$-indistinguishable when users hold the datasets in $x$ and $x'$.
\end{definition}

\noindent\underline{Equivalence of a local protocol $P$ being LDP to a continual mechanism $\cM_P$ being DP.}
Let $P$ be a local protocol as defined above, and let $\sim$ denote the neighboring relation on the sequences of users’ datasets. Define $f_\sim$ to output $\top$ only when it receives a single message of the form $(x,x')$ where $x \sim x'$. We construct a continual mechanism $\cM_P$ that receives a dataset sequence as (a single) input, simulates the interaction of the server with the users in the protocol $P$, and returns the view of the server as output. (Indeed, $\cM_P$ is a NIM, which is a special case of CMs.) We will show that $P$ is $(\eps,\delta)$-LDP w.r.t. $\sim$ if and only if $\cM_P$ is $(\eps,\delta)$-DP w.r.t. $f_\sim$. We define $\cM_P$ as
$$\cM_P= \calS_P\circstar \extconcomp,$$
where $\calS_P$ is an IPM modeling the server’s behavior. 
The mechanism $\calS_P$ is specifically defined as follows: $\calS_P$ accepts only one left message of the form $x=(x_1, \dots, x_n)$, where each $x_i$ is a dataset for user $i$. Upon receiving $x$, $\calS_P$ takes the following actions: (1) It creates $n$ user mechanisms by sending appropriate creation queries to $\extconcomp$, and (2) it sends each $x_i$ as the first message to the corresponding user mechanism through $\extconcomp$. (3) It then interacts with the user mechanisms according to the server’s interaction rules defined in protocol $P$. (4) Once the server decides to halt, $\calS_P$ outputs the server’s view as its (left) response and then terminates. This view consists of the server’s internal randomness (i.e., random coins used by $\calS_P$) and the transcript of all messages sent and received by the server. Thus, it does not include the creation queries or the initial dataset messages sent by $\calS_P$.

\begin{claim}\label{cla:p-ldp-mp-dp}
    The local protocol $P$ is $(\eps, \delta)$-LDP w.r.t. the neighboring relation $\sim$ if and only if the CM $\cM_P$ is $(\eps, \delta)$-DP w.r.t. the verification function $f_\sim$.
\end{claim}
\begin{proof}
    The claim holds by the construction of $\cM_P$ and Definition~\ref{def:ldp}.
\end{proof}

This equivalence allows us to analyze the privacy of the protocol $P$ by examining the privacy of the corresponding CM $\cM_P$, which has a modular design, using Theorem~\ref{thm:dp-complex-mech}.

\begin{remark}\label{rem:sample-generating-mech}
    While the server in $P$ may use randomness, Theorem~\ref{thm:dp-complex-mech} requires $\calS_P$ to be deterministic. This can be handled by defining $\calS_P$ to request $\extconcomp$ to create a $(0,0)$-DP mechanism that generates random coins or random samples whenever queried.
\end{remark}

\begin{claim}\label{cla:dp-sat-dest-and-resp}
    The IPM $\calS_P$ satisfies the destination and response properties defined in Section~\ref{subsec:properties} and required by Theorem~\ref{thm:dp-complex-mech}.
\end{claim}

\begin{proof}
    By construction, upon receiving the left message $x$, the IPM $\calS_P$ in $\cM_P$ first sends exactly $2n$ messages to $\extconcomp$: $n$ creation queries and $n$ dataset messages. Although the content of these dataset messages depends on $x$, the decision of $\calS_P$ sending a dataset message to its right mechanism, $\extconcomp$, is completely independent of $x$. (If the server is randomized, $\calS_P$ additionally creates an auxiliary $(0,0)$-DP mechanism. This creation query to $\extconcomp$ is also independent of $x$.)

    After these initial interactions, $\calS_P$ proceeds to interact with the user mechanisms according to the server’s specification in $P$. The decision of whether to continue the interaction or to return a left response and halt is solely determined by the control flow of the server, which depends only on the users' (and the auxiliary mechanism's) responses and not directly on $x$. (The server has no access to $x$.) Hence, $\calS_P$ satisfies the destination property.
    
    Moreover, the (single) left response of $\calS_P$ is the server’s view, consisting of the messages exchanged with (the auxiliary mechanism and) the user mechanisms after their dataset is set. By design, this depends only on the privatized outputs of the mechanisms executed by $\extconcomp$. Therefore, $\calS_P$ also satisfies the response property.
\end{proof}

To apply Theorem~\ref{thm:dp-complex-mech}, it remains to show that for a verification function $f$, the IPM $\calS_P$ satisfies the $f_\sim \to f$ property (see Section~\ref{subsec:properties}) and the $f$-concurrent composition of CMs is DP. We now propose a verification function $f$ for which the $f_\sim \to f$ property holds. Intuitively, this verification function enforces that all user mechanisms are first created, and that each is assigned a pair of datasets $(x_i, x_i')$ satisfying $x \sim x'$. After this, queries $q$ can be asked concurrently from the users. Local protocols are typically defined so that the queries asked from each mechanism satisfy certain properties~$\prop$, which play a key role in the privacy analysis. To later be able to show privacy guarantees for the $f$-concurrent composition of CMs, we also define $f$ to enforce the queries asked from each user to satisfy $\prop$. The formal definition of $f$ is given below. For convenience, we refer to the $i$-th created mechanism as user~$i$.

\begin{definition}[$\sim_i$, $f_{\sim_i, \prop_i}$, and $f_{P,\sim}^{\eps_1,\delta_1,\dots,\eps_n,\delta_n}$]\label{def:ver-func-ldp-neighbor-prop}
    Let $P$ be a local protocol with $n$ users, where the queries issued by the server to each user~$i\in[n]$ satisfy a property~$\prop_i$. Let $\sim$ denote a neighboring relation on the sequences of users’ datasets. For every $i\in [n]$, define the neighbor relation~$\sim_i$ on datasets of user~$i$ as follows: for datasets $x_i,x_i'$, we have $x_i\sim_i x_i'$ if and only if there exist datasets $x_j,x_j'$ for all $j\in[n]\setminus\{i\}$ such that $x\sim x'$. (For example, assume $x$ represents a graph with $n$ vertices, where $x_i$ denotes the set of vertices adjacent to vertex $i$. Suppose $x\sim x'$ iff the corresponding graphs of $x$ and $x'$ differ by an edge. Then $x_i\sim_i x_i'$ iff the adjacency sets $x_i$ and $x_i'$ of vertex $i$ differ in the presence or absence of a single vertex.)
    
    Define the verification function $f_{\sim_i,\prop_i}$ to accept a message sequence if and only if (1) the first message is a dataset pair $(x_i,x_i')$ with $x_i\sim_i x_i'$, and (2) all subsequent messages are pairs of identical queries $(q,q)$, where the sequence of~$q$ values satisfies~$\prop_i$.  
    
    For every $i\in[n]$, let $\cM_i$ denote the CM associated with user~$i$ and let $s_i$ be its single initial state. Assume $\cM_i$ is $(\eps_i,\delta_i)$-DP w.r.t. $f_{\sim_i,\prop_i}$.  
    Then the verification function $f_{P,\sim}^{\eps_1,\delta_1,\dots,\eps_n,\delta_n}$ is defined as follows: for every $t\in\N$ and message sequence $(m_1,\dots,m_t)$, we have $f_{P,\sim}^{\eps_1,\delta_1,\dots,\eps_n,\delta_n}(m_1,\dots,m_t)=\top$ if and only if
    \begin{enumerate}[label=(\roman*), left=0pt]
        \item the sequence $(m_1,\dots,m_t)$ is suitable; and
        \item it is a prefix of a sequence consisting of the following messages, in order:
        \begin{enumerate}[label=(\alph*)]
            \item pairs of identical creation queries $(\alpha_i,\alpha_i)$ for each user $i\in[n]$, where $\alpha_i=(\cM_i,s_i,\eps_i,\delta_i,f_{\sim_i,\prop_i})$ (and possibly an additional pair for an auxiliary sample-generating mechanism if the server is randomized);
            \item dataset-assignment messages $((x_i,i),(x_i',i))$ for each user~$i$, where $x\sim x'$; and
            \item arbitrarily many identical query messages $((q,i),(q,i))$ directed to the same user~$i$.
        \end{enumerate}
    \end{enumerate}
    Note that condition~(i) ensures that the query sequence for each user~$i$ remains valid with respect to~$f_{\sim_i,\prop_i}$.
\end{definition}

\begin{claim}\label{cla:sp-sat-f-map}
    The IPM $\calS_P$ satisfies the $f_\sim \to f_{P, \sim}^{\eps_1, \delta_1, \dots, \eps_n, \delta_n}$ property.
\end{claim}

\begin{proof}
    By definition, every $f_\sim$-valid message sequence consists of a single message of the form $(x,x')$ with $x \sim x'$. Consider two executions of $\calS_P$ receiving $x$ and $x'$ as their respective left messages. By definition, both IPMs first identically create user mechanisms. Then one instance assigns datasets $x_i$ to the users by sending the right messages $(x_i, i)$, while the other assigns datasets $x_i'$ by sending $(x_i', i)$. Once all datasets have been assigned, the future behavior of both instances becomes independent of their left inputs—their actions are determined solely by the server description in $P$, which has no access to the datasets and depends only on the (privatized) responses returned by the user mechanisms. Hence, if both instances of $\calS_P$ receive identical right messages (from user mechanisms), they generate identical right queries to the same users, i.e., identical messages of the form $(q, i)$. Consequently, the sequence $\sigma$ formed by pairing the right messages produced by these two instances satisfies condition (ii) in Definition~\ref{def:ver-func-ldp-neighbor-prop}. We will show that this sequence also satisfies condition (i) in this definition, and thus it is $f_{P,\sim}$-valid. Consequently, $\calS_P$ satisfies the $f_\sim \to f_{P,\sim}$ property. 

    For condition (i), we recall that a message sequence is \emph{suitable} if each message is either (i) a pair of identical creation queries $(\cM, s, \eps, \delta, f)^2$, where the CM~$\cM$ with single initial state~$s$ is $(\eps,\delta)$-DP w.r.t.~$f$, or (ii) a message of the form $\big((q_0,\ell),(q_1,\ell)\big)$, where at least $\ell$ mechanisms have already been created and $q_0,q_1$ are valid inputs for the $\ell$-th one. Let $\cM_\ell$ denote the $\ell$-th created mechanism and let $f_\ell$ be its associated verification function. The suitability condition also requires that the sequence of query pairs $(q_0,q_1)$ extracted from messages of the form $\big((q_0,\ell),(q_1,\ell)\big)$ be $f_\ell$-valid. By construction, the paired sequence~$\sigma$ satisfies these conditions and is therefore suitable, finishing the proof.
\end{proof}

Therefore, by Claim~\ref{cla:dp-sat-dest-and-resp} and Claim~\ref{cla:sp-sat-f-map}, Theorem~\ref{thm:dp-complex-mech} is applicable. This theorem combined with Claim~\ref{cla:p-ldp-mp-dp} implies:
\begin{theorem}\label{thm:ldp-to-concomp}
    Let $P$ be a local protocol with $n$ users, where the queries asked by the server of each user $i\in[n]$ satisfy a property~$\prop_i$. Let $\sim$ be a neighbor relation on sequences of $n$ user datasets, and for each $i\in[n]$, define the verification function $f_{\sim_i,\prop_i}$ as in Definition~\ref{def:ver-func-ldp-neighbor-prop}. Suppose that for parameters $\eps_1,\dots,\eps_n,\eps\geq 0$ and $0\leq \delta_1,\dots,\delta_n,\delta\leq 1$, the following hold:
    \begin{enumerate}[left=0pt,label=(\roman*)]
        \item for every $i\in[n]$, the CM corresponding to user~$i$ is $(\eps_i,\delta_i)$-DP w.r.t. $f_{\sim_i,\prop_i}$; and
        \item the $f_{P,\sim}^{\eps_1,\delta_1,\dots,\eps_n,\delta_n}$-concurrent composition of CMs is $(\eps,\delta)$-DP.
    \end{enumerate}
    Then the protocol~$P$ is $(\eps,\delta)$-LDP w.r.t. $\sim$.
\end{theorem}

We next simplify this result for a choice of the neighboring relation $\sim$ based on a shared neighbor relation for the users. Let $\sim^*$ denote a neighbor relation on (individual) users’ datasets. For $k\in\N$, define $\sim$ on sequences of $n$ user datasets as follows: For $x,x'$, we have $x\sim x'$ if and only if they differ in at most $k$ indices, and for each differing index $i$, we have $x_i\sim^*x_i'$. We refer to $\sim$ as the \emph{$k$-sparse parallel-composition neighbor relation with user-neighbor relation~$\sim^*$}. We note that the relation~$\sim^*$ is implicitly assumed to be reflexive, meaning each dataset is related to itself. Hence, the neighbor relation~$\sim_i$ associated to each user~$i$ in Definition~\ref{def:ver-func-ldp-neighbor-prop} equals~$\sim^*$.

\begin{corollary}\label{cor:ldp-parallel-concomp}
    Let $P$ be a local protocol with $n$ users following identical procedures, and let $\sim^*$ be a neighbor relation on individual users’ datasets. For $k\in\N$, let $\sim_{k,\sim^*}$ denote the $k$-sparse parallel-composition relation with user-neighbor relation~$\sim^*$, defined above. Assume that in protocol~$P$, the queries asked by the server of each user satisfy a property~$\prop^*$. Let $\eps'> 0$, $\eps\geq 0$, and $0\leq \delta',\delta\leq 1$. Suppose the following conditions hold:
    \begin{itemize}
        \item The CM corresponding to each user is $(\eps', \delta')$-DP w.r.t. the verification function $f_{\sim^*, \prop^*}$ from Definition~\ref{def:ver-func-ldp-neighbor-prop}. 
        \item The composition of $k$ instances of the randomized response mechanism $\RR_{\eps',\delta'}$ is $(\eps,\delta)$-DP.
    \end{itemize}
    Then the protocol~$P$ is $(\eps,\delta)$-LDP w.r.t. $\sim$.
\end{corollary}

To prove Corollary~\ref{cor:ldp-parallel-concomp}, we need the following observation:

\begin{observation}\label{obs:f2-concomp-imply-f1-concomp}
    Let $f_1$ and $f_2$ be verification functions such that every message sequence that is $f_1$-valid is also $f_2$-valid. If a CM~$\cM$ is $(\eps,\delta)$-DP w.r.t.~$f_2$, then $\cM$ is also $(\eps,\delta)$-DP w.r.t.~$f_1$.
\end{observation}

\begin{proof}
    By assumption, the IPMs $\V[f_1]\circstar\I\circstar\cM$ and $\V[f_1]\circstar\V[f_2]\circstar\I\circstar\cM$ are equivalent. Hence, by the post-processing lemma, the claim follows immediately.
\end{proof}

\begin{proof}[Proof of Corollary~\ref{cor:ldp-parallel-concomp}]
    Recall the verification function $f^{(\eps', \delta')_{i=1}^{k}}_{\infty, \mathit{FPC}}$ from Definition~\ref{def:ver-func-par-fixed-param-comp-fpc}. This verification function permits the creation of an arbitrary number of CMs over time, but guarantees that at most $k$ of them—each with privacy parameters $(\eps', \delta')$—are assigned pairs of non-identical inputs. Moreover, it enforces that if the first message directed to a mechanism contains identical inputs, then all subsequent messages to that mechanism must also contain identical inputs. Comparing the verification functions $f_{P, \sim}^{(\eps', \delta')_{i=1}^{n}}$ (Definition~\ref{def:ver-func-ldp-neighbor-prop}) and $f^{(\eps', \delta')_{i=1}^{k}}_{\infty, \mathit{FPC}}$, we observe that every $f_{P, \sim}^{(\eps', \delta')_{i=1}^{n}}$-valid message sequence is also $f^{(\eps', \delta')_{i=1}^{k}}_{\infty, \mathit{FPC}}$-valid. Applying Observation~\ref{obs:f2-concomp-imply-f1-concomp} with $\cM = \extconcomp$, we conclude that if the $f^{(\eps', \delta')_{i=1}^{k}}_{\infty, \mathit{FPC}}$-concurrent composition of CMs is $(\eps,\delta)$-DP, then so is the $f_{P, \sim}^{(\eps', \delta')_{i=1}^{n}}$-concurrent composition of CMs.
    
    By Theorem~\ref{thm:parallel-fixed-param-comp-fpc}, if the composition of $k$ randomized response mechanisms $\RR_{\eps', \delta'}$ is $(\eps,\delta)$-DP, then the $f^{(\eps', \delta')_{i=1}^{k}}_{\infty, \mathit{FPC}}$-concurrent composition of CMs is also $(\eps,\delta)$-DP. Hence, under the second assumption of Corollary~\ref{cor:ldp-parallel-concomp}, the $f_{P, \sim}^{(\eps', \delta')_{i=1}^{n}}$-concurrent composition is $(\eps,\delta)$-DP. Furthermore, by the definition of~$\sim$, for each $i\in[n]$, the neighbor relation $\sim_i$ in Definition~\ref{def:ver-func-ldp-neighbor-prop} coincides with~$\sim^*$. Thus, the first assumption of the corollary is equivalent to requiring that each user’s CM is $(\eps_i,\delta_i)$-DP w.r.t. $f_{\sim_i,\prop_i}$. Thus, applying Theorem~\ref{thm:ldp-to-concomp} yields that the protocol~$P$ is $(\eps,\delta)$-LDP w.r.t. $\sim$, completing the proof.
\end{proof}

\begin{remark}
    In real-world deployments, users may join or leave the system dynamically, and data records may be shared across a subset of users. Furthermore, a user’s local dataset may evolve over time. To cover these scenarios, users are CMs rather than IMs, and in the above discussions, instead of $\cM_P$ receiving only a single message $x=(x_1, \dots, x_n)$, this mechanism receives creation queries for creating a new user, deletion messages informing that a specific user is gone, and messages of the form $(u, j)$ requesting $\cM_P$ to forward the dataset-update $u$ to user $j$. Without repeating the whole argument, we note that Theorem~\ref{thm:ldp-to-concomp} and Corollary~\ref{cor:ldp-parallel-concomp} can be adapted and applied in this general scenario and our concurrent composition theorems for CMs (rather than IMs) can be used to simplify the privacy analysis of these local protocols.
\end{remark}

\begin{remark}
    In this section, we assumed that all components (i.e., the server and the users) are honest, meaning that they follow the local protocol. We also assumed that there is no collusion, that is, parties do not share their internal states with each other. Suppose now that $k$ users hold private data, while the remaining $n-k$ users and the server may be dishonest and collude to learn this data. We can model the colluding parties as a single adversary that interacts concurrently with the $k$ target users. In this case, privacy of the local protocol follows from the concurrent composition of the IMs corresponding to these $k$ users. If the users' datasets evolve over time, then the concurrent composition of CMs must be used instead.
\end{remark}
\section{Application (Core Decomposition via Multi-Dimensional SVT)}\label{sec:app:core-decomposition}
In this section, we review the local protocol of Dhulipala, Henzinger, Li, Liu, Sricharan, and Zhu~\cite{dhulipala2025near} and demonstrate how its privacy analysis can be simplified using our concurrent parallel composition theorem for IMs. They study the \emph{core decomposition problem} in the LDP model, where each user corresponds to a (publicly known) vertex of a graph and holds the set of adjacent vertices as its private dataset. A core decomposition protocol is said to satisfy $(\eps,\delta)$-\emph{edge-LDP} if it is $(\eps,\delta)$-LDP w.r.t. the neighbor relation $\sim$, where two sequences of user datasets are considered neighbors whenever the corresponding graphs differ by at most one edge.

\medskip\noindent\underline{The core decomposition problem.}
Let $G=(V,E)$ be a graph, where $V(G)$ and $E(G)$ denote its vertex and edge sets, respectively. For a subset $U\subseteq V(G)$, the subgraph of $G$ induced by $U$ is denoted by $G[U]$. The degree of a vertex $v\in V(G)$ is written as $\deg_G(v)$. The \emph{core number} of a vertex $v$, denoted $k_G(v)$, is the largest integer $k$ such that $v$ belongs to a subgraph $H\subseteq G$ in which every vertex $u\in V(H)$ satisfies $\deg_H(u)\ge k$. The goal of the core decomposition problem is to compute the \emph{coreness vector} 
\[
\vec{k}_G = (k_G(v) : v \in V),
\]
which assigns to each vertex its coreness value.

\medskip\noindent\underline{Overview of the Local Protocol by \cite{dhulipala2025near}.}
Dhulipala et~al.~\cite{dhulipala2025near} propose a local protocol $P$ for the core decomposition problem, inspired by the classical \emph{peeling algorithm}. Let $G = (V, E)$ be a graph stored by $n=|V|$ users: each user corresponds to exactly one vertex $v \in V$ and stores a private dataset $N_v$, defined as
$$N_v=\left\{u\in V\mid \{v, u\}\in E\right\}.$$
The server knows only the vertex set $V$, while each user $v$ knows both $V$ and its private set of neighbors $N_v$. Next, we will describe the server and user mechanisms in $P$.

\vspace{0.3em}
\noindent\textbf{Server.}
The server maintains four variables: (1) the current core value $d$, initialized to $1$; (2) the active vertex set $A$, initialized to $V$; and (3) the estimated coreness vector $\vec{k}$, initialized to $\vec{0} \in \R^n$. The server repeats the following procedure $n$ times: 
\begin{enumerate}
    \item Send the current active vertex set $A$ and core value $d$ to all users $v \in A$ and ask whether their number of adjacent vertices in $A$ is less than $d$, i.e., $|N_v\cap A|< d$.
    \item Let $S$ denote the subset of vertices in $A$ responding with $\top$, indicating that $|N_v\cap A|< d$. Update $\vec{k}$ by assigning $k_v = d$ for all $v \in S$.
    \item Remove the vertices in $S$ from $A$.
\end{enumerate}

\vspace{0.3em}
\noindent\textbf{Users.}
For each $v \in V$, let $\cM_v$ denote the user IM corresponding to $v$. This mechanism is a variant of the \emph{Sparse Vector Technique (SVT)} that holds $N_v$ as its dataset and maps inputs $(A, d)$—where $A \subseteq V$ and $d \in \N$—to either $\top$ or $\bot$, indicating whether $|N_v \cap A| < d$. Two datasets are neighbors if they differ in at most one vertex, i.e., the adjacency set of $v$ changes in at most one element. The integer $d$ is a threshold, and the set $A$ represents a function $f_A(\cdot) = |\cdot \cap A|$, which is $1$-sensitive because it changes by at most one on neighboring vertex sets. Modeling $\cM_v$ as a continual mechanism, differential privacy for $\cM_v$ is defined with respect to a verification function~$g$ that requires the first message to be a pair of identical or neighboring vertex sets, and every subsequent message to be a pair of identical vertex sets (representing $1$-sensitive functions). We say that $\cM_v$ is $(\eps,\delta)$-DP if it satisfies $(\eps,\delta)$-DP with respect to~$g$ against adaptive adversaries.

\begin{lemma}[\cite{dhulipala2025near}]
    Consider the local protocol $P$ for the core decomposition problem described above. If every user (SVT) mechanism satisfies $(\eps/2, 0)$-DP, then the protocol $P$ satisfies $(\eps, 0)$-LDP. Moreover, for a graph $G = (V, E)$ with $n = |V|$ and maximum degree $\Delta$, the following holds with high probability: for every vertex $v \in V$, the coreness value estimated by the server is in the range $k_G(v) \pm O(\log n \log \Delta)$ where $k_G(v)$ is the true core number of $v$.
\end{lemma}

To prove the privacy part of this lemma, Dhulipala et al.~\cite{dhulipala2025near} analyze the privacy of a mechanism called \emph{Multi-Dimensional SVT}, which generalizes SVT by mapping, at each step, a vector of $n$ $1$-sensitive functions and a vector of $n$ thresholds to an element of $\{\top,\bot\}^n$. Their analysis is from scratch and technically involved. We observe that multi-dimensional SVT is a special case of the concurrent parallel composition of $n$ SVT mechanisms, where the adversary must issue queries to the SVT instances in a round-robin pattern. So we can simplify their analysis using Corollary~\ref{cor:parallel-fixed-param-comp-im} and Theorem~\ref{thm:dp-complex-mech}. However, to illustrate the use of Corollary~\ref{cor:ldp-parallel-concomp}, we instead analyze the entire local protocol directly.

\medskip\noindent\underline{Simplified Privacy Analysis.}
Recall that in the local model, a graph with public vertex set $V$ of size $n = |V|$ is represented as a sequence of adjacency sets $(N_v)_{v\in V}$, where each $N_v$ corresponds to the dataset of the user mechanism associated with vertex~$v$. Let $\sim^*$ denote the neighbor relation for user SVT mechanisms. (For every $N, N'\subseteq V$, we have $N \sim^* N'$ iff the sets differ in at most one vertex.) Let $\sim_{2, \sim^*}$ denote the $2$-sparse parallel-composition neighbor relation with user-neighbor relation $\sim^*$. By definition, the local protocol $P$ of~\cite{dhulipala2025near} satisfies $(\eps,\delta)$-edge-LDP if, for every two sequences of $n$ (adjacency) vertex sets whose graphs differ in exactly one edge, the server’s views are $(\eps,\delta)$-indistinguishable. Since such sequences differ in at most two adjacency sets, it suffices to show that $P$ is $(\eps,\delta)$-LDP w.r.t.~$\sim_{2,\sim^*}$.

By construction, functions corresponding to the vertex sets $A$ sent by the server are all $1$-sensitive; let this property be~$\prop^*$. Hence, each user’s SVT mechanism is $(\eps/2,0)$-DP w.r.t. the verification function $f_{\sim^*,\prop^*}$ from Definition~\ref{def:ver-func-ldp-neighbor-prop}. By basic composition, the composition of two instances of $\RR_{\eps/2,0}$ is $(\eps,0)$-DP. Therefore, by Corollary~\ref{cor:ldp-parallel-concomp}, the protocol~$P$ is $(\eps,0)$-LDP w.r.t. $\sim_{2, \sim^*}$, and therefore $(\eps,0)$-edge-LDP.
\section{Reduction to Interactive Randomized Response}\label{sec:post-irr}
Lyu~\cite{lyu2022composition} shows that for every $(\eps, \delta)$-DP IM $\cM$ and every pair of neighboring initial states $s_0, s_1$ for $\cM$, there exists an IPM $\cP$ such that for each $b\in \zo$, the mechanisms $\cM(s_b)$ and $\cP\circstar\RR_{\eps, \delta}(b)$ are equivalent. This means that $\cP$ can simulate $\cM(s_b)$ given access only to the output of a randomized response mechanism $\RR_{\eps, \delta}(b)$. In this section, inspired by this work, we construct a new IPM $\cP$ that, rather than simulating $\cM$ based on the outcome of a single instance of $\RR_{\eps, \delta}$, operates on the output of $\RR_{0, \delta}$ and only when necessary uses the output of $\RR_{\eps, 0}$. 
In other words, instead of interacting with the NIM $\RR_{\eps, \delta}$, $\cP$ interacts with the IM $\irr_{\eps, \delta}$, and under certain conditions, it queries $\irr_{\eps, \delta}$ at most once.

In addition, our construction removes or relaxes the following implicit assumptions in Lyu's proof:
\begin{itemize}
    \item Lyu's result relies on the existence of a constant $T$ that upper bounds the number of queries adversaries ask. This restriction comes from a backward definition of a so-called control function. We define a function which is the limit over this control function and, using monotone convergence, show that this ``limit function'' satisfies all the necessary conditions to be used in the rest of Lyu's proof. This extends the result to an unbounded number of rounds of interaction. 
    \item Lyu's proof uses an algorithmic construction iterating over the answer set, implying that the answer set must be finite. Instead with a mathematical approach, we extend the proof to work with a countably infinite answer set.
\end{itemize}

Our result is formalized in Lemma~\ref{lem:new}, which is restated below:
\begin{lemma*}[Restatement of Lemma~\ref{lem:new}]
    For $\eps > 0$ and $0 \leq \delta \leq 1$, let $\cM:S_\cM\times Q_\cM\to S_\cM\times A_\cM$ be an $(\eps, \delta)$-DP IM w.r.t. a neighbor relation $\sim$. Suppose that $\cM$ has discrete answer distributions. Then, for every two neighboring initial states $s_0$ and $s_1$, there exists an IPM $\cP$ such that for each $b\in \zo$:
    \begin{itemize}
        \item The mechanisms $\cM(s_b)$ and $\cP \circstar \irr_{\eps, \delta}(b)$ are equivalent.
        \item For every $k \in \N$ and every query sequence $(q_1, \dots, q_k) \in Q_{\cM}^k$, if the distributions of answers produced by $\cM(s_0)$ and $\cM(s_1)$ to $(q_1, \dots, q_k)$ are identical, then when $\cP \circstar \irr_{\eps, \delta}(b)$ receives the queries $q_1, \dots, q_k$, the IPM $\cP$ interacts with $\irr_{\eps, \delta}(b)$ only once. Otherwise, $\cP$ interacts with $\irr_{\eps, \delta}(b)$ at most twice.
    \end{itemize}
\end{lemma*}

\begin{notation}
    For short, let $Q=Q_\cM$ and $A=A_\cM$. For each $t\in \N$ and $b \in \zo$, we define the function $\mu_t^b:(Q\times A)^t\to [0,1]$ such that $\mu_t^b(q_1, \dots, a_t)$ denotes the probability of $\cM(s_b)$ returning answers $a_1, \dots, a_t$ to the queries $q_1, \dots, q_t$.  
\end{notation}

Fixing an upper bound $T$ on the number of queries, Lyu~\cite{lyu2022composition} reduces the construction of $\cP$ to defining a sequence of functions $\phi_t^0$ and $\phi_t^1$ for $t\in[T]$ that satisfy three conditions. They then inductively construct $\phi_t^b$ and verify that it satisfies the required conditions. The definition of these conditions is based on a sequence of control function $\rml_{t, T}^0$ and $\rml_{t, T}^1$, which have backward recursive definitions for $t$ from $T$ to $1$.

In Section~\ref{subsec:rmlower-rml-T}, we overview the definition of $\rml_{t, T}^b$ and its properties. In Section~\ref{subsec:rmlower-rml-limit-def}, we introduce a new control function $\rml_{t}^b$ with no dependency on $T$. The goal is to use $\rml_{t}^b$ instead of $\rml_{t, T}^b$. In Section~\ref{subsec:rmlower-rml-limit-prop}, we show some properties for $\rml_{t}^b$, required in later proofs.

In Section~\ref{subsec:reduction}, we summarize the reduction introduced in~\cite{lyu2022composition}, with the control functions $\rml_{t, T}^b$ replaced by $\rml_{t}^b$. We also modify this reduction by adding a new condition to the original three that the functions $\phi_t$ must satisfy. This additional condition is used to limit the number of interactions between $\cP$ and $\irr_{\eps, \delta}(b)$. Finally, in Section~\ref{subsec:construction}, we provide a new construction of $\phi_t$ and show that it satisfies all four conditions.

\subsection{Control Functions $\rmlower_{t,T}$ and $\rml_{t,T}$.}\label{subsec:rmlower-rml-T}
Lyu~\cite{lyu2022composition} introduces the following control function:

\begin{definition}[$\rmlower_{t,T}$]\label{def:rmlower-T}
    For $T\in\N$, $b\in\zo$, and $t \in [T]$, the function $\rmlower_{t,T}^b:(Q\times A)^{t-1}\times Q\to [0,1]$ is defined recursively as follows: for each transcript $(q_1, \dots, a_{t-1}, q_t)\in (Q\times A)^{t-1}\times Q$, 
    \begin{align*}
        \rmlower_{t,T}^b(q_1, \dots, a_{t-1}, q_t)=\begin{cases}
            \sum\limits_{a_t\in A}\sup_{q_{t+1}\in Q}\rmlower_{t+1,T}^b(q_1, \dots a_{t}, q_{t+1}), & t<T\\
            \sum\limits_{a_t\in A}\max\{0, \mu_t^b(q_1, \dots, a_t)-e^\eps\mu_t^{1-b}(q_1, \dots, a_t)\}, & t=T 
        \end{cases}
    \end{align*}
\end{definition}

We introduce another control function $\rml_{t,T}$ based on $\rmlower_{t, T}$:

\begin{definition}[$\rml_{t,T}^b$]\label{def:rml-T}
    For $T\in\N$, $b\in\zo$, and $t \in [T]$, the function $\rml_{t,T}^b:(Q\times A)^t\to [0,1]$ is defined recursively as follows: for each transcript $(q_1, \dots a_t)\in (Q\times A)^t$,
    \begin{align*}
        \rml_{t,T}^b(q_1, \dots, a_t)=
        \begin{cases}
            \sup_{q_{t+1}\in Q}\rmlower_{t+1,T}^b(q_1, \dots, a_t, q_{t+1}),& t<T\\
            \max\{0, \mu_t^b(q_1, \dots, a_t)-e^\eps\mu_t^{1-b}(q_1, \dots, a_t)\},& t=T
        \end{cases}
    \end{align*}
\end{definition}

The following corollary follows directly from the definitions of $\rmlower_{t, T}^b$ and $\rml_{t, T}^b$:
\begin{corollary}\label{cor:sum-rml-eq-rmlower}
    For every $T\in\N$, $b\in\zo$, $t \in [T]$, and $(q_1, \dots, a_{t-1}, q_t)\in (Q\times A)^{t-1}\times Q$, 
    $$\rmlower_{t,T}^b(q_1, \dots, a_{t-1}, q_t)=\sum_{a_t\in A}\rml_{t,T}^b(q_1, \dots, a_t).$$
\end{corollary}

Lyu~\cite{lyu2022composition} shows the following two claims about $\rmlower_{t, T}^b$:
\begin{claim}[\cite{lyu2022composition}]\label{cla:rmlower-ub-delta}
    For each $b\in\zo$, $T\in\N$, and $q_1\in Q$, we have $\rmlower_{1,T}^b(q_1)\leq \delta$.
\end{claim}

\begin{claim}[\cite{lyu2022composition}]\label{cla:rmlower-upper-ineq}
    For each $b\in\zo$, $T\in\N$, $t\in [T]$,\footnote{In the original claim in \cite{lyu2022composition}, $t$ is assumed to be at most $T-1$, which is a typo.} and transcript $(q_1, \dots, a_{t-1}, q_t)\in (Q\times A)^{t-1}\times Q$, we have
    \begin{align*}
        \rmlower_{t,T}^b(q_1, a_1, \dots, q_t)&-e^{-\eps}\rmlower_{t,T}^{1-b}(q_1, a_1, \dots, q_t) \leq \mu_{t-1}^b(q_1, \dots, a_{t-1})- e^{-\eps}\mu_{t-1}^{1-b}(q_1, \dots, a_{t-1}).
    \end{align*}
\end{claim}

Claim~\ref{cla:rmlower-ub-delta} combined with the definition of $\rml_{t, T}^b$ imply:
\begin{corollary}\label{cor:rml-ub-delta}
    For each $b\in\zo$ and $T\in\N$, we have $\rml_{0,T}^b()\leq \delta$.
\end{corollary}

In some part of the proof, Lyu~\cite{lyu2022composition} wrongly refers to Claim~\ref{cla:rmlower-upper-ineq}, while the following lemma is required. For completeness, we prove this lemma and use it in the next section.

\begin{lemma}\label{lem:rml-upper-ineq}
    For each $b\in\zo$, $T\in\N$, $t\in [T]\cup\{0\}$, and transcript $(q_1, \dots, a_t)\in (Q\times A)^t$, we have
    \begin{align*}
        \rml_{t,T}^b(q_1, \dots, a_t)&-e^{-\eps}\rml_{t,T}^{1-b}(q_1, \dots, a_t) \leq \mu_{t}^b(q_1, \dots, a_t)- e^{-\eps}\mu_{t}^{1-b}(q_1, \dots, a_{t}).
    \end{align*}
\end{lemma}

\begin{proof}
    For $t=T$, we need to show that
    \begin{align*}
        &\max\{0, \mu_t^b(q_1, \dots, a_t)-e^\eps\mu_t^{1-b}(q_1, \dots, a_t)\}- e^{-\eps}\max\{0, \mu_t^{1-b}(q_1, \dots, a_t)-e^\eps\mu_t^{b}(q_1, \dots, a_t)\}\\
        &\leq \mu_t^b(q_1, \dots, a_{t})- e^{-\eps}\mu_t^{1-b}(q_1, \dots, a_{t})
    \end{align*}
    For short, we refer to the LHS and RHS as $A$ and $B$, respectively.
    There are three cases:
    \begin{itemize}
        \item If $\mu_t^b(q_1, \dots, a_t)\leq e^{-\eps}\mu_t^{1-b}(q_1, \dots, a_t)$, then $A=\mu_t^b(q_1, \dots, a_t)-e^{-\eps}\mu_t^{1-b}(q_1, \dots, a_t),$
        which equals $B$.
        \item If $\mu_t^b(q_1, \dots, a_t)\geq e^{\eps}\mu_t^{1-b}(q_1, \dots, a_t)$, then, $A = \mu_t^{b}(q_1, \dots, a_t)-e^{\eps}\mu_t^{1-b}(q_1, \dots, a_t),$
        which is at most $B$ as $e^\eps\geq e^{-\eps}$.
        \item Otherwise, if $e^{-\eps}\mu_t^{1-b}(q_1, \dots, a_t)< \mu_t^b(q_1, \dots, a_t)< e^{\eps}\mu_t^{1-b}(q_1, \dots, a_t)$, $A=0$ is at most $B$ since by the condition of the case distinction, $B=\mu_t^b(q_1, \dots, a_{t})- e^{-\eps}\mu_t^{1-b}(q_1, \dots, a_{t})$ is positive.
    \end{itemize}
    For $t<T$, we have
    \begin{align*}
        &\rml_{t,T}^b(q_1, a_1, \dots, q_t)-e^{-\eps}\rml_{t,T}^{1-b}(q_1, a_1, \dots, q_t)\\
        &=\sup_{q_{t+1}\in Q}\rmlower_{t+1,T}^b(q_1, a_1, \dots, q_{t+1})- e^{-\eps}\sup_{q_{t+1}\in Q}\rmlower_{t+1,T}^{1-b}(q_1, a_1, \dots, q_{t+1})\\
        &\leq \sup_{q_{t+1}\in Q}\left(\rmlower_{t,T}^b(q_1, a_1, \dots, q_{t+1})- e^{-\eps}\rmlower_{t,T}^{1-b}(q_1, a_1, \dots, q_{t+1})\right)\\
        &\leq \mu_t^b(q_1, \dots, a_{t})- e^{-\eps}\mu_t^b(q_1, \dots, a_{t}),
    \end{align*}
    where the last inequality holds because for every $q_{t+1}\in Q$, by Claim~\ref{cla:rmlower-upper-ineq}
    \begin{align*}
        \rmlower_{t+1,T}^b(q_1, a_1, \dots, q_{t+1})&-e^{-\eps}\rmlower_{t+1,T}^{1-b}(q_1, a_1, \dots, q_{t+1}) \leq \mu_t^b(q_1, \dots, a_{t})- e^{-\eps}\mu_t^b(q_1, \dots, a_{t}).
    \end{align*}
\end{proof}

\subsubsection{Control Function $\rml_t^b=\lim_{T \ge t, T \to \infty}\rml_{t,T}$ (Definition).}\label{subsec:rmlower-rml-limit-def}
This section introduces our main new technical contribution.
To remove the assumption of having an upper bound $T$ on the number of queries, we introduce $\rml_t^b$ as the limit of $\rml_{t,T}^b$ and show that it satisfies all properties needed by the proof of~\cite{lyu2022composition}.
Thus, differently from~\cite{lyu2022composition}, we will use $\rml_t^0$ and $\rml_t^1$ to define the function $\phi_t$.

\begin{definition}[$\rml_t^b$]
    For each $b\in\zo$ and $t \in \N$, the function $\rml_t^b:(Q\times A)^t\to [0,1]$ is defined as follows: for each transcript $(q_1, \dots a_t)\in (Q\times A)^t$,
    $$\rml_t^b(q_1, \dots, a_{t-1}, q_t)=\lim_{T \ge t, T \to \infty}\rml_{t,T}^b(q_1, \dots, a_{t-1}, q_t).$$
\end{definition}

To ensure $\rml_t^b$ is well-defined, we need to show that the limit of $\rml_{t,T}^b$ exists. We will upper bound the values of $\rml_{t,T}^b$ by $1$ and show that this function is non-decreasing in $T$. As a result, for each $t\in \N\cup\{0\}$ and $(q_1, \dots, a_t)\in (Q\times A)^t$, the limit $\lim_{T \ge t, T \to \infty}\rml_{t,T}^b(q_1, \dots, a_t)$ exists. 
We need the following lemma throughout the section:

\begin{lemma}\label{lem:sum-mu-eq-mu}
    For every $t\in \N$, transcript $(q_1, \dots a_{t-1})\in (Q\times A)^{t-1}$, and query $q_t\in Q$, we have
    $$\mu_{t-1}^b(q_1, a_1, \dots, q_{t-1}, a_{t-1}) = \sum_{a_t \in A} \mu_{t}^b(q_1, a_1, \dots, q_t, a_t).$$
\end{lemma}

\begin{proof}
    For $a_t\in A$, let $\mu_t^b(a_t\mid q_1, a_1, \dots, q_{t-1}, a_{t-1}, q_t)$ denote the probability that $\cM(s_b)$ returns the answer $a_t$ in response to $q_t$ conditioned on $q_1, a_1, \dots, q_{t-1}, a_{t-1}$ being the history of previous queries and answers. By the chain rule, we have
    \begin{align*}
        \mu_t^b(q_1, a_1, \dots, q_t, a_t)&=\sum_{a_t\in A}\mu_{t-1}^b(q_1, a_1, \dots, q_{t-1}, a_{t-1})\cdot \mu_t^b(a_t\mid q_1, \dots, a_{t-1}, q_t)\\
        &=\mu_{t-1}^b(q_1, a_1, \dots, q_{t-1}, a_{t-1})\sum_{a_t\in A}\mu_t^b(a_t\mid q_1, \dots, a_{t-1}, q_t)\\
        &=\mu_{t-1}^b(q_1, a_1, \dots, q_{t-1}, a_{t-1})
    \end{align*}   
\end{proof}

To upper bound $\rml_{t,T}$, we first need to upper bound $\rmlower_{t,T}$:

\begin{lemma}\label{lem:rmlower-leq-mu-T}
    For every $T\in\N$, $b\in\zo$, $t\in [T]$, and transcript $(q_1, \dots a_{t-1}, q_t)\in (Q\times A)^{t-1}\times Q$,
    $$\rmlower_{t, T}^b(q_1, \dots a_{t-1}, q_t)\leq \mu_t^b(q_1, \dots, a_{t-1}).$$
\end{lemma}

\begin{proof}
    Fix $T\geq t$. The proof is by induction on $t$ from $T$ to $1$. For $t=T$, by definition and Lemma~\ref{lem:sum-mu-eq-mu}, we have
    \begin{align*}
        \rmlower_{t, T}^b(q_1, \dots a_{t-1}, q_t)&=\sum_{a_t\in A}\max\{0, \mu_t^b(q_1, \dots, a_t)-e^\eps\mu_t^{1-b}(q_1, \dots, a_t)\}\\
        &\leq \sum_{a_t\in A}\max\{0, \mu_t^b(q_1, \dots, a_t)\}= \mu_{t-1}^b(q_1, \dots, a_{t-1})
    \end{align*}
    For $t<T$, by the inductive hypothesis, 
    \begin{align*}
        \rmlower_{t, T}^b(q_1, \dots, a_{t-1}, q_t)&\leq \sum_{a_t\in A}\sup_{q_{t+1}\in Q}\mu_{t-1}^b(q_1, \dots a_{t}) =\sum_{a_t\in A}\mu_{t-1}^b(q_1, \dots a_{t})= \mu_{t-1}^b(q_1, \dots, a_{t-1}),
    \end{align*}
    completing the induction.
\end{proof}

The following lemma implies that the values of $\rml_{t,T}$ are upper bounded by $1$.

\begin{lemma}\label{lem:rml-leq-mu-T}
    For every $T\in\N$, $b\in\zo$, $t\in [T]\cup\{0\}$, and transcript $(q_1, \dots a_t)\in (Q\times A)^t$,
    $$\rml_{t, T}^b(q_1, \dots a_t)\leq \mu_t^b(q_1, \dots, a_t).$$
\end{lemma}
\begin{proof}
    For $t=T$, by definition and Lemma~\ref{lem:sum-mu-eq-mu}, we have
    \begin{align*}
        \rml_{t,T}^b(q_1, \dots, a_t)&=\max\{0, \mu_t^b(q_1, \dots, a_t)-e^\eps\mu_t^{1-b}(q_1, \dots, a_t)\}\\ 
        &\leq \max\{0, \mu_t^b(q_1, \dots, a_t)\} = \mu_t^b(q_1, \dots, a_t)
    \end{align*}
    For $t<T$, by Lemma~\ref{lem:rmlower-leq-mu-T},
    \begin{align*}
        \rml_{t,T}^b(q_1, \dots, a_t)&=\sup_{q_{t+1}\in Q}\rmlower_{t+1,T}^b(q_1, \dots, a_t, q_{t+1}) \leq \sup_{q_{t}\in Q}\mu_t^b(q_1, \dots, a_t)=\mu_t^b(q_1, \dots, a_t)
    \end{align*}
\end{proof}

To ensure the function $\rml_t^b$ is well-defined, it remains to show that $\rml_{t,T}^b$ is non-decreasing in $T$. To prove so, we first need to show:

\begin{lemma}\label{lem:rmlower-non-decreasing}
    For each $b\in\zo$ and $t\in\N$, the function $\rmlower_{t,T}^b$ is monotonically non-decreasing in $T$.
\end{lemma}

\begin{proof}
    We need to show that for every $T\in \N$, $b\in\zo$, $t\leq T$, and transcript $(q_1, \dots, a_{t-1}, q_t)\in (Q\times A)^{t-1}\times Q$, 
    $$\rmlower_{t,T}^b(q_1, \dots, a_{t-1}, q_t) \le \rmlower_{t,T+1}^b(q_1, \dots, a_{t-1}, q_t).$$
    Fix $T\in\N$ and $b\in\zo$. We show this by reverse induction on $t$: The base case is $t=T$. By Lemma~\ref{lem:sum-mu-eq-mu}, for every $(q_1, \dots, a_{t-1}, q_t)$, we have
    \begin{align*}
    &\rmlower_{T,T}^b (q_1, \dots, a_{T-1}, q_T) = \sum\limits_{a_T\in A}\max\{0, \mu_T^b(q_1, \dots, a_T)-e^\eps\mu_T^{1-b}(q_1, \dots, a_T)\}\\
    &= \sum\limits_{a_T\in A} \sup_{q_{T+1} \in Q} \max\{0, \sum\limits_{a_{T+1}\in A}\mu_{T+1}^b(q_1, \dots, a_{T+1})-e^\eps\mu_{T+1}^{1-b}(q_1, \dots, a_{T+1})\}\\
    &\le \sum\limits_{a_T\in A}\sup_{q_{T+1}\in Q} 
    \sum\limits_{a_{T+1}\in A}\max\{0, \mu_{T+1}^b(q_1, \dots, a_{T+1})-e^\eps\mu_{T+1}^{1-b}(q_1, \dots, a_{T+1})\}\\
    &= \sum\limits_{a_{T}\in A}\sup_{q_{T+1}\in Q} \rmlower_{T+1,T+1}(q_1, \dots, a_T, q_{T+1}) \\
    &= \rmlower_{T,T+1}^b (q_1, \dots, a_{T-1}, q_{T}) 
    \end{align*}
    For the induction step, assume that $t < T$ and that for every $(q_1, \dots, a_{t}, q_{t+1})$, we have 
    $$\rmlower_{t+1,T}^b(q_1, \dots, a_{t}, q_{t+1}) \le \rmlower_{t+1,T+1}^b(q_1, \dots, a_{t}, q_{t+1}).$$
    We need to show $\rmlower_{t,T}^b(q_1, \dots, a_{t-1}, q_t) \le \rmlower_{t,T+1}^b(q_1, \dots, a_{t-1}, q_t)$ for every $(q_1, \dots, a_{t-1}, q_t)$.
    By definition, 
    \begin{align*}
    \rmlower_{t,T}^b(q_1, \dots, a_{t-1}, q_t)& = 
      \sum\limits_{a_t\in A}\sup_{q_{t+1}\in Q}\rmlower_{t+1,T}^b(q_1, \dots a_{t}, q_{t+1}) \\
      & \le \sum\limits_{a_t\in A}\sup_{q_{t+1}\in Q}\rmlower_{t+1,T+1}^b(q_1, \dots a_{t}, q_{t+1}) \\
      & =  \rmlower_{t,T+1}^b(q_1, \dots, a_{t-1}, q_t)
      \end{align*}
\end{proof}

\begin{lemma}\label{lem:rml-non-decreasing}
    For each $b\in\zo$ and $t\in\N\cup\{0\}$, the function $\rml_{t,T}^b$ is monotonically non-decreasing in $T$.
\end{lemma}

\begin{proof}
    We need to show that for every $T\in \N$, $t\in [T]\cup \{0\}$, and transcript $(q_1, \dots, a_t)\in (Q\times A)^t$, 
    $$\rml_{t,T}(q_1, \dots, a_t) \le \rml_{t,T+1}(q_1, \dots, a_t)$$
    Fix $T\in\N$ and $b\in\zo$. We show this by reverse induction on $t$ from $T$ to $0$. 
    
    The base case is $t=T$. By Lemma~\ref{lem:sum-mu-eq-mu}, for every $(q_1, \dots, a_T)$ we have
    \begin{align*}
    &\rml_{T,T}^b (q_1, \dots, a_T) = \max\{0, \mu_T^b(q_1, \dots, a_T)-e^\eps\mu_T^{1-b}(q_1, \dots, a_T)\} \\
    &= \sup_{q_{T+1} \in Q} \max\{0, \sum\limits_{a_{T+1}\in A}\mu_{T+1}^b(q_1, \dots, a_{T+1})-e^\eps\mu_{T+1}^{1-b}(q_1, \dots, a_{T+1}))\}\\
    &\le \sup_{q_{T+1}\in Q} 
    \sum\limits_{a_{T+1}\in A}\max\{0, \mu_{T+1}^b(q_1, \dots, a_{T+1})-e^\eps\mu_{T+1}^{1-b}(q_1, \dots, a_{T+1})\}\\
    &= 
    \sup_{q_{T+1}\in Q} \rmlower_{T+1,T+1}(q_1, \dots, a_T, q_{T+1}) \\
    &= \rml_{T,T+1}^b (q_1, \dots, a_{T-1}, q_{T}, a_T) 
    \end{align*}

    For the induction step, assume that $t < T$ and that for every $(q_1, \dots, a_{t+1})$, we have $\rml_{t+1,T}^b(q_1, \dots, a_{t+1}) \le \rml_{t+1,T+1}^b(q_1, \dots, a_{t+1})$. We need to show $\rml_{t,T}^b(q_1, \dots, a_t) \le \rml_{t,T+1}^b(q_1, \dots, a_t)$ for every $(q_1, \dots, a_t)$.
    By definition of $\rml_{t,T+1}$ and Lemma~\ref{lem:rmlower-non-decreasing}, we have
    \begin{align*}
    \rml_{t,T}^b(q_1, \dots, q_t, a_t)& = 
      \sup_{q_{t+1}\in Q}\rmlower_{t+1,T}^b(q_1, \dots a_{t}, q_{t+1}) \\
      & \le \sup_{q_{t+1}\in Q}\rmlower_{t+1,T+1}^b(q_1, \dots a_{t}, q_{t+1}) \\
      & =  \rml_{t,T+1}^b(q_1, \dots, q_t, a_t)
      \end{align*}
\end{proof}

\subsubsection{Control Function $\rml_t^b=\lim_{T \ge t, T \to \infty}\rml_{t,T}$ (Properties).}\label{subsec:rmlower-rml-limit-prop}
In this section, we present the properties of $\rml_t^b$, required in the next sections. The definition of $\rml_t^b$ combined with Corollary~\ref{cor:rml-ub-delta}, Lemma~\ref{lem:rml-upper-ineq}, and Lemma~\ref{lem:rml-leq-mu-T} gives the following corollaries, respectively:
\begin{corollary}\label{cor:rml-ub-delta-limit}
    For each $b\in\zo$, we have $\rml_0^b()\leq \delta$.
\end{corollary}

\begin{corollary}\label{cor:rml-upper-ineq-limit}
    For each $b\in\zo$, $t\in\N\cup\{0\}$, and transcript $(q_1, \dots, a_t)\in (Q\times A)^t$, we have
    \begin{align*}
        \rml_t^b(q_1, \dots, a_t)&-e^{-\eps}\rml_t^{1-b}(q_1, \dots, a_t) \leq \mu_{t}^b(q_1, \dots, a_t)- e^{-\eps}\mu_{t}^{1-b}(q_1, \dots, a_{t}).
    \end{align*}
\end{corollary}

\begin{lemma}\label{lem:rml-leq-mu-limit}
    For each $b\in\zo$, $t\in \N$, and transcript $(q_1, \dots a_t)\in (Q\times A)^t$,
    $$\rml_t^b(q_1, \dots a_t)\leq \mu_t^b(q_1, \dots, a_t).$$
\end{lemma}

To show the next property, we need the following known fact:
\begin{lemma}[Monotone Convergence]\label{lem:monotone-convergence}
    Let $f:\mathbb{N}^2 \to \mathbb{R}$ be a function such that $\sum_m f(m,1) > -\infty$ and that is non-decreasing in the second parameter. Then 
    $$\lim_{n \to \infty} \sum_{m=1}^\infty f(m,n) = \sum_{m=1}^\infty \lim_{n \to \infty} f(m,n)$$
\end{lemma}

Lemma~\ref{lem:monotone-convergence} combined with Lemma~\ref{lem:rml-non-decreasing} implies:
\begin{corollary}\label{cor:swap-lim-and-sum-rml}
    For every $b\in\zo$, $t\in \N$, and $(q_1, \dots, a_{t-1}, q_t)\in (Q\times A)^{t-1}\times Q$,
    $$\sum_{a_t\in A}\lim_{T \ge t, T \to \infty}\rml_{t, T}^b(q_1, \dots, a_t)= \lim_{T \ge t, T \to \infty}\sum_{a_t\in A}\rml_{t, T}^b(q_1, \dots, a_t)$$
\end{corollary}

\begin{lemma}\label{lem:rml-geq-sum-rml-limit}
    For every $b\in \{0,1\}$, $t \in \N$,
    transcript $(q_1, \dots a_{t-1})\in (Q\times A)^{t-1}$, and query $q_t\in Q$,
    $$\rml_{t-1}^b(q_1, \dots a_{t-1}) \geq \sum_{a_t \in A}\rml_{t}^b(q_1, \dots a_t)$$
\end{lemma}
\begin{proof}
    First, we will show that for every $T>t-1$, 
    $$\rml_{t-1, T}^b(q_1, \dots a_{t-1}) \geq \sum_{a_t \in A}\rml_{t, T}^b(q_1, \dots a_t)$$
    Since $t-1\neq T$, by definition of $\rml_{t-1,T}^b$ and Corollary~\ref{cor:sum-rml-eq-rmlower}, we have
    \begin{align*}
        \rml_{t-1, T}^b(q_1, \dots a_{t-1}) &= \sup_{q\in Q}\rmlower_{t, T}(q_1, \dots a_{t-1}, q) \geq \rmlower_{t, T}(q_1, \dots a_{t-1}, q_t) = \sum_{a_t \in A}\rml_{t, T}^b(q_1, \dots a_t).
    \end{align*}

    Thus, by Corollary~\ref{cor:swap-lim-and-sum-rml}, we have
    \begin{align*}
        \rml_{t-1}^b(q_1, \dots a_{t-1}) &= \lim_{T \ge t-1, T \to \infty}\rml_{t-1, T}^b(q_1, \dots a_{t-1})= \lim_{T \ge t, T \to \infty}\rml_{t-1, T}^b(q_1, \dots a_{t-1})\\
        &\geq \lim_{T \ge t, T \to \infty}\sum_{a_t \in A}\rml_{t, T}^b(q_1, \dots a_t)= \sum_{a_t \in A}\lim_{T \ge t, T \to \infty}\rml_{t, T}^b(q_1, \dots a_t)= \sum_{a_t \in A}\rml_{t}^b(q_1, \dots a_t)
    \end{align*}
\end{proof}

\subsection{Reduction.}\label{subsec:reduction}
To define $\cP$ and simulate $\cM$, Lyu~\cite{lyu2022composition} constructs four families of probability mass functions (PMFs), corresponding to each outcome of $\RR_{\eps, \delta}$. Intuitively, each family describes an IM and consists of PMFs over the answer set $A$ for every $t\in [T]$ and every history $(q_1, \dots, a_{t-1}, q_t)\in (Q\times A)^t\times Q$. These PMFs represent the probability of returning each answer in response to $q_t$, conditioned on $q_1, \dots, a_{t-1}$ being the history of queries and answers. Similarly, we define $\cP$ based on four families of PMFs: $\calF_{(\bot, 0)}$, $\calF_{(\bot, 1)}$, $\calF_{(\top, 0)}$, and $\calF_{(\top, 1)}$, corresponding to the four possible output sequences of $\irr_{\eps, \delta}$. We note that in our construction, each family includes a PMF over $A$ for every $t\in \N$ and every history $(q_1, \dots, a_{t-1}, q_t)$. Moreover, the definitions of these PMFs are not the same as the ones in~\cite{lyu2022composition}.

\subsubsection{Constructing $\cP$.}\label{subsubsec:reduction-P}
In~\cite{lyu2022composition}, the post-processing mechanism initially interacts with $\RR_{\eps, \delta}$ and stores its output $r=(\tau, c)\in\{\top, \bot\}\times\zo$. It then answers the adversary's queries by itself using the family of PMFs corresponding to $r$. Our high-level idea is to construct the families $\calF_{(\bot, 0)}$, $\calF_{(\bot, 1)}$, $\calF_{(\top, 0)}$, and $\calF_{(\top, 1)}$ in such a way that, under certain conditions, knowledge of $\tau$ alone suffices to choose the appropriate PMF. Recall that the IPM $\cP$ interacts with $\irr_{\eps, \delta}$ instead of $\RR_{\eps, \delta}$. We design $\cP$ to initially interact with $\irr_{\eps, \delta}$ and store its response. $\cP$ then interacts with $\irr_{\eps, \delta}$ for the second time when it cannot choose a PMF using only the first response.

Specifically, $\cP$ has a single initial state $\left((), ()\right)$, where $()$ represents an empty sequence. We denote the state of $\cP$ by a pair $(r, h)$, where $r$ is the sequence of responses from $\irr_{\eps, \delta}$, and $h$ is the history of queries and answers exchanged with the adversary: right after receiving a message from $\irr_{\eps, \delta}$, $\cP$ appends it to $r$, and right before returning an answer $a$ to a query $q$ from the adversary, $\cP$ appends both $q$ and $a$ to $h$.

To generate a response $a$ to query $q$, $\cP$ proceeds as follows: If $q$ is the first query (i.e., $r=()$), then $\cP$ initially interacts with $\irr_{\eps, \delta}$ and appends its response to $r$. After this, whether $q$ is the first query or not, the first time the following condition does \textit{not} hold, $\cP$ interacts with $\irr_{\eps, \delta}$ for the second time and appends its answer to $r$: \emph{the PMFs corresponding to $(h, q)$ in $\calF_{(\top, 0)}$ and $\calF_{(\top, 1)}$ are identical, and the PMFs corresponding to $(h, q)$ in $\calF_{(\bot, 0)}$ and $\calF_{(\bot, 1)}$ are identical.} After that, $\cP$ never interacts with $\irr_{\eps, \delta}$ again.

After the potential interactions with $\irr_{\eps, \delta}$, $\cP$ generates the answer $a$ as follows:
\begin{itemize}
    \item If $r \in \{\top, \bot\} \times \zo$, then $\cP$ samples $a$ from the PMF corresponding to $(h, q)$ in $\calF_r$.
    \item Otherwise, by design, $r\in\{\top, \bot\}$, and the PMFs in $\calF_{(r, 0)}$ and $\calF_{(r, 1)}$ corresponding to $(h, q)$ are identical. In this case, $\cP$ samples $a$ from this common PMF.
\end{itemize}

\subsubsection{Constructing $\calF_{(\bot, b)}$s.}\label{subsubsec:reduction-F-bot}
Using the control function $\rmlower_{t, T}$, Lyu~\cite{lyu2022composition} introduces three conditions and constructs a sequence of functions $\phi_t:(Q\times A)^t \to [0,1]^2$ for $t \in \{1, \dots, T\}$ that satisfy these conditions. The families of PMFs are defined then based on these functions. We remove the dependence on the bound $T$ by replacing $\rmlower_{t, T}$ with $\rml_t$ in the conditions, and extend the construction of $\phi_t$ to all $t \in \N$. Furthermore, we introduce an additional condition to those given in~\cite{lyu2022composition}, which we will later use to bound the number of times $\cP$ interacts with $\irr_{\eps, \delta}$.

Let $\phi_t^0$ and $\phi_t^1$ denote the functions that return the first and second components, respectively, of the output of $\phi_t$. Intuitively, each $\phi_t^b$ represents an IM: for every query sequence $q_1, \dots, q_t$, it assigns a probability to each answer sequence $a_1, \dots, a_t$. That is, for any query sequence $q_1, \dots, q_t$, the function $\phi_t^b(q_1, \cdot, q_2, \cdot, \dots, q_t, \cdot)$ defines a PMF over $A^t$.

For $t=0$, we define $\phi_t()=(1, 1)$. We defer the detailed construction of the functions $\phi_t$ to Section~\ref{subsec:construction}. However, we assume the existence of $\phi_t$ satisfying the following conditions for each $b\in \zo$:

\begin{enumerate}[label=({\Roman*})]
    \item \label{item:cond1} For every transcript $(q_1, \dots, a_t)\in (Q\times A)^t$, 
    $$\delta\phi_t^b(q_1, \dots, a_t) \geq \rml_t^b(q_1, \dots, a_t).$$
    \item \label{item:cond2} For every transcript $(q_1, \dots, a_t)\in (Q\times A)^t$,
    $$\delta\phi_t^b(q_1, \dots, a_t) -  e^{-\eps} \delta\phi_t^{1-b}(q_1, \dots, a_t)\leq \mu_t^b(q_1, \dots, a_t) - e^{-\eps} \mu_t^{1-b}(q_1, \dots, a_t).$$
    \item \label{item:cond3} For every transcript $(q_1, \dots, a_{t-1}, q_t)\in (Q\times A)^{t-1}\times Q$,
    $$\sum_{a\in A} \phi_t^b(q_1, \dots, a_{t-1}, q_t,a) = \phi_{t-1}^b(q_1, \dots, a_{t-1}).$$
    \item \label{item:cond4} For each query sequence $(q_1, \dots, q_t)\in Q^t$, if $\mu_t^0(q_1, \dots, a_t)=\mu_t^1(q_1, \dots, a_t)$ for every answer sequence $(a_1, \dots, a_t)\in A^t$, then 
    $$\phi_t^0(q_1, \dots, a_t)=\phi_t^1(q_1, \dots, a_t)\geq \min\{\mu_t^0(q_1, \dots, a_t), \rml_t^0(q_1, \dots, a_t)+\rml_t^1(q_1, \dots, a_t)\}$$
    for every $(a_1, \dots, a_t)\in A^t$.
\end{enumerate}

For each $b\in \zo$, the family $\calF_{(\bot, b)}$ consists of functions $f_{q_1, \dots, a_{t-1}, q_t}^b:A\to [0,1]$ defined for every $(q_1, \dots, a_{t-1}, q_t)\in (Q\times A)^*\times Q$, where $f_{q_1, \dots, a_{t-1}, q_t}^b$ is defined as follows: 
\begin{itemize}[left=10pt]
    \item If $\phi_{t-1}(q_1, \dots, a_{t-1})>0$, then for each answer $a_t\in A$, $$f_{q_1, \dots, a_{t-1}, q_t}^b(a_t)=\frac{\phi_{t}(q_1, \dots, a_{t})}{\phi_{t-1}(q_1, \dots, a_{t-1})}.$$
    \item Otherwise, if $\phi_{t-1}(q_1, \dots, a_{t-1})=0$, $f_{q_1, \dots, a_{t-1}, q_t}^b(a_t)=0$ for all $a_t \in A$, except for one fixed arbitrary answer $a^* \in A$, for which $f_{q_1, \dots, a_{t-1}, q_t}^b(a^*) = 1$. This case is included only for formality, as the corresponding PMFs are never used by $\cP$.
\end{itemize}

Since $\phi_{t-1}$ and $\phi_t$ are non-negative, $f_{q_1, \dots, a_{t-1}, q_t}^b(a_t)$ is also non-negative. Furthermore, by \ref{item:cond3}, we have $\sum_{a_t\in A}f_{q_1, \dots, a_{t-1}, q_t}^b(a)=1$, implying that $f_{q_1, \dots, a_{t-1}, q_t}^b(a_t)$ is a PMF.

\subsubsection{Constructing $\calF_{(\top, b)}$s.}\label{subsubsec:reduction-F-top}
Similar to the previous section, $\calF_{(\top, 0)}$ and $\calF_{(\top, 1)}$ are defined based on a sequence of functions $\psi_t:(Q\times A)^t\to [0,1]^2$ for each $t\in \N\cup \{0\}$. For $t=0$, $\psi_t()=(1,1)$. For $t\in \N$, we define $\psi_t$ as follows: 
\begin{itemize}[left=10pt]
    \item If $\delta=1$, then $\psi_t$ is identical to $\phi_t$.
    \item Otherwise, for each $b\in\zo$ and $(q_1, \dots, a_t)\in (Q\times A)^t$, $\psi_t^b(q_1, \dots, a_t)$ equals
    \begin{align*}\label{def:psi}
        \frac{1}{(1-\delta)(e^\eps-1)}\Big(e^\eps\mu_t^b(q_1, \dots, a_t) -\delta e^\eps\phi_t^b(q_1, \dots, a_t) -\mu_t^{1-b}(q_1, \dots, a_t) +\delta \phi_t^{1-b}(q_1, \dots, a_t)\Big).
    \end{align*}
\end{itemize}

For each $b\in\zo$, the family $\calF_{(\top, b)}$ consists of functions $g_{q_1, \dots, a_{t-1}, q_t}^b:A\to [0,1]$ for each $(q_1, \dots, a_{t-1}, q_t)\in (Q\times A)^*\times Q$, where $g_{q_1, \dots, a_{t-1}, q_t}^b$ is defined as follows: 
\begin{itemize}[left=10pt]
    \item If $\psi_{t-1}(q_1, \dots, a_{t-1})>0$, then for each answer $a_t\in A$, 
    $$g_{q_1, \dots, a_{t-1}, q_t}^b(a_t)=\frac{\psi_{t}(q_1, \dots, a_{t})}{\psi_{t-1}(q_1, \dots, a_{t-1})}.$$
    \item Otherwise, $g_{q_1, \dots, a_{t-1}, q_t}^b(a_t)=0$ for all $a_t \in A$, except for one fixed arbitrary answer $a^* \in A$, for which $g_{q_1, \dots, a_{t-1}, q_t}^b(a^*) = 1$.
\end{itemize}

By condition \ref{item:cond2}, $\psi_t^b$ is non-negative. Also, by Lemma~\ref{lem:sum-mu-eq-mu} and condition \ref{item:cond3}, 
$$\sum_{a\in A} \psi_t^b(q_1, \dots, a_{t-1}, q_t,a) = \psi_{t-1}^b(q_1, \dots, a_{t-1}),$$
implying that $g_{q_1, \dots, a_{t-1}, q_t}^b$ is a PMF.

\subsubsection{Proof of Lemma~\ref{lem:new}.}\label{subsubsec:reduction-proof}
By definition of $\irr_{\eps, \delta}$ and $\cP$, for every $t\in \N$ and every transcript $(q_1, \dots, a_t)\in (Q\times A)^t$, the probability of $\cP\circstar\irr_{\eps, \delta}(b)$ returning the answers $a_1, \dots, a_t$ to the queries $q_1, \dots, q_t$ is
\begin{align*}
    \delta \phi_t^b(q_1, \dots, a_t) + (1-\delta)\frac{e^\eps}{1+e^\eps}\psi_t^b(q_1, \dots, a_t)+(1-\delta)\frac{1}{1+e^\eps}\psi_t^{1-b}(q_1, \dots, a_t).
\end{align*}
By the definition of $\psi_t$, this probability equals
\begin{align*}
    &\delta \phi_t^b(q_1, \dots, a_t) + \frac{e^\eps}{(1+e^\eps)(e^\eps-1)}\Big(e^\eps\mu_t^b(q_1, \dots, a_t) \\
    &\quad-\delta e^\eps\phi_t^b(q_1, \dots, a_t)-\mu_t^{1-b}(q_1, \dots, a_t) +\delta \phi_t^{1-b}(q_1, \dots, a_t)\Big) \\
    &\quad+\frac{1}{(1+e^\eps)(e^\eps-1)}\Big(e^\eps\mu_t^{1-b}(q_1, \dots, a_t) -\delta e^\eps\phi_t^{1-b}(q_1, \dots, a_t)-\mu_t^{b}(q_1, \dots, a_t) +\delta \phi_t^{b}(q_1, \dots, a_t)\Big)\\
    &=\mu_t^{b}(q_1, \dots, a_t).
\end{align*}
Thus, the views of any adversary interacting with $\cM(s_b)$ and $\cP\circstar\irr_{\eps, \delta}(b)$ are identically distributed, implying that these mechanisms are equivalent. 

Furthermore, for every query sequence $(q_1, \dots, q_k)$, if the answer distributions of $\cM(s_0)$ and $\cM(s_1)$ to $(q_1, \dots, q_k)$ are identical, then by Lemma~\ref{lem:sum-mu-eq-mu}, for every $t \leq k$ and every answer sequence $(a_1, \dots, a_t) \in A^t$, we have
$$\mu_t^0(q_1, \dots, a_t)=\mu_t^1(q_1, \dots, a_t).$$
Therefore, by condition~\ref{item:cond4}, for every $t\leq k$ and every $(a_1, \dots, a_t)\in A^t$,
$$\phi_t^0(q_1, \dots, a_t)=\phi_t^1(q_1, \dots, a_t).$$
Consequently,
$$\psi_t^0(q_1, \dots, a_t)=\psi_t^1(q_1, \dots, a_t).$$
By the definitions of $\calF_{(\top, 0)}$, $\calF_{(\top, 1)}$, $\calF_{(\bot, 0)}$, and $\calF_{(\bot, 1)}$, this implies that for every $t \leq k$ and answer sequence $(a_1, \dots, a_{t-1})$, the PMF associated with $(q_1, \dots, a_{t-1}, q_t)$ is identical in both $\calF_{(\bot, 0)}$ and $\calF_{(\bot, 1)}$, and likewise in $\calF_{(\top, 0)}$ and $\calF_{(\top, 1)}$. Hence, by the construction of $\cP$, this mechanism interacts with $\irr_{\eps, \delta}$ only once. Moreover, by design, $\cP$ never interacts with $\irr_{\eps, \delta}$ more than twice.

\subsection{Construction of $\phi_t$}\label{subsec:construction}
For each $q_1, \dots, a_{t-1}, q_t$, Lyu~\cite{lyu2022composition} 
initializes the values of $\phi_t(q_1, \dots, a_{t-1}, q_t, \cdot)$ with the ones from $\rml_{t, T}(q_1, \dots, a_{t-1}, q_t, \cdot)$. 
(Removing the upper bound assumption, they are initialized by the values of $\rml_t(q_1, \dots, a_{t-1}, q_t, \cdot)$.) Then, they iterate over all possible answers $a_t \in A$ in a fixed order and gradually increase the value of $\phi_t(q_1, \dots, a_t)$ to satisfy conditions \ref{item:cond2} and \ref{item:cond3}, while ensuring that condition \ref{item:cond1} remains valid. This iterative procedure requires $A$ to be finite; otherwise, the construction algorithm does not terminate.  Moreover, the fixed iteration order can lead to an imbalance in which earlier values grow significantly larger than those that come later in the order. In our new construction, we control and balance this increase. 

Intuitively, we introduce a temporary upper bound and increase the values of the function $\rml_t(q_1, \dots, a_{t-1}, q_t, \cdot)$ to reach the minimum of the original upper bounds and the new temporary ones. This results in an intermediate function, denoted by $\xi_t(q_1, \dots, a_{t-1}, q_t, \cdot)$. Then, we remove the temporary upper bound and continue increasing the values of $\xi_t(q_1, \dots, a_{t-1}, q_t, \cdot)$ (if needed) to meet the original upper bounds. The final function is the desired $\phi_t(q_1, \dots, a_{t-1}, q_t, \cdot)$. We will show that if $\mu_t^0(q_1, \dots, a_t)=\mu_t^1(q_1, \dots, a_t)$ for every answer sequence $(a_1, \dots, a_t)\in A^t$, then the functions $\xi_t(q_1, \dots, a_{t-1}, q_t, \cdot)$ and $\phi_t(q_1, \dots, a_{t-1}, q_t, \cdot)$ are identical. We set the temporary upper bound in a way that this implies the condition~\ref{item:cond4} holds.

Formally, we inductively construct $\phi_t$ for each $t \in \N$, assuming that $\phi_{t-1}$ exists and satisfies conditions \ref{item:cond1}, \ref{item:cond2} and \ref{item:cond4}. We then show that $\phi_t$ satisfies all conditions from \ref{item:cond1} to \ref{item:cond4}. For the base case $t = 0$, we define $\phi_0() = (1, 1)$. By Corollary~\ref{cor:rml-ub-delta-limit}, this choice satisfies condition \ref{item:cond1}. Moreover, for $\phi_0^0() = \phi_0^1() = 1$ and $\mu_0^0() = \mu_0^1() = 1$, condition~\ref{item:cond2} is $\delta\cdot 1- e^{-\eps}\cdot \delta\cdot 1\leq 1-e^{-\eps}\cdot 1$, which is true because $\delta\leq 1$. Furthermore, since $\phi_0^0() = \phi_0^1() = 1 \geq \min\{1, \delta+\delta\}$, condition \ref{item:cond4} is also satisfied.

\begin{remark}\label{rem:inductive-hypothesis}
    In the rest of this section, we assume $\phi_{t-1}$ exists and satisfies conditions \ref{item:cond1}, \ref{item:cond2} and \ref{item:cond4}.
\end{remark}

Section~\ref{subsubsec:construction-preliminaries} outlines the preliminaries required for defining $\xi_t$ and $\phi_t$. In Section~\ref{subsubsec:construction-xi}, we define the intermediate function $\xi_t$. Section~\ref{subsubsec:construction-phi} presents the new construction of $\phi_t$, and finally, in Section~\ref{subsubsec:construction-conditions}, we prove that $\phi_t$ satisfies conditions \ref{item:cond1} to \ref{item:cond4}.

\subsubsection{Preliminaries.}\label{subsubsec:construction-preliminaries}
In this section, we introduce the fundamental concepts and notations needed for defining $\phi_t$ when $A$ is not finite.

\begin{definition}[Well-Order]
    A well-order on a set $X$ is a total order (i.e., for every $x,y\in X$, either $x<y$, $x>y$, or $x=y$) such that every non-empty subset of $X$ has a least element.
\end{definition}

Since $A$ is countable, $A$ can be injected into $\N$. Listing elements of $A$ in the order given by their images under this injection gives a well‐order $<_A$ on $A$. Note that for any non-empty subset $B\subseteq A$, the image of $B$ in $\N$ has a least element, whose preimage is the least element of $B$ under $<_A$. We will later use the following remark when constructing $\phi_t$.

\begin{remark}\label{rem:smallest-or-predecessor}
    By the definition of $<_A$, each element of $A$ is either the least element of $A$ or has an immediate predecessor in $A$.
\end{remark}

\begin{definition}[$\Gamma_{<a}$ and $\Gamma_{>a}$]
    For any answer $a\in A$, $\Gamma_{<a}=\{a'\in A\mid a'<_A a\}$ and $\Gamma_{>a}=\{a'\in A\mid a<_A a'\}$.
\end{definition}

\begin{lemma}[Transfinite Recursion~\cite{suppes2012axiomatic}]\label{lem:transfinite}
    Let $f$ be a function recursively defined over $A$ according to the well-order $<_A$. If for each answer $a\in A$, the value of $f(a)$ is uniquely determined by the values $f(a')$ for every $a' \in \Gamma_{<a}$, then $f$ exists and is unique.
\end{lemma}

In the next section, we construct functions $f : A \to [0,1]^2$ and apply Lemma~\ref{lem:transfinite} to prove that they are well-defined. To apply this lemma, we must ensure that each value $f(a)$ is uniquely determined based on the previously defined values $f(a')$ for all $a' \in \Gamma_{<a}$. Thus, to define $f(a)$, we first select a subset $\mathcal{C} \subseteq [0,1]^2$ of desired values satisfying certain properties derived from the values of $f(a')$ for $a' \in \Gamma_{<a}$. We then set $f(a)$ to the lexicographically maximal element of $\mathcal{C}$. The following definitions and lemma formalizes this and guarantee that $f(a)$ is uniquely defined in this way.

\begin{definition}[Maximal Pair]
    Let $\mathcal{C} \subseteq [0,1]^2$. A pair $(x_0, x_1) \in \mathcal{C}$ is said to be \emph{maximal} in $\mathcal{C}$ if there does not exist a pair $(y_0, y_1) \in \mathcal{C}$ such that $(x_0, x_1) \neq (y_0, y_1)$, $x_0 \leq y_0$, and $x_1 \leq y_1$.
\end{definition}

\begin{definition}[Lexicographic Order $<_{lex}$]
    The relation $<_{lex}$ is a binary ordering on $[0,1]^2$ defined as follows: for any $(r_1, r_2), (r_3, r_4) \in [0,1]^2$, we have $(r_1, r_2)<_{lex} (r_3, r_4)$ if and only if $r_1\leq r_3$, or $r_1=r_3$ and $r_2\leq r_4$.
\end{definition}

\begin{lemma}\label{lem:compact-maximal}
    Let $\mathcal{C}$ be a non-empty compact subset of $[0,1]^2$.
    \begin{enumerate}[label=(\alph*)]
        \item  $(x_0, x_1)<_{lex}(x_0^*, x_1^*)$ for all $(x_0, x_1)\in \mathcal{C}$.
        \item $(x_0^*, x_1^*)$ is maximal in $\mathcal{C}$.
    \end{enumerate}
\end{lemma}
\begin{proof}
    (a): Define $x_0^*=\sup\{x_0\mid (x_0, x_1)\in \mathcal{C}\}$. Since $\mathcal{C}$ is non-empty, the set $\{x_0\mid (x_0, x_1)\in \mathcal{C}\}$ is non-empty. Also, by definition, $x_0^*\leq 1$. As we are taking a continuous projection $(x_0, x_1)\to x_0$ of a compact set $\mathcal{C}$, this supremum is actually a maximum, so there exists $(x_0, x_1)\in \mathcal{C}$ such that $x_0=x_0^*$.

    Define $x_1^*=\sup\{x_1\mid (x_0^*, x_1)\in \mathcal{C}\}$. Again, the set $\{x_1\mid (x_0^*, x_1)\in \mathcal{C}\}$ is non-empty since there exists $(x_0, x_1)\in \mathcal{C}$ with $x_0=x_0^*$. By the same compactness argument, this supremum is also attained. Hence, $(x_0^*, x_1^*)\in \mathcal{C}$. By definition, $(x_0, x_1)<_{lex}(x_0^*, x_1^*)$ for every $(x_0, x_1)\in \mathcal{C}_{q_1, \dots, a_t}$, which trivially implies that $(x_0^*, x_1^*)$ is the only pair in $\mathcal{C}$ satisfying this property.

    (b): Let $(x_0, x_1)$ be a pair in $\mathcal{C}$ satisfying $x_0\geq x_0^*$ and $x_1\geq x_1^*$. Since $x_0\geq x_0^*$ and $(x_0, x_1)\in \mathcal{C}$, by definition of $x_0^*$, we have $x_0^*=x_0$. Since $x_1\geq x_1^*$, $x_0^*=x_0$, and $(x_0, x_1)\in \mathcal{C}$, by definition of $x_1^*$, we have $x_1^*=x_1$. Therefore, $(x_0, x_1)=(x_0^*, x_1^*)$, implying that $(x_0^*, x_1^*)$ is maximal.
\end{proof}

\subsubsection{Defining $\xi_t$ and showing its existence.}\label{subsubsec:construction-xi}
To define the function $\xi_t:(Q\times A)^t\to [0,1]^2$, we fix $(q_1, \dots, a_{t-1}, q_t)\in (Q\times A)^{t-1}\times Q$ and recursively define $\xi_t(q_1, \dots, a_{t-1}, q_t, \cdot)$ for each answer $a_t\in A$ according to the order $<_A$. We define $\xi_t(q_1, \dots, a_t)=(x_0^*, x_1^*)$, where $x_0^*$ and $x_1^*$ are uniquely chosen from the maximal feasible solutions of a constraint satisfaction problem with two variables $x_0$ and $x_1$. The constraints for this problem involve the known functions $\mu_t$, $\phi_{t-1}$, $\rml_t$, and the known values $\xi_t(q_1,\dots, a_{t-1}, q_t, a)$ for every $a\in\Gamma_{<a_t}$.

Consider the following constraints for the variables $x_0$ and $x_1$:

\begin{enumerate}[label={($\xi$.\arabic*)}]
    \item \label{item:cond-xi-geq-rml} For each $b\in\zo$, 
    $$x_b\geq \rml_t^b(q_1, \dots, a_{t}).$$
    \item \label{item:cond-xi-normal} For each $b\in\zo$,  
    \begin{align*}
        \sum_{a\in\Gamma_{<a_t}} &\xi_t^b(q_1, \dots, a_{t-1}, q_t, a) + x_b +\sum_{a\in\Gamma_{>a_t}} \rml_t^b(q_1, \dots, a_{t-1}, q_t, a) \leq \phi_{t-1}^b(q_1, \dots, a_{t-1})
    \end{align*}
    \item \label{item:cond-xi-upper} For each $b\in\zo$,  
    \begin{align*}
        x_b- e^{-\eps} x_{1-b}
        \leq \mu_t^b(q_1, \dots, a_{t}) - e^{-\eps} \mu_t^{1-b}(q_1, \dots, a_{t})
    \end{align*} 
    \item \label{item:cond-xi-extra} For each $b\in\zo$,  
    \begin{align*}
        x_b \leq \min\{\mu_t^b(q_1, \dots, a_{t}), \rml_t^0(q_1, \dots, a_t)+\rml_t^1(q_1, \dots, a_t)\}
    \end{align*} 
\end{enumerate}

Let the set $\mathcal{C}_{q_1, \dots, a_t}$ denote the set of feasible solutions $(x_0, x_1)$ that satisfy all the constraints above. Each constraint is a non-strict inequality of the form $h_i(x_0, x_1) \geq 0$, where each $h_i$ is a continuous function. Since the solution set of a single non-strict inequality defined by a continuous function is closed, and the intersection of finitely many closed sets is also closed, the set $\mathcal{C}_{q_1, \dots, a_t}$ is closed. Furthermore, as it is a subset of the compact set $[0,1]^2$, it is itself compact. If $\mathcal{C}_{q_1, \dots, a_t}$ is nonempty, then by Lemma~\ref{lem:compact-maximal}, the lexicographically maximum pair $(x_0^*, x_1^*)$ of $\mathcal{C}_{q_1, \dots, a_t}$ exists, belongs to $\mathcal{C}_{q_1, \dots, a_t}$, and is maximal in $\mathcal{C}_{q_1, \dots, a_t}$. We define $\xi_t(q_1, \dots, a_{t})=(x_0^*, x_1^*)$.

\begin{claim}\label{cla:xi-exists}
    The function $\xi_t$ is well-defined and exists.
\end{claim}
\begin{proof}
    The proof consists of two steps: (i) proving that $\mathcal{C}_{q_1, \dots, a_t}$ is non-empty, and (ii) applying the transfinite recursion theorem (Lemma~\ref{lem:transfinite}) and concluding that $\xi_t$ exists.

    \medskip\noindent\textbf{(i):}
    Fix any $(q_1, \dots, a_{t-1}, q_t)\in (Q\times A)^{t-1}\times Q$. By Remark~\ref{rem:smallest-or-predecessor}, each answer in $A$ is either the smallest element of $A$ or has an immediate predecessor. We will inductively show that for every $a_t\in A$, the values $\rml_t^0(q_1, \dots, a_t)$ and $\rml_t^1(q_1, \dots, a_t)$ satisfy the constraints for defining $\xi_t(q_1, \dots, a_t)$, and consequently, $\mathcal{C}_{q_1, \dots, a_t}$ is not empty.

    Trivially, $\rml_t^b(q_1, \dots, a_t)$ satisfies Constraint \ref{item:cond-xi-geq-rml}. Moreover, by Corollary~\ref{cor:rml-upper-ineq-limit}, it satisfies Constraint \ref{item:cond-xi-upper}. To prove Constraint~\ref{item:cond-xi-normal} is satisfied, we need to show that
    \begin{equation}\label{eq:xi-normal-feasible}
        \begin{aligned}
        \sum_{a\in\Gamma_{<a_t}} \xi_t^b(q_1, \dots, a_{t-1}, q_t, a) &+\sum_{a\in\Gamma_{>a_t}\cup \{a_t\}} \rml_t^{b}(q_1, \dots, a_{t-1}, q_t, a)\leq \delta\phi_{t-1}^b(q_1, \dots, a_{t-1}),
    \end{aligned}
    \end{equation}
    We prove this by induction on $a_t\in A$:

    \medskip\noindent\underline{Base Case:} 
    Suppose $a_t$ is the least element of $A$. Then, $\Gamma_{<a_t}$ is empty, reducing our goal to showing that 
    \begin{align*} 
        \sum_{\substack{a\in A}} \rml_t^{b}(q_1, \dots, a_{t-1}, q_t, a)\leq \delta\phi_{t-1}^b(q_1, \dots, a_{t-1}).
    \end{align*}
    We know that $\phi_{t-1}(q_1, \dots, a_{t-1})$ satisfies condition \ref{item:cond1}. (see Remark~\ref{rem:inductive-hypothesis}.) Thus, by Lemma~\ref{lem:rml-geq-sum-rml-limit},
    \begin{align*}
        \delta\phi_{t-1}^b(q_1, \dots, a_{t-1})&\geq \rml_{t-1}^b(q_1, \dots, a_{t-1}) \geq \sum_{\substack{a\in A}} \rml_t^{b}(q_1, \dots, a_{t-1}, q_t, a),
    \end{align*}
    finishing the proof for the base case.

    \noindent\underline{Inductive Step:}
    Let $a^*$ denote the immediate predecessor of $a_t$. By inductive hypothesis, $\xi_t(q_1, \dots, a_{t-1}, q_t, a^*)$ exists and satisfies Constraint~\ref{item:cond-xi-upper}. Since $\Gamma_{<a_t} = \Gamma_{<a^*} \cup \{a^*\}$, Constraints \ref{item:cond-xi-normal} for defining $\xi_t(q_1, \dots, a_{t-1}, q_t, a^*)$ is identical to (\ref{eq:xi-normal-feasible}). Thus, by inductive hypothesis, inequality~(\ref{eq:xi-normal-feasible}) holds, completing the induction.

    It remains to prove that Constraint~\ref{item:cond-xi-extra} is satisfied. By Lemma~\ref{lem:rml-leq-mu-limit}, $\rml_t^b(q_1, \dots, a_t)\leq \mu_t^b(q_1, \dots, a_t)$. Also, as $\rml_t^{1-b}$ is non-negative, $\rml_t^b(q_1, \dots, a_t)\leq \rml_t^0(q_1, \dots, a_t)+\rml_t^1(q_1, \dots, a_t)$. Thus, 
    $$\rml_t^b(q_1, \dots, a_t)\leq \min\{ \mu_t^b(q_1, \dots, a_t), \rml_t^0(q_1, \dots, a_t)+\rml_t^1(q_1, \dots, a_t)\},$$
    and Constraint~\ref{item:cond-xi-extra} is satisfied.

    \medskip\noindent\textbf{(ii):}
    For every $(q_1, \dots, a_{t-1}, q_t)$, the values of the function $\xi_t(q_1, \dots, a_{t-1}, q_t, \cdot)$ is uniquely determined by the preceding values of this function in the well-order $<_A$. Thus, by the transfinite recursion theorem, the function $\xi_t(q_1, \dots, a_{t-1}, q_t, \cdot)$ exists and is uniquely defined. Consequently, the function $\xi_t$ exists and is uniquely defined.
\end{proof}

\subsubsection{Defining $\phi_t$ and showing its existence.}\label{subsubsec:construction-phi}
To define the function $\phi_t:(Q\times A)^t\to [0,1]^2$, we fix $(q_1, \dots, a_{t-1}, q_t)\in (Q\times A)^{t-1}\times Q$ and recursively define $\phi_t(q_1, \dots, a_{t-1}, q_t, \cdot)$ for each answer $a_t\in A$ based on the $<_A$ order. Similar to the definition of $\xi_t$, for each $(q_1, \dots, a_t)$, we determine $\phi_t(q_1, \dots, a_t)$ using a set of constraints. These constraints involve the known functions $\xi_t$, $\mu_t$, $\phi_{t-1}$, $\rml_t$, and the known values $\phi_t(q_1,\dots, a_{t-1}, q_t, a)$ for every $a\in\Gamma_{<a_t}$.

Consider the following set of constraints for the variables $x_0$ and $x_1$:

\begin{enumerate}[label={($\phi$.\arabic*)}]
    \item \label{item:cond-phi-geq-xi} For each $b\in\zo$, 
    $$\delta x_b\geq \xi_t^b(q_1, \dots, a_{t}).$$
    \item \label{item:cond-phi-normal} For each $b\in\zo$,  
    \begin{align*}
        \sum_{a\in\Gamma_{<a_t}} &\phi_t^b(q_1, \dots, a_{t-1}, q_t, a) + x_b +\sum_{a\in\Gamma_{>a_t}}{\frac{1}{\delta}} \xi_t^b(q_1, \dots, a_{t-1}, q_t, a) \leq \phi_{t-1}^b(q_1, \dots, a_{t-1})
    \end{align*}
    \item \label{item:cond-phi-upper} For each $b\in\zo$,  
    \begin{align*}
        \delta x_b- e^{-\eps}\delta x_{1-b}
        \leq \mu_t^b(q_1, \dots, a_{t}) - e^{-\eps} \mu_t^{1-b}(q_1, \dots, a_{t})
    \end{align*} 
\end{enumerate}

Let $\mathcal{C}_{q_1, \dots, a_t}$ denote the set of feasible solutions $(x_0, x_1)$ satisfying the constraints above. As in Section~\ref{subsubsec:construction-xi}, one can show that $\mathcal{C}_{q_1, \dots, a_t}$ is compact. Therefore, if $\mathcal{C}_{q_1, \dots, a_t}$ is nonempty, then by Lemma~\ref{lem:compact-maximal}, the lexicographically maximum pair $(x_0^*, x_1^*)$ of $\mathcal{C}_{q_1, \dots, a_t}$ exists, belongs to $\mathcal{C}_{q_1, \dots, a_t}$, and is maximal in $\mathcal{C}_{q_1, \dots, a_t}$.

Define $\mathcal{G}_{q_1, \dots, a_t} \subseteq \mathcal{C}_{q_1, \dots, a_t}$ to be the set of all maximal elements in $\mathcal{C}_{q_1, \dots, a_t}$. By the definition of maximality, this set either contains no pair of the form $(z, z)$ with identical components, or it contains exactly one such pair, denoted $(z^*, z^*)$. We define $\phi_t(q_1, \dots, a_t)=(z^*, z^*)$ if there exists $(z^*, z^*)\in \mathcal{G}_{q_1, \dots, a_t}$ and $\phi_t(q_1, \dots, a_{t})=(x_0^*, x_1^*)$ otherwise. In either case, $\phi_t(q_1, \dots, a_{t})$ is a maximal solution for conditions~\ref{item:cond-phi-geq-xi} to \ref{item:cond-phi-upper}. 

\begin{claim}\label{cla:phi-exists}
    The function $\phi_t$ is well-defined and exists.
\end{claim}
\begin{proof}
    The proof consists of two steps: (i) proving that $\mathcal{C}_{q_1, \dots, a_t}$ is non-empty, and (ii) applying the transfinite recursion theorem (Lemma~\ref{lem:transfinite}) and concluding that $\phi_t$ exists.

    \medskip\noindent\textbf{(i):}
    Fix any $(q_1, \dots, a_{t-1}, q_t)\in (Q\times A)^{t-1}\times Q$. By Remark~\ref{rem:smallest-or-predecessor}, each answer in $A$ is either the smallest element of $A$ or has an immediate predecessor. We will inductively show that for every $a_t\in A$, the values $\frac{1}{\delta}\xi_t^0(q_1, \dots, a_t)$ and $\frac{1}{\delta}\xi_t^1(q_1, \dots, a_t)$ satisfy all the constraints for the optimization problem defining $\phi_t(q_1, \dots, a_t)$, and consequently, $\mathcal{C}_{q_1, \dots, a_t}$ is not empty. 

    First, note that $\frac{1}{\delta}\xi_t^b(q_1, \dots, a_t)$ trivially satisfies Constraint \ref{item:cond-phi-geq-xi}. Moreover, by Constraint~\ref{item:cond-xi-upper}, it satisfies Constraint \ref{item:cond-phi-upper}. To prove Constraint \ref{item:cond-phi-normal} is satisfied, it suffices to show that
    \begin{equation}\label{eq:phi-normal-feasible}
        \begin{aligned}
        \sum_{a\in\Gamma_{<a_t}} \phi_t^b(q_1, \dots, a_{t-1}, q_t, a) &+\sum_{a\in\Gamma_{>a_t}\cup \{a_t\}} \frac{1}{\delta}\xi_t^{b}(q_1, \dots, a_{t-1}, q_t, a)\leq \phi_{t-1}^b(q_1, \dots, a_{t-1}),
    \end{aligned}
    \end{equation}
    We prove this by induction on $a_t\in A$:

    \medskip\noindent\underline{Base Case:} 
    Suppose $a_t$ is the least element of $A$. Then, $\Gamma_{<a_t}$ is empty, reducing our goal to showing that
    \begin{align*} 
        \sum_{\substack{a\in A}} \xi_t^{b}(q_1, \dots, a_{t-1}, q_t, a)\leq \delta\phi_{t-1}^b(q_1, \dots, a_{t-1}).
    \end{align*}
    Let $a+t=a_1'<_A a_2'<_A \dots$ denote the ordered elements of $A$. As $\rml_t$ is non-negative, by Constraint~\ref{item:cond-xi-normal}, for every $n\in \N$, we have
    \begin{align*} 
        \sum_{\substack{a\in \Gamma_{<a_n'}\cup\{a_n'\}}} \xi_t^{b}(q_1, \dots, a_{t-1}, q_t, a)\leq \delta\phi_{t-1}^b(q_1, \dots, a_{t-1}).
    \end{align*}
    Taking $\lim_{n\to \infty}$ when $A$ is infinite or setting $n=|A|$ when $A$ is finite gives inequality~(\ref{eq:phi-normal-feasible}).
    
    \noindent\underline{Inductive Step:}
    Let $a^*$ be the immediate predecessor of $a_t$. Since $\Gamma_{<a_t} = \Gamma_{<a^*} \cup \{a^*\}$, Constraints \ref{item:cond-phi-normal} for defining $\phi_t(q_1, \dots, a_{t-1}, q_t, a^*)$ is identical to (\ref{eq:phi-normal-feasible}). By inductive hypothesis, $\phi_t(q_1, \dots, a_{t-1}, q_t, a^*)$ exists and satisfies \ref{item:cond-phi-normal}. Thus, inequality (\ref{eq:phi-normal-feasible}) holds, completing the induction.

    \medskip\noindent\textbf{(ii):}
    For every $(q_1, \dots, a_{t-1}, q_t)$, the values of the function $\phi_t(q_1, \dots, a_{t-1}, q_t, \cdot)$ is uniquely determined by the preceding values of this function in the well-order $<_A$. Thus, by the transfinite recursion theorem, the function $\phi_t(q_1, \dots, a_{t-1}, q_t, \cdot)$ exists and is uniquely defined. Consequently, the function $\phi_t$ exists and is uniquely defined.
\end{proof}

\subsubsection{Satisfying Conditions \ref{item:cond1} to \ref{item:cond4}}\label{subsubsec:construction-conditions}
Our objective was to construct $\phi_t^b$ so that it fulfills conditions \ref{item:cond1} to \ref{item:cond4}. Combination of Constraints \ref{item:cond-phi-geq-xi} and \ref{item:cond-xi-geq-rml} gives condition \ref{item:cond1}. Moreover, Constraint~\ref{item:cond-phi-upper} directly implies condition~\ref{item:cond2}. The proof of the following claim is identical to a proof in \cite{lyu2022composition}:
\begin{claim}
    $\phi_t$ satisfies condition \ref{item:cond3}.
\end{claim}
\begin{proof}
    For each $b\in\zo$ and every transcript $(q_1, \dots, a_{t-1}, q_t)$, by Constraint \ref{item:cond-phi-normal} and non-negativity of $\xi^b$, for every $a_t$,
    $$\sum_{a\in \Gamma_{<a_t}\cup \{a_t\}} \phi_t^b(q_1, \dots, a_t) \leq \phi_{t-1}^b(q_1, \dots, a_{t-1}),$$
    which in both cases where $A$ is finite and countable implies that
    \begin{align*}
        \sum_{a\in A} \phi_t^b(q_1, \dots, a_{t-1}, q_t,a) \leq \phi_{t-1}^b(q_1, \dots, a_{t-1}).
    \end{align*}
    For the sake of contradiction, suppose there exists $b\in\zo$ and $(q_1, \dots, a_{t-1}, q_t)$ such that 
    \begin{align}\label{eq:contradiction}
        \sum_{a\in A} \phi_t^b(q_1, \dots, a_{t-1}, q_t,a) < \phi_{t-1}^b(q_1, \dots, a_{t-1}).
    \end{align}
    Fix $b$. As $\phi_t^b$ is lower bounded by $\xi_t^b/\delta$ and the values of $\phi_t$ are maximal solutions, this strict inequality implies that for every $a_t \in A $, Constraint \ref{item:cond-phi-normal} for $b$ is not tight. Consequently, for every $a_t \in A$, we must have had
    \begin{align*}
        \delta\phi_t^{b}(q_1, \dots, a_t)- e^{-\eps}\delta\phi_t^{1-b}(q_1, \dots, a_t)
        = \mu_t^{b}(q_1, \dots, a_{t})-e^{-\eps}\mu_t^{1-b}(q_1, \dots, a_{t})
    \end{align*}
    otherwise, the value of $\phi_t^b(q_1, \dots, a_t)$ could have increased, which contradicts the maximality.
    Taking a sum over all $a_t\in A$, we get
    \begin{equation}\label{eq:tight}
        \begin{aligned}
            &\sum_{a_t\in A} \delta \phi_t^{b}(q_1, \dots, a_t)- e^{-\eps} \delta\sum_{a_t\in A} \phi_t^{1-b}(q_1, \dots, a_t) = \mu_{t-1}^{b}(q_1, \dots, a_{t-1})-e^{-\eps}\mu_{t-1}^{1-b}(q_1, \dots, a_{t-1})
        \end{aligned}
    \end{equation}
    By \ref{item:cond2} and Remark~\ref{rem:inductive-hypothesis}, we know that the RHS is at least
    $$\delta\phi_{t-1}^b(q_1, \dots, a_{t-1}) - e^{-\eps} \delta \phi_{t-1}^{1-b}(q_1, \dots, a_{t-1})$$
    Plugging this to (\ref{eq:tight}) and reformulating it, we get
    \begin{align*}
        \sum_{a_t\in A} \phi_t^{b}(q_1, \dots, a_t)&-\phi_{t-1}^{b}(q_1, \dots, a_{t-1}) \geq  e^{-\eps}(\sum_{a_t\in A} \phi_t^{1-b}(q_1, \dots, a_t)- \phi_{t-1}^{1-b}(q_1, \dots, a_{t-1}))
    \end{align*}
    By (\ref{eq:contradiction}), the LHS is negative, implying that the RHS is negative too. 
    Thus, $\sum_{a\in A} \phi_t^{1-b}(q_1, \dots, a_{t-1}, q_t,a) < \phi_{t-1}^{1-b}(q_1, \dots, a_{t-1})$. Hence, by an identical argument, we have
    \begin{align*}
        \sum_{a_t\in A} \phi_t^{1-b}(q_1, \dots, a_t)&-\phi_{t-1}^{1-b}(q_1, \dots, a_{t-1}) \geq  e^{-\eps}(\sum_{a_t\in A} \phi_t^{b}(q_1, \dots, a_t)- \phi_{t-1}^{b}(q_1, \dots, a_{t-1})),
    \end{align*}
    which is a contradiction when $\eps>0$, implying that condition \ref{item:cond3} holds.
\end{proof}
To show $\phi_t$ satisfies conditions \ref{item:cond4}, we need the following lemma:
\begin{lemma}\label{lem:xi-eq-min}
    For any query sequence $(q_1, \dots, q_t)\in Q^t$, if $\mu_t^0(q_1, \dots, a_t)=\mu_t^1(q_1, \dots, a_t)$ for every answer sequence $(a_1, \dots, a_t)\in A^t$, then 
    $$\xi_t^0(q_1, \dots, a_t)=\xi_t^1(q_1, \dots, a_t)=\min\{\mu_t^0(q_1, \dots, a_t), \rml_t^0(q_1, \dots, a_t)+\rml_t^1(q_1, \dots, a_t)\}$$
    for every answer sequence $(a_1, \dots, a_t)\in A^t$.
\end{lemma}
\begin{proof}
    Fix $(q_1, \dots, a_{t-1}, q_t)\in (Q\times A)^{t-1}\times Q$. We will inductively show that for every $a_t\in A$, 
    $$\xi_t^0(q_1, \dots, a_t)=\xi_t^1(q_1, \dots, a_t)=\min\{\mu_t^0(q_1, \dots, a_t), \rml_t^0(q_1, \dots, a_t)+\rml_t^1(q_1, \dots, a_t)\}.$$
    
    Recall $\xi_t(q_1, \dots, a_t)$ is defined as the lexicographically maximum feasible solution $(x_0^*, x_1^*)$. By Constraint~\ref{item:cond-xi-extra}, for each $b\in\zo$, we have
    $$x_b^*\leq\min\{\mu_t^b(q_1, \dots, a_t), \rml_t^0(q_1, \dots, a_t)+\rml_t^1(q_1, \dots, a_t)\}.$$
    Therefore, to prove the lemma, it suffices to show
    $$x_0=x_1=\min\{\mu_t^0(q_1, \dots, a_t), \rml_t^0(q_1, \dots, a_t)+\rml_t^1(q_1, \dots, a_t)\}$$
    is a feasible solution:

    \begin{itemize}[left=0pt]
        \item \ref{item:cond-xi-geq-rml}: By Lemma~\ref{lem:rml-leq-mu-limit} and non-negativity of $\xi_t$, this constraint is satisfied.
        \item \ref{item:cond-xi-normal}: To show Constraint~\ref{item:cond-xi-normal} is satisfied, we need to show for each $b\in \zo$,
        \begin{align*}
            \sum_{a\in\Gamma_{<a_t}\cup\{a_t\}} &\min\{\mu_t^b(q_1, \dots, a_{t-1}, q_t, a), \rml_t^0(q_1, \dots, a_{t-1}, q_t, a)+\rml_t^1(q_1, \dots, a_{t-1}, q_t, a)\}\\
            &+\sum_{a\in\Gamma_{>a_t}}\rml_t^b(q_1, \dots, a_{t-1}, q_t, a) \leq \phi_{t-1}^b(q_1, \dots, a_{t-1})
        \end{align*}
        Since 
        $$\rml_t^b(q_1, \dots, a_{t-1}, q_t, a)\leq \min\{\mu_t^b(q_1, \dots, a_{t-1}, q_t, a), \rml_t^0(q_1, \dots, a_{t-1}, q_t, a)+\rml_t^1(q_1, \dots, a_{t-1}, q_t, a)\},$$ 
        it suffices to prove
        \begin{align*}
            \sum_{a\in A} &\min\{\mu_t^b(q_1, \dots, a_{t-1}, q_t, a), \rml_t^0(q_1, \dots, a_{t-1}, q_t, a)+\rml_t^1(q_1, \dots, a_{t-1}, q_t, a)\} \leq \phi_{t-1}^b(q_1, \dots, a_{t-1})
        \end{align*}
        By condition~\ref{item:cond4} and Remark~\ref{rem:inductive-hypothesis}, 
        \begin{align*}
            \phi_{t-1}^b(q_1, \dots, a_{t-1})\geq \min\{\mu_{t-1}^b(q_1, \dots, a_{t-1}), \rml_{t-1}^0(q_1, \dots, a_{t-1})+\rml_{t-1}^1(q_1, \dots, a_{t-1})\}
        \end{align*}
        Therefore, by Lemma~\ref{lem:sum-mu-eq-mu} and Lemma~\ref{lem:rml-geq-sum-rml-limit},
        \begin{align*}
            &\phi_{t-1}^b(q_1, \dots, a_{t-1})\geq \min\{\mu_{t-1}^b(q_1, \dots, a_{t-1}), \rml_{t-1}^0(q_1, \dots, a_{t-1})+\rml_{t-1}^1(q_1, \dots, a_{t-1})\}\\
            &\geq \min\left\{\sum_{a\in A} \mu_t^b(q_1, \dots, a_{t-1}, q_t, a), \sum_{a\in A} \rml_t^0(q_1, \dots, a_{t-1}, q_t, a)+ \sum_{a\in A} \rml_t^1(q_1, \dots, a_{t-1}, q_t, a)\right\}\\
            &\geq \sum_{a\in A} \min\{\mu_t^b(q_1, \dots, a_{t-1}, q_t, a), \rml_t^0(q_1, \dots, a_{t-1}, q_t, a)+\rml_t^1(q_1, \dots, a_{t-1}, q_t, a)\},
        \end{align*}
        completing the proof.
        \item \ref{item:cond-xi-upper}: Recall that $\mu_t^0(q_1, \dots, a_t)=\mu_t^1(q_1, \dots, a_t)$ for every answer sequence $(a_1, \dots, a_t)\in A^t$. To prove this constraint is satisfied, we need to show
        $$(1-e^{-\eps})\min\{\mu_t^0(q_1, \dots, a_t), \rml_t^0(q_1, \dots, a_t)+\rml_t^1(q_1, \dots, a_t)\}\leq (1-e^{-\eps}) \mu_t^0(q_1, \dots, a_t),$$
        which is trivially true.
        \item \ref{item:cond-xi-extra}: This constraint is trivially satisfied as equality holds.
    \end{itemize}
\end{proof}
\begin{claim}
    $\phi_t$ satisfies conditions \ref{item:cond4}.
\end{claim}
\begin{proof}
    For any query sequence $(q_1, \dots, q_t)\in Q^t$, if 
    \begin{equation}\label{eq:mu-0-eq-mu-1}
        \mu_t^0(q_1, \dots, a_t)=\mu_t^1(q_1, \dots, a_t)\tag{$\star$}
    \end{equation}
    for every answer sequence $(a_1, \dots, a_t)\in A^t$, then by Lemma~\ref{lem:xi-eq-min}, 
    \begin{equation}\label{eq:xi-0-eq-xi-1}
        \xi_t^0(q_1, \dots, a_t)=\xi_t^1(q_1, \dots, a_t)=\min\{\mu_t^0(q_1, \dots, a_t), \rml_t^0(q_1, \dots, a_t)+\rml_t^1(q_1, \dots, a_t)\}\tag{$\star\star$}
    \end{equation}
    for all $(a_1, \dots, a_t)\in A^t$.
    Moreover, by Lemma~\ref{lem:sum-mu-eq-mu} and \ref{eq:mu-0-eq-mu-1}, $\mu_{t-1}^0(q_1, \dots, a_{t-1})=\mu_{t-1}^1(q_1, \dots, a_{t-1})$ for all $(a_1, \dots, a_t)\in A^t$. Thus, by Remark~\ref{rem:inductive-hypothesis} and condition~\ref{item:cond4}, 
    \begin{equation}\label{eq:phi-0-eq-phi-1}
        \phi_{t-1}^0(q_1, \dots, a_{t-1})=\phi_{t-1}^1(q_1, \dots, a_{t-1}).\tag{$\star\star\star$}
    \end{equation}
    Therefore, by equations (\ref{eq:mu-0-eq-mu-1}), (\ref{eq:xi-0-eq-xi-1}) and (\ref{eq:phi-0-eq-phi-1}), Constraints~\ref{item:cond-phi-geq-xi}, \ref{item:cond-phi-normal}, and \ref{item:cond-phi-upper} are symmetric for $x_0$ and $x_1$. Hence, there exists a maximal feasible solution $(z^*, z^*)$ satisfying all the constraints. By definition of $\phi_t(q_1, \dots, a_t)$, we have $\phi_t(q_1, \dots, a_t)=(z^*, z^*)$. Also, by Constraint~\ref{item:cond-phi-geq-xi}, $z^*$ is at least $\min\{\mu_t^0(q_1, \dots, a_t), \rml_t^0(q_1, \dots, a_t)+\rml_t^1(q_1, \dots, a_t)\}$. 
    
\end{proof}

\section{$f^{\delta}_{\RR}$-concurrent composition}\label{sec:filter-comp-rr}
In this section, we prove: 
\begin{lemma*}[Restatement of Lemma~\ref{lem:filter-comp-eps-eq-zero}]
    For every $0\leq\delta\leq 1$, the $f^\delta_{\RR}$-concurrent composition of continual mechanisms is $(0, \delta)$-differentially private.
\end{lemma*}
In Section~\ref{sec:concomp-cm}, we assumed that all composed CMs have discrete answer distributions. Accordingly, we also prove Lemma~\ref{lem:filter-comp-eps-eq-zero} under this assumption. During this section, we assume the domain of all variables is finite or countable. For every random variable $Y$, we denote its support by $\supp(Y)$.

We begin by proving some intermediate results needed for the final proof. Vadhan et al.~\cite{vadhan2022concurrent} show that to guarantee an IM is differentially private, it suffices to consider only deterministic adversaries. Specifically, the following holds:

\begin{lemma}[\cite{vadhan2022concurrent}]
    An IM $\cM$ is $(\eps, \delta)$-DP w.r.t. a neighboring relation $\sim$ if and only if for every deterministic adversary $\cA$ and every pair of neighboring initial states $s\sim s'$, the random variables $\View(\cA, \cM(s))$ and $\View(\cA, \cM(s'))$ are $(\eps, \delta)$-indistinguishable.
\end{lemma}

By Definition~\ref{def:dp-cm}, this result extends naturally to CMs:

\begin{corollary}\label{cor:dp-deterministic-adversary-cm}
    A CM $\cM$ is $(\eps, \delta)$-DP w.r.t. a verification function $f$ if and only if for every deterministic adversary $\cA$, the random variables $\View(\cA, \V[f]\circstar\I(0)\circstar\cM)$ and $\View(\cA, \V[f]\circstar\I(1)\circstar\cM)$ are $(\eps, \delta)$-indistinguishable.
\end{corollary}

The definition of $(0, \delta)$-DP CMs can be equivalently expressed using the notion of \emph{total variation}:

\begin{definition}[Total Variation]\label{def:tv}
    Let $Y^0$ and $Y^1$ be two variables with the same domain $\calY$. The \emph{total variation distance} between $Y^0$ and $Y^1$, denoted $\tv(Y^0, Y^1)$, is defined as 
    $$\tv(Y^0, Y^1)=\sup_{S\subseteq \calY} \Pr[Y^0\in S]-\Pr[Y^1\in S].$$
\end{definition}

\begin{lemma}\label{lem:dp-eq-tv}
    Let $Y^0$ and $Y^1$ be two random variables over the same domain $\calY$. For every $0\leq \delta\leq 1$, $Y^0$ and $Y^1$ are $(0, \delta)$-indistinguishable if and only if $\tv(Y^0, Y^1)\leq \delta$.
\end{lemma}
\begin{proof}
    By definition, $Y^0$ and $Y^1$ are $(0, \delta)$-indistinguishable if and only if for every $S\subseteq\calY$, $\Pr[Y^0\in S]-\Pr[Y^1\in S]\leq \delta$. Therefore, $Y^0$ and $Y^1$ are $(0, \delta)$-indistinguishable if and only if $\tv(Y^0, Y^1)=\sup_{S\subseteq\calY} \Pr[Y^0\in S]-\Pr[Y^1\in S]\leq \delta$.
\end{proof}

\begin{lemma}\label{lem:tv-extend-inf}
    Let $\Sigma$ be a (discrete) set, and let $V^0$ and $V^1$ be random variables over $\Sigma^*$, i.e., finite-length sequences over $\Sigma$ that terminate with probability $1$. Let $\$$ be a symbol not in $\Sigma$. For $b\in\zo$, define the infinite sequence $W^b=(W_1^b, W_2^b, \dots)$ by padding $V^b$ with an infinite sequence of $\$$ symbols. Suppose that for some $0 \leq \delta \leq 1$, the following holds:
    \begin{align*}
        \forall k > 0, \quad \tv\left((W_1^0, \dots, W_k^0), (W_1^1, \dots, W_k^1)\right) \leq \delta.
    \end{align*}
    Then,
    \begin{align*}
        \tv(V^0, V^1) \leq \delta.
    \end{align*}
\end{lemma}

\begin{proof}
    We have
    \begin{align*}
        \tv(V^0, V^1) &= \frac{1}{2}\sum_{\sigma\in \Sigma^*}\big|\Pr[V^0=\sigma]-\Pr[V^1=\sigma]\big|\\
        &= \frac{1}{2}\sum_{\sigma\in \Sigma^*}\big|\Pr[W^0=\sigma\cdot\$^\omega]-\Pr[W^1=\sigma\cdot\$^\omega]\big|\\
        &= \frac{1}{2}\sum_{\sigma'\in \Sigma^*\cdot\$^\omega}\big|\Pr[W^0=\sigma']-\Pr[W^1=\sigma']\big|= \tv(W^0, W^1)
    \end{align*}
    Thus, we need to show $\tv(W^0, W^1)\leq \delta$.

    For each $i\in \N$ and $b\in\zo$, let $W_{1:i}^b = (W_1^b, \dots, W_i^b)$ denote the prefix of length $i$, and let $W_{i:}^b = (W_i^b, W_{i+1}^b, \dots)$ denote the infinite suffix of $W^b$. Define:
    \begin{align*}
        T=\tv(W^0, W^1), \quad T_i = \tv(W_{1:i}^0, W_{1:i}^1), \quad \text{and} \quad \alpha(i) = \Pr[W_i^0 \neq \$] + \Pr[W_i^1 \neq \$].
    \end{align*}
    We will show that
    \begin{enumerate}
        \item $T_i$ is non-decreasing in $i$.
        \item For all $i\in\N$, we have $T-T_i\leq \alpha(i+1)$.
    \end{enumerate}
    Since each $T_i \leq 1$ and the sequence is non-decreasing, the limit $T^* = \lim_{i \to \infty} T_i$ exists. Moreover, since $W_i^0=\$$ if and only if $|V^b|<i$ and $V^b$ is finite with probability $1$, the probability $\Pr[W_i^b\neq\$]$ converges to zero as $i\to \infty$. Therefore, 
    $$T - T^*= \lim_{i \to \infty}T-T_i\leq \lim_{i \to \infty} \alpha(i)=0.$$
    Since $T_i \leq \delta$ for all $i$, we conclude that $T^* \leq \delta$. Thus, $T=\tv(W^0, W^1)\leq \delta$.
    
    It remains to prove (1) and (2). Since $W_{1:i}^0$ and $W_{1:i}^1$ in $T_i = \tv(W_{1:i}^0, W_{1:i}^1)$ equal $W_{1:i+1}^0$ and $W_{1:i+1}^1$ in $T_i = \tv(W_{1:i+1}^0, W_{1:i+1}^1)$ with their last entries removed, the data-processing inequality implies $T_i\leq T_{i+1}$.
    For (2), 
    define $X_i^b=W_{1:i}^b\$^{\omega}$ for each $b\in\zo$. By definition, $T_i=\tv(X_i^0,X_i^1)$. We have
    \begin{align*}
        T-T_i&= \tv(W^0, W^1)-\tv(X_i^0,X_i^1)\\
        &= \sup_S \left(\Pr[W^0\in S]- \Pr[W^1\in S]\right)- \sup_S \left(\Pr[X^0\in S]- \Pr[X^1\in S]\right)\\
        &\leq \sup_S \left(\Pr[W^0\in S]- \Pr[W^1\in S] - \Pr[X^0\in S]+ \Pr[X^1\in S] \right)\\
        &\leq \sup_S \left(\Pr[W^0\in S]- \Pr[X^0\in S]\right)+\sup_S \left(\Pr[X^1\in S]- \Pr[W^1\in S] \right)\\
        &= \Pr[W^0\neq X_i^0]+\Pr[W^1\neq X_i^1]
    \end{align*}
    Since once an entry of $W^b$ equals $\$$ all subsequent entries remain $\$$, by definition of $X_i^b$, we have $W^b \neq X_i^b$ if and only if $W_{i+1}^b \neq \$$. Therefore,
    \begin{align*}
        T-T_i\leq \Pr[W_{i+1}^0\neq \$]+\Pr[W_{i+1}^1\neq \$]=\alpha(i+1),
    \end{align*}
    finishing the proof.
\end{proof}

\begin{lemma}[The Maximal Coupling Lemma]\label{lem:maximal-couple}
    Let $Y^0$ and $Y^1$ be random variables with domain $\calY$. There exists a joint distribution for a pair of random variables $(Z^0, Z^1)$ on $\calY\times \calY$, called a \emph{coupling}, such that:
    \begin{enumerate}[left=0pt]
        \item [(i)]\label{item:max-coupling-identical-marginals} For every $b\in \{0,1\}$, the marginal distribution of $Z^b$ is identical to the distribution of $Y^b$.
        \item [(ii)]\label{item:max-coupling-sum-mins} Let $S^0=\supp(Y^0)$ and $S^1=\supp(Y^1)$. The probability of the event $Z^0=Z^1$ is maximized and satisfies
        $$\Pr[Z^0=Z^1]=\sum_{y\in \calY}\min\left\{\Pr[Z^0=y], \Pr[Z^1=y]\right\}= \sum_{y\in S^0\cap S^1}\min\left\{\Pr[Z^0=y], \Pr[Z^1=y]\right\}.$$
    \end{enumerate}
\end{lemma}

\begin{lemma}[The Coupling Lemma]\label{lem:coupling-ineq}
    Let $Y^0$ and $Y^1$ be random variables with domain $\calY$. For every (not necessarily maximal) coupling $(Z^0, Z^1)$ of $Y^0$ and $Y^1$, where the marginals $Z^0$ and $Z^1$ are distributed as $Y^0$ and $Y^1$, we have
    $$\tv(Y^0, Y^1)\leq \Pr[Z^0\neq Z^1].$$
\end{lemma}

\begin{lemma}\label{lem:tv-couple}
    Let $k\in \N$, and let $Y^0=(Y_1^0, \dots, Y_k^0)$ and $Y^1=(Y_1^1, \dots, Y_k^1)$ be random variables. For $i\in[k-1]$ and $b\in\zo$, set $S_i^b=\supp(W_{1:i}^b)$ and $S_i=S_i^0\cap S_i^1$, with $S_0=S_0^b=\emptyset$. For $i\in[k]$, let $\Delta_i:S_{i-1}\to [0,1]$ be a function, and fix $\gamma\in [0,1]$. Suppose the following conditions hold:
    \begin{itemize}
        \item For every $i\in [k]$ and every $y_{1:i-1}\in S_{i-1}$, 
        $$\tv\left((Y_i^0|Y_{1:i-1}^0=y_{1:i-1}), (Y_i^1|Y_{1:i-1}^1=y_{1:i-1})\right)\leq \Delta_i(w_{1:i-1}).$$
        \item For every $y_{1:k-1}\in S_{k-1}$, $1-\prod_{j=1}^i \left(1-\Delta_j(w_{1:j})\right)\leq \gamma$.
    \end{itemize}
    Then, 
    $$\tv(Y^0, Y^1)\leq \gamma.$$
\end{lemma}

\begin{proof}
    We construct a coupling $\left(Z^0=(Z_1^0, \dots, Z_k^0), Z^1=(Z_1^1, \dots, Z_k^1)\right)$ for $(Y^0, Y^1)$ and show that $\Pr[Z^0 \neq Z^1] \leq \gamma$. By Lemma~\ref{lem:coupling-ineq}, this implies $\tv(Y^0,Y^1)\leq \gamma$. Throughout the proof, let $\calY$ denote the union of all possible outcomes of $Y_i^b$ over $i\in[k]$ and $b\in\zo$.

    Define $(Z_1^0,Z_1^1)$ to be the maximal coupling of $Y_1^0$ and $Y_1^1$ given by Lemma~\ref{lem:maximal-couple}. For $i=2$ to $k$, given samples $(Z_{1:i-1}^0,Z_{1:i-1}^1)$, we define $(Z_i^0,Z_i^1)$ as follows:
    \begin{itemize}[left=0pt]
        \item If the histories are identical, i.e.\ $Z_{1:i-1}^0=Z_{1:i-1}^1=y_{1:i-1}$, then $(Z_i^0,Z_i^1)$ is drawn from the maximal coupling of teh distributions of ${Y_i^0\mid Y_{1:i-1}^0=y_{1:i-1}}$ and ${Y_i^1\mid Y_{1:i-1}^1=y_{1:i-1}}$, given by Lemma~\ref{lem:maximal-couple}.
        \item Otherwise, if $Z_{1:i-1}^0=y_{1:i-1}^0$ and $Z_{1:i-1}^1=y_{1:i-1}^1$ differ, then $Z_i^0$ and $Z_i^1$ are drawn independently from their respective conditional distributions ${Y_i^0\mid Y_{1:i-1}^0=y_{1:i-1}^0}$ and ${Y_i^1\mid Y_{1:i-1}^1=y_{1:i-1}^1}$.
    \end{itemize}
    By the chain rule, the probability $\Pr[Z^0=Z^1]=\sum_{y_{1:k}\in S_k}\Pr[Z^0= Z^1\nospaceeq y_{1:k}]$ equals
    \begin{align*}
        &\sum_{y_{1:k}\in S_k} \prod_{i=1}^k \Pr\!\left[ Z_i^0\nospaceeq Z_i^1\nospaceeq y_i \,\middle|\, Z_{1:i-1}^0\nospaceeq Z_{1:i-1}^1\nospaceeq y_{1:i-1}\right]=\\
        &\sum_{y_{1:k-1}\in S_{k-1}} \prod_{i=1}^{k-1} \Pr\!\left[ Z_i^0\nospaceeq Z_i^1\nospaceeq y_i \,\middle|\, Z_{1:i-1}^0\nospaceeq Z_{1:i-1}^1\nospaceeq y_{1:i-1}\right]\negspace\sum_{\substack{y_k\in \calY\\\text{s.t. } y_{1:k}\in S_k}}\negspace\Pr\!\left[ Z_k^0\nospaceeq Z_k^1\nospaceeq y_k \,\middle|\, Z_{1:k-1}^0\nospaceeq Z_{1:k-1}^1\nospaceeq y_{1:k-1}\right].
    \end{align*}
    Conditioning on $Z_{1:k-1}^0=Z_{1:k-1}^1= y_{1:k-1}$, the pair $(Z_k^0, Z_k^1)$ is the maximal coupling of ${Y_i^0\mid Y_{1:i-1}^0=y_{1:i-1}^0}$ and ${Y_i^1\mid Y_{1:i-1}^1=y_{1:i-1}^1}$. The innermost sum is taken over the common support of these conditional random variables. Thus, by Lemma~\ref{lem:maximal-couple}, this summation equals
    $$1-\tv\left((Y_{k}^0\mid Y_{1:k-1}^0=y_{1:k-1}), (Y_{k}^1\mid Y_{1:k-1}^1=y_{1:k-1})\right).$$
    Hence, by the first assumption, 
    \begin{align}\label{eq:sum-k-geq-delta-k}
        \sum_{\substack{y_k\in \calY\\\text{s.t. } y_{1:k}\in S_k}}\negspace\Pr\!\left[ Z_k^0\nospaceeq Z_k^1\nospaceeq y_k \,\middle|\, Z_{1:k-1}^0\nospaceeq Z_{1:k-1}^1\nospaceeq y_{1:k-1}\right]\geq 1-\Delta_k(y_{1:k-1})
    \end{align}
    By the second assumption, we know
    $$1-\Delta_k(y_{1:k-1})\geq \frac{1-\gamma}{\prod_{i=1}^{k-1}(1-\Delta_i(y_{1:i-1}))}.$$
    Therefore,
    \begin{align*}
        \Pr[Z^0=Z^1]
        &\geq (1-\gamma) \sum_{y_{1:k-1}\in S_{k-1}} \prod_{i=1}^{k-1} \frac{\Pr\!\left[ Z_i^0\nospaceeq Z_i^1\nospaceeq y_i \,\middle|\, Z_{1:i-1}^0\nospaceeq Z_{1:i-1}^1\nospaceeq y_{1:i-1}\right]}{1-\Delta_{i}(y_{1:i-1})}\\
        &= (1-\gamma) \sum_{y_{1:k-2}\in S_{k-2}} \prod_{i=1}^{k-2} \frac{\Pr\!\left[ Z_i^0\nospaceeq Z_i^1\nospaceeq y_i \,\middle|\, Z_{1:i-1}^0\nospaceeq Z_{1:i-1}^1\nospaceeq y_{1:i-1}\right]}{1-\Delta_{i}(y_{1:i-1})}\\
        &\qquad\qquad\qquad\qquad\quad\sum_{\substack{y_{k-1}\in \calY\\\text{s.t. } y_{1:k-1}\in S_{k-1}}}\negspace\frac{\Pr\!\left[ Z_{k-1}^0\nospaceeq Z_{k-1}^1\nospaceeq y_{k-1} \,\middle|\, Z_{1:k-2}^0\nospaceeq Z_{1:k-2}^1\nospaceeq y_{1:k-1}\right]}{1-\Delta_{k-1}(y_{1:k-2})}.
    \end{align*}
    The term $\Delta_{k-1}(y_{1:k-2})$ is independent of $y_{k-1}$. Moreover, by an argument identical to the one for Inequality~\ref{eq:sum-k-geq-delta-k}, we have
    $$\sum_{\substack{y_{k-1}\in \calY\\\text{s.t. } y_{1:k-1}\in S_{k-1}}}\negspace\Pr\!\left[ Z_{k-1}^0\nospaceeq Z_{k-1}^1\nospaceeq y_{k-1} \,\middle|\, Z_{1:k-2}^0\nospaceeq Z_{1:k-2}^1\nospaceeq y_{1:k-1}\right]\geq 1-\Delta_{k-1}(y_{1:k-2}).$$
    Thus,
    \begin{align*}
        \Pr[Z^0=Z^1] \geq 
        (1-\gamma) \sum_{y_{1:k-2}\in S_{k-2}} \prod_{i=1}^{k-2} \frac{\Pr\!\left[ Z_i^0\nospaceeq Z_i^1\nospaceeq y_i \,\middle|\, Z_{1:i-1}^0\nospaceeq Z_{1:i-1}^1\nospaceeq y_{1:i-1}\right]}{1-\Delta_{i}(y_{1:i-1})}
    \end{align*}
    A backward iteration with the same proof shows that $\Pr[Z^0=Z^1]\geq 1-\gamma$. Consequently, $\Pr[Z^0\neq Z^1]\leq \gamma$, and by Lemma~\ref{lem:coupling-ineq}, $\tv(Y^0, Y^1)\leq \gamma$. 
\end{proof}

\begin{proof}[Proof of Lemma~\ref{lem:filter-comp-eps-eq-zero}]
    By Corollary~\ref{cor:dp-deterministic-adversary-cm} and the definition of  concurrent composition using the continual mechanism $\extconcomp$, it suffices if we show that for every deterministic adversary $\cA$, the random variables $\View(\cA, \V[f^{\delta}_{\RR}]\circstar\I(0)\circstar\extconcomp)$ and $\View(\cA, \V[f^{\delta}_{\RR}]\circstar\I(1)\circstar\extconcomp)$ are $(0, \delta)$-indistinguishable. Fix a deterministic adversary $\cA$.
    
    Since $\cA$ is deterministic, its sequence of random coins is empty (or predetermined), and its $i$-th query message to the verifier $\V[f^{\delta}_{\RR}]$ can be determined by the first $i-1$ responses from the verifier. Therefore, there exists a post-processing function $g$ such that for every $b\in\zo$, $g(V^b)$ is equivalently distributed as $\View(\cA, \V[f^{\delta}_{\RR}]\circstar\I(b)\circstar\extconcomp)$, where $V^b$ is the (finite) sequence of responses $V^b_i$ sent by $\V[f^{\delta}_{\RR}]\circstar\I(b)\circstar\extconcomp$ to $\cA$. By the (non-interactive) post-processing lemma, it suffices to prove that the random variables $V^0$ and $V^1$ are $(0, \delta)$-indistinguishable.

    For each $b\in\zo$, define $W^b=(W_1^b, W_2^b, \dots)$ as the sequence $V^b$ padded with an infinite number of $\$$ symbols. For short, we denote the subsequences $(W_1^b, \dots, W_i^b)$ and $(W_i^b, W_{i+1}^b, \dots)$ by $W_{1:i}^b$ and $W_{i:}^b$, respectively. We also define $S_i^b=\supp(W_{1:i}^b)$. By Lemma~\ref{lem:dp-eq-tv}, to show $V^0$ and $V^1$ are $(0, \delta)$-indistinguishable, we must prove $\tv(V^0, V^1)\leq \delta$. Moreover, by Lemma~\ref{lem:tv-extend-inf}, to prove $\tv(V^0, V^1)\leq \delta$, it suffices to prove for every $k\in\N$, $\tv(W_{1:k}^0, W_{1:k}^1)\leq \delta$. Fix $k\in\N$. In the rest of the proof, we will show that $\tv(W_{1:k}^0, W_{1:k}^1)\leq \delta$.
    
    As $\cA$ is deterministic, for every $i\in\N$ and $w_{1:i-1}=(w_1, \dots, w_{i-1})\in S_{i-1}^b$, the random variable $W_i^b$ conditioned on $W_{1:i-1}=w_{1:i-1}$ is distributed as follows:
    \begin{itemize}
        \item If $w_{i-1}=\$$ (i.e., the communication has halted before) or the design of $\cA$ is such that it halts the communication after receiving the $(i-1)$-th response $w_{i-1}$ given the history $w_{1:i-1}$, then $W_{1:i-1}=\$$ with probability $1$.
        \item Otherwise, $W_{1:i-1}$ is equivalently distributed to $\RR_{0, \delta_i'}(b)$, where $\delta_i'=\delta_i'(w_{1:i-1}, \cA)$ is determined by the history $w_{1:i-1}$ and the design of $\cA$.
    \end{itemize}
    Define $\Delta_i(w_{1:i-1})=0$ in the former case and $\Delta_i(w_{1:i-1})=\delta_i'(w_{1:i-1}, \cA)$ in the latter. By Lemmas~\ref{lem:dp-eq-tv} and~\ref{lem:rr-is-dp}, for every $i\in\N$ and every $w_{1:i-1}\in S_{i-1}^b$, 
    $$\tv\left((W_i^0|W_{1:i-1}^0=w_{1:i-1}), (W_i^1|W_{1:i-1}^1=w_{1:i-1})\right)\leq \Delta_i(w_{1:i-1}).$$
    Fix $w_{1:k-1}\in S_{k-1}^0$. Let $j$ be the smallest index such that $w_j=\$$. If there is no such index, set $j=k$. By Definition~\ref{def:ver-func-sum-delta-rr}~(iii), $1-\prod_{i=1}^{j}(1-\Delta_i(w_{1:i-1}))\leq \delta$ for every $w_{1:i-1}\in S_{i-1}^0\cap S_{i-1}^1$. Moreover, by the definition of the functions $\Delta_i$, we have $\prod_{i=j+1}^{k}(1-\Delta_i(w_{1:i-1}))=1$. Therefore, $1-\prod_{i=1}^k (1-\Delta_i(w_{1:i-1}))\leq \delta$. Thus, by Lemma~\ref{lem:tv-couple},
    \begin{align*}
        \tv(W_{1:k}^0, W_{1:k}^1)\leq \delta,
    \end{align*}
    finishing the proof.
\end{proof}

\section{Concurrent Filter Composition}\label{sec:filter}
In the parallel composition of continual mechanisms, the adversary sends pairs of identical queries to all but $k$ mechanisms with predetermined privacy parameters. In this section, we study a concurrent composition where the adversary can ask pairs of non-identical queries from all created mechanisms and choose the privacy parameters of all mechanisms adaptively. To limit the adversary’s power in this more general setting, we impose restrictions on the sequence of adaptively chosen privacy parameters. These restrictions are captured using a class of functions known as \emph{filters}:

\begin{definition}\label{def:filter}
    A \emph{filter} is a function $\filt:\left([0, \infty)\times [0, 1]\right)^*\to \{\top, \bot\}$ that takes a finite sequence of privacy parameters $(\eps_i, \delta_i)$ and returns either $\top$ or $\bot$. The output $\bot$ indicates that the privacy budget has been exhausted, and the interaction must be halted.
\end{definition}

Given a filter $\filt$, we define a verification function $f^\filt$ as follows and refer to the $f^\filt$-concurrent composition of CMs as the concurrent $\filt$-filter composition of CMs. The verification function $f^\filt$ verifies the validity of message formats, extracts the privacy parameters from the creation queries, and applies $\filt$ on them.

\begin{definition}\label{def:ver-func-filter}
    Let $\filt$ be a filter. For every $t\in\N$ and every input sequence of messages $m_1, \dots, m_t$, the verification function $f^\filt$ returns $\top$ if and only if all of the following conditions hold:
    \begin{enumerate}[left=0pt]
        \item [(i)] The sequence $m_1, \dots, m_t$ is suitable.
        \item [(ii)] Let $\ell$ denote the number of pairs of identical creation queries in the sequence $m_1, \dots, m_t$, and let $\sigma=\left((\eps_j', \delta_j')\right)_{j=1}^\ell$ be the corresponding sequence of privacy parameters for the mechanisms created in those pairs. It must hold that $\filt(\sigma)=\top$.
    \end{enumerate}
\end{definition}

Rogers et al.~\cite{rogers2016privacy} studied the $\filt$-filter composition of NIMs for the first time. More recently, Haney et al.~\cite{haney2023concurrent} showed that if the $\filt$-filter composition of NIMs is $(\eps, \delta)$-DP, then the same guarantee extends to the concurrent $\filt$-filter composition of IMs. We further generalize this result to the concurrent $\filt$-filter composition of CMs. To state our result, we need to define the $\filt$-filter composition of randomized response mechanisms:

\begin{definition}[$f^\filt_{\RR}$]
    Let $\filt$ be a filter. For every $t\in\N$ and every input sequence of messages $m_1, \dots, m_t$, the function $f^\filt_{\RR}$ returns $\top$ if and only if all of the following conditions hold:
    \begin{enumerate}[left=0pt]
        \item [(i)] $f^\filt(m_1, \dots, m_t)=\top$.
        \item [(ii)] All mechanisms created in $m_1, \dots, m_t$ are randomized response mechanisms.
    \end{enumerate}
\end{definition}
We refer to the $f^\filt_{\RR}$-concurrent composition of CMs as the $\filt$-filter composition of randomized response mechanisms.

\begin{theorem}\label{thm:filter-comp-cm}
    For every filter $\filt$, if the $\filt$-filter composition of randomized response mechanisms is $(\eps, \delta)$-DP, then the concurrent $\filt$-filter composition of CMs is also $(\eps, \delta)$-DP.
\end{theorem}
\begin{proof}
    Let $f=f^\filt$ and $f_{\RR}=f^\filt_{\RR}$. This proof is identical to the proof of Theorem~\ref{thm:comp-fixed-cm}, except for the following modifications: First, the right mechanism of the IPM $\cP$ is $\V[f_{\RR}]\circstar\I(b)\circstar\extconcomp$ instead of $\comp[\RR_{\eps_1, \delta_1}, \dots, \RR_{\eps_k, \delta_k}]\left((b)_{i=1}^k\right)$. Second, the variable $\sigma_1$ is set differently. In Theorem~\ref{thm:comp-fixed-cm}, $\cP$ initially interacted with its right mechanism, $\comp[\RR_{\eps_1, \delta_1}, \dots, \RR_{\eps_k, \delta_k}]\left((b)_{i=1}^k\right)$, and set $\sigma_1$ to its output sequence $(r_1, \dots, r_k)$, without ever changing it. Here, $\cP$ interacts with its right mechanism, $\V[f_{\RR}]\circstar\I(b)\circstar\extconcomp$, each time it receives a pair of creation queries as a left message and appends a randomized response $r_i$ to $\sigma_1$.
    
    Specifically, as before, the state of $\cP$ consists of two sequences $\sigma_1$ and $\sigma_2$, both initialized empty. Upon receiving a left message of the form $(\cM_i, s_i, \eps_i, \delta_i, f_i)^2$, $\cP$ takes the following steps: First, $\cP$ appends the pair $(\cT_i, t_i)$ to the sequence $\sigma_2$, where $\cT_i$ is the IPM corresponding to $\cM_i$ from Lemma~\ref{lem:cm-post-rr-lyu} and $t_i$ is its initial state. Then, it sends the message $(\beta_{\eps_i, \delta_i}, \beta_{\eps_i, \delta_i})$ to its right mechanism $\V[f_{\RR}]\circstar\I(b)\circstar\extconcomp$, where $\beta_{\eps_i, \delta_i}$ is the creation query corresponding to $\RR_{\eps_i, \delta_i}$. Upon receiving the acknowledgment response $\top$ from $\V[f_{\RR}]\circstar\I(b)\circstar\extconcomp$, $\cP$ discards it and sends the message $\left((0,i), (1,i)\right)$ to its right mechanism again. This message will be transformed by $\I(b)$ to $(b, i)$. Receiving $(b, i)$, the mechanism $\extconcomp$ sends $b$ to $\RR_{\eps_i, \delta_i}$ and returns its response $r_i$, which will be forwarded unchanged to $\cP$. $\cP$ then appends $r_i$ to $\sigma_1$. Finally, $\cP$ returns the acknowledgment message $\top$ to its left mechanism.
    
    Upon receiving a left message $\left((q_0, i), (q_1, i)\right)$, exactly the same as in the proof of Theorem~\ref{thm:comp-fixed-cm}, $\cP$ uses $r_i$ and $(\cT_i, t_i)$ to generate a left response.
    To ensure that the privacy argument from Theorem~\ref{thm:comp-fixed-cm} still holds under this modified construction, we must verify that $\cP$ can always provide each $\cT_i$ with a fresh output of $\RR_{\eps_i, \delta_i}(b)$. This follows directly from the definitions of $f$ and $f_{\RR}$ and the design of $\cP$. Specifically, the verifier $\V[f_{\RR}]$ never halts the communication and always returns the requested randomized response.
\end{proof}

\section{R\'enyi DP}\label{sec:rdp}
In this section, we show concurrent composition theorems for R\'enyi differential privacy (RDP), defined as follows.

\begin{definition}[R\'enyi Divergence]
    For $\alpha > 1$, the $\alpha$-R\'enyi divergence between two distributions $P$ and $Q$ is defined as
    \[
    D_{\alpha}(P \| Q)
    = \frac{1}{\alpha - 1} \log \mathbb{E}_{y \sim Q}
    \left[
      \left(\frac{P(y)}{Q(y)}\right)^{\alpha}
    \right].
    \]
    By taking the limit as $\alpha\to\infty$, we obtain the $\infty$-R\'enyi divergence,
    \[
    D_{\infty}(P \| Q)
    = \lim_{\alpha\to\infty} D_{\alpha}(P \| Q)
    = \log \max_{y} \frac{P(y)}{Q(y)}.
    \]
\end{definition}

\begin{definition}[R\'enyi Differential Privacy (RDP)~\cite{mironov2017renyi}]\label{def:rdp-nim}
    For $1<\alpha \leq \infty$ and $\eps \geq 0$, a NIM $\cM : \calX \to \calY$ is said to be $(\alpha, \eps)$-R\'enyi differentially private (or $(\alpha, \eps)$-RDP) w.r.t. a neighbor relation $\sim$ if for all neighboring datasets $x_0, x_1 \in \calX$, the $\alpha$-R\'enyi divergence between the output distributions satisfies
    \[
    D_{\alpha}\left(\cM(x_0), \cM(x_1)\right) \leq \eps \qquad \text{ and }\qquad D_{\alpha}\left(\cM(x_1), \cM(x_0)\right)\leq \eps.
    \]
\end{definition}

Note that the above definition includes pure DP as a special case. Lyu~\cite{lyu2022composition} and Vadhan et al.~\cite{vadhan2022concurrent} extend this definition to IMs, which we further generalize to CMs:

\begin{definition}[RDP for CMs]\label{def:rdp-cm}
    For $1<\alpha \leq \infty$ and $\eps\geq 0$, a continual mechanism $\cM$ is said to be $(\alpha, \eps)$-R\'enyi differentially private (or $(\alpha, \eps)$-RDP) w.r.t. a verification function $f$ against adaptive adversaries if for every adversary $\cA$, the random variables $V^0=\View(\cA, \V[f]\circstar\I(0)\circstar\cM)$ and $V^1=\View(\cA, \V[f]\circstar\I(1)\circstar\cM)$ satisfy
    \[
    D_{\alpha}(V^0, V^1) \leq \eps \qquad \text{ and }\qquad D_{\alpha}(V^1, V^0)\leq \eps.
    \]
\end{definition}

Recall the concurrent composition of $k$ CMs with fixed privacy parameters $\{(\eps_i, \delta_i)\}_{i=1}^k$, which was formally defined via a verification function $f^{\eps_1, \delta_1, \dots, \eps_k, \delta_k}$ in Definition~\ref{def:ver-func-fixed-param-comp-cm}. This verification function ensures that every creation query $(\cM_j', s_j', \eps_j', \delta_j', f_j')$ is truthful, meaning that the mechanism $\cM_j'$ is $(\eps_j', \delta_j')$-DP w.r.t. the verification function $f_j'$, and the multiset of parameters $(\eps_j', \delta_j')$ for all created mechanisms forms a sub-multiset of the predetermined parameters $\{(\eps_i, \delta_i)\}_{i=1}^k$.

For RDP, creation queries are of the form $(\cM_j', s_j', \alpha_j', \eps_j', f_j')$, and the corresponding concurrent composition verification function must ensure that each $\cM_j'$ is $(\alpha_j', \eps_j')$-RDP with respect to $f_j'$. Fixing $\alpha$ and $\eps_1, \dots, \eps_k$, this verification function must check that all created mechanisms satisfy $\alpha_j' = \alpha$, and that the multiset of parameters $\eps_j'$ is a sub-multiset of the prefix multiset $\{\eps_i\}_{i=1}^k$. With these modifications to Definition~\ref{def:ver-func-fixed-param-comp-cm}, we denote the resulting verification function by $f^{\eps_1, \dots, \eps_k}_{\mathrm{RDP}, \alpha}$ and refer to the $f^{\eps_1, \dots, \eps_k}_{\mathrm{RDP}, \alpha}$-concurrent composition of CMs as the \emph{$\alpha$-RDP concurrent composition} of CMs with parameters $\eps_1, \dots, \eps_k$.

Lyu~\cite{lyu2022composition} analyzed the concurrent composition of $k$ fixed IMs $\cM_1, \dots, \cM_k$, where each $\cM_i$ is $(\alpha, \eps_i)$-RDP w.r.t. a neighbor relation $\sim_i$. The analysis considers an adversary that concurrently interacts with $\cM_1, \dots, \cM_k$, initialized with states (datasets) $s_1, \dots, s_k$, and compares this interaction to one where the mechanisms are initialized with neighboring states $s_1', \dots, s_k'$, such that $s_i \sim_i s_i'$ for all $i \in [k]$. Lyu shows that the $\alpha$-R\'enyi divergence between the resulting adversarial view distributions is at most $\eps = \sum_{i=1}^k \eps_i$. This result was originally proved to hold assuming that the number of queries made by the adversary is bounded by a fixed integer $T$. However, Haney et al.~\cite{haney2023concurrent} later demonstrated that this assumption can be removed. (See the application of the result of Van Erven and Harremoës~\cite{van2014Renyi} in their paper.)

We can replace each IM $\cM_i$ by a \emph{universal interactive mechanism} $\mathcal{UM}_i$ whose initial state space is $\zo$. Given an initial state $b \in \zo$, the mechanism $\mathcal{UM}_i$ receives as its first message a creation query of the form $(\cM_i', s_i', \alpha, \eps_i, f_i)$ and executes $\V[f_i] \circstar \I(b) \circstar \cM_i'(s_i')$ from that point on. If the initial message is malformed or if $\cM_i'$ does not satisfy $(\alpha, \eps_i)$-RDP w.r.t. $f_i$, then $\mathcal{UM}_i$ immediately halts the interaction. It is straightforward to verify that $\mathcal{UM}_i$ is $(\alpha, \eps_i)$-RDP w.r.t. the neighbor relation $\sim_{01}$ defined by $0 \sim_{01} 1$. By construction, the concurrent composition of the universal mechanisms $\mathcal{UM}_1, \dots, \mathcal{UM}_k$ is equivalent to the $\alpha$-RDP concurrent composition of CMs with parameters $\eps_1, \dots, \eps_k$. Hence, the result of Lyu~\cite{lyu2022composition} for IMs implies the following.

\begin{corollary}\label{cor:comp-fixed-param-rdp}
    For $1<\alpha \leq \infty$, $k\in\N$, and $\eps_1, \dots, \eps_k\geq 0$, the $\alpha$-RDP concurrent composition of CMs with parameters $\eps_1, \dots, \eps_k$ is $(\alpha, \sum_{i=1}^k\eps_i)$-RDP.
\end{corollary}

Previously, the concurrent \emph{parallel} composition of CMs was defined via the verification function $f^{\eps_1, \delta_1, \dots, \eps_k, \delta_k}_\infty$ in Definition~\ref{def:ver-func-par-fixed-param-comp-cm}. Analogously, we define the verification function $f^{\eps_1, \dots, \eps_k}_{\infty, \mathrm{RDP}, \alpha}$ for RDP with the following modifications: (i) Creation queries are of the form $(\cM_j', s_j', \alpha_j', \eps_j', f_j')$ rather than $(\cM_j', s_j', \eps_j', \delta_j', f_j')$. (ii) The verification function $f^{\eps_1, \dots, \eps_k}_{\infty, \mathrm{RDP}, \alpha}$ ensures that each $\cM_j'$ is $(\alpha_j', \eps_j')$-RDP w.r.t. $f_j'$, and that all $\alpha_j'$ equal the fixed $\alpha$. We refer to the $f^{\eps_1, \dots, \eps_k}_{\infty, RDP, \alpha}$-concurrent composition as \emph{the $\alpha$-RDP concurrent $k$-sparse parallel composition of CMs with privacy parameters $\eps_1, \dots, \eps_k$}. The following theorem shows that this concurrent composition is not necessarily R\'enyi private:

\begin{theorem}\label{thm:counter-example-rdp}
    For every $1<\alpha < \infty$ and $\eps \geq 0$, the $f^{\eps}_{\infty, RDP, \alpha}$-concurrent composition is not $(\alpha, \eps')$-DP for every $\eps'\geq 0$.
\end{theorem}

\begin{proof}
    Recall the proof of Theorem~\ref{thm:counter-example}, in which an adversary $\cA_\delta$ repeatedly created instances of a $(0, \delta)$-DP CM, denoted $\cM_\delta$, with a pair of zeros as its first input pair until one mechanism returned $\bot$ as response. At that point, $\cA_\delta$ sent the input pair $(0,1)$ to the last mechanism, and the mechanism responded with the bit held by the identifier, thereby revealing it. In that construction, $\cM_\delta$ was $(0, \delta)$-DP w.r.t. a verification function $g$ that was defined to accept only message sequences of length at most two, where each message was a pair of bits and at most one pair included different bits. 

    This proof is identical to the proof of Theorem~\ref{thm:counter-example}, except that the adversary $\cA_{\alpha, \eps, \eps'}$ creates instances of a mechanism $\cM_{\alpha, \eps, \eps'}$ that is $(\alpha, \eps)$-RDP w.r.t. the (same) verification function $g$, and instead of a full disclosure of the identifier's bit, we will show that the $f^{\eps}_{\infty, RDP, \alpha}$-concurrent composition fails to satisfy $(\alpha, \eps')$-DP.

    \medskip\noindent\underline{Definition of $\cM_{\alpha, \eps, \eps'}$:}
    Like $\cM_\delta$, the mechanism $\cM_{\alpha, \eps, \eps'}$ receives single-bit inputs and outputs a response in $\{0,1,\top, \bot\}$. Upon receiving its first input bit, $\cM_{\alpha, \eps, \eps'}$ outputs $\bot$ with probability
    \[
        p = \min\left\{\frac{1}{2}, \eps\cdot(\alpha-1)\cdot6^{-\alpha} \cdot 2^{-\eps'\cdot(\alpha-1)}\right\},
    \]
    and outputs $\top$ otherwise. 
    If the first response was $\top$, then on its second input, $\cM_{\alpha, \eps, \eps'}$ again outputs $\bot$ with probability $p$ and $\top$ otherwise. However, if the first response was $\bot$, then on its next input it outputs a bit correlated with the input bit. In the $(0, \delta)$-DP construction, $\cM_\delta$ revealed the input bit exactly, resulting in a full disclosure of the private bit. In contrast, $\cM_{\alpha, \eps, \eps'}$ outputs the true input bit with probability
    \[
        q = 1 - \frac{3}{4}2^{-2-\eps'\frac{\alpha-1}{\alpha}},
    \]
    and outputs the opposite bit otherwise. Note that $q\geq 1-\frac{3}{4}2^{-2}> \frac{3}{4}$ and $q \leq 1$.

    \medskip\noindent\underline{$\cM_{\alpha, \eps, \eps'}$ is $(\alpha, \eps)$-RDP w.r.t. $g$:}
    By the definition of the verification function $g$, the mechanism $\V[g] \circstar \I \circstar \cM_{\alpha, \eps, \eps'}$ interacts with an adversary at most twice and accepts a pair of bits each time, where at most one pair contains differing bits. By construction of $\cM_{\alpha, \eps, \eps'}$, the $\alpha$-R\'enyi divergence between the adversary’s views is maximized when the adversary first sends identical bits, and when observing the response $\bot$, it sends a pair of differing bits as the second query. Let $\cA$ denote such an adversary. For $b \in \zo$, we denote the view of $\cA$ interacting with $\V[g] \circstar \I(b) \circstar \cM_{\alpha, \eps, \eps'}$ by $V_\cA^b$.

    For every view $v$ where the first response of $\cM_{\alpha, \eps, \eps'}$ is $\top$, the probabilities $\Pr[V_\cA^0 = v]$ and $\Pr[V_\cA^1 = v]$ are equal; hence,
    $$\left(\frac{\Pr[V_\cA^0 = v]}{\Pr[V_\cA^1 = v]}\right)^\alpha=1.$$
    For every view $v$ whose first response is $\bot$, the term $\bigl(\frac{\Pr[V_\cA^0 = v]}{\Pr[V_\cA^1 = v]}\bigr)^\alpha$ equals $(\frac{q}{1-q})^\alpha$ if the second response is $0$, and $(\frac{1-q}{q})^\alpha$ if the second response is $1$. Since $q\in [3/4,1]$, we have $\frac{1-q}{q}\leq \frac{1}{1-q}$. Thus, both $(\frac{q}{1-q})^\alpha$ and $(\frac{1-q}{q})^\alpha$ are upper bounded by $(\frac{1}{1-q})^\alpha$. The probability that the first response in the random variable in $V_\cA^0$ equals $\bot$ is $p$. Thus, for the worst-case adversary $\cA$ we have
    \begin{align*}
        D_\alpha(V_\cA^0\|V_\cA^1)
        &\leq \frac{1}{\alpha-1}\log\left((1-p)\cdot 1+ p\cdot\left(\frac{1}{1-q}\right)^\alpha\right)\\
        &\leq \frac{1}{\alpha-1}\log\left(1+ p\cdot\left(\frac{4}{3}2^{2+\eps'\frac{\alpha-1}{\alpha}}\right)^\alpha\right)\\
        &\leq \frac{1}{\alpha-1}\log\left(1+ p\cdot 6^\alpha \cdot 2^{\eps'(\alpha-1)}\right)\\
        &\leq \frac{1}{\alpha-1}\log\left(1+ \eps\cdot(\alpha-1)\right)\\
        &\leq \eps 
    \end{align*}
    In the last inequality, we use the fact that $\log(1+x)\leq x$ for every $x\geq 0$.
    By symmetry, $D_\alpha(V_\cA^1\| V_\cA^0) \leq \eps$ as well, and therefore $\cM_{\alpha, \eps, \eps'}$ satisfies $(\alpha, \eps)$-RDP w.r.t. $g$.

    \medskip\noindent\underline{The $f^{\eps}_{\infty, \mathrm{RDP}, \alpha}$-concurrent composition is not $(\alpha, \eps')$-DP:}
    For $b \in \zo$, let $V^b$ denote the view of $\cA_{\alpha, \eps, \eps'}$ interacting with 
    $\V[f^{\eps}_{\infty, \mathrm{RDP}, \alpha}] \circstar \I(b) \circstar \extconcomp$. 
    We show that $D_\alpha(V^0 \| V^1) > \eps'$.

    By definition, $\cA_{\alpha, \eps, \eps'}$ eventually receives a $\bot$ response from an instance of $\cM_{\alpha, \eps, \eps'}$. It then sends $(0,1)$ to that instance and receives an answer $a \in \zo$. For every view $v$ where $a = 0$,
    \[
    \left(\frac{\Pr[V^0=v]}{\Pr[V^1=v]}\right)^{\alpha}
    = \left(\frac{q}{1-q}\right)^{\alpha},
    \]
    and for every view $v$ with $a = 1$, 
    \[
    \left(\frac{\Pr[V^0=v]}{\Pr[V^1=v]}\right)^{\alpha}
    = \left(\frac{1-q}{q}\right)^{\alpha}.
    \]
    When sampling $v$ from $V^0$, the event $a = 0$ occurs with probability $q$. Hence,
    \begin{align*}
        D_\alpha(V^0\|V^1)
        &= \frac{1}{\alpha-1}
           \log\left(
              q \cdot \left(\frac{q}{1-q}\right)^{\alpha}
              + (1-q)\cdot \left(\frac{q}{q}\right)^{\alpha}
           \right)\\
        &\ge \frac{1}{\alpha-1}
           \log\left(
              q \left(\frac{q}{1-q}\right)^{\alpha}
           \right)\\
        &\ge \frac{1}{\alpha-1}
           \log\left(
              \frac{3}{4}\cdot
              \left(\frac{3}{4(1-q)}\right)^{\alpha}
           \right)\\
        &= \frac{1}{\alpha-1}
           \log\left(
              \frac{3}{4}\cdot 
              \left(
                 \frac{3}{4}\cdot
                 \frac{1}{\frac{3}{4}\cdot 2^{-2-\eps'\frac{\alpha-1}{\alpha}}}
              \right)^{\alpha}
           \right)\\
        &= \frac{1}{\alpha-1}
           \log\left(
              \frac{3}{4}\cdot 2^{2+\eps'(\alpha-1)}
           \right)\\
        &> \eps'.
    \end{align*}
\end{proof}

As in Section~\ref{subsec:parallel-restricted-ver-func}, to overcome this lower bound, we modify the verification function $f^{\eps_1, \dots, \eps_k}_{\infty, \mathrm{RDP}, \alpha}$ to only create mechanisms $\cM_j'$ whose associated verification function $f_j'$ is first-pair consistent (see Definition~\ref{def:fpc}). Let $f^{\eps_1, \dots, \eps_k}_{\infty, \mathit{FPC}, \mathrm{RDP}, \alpha}$ denote this verification function.

\begin{theorem}\label{thm:parallel-comp-fpc-rdp}
    For $\alpha >1$, $k\in\N$, and $\eps_1, \dots, \eps_k\geq 0$, the $f^{\eps_1, \dots, \eps_k}_{\infty, \mathit{FPC}, \mathrm{RDP}, \alpha}$-concurrent composition of CMs is $(\alpha, \sum_{i=1}^k\eps_i)$-RDP.
\end{theorem}

\begin{proof}
    The proof is identical to the proof of Theorem~\ref{thm:parallel-fixed-param-comp-fpc}, with a few modifications in the construction of $\cP$. Recall that in the proof of Theorem~\ref{thm:parallel-fixed-param-comp-fpc}, receiving the $j$-th pair of identical creation queries $(\cM_j', s_j', \eps_j', \delta_j', f_j')^2$, the IPM $\cP$ temporarily stored the creation query until it received the first query pair for mechanism~$\cM_j$, i.e., a message of the form $\left((q_0,j),(q_1,j)\right)$. Omitting the implementation details, the idea was that $\cP$ instantiated either an IM $\V[f_j']\circstar\I(0)\circstar\cM_j'$ if $q_0=q_1$ on the first query pair or an IPM~$\cT_j$ otherwise, where $\cT_j$ was the IPM corresponding to~$\cM_j'$ from Lemma~\ref{lem:cm-post-rr-lyu}. This lemma constructed $\cT_j$ so that the IMs $\cT_j\circstar \RR_{\eps_j',\delta_j'}(b)$ and $\V[f_j']\circstar\I(b)\circstar\cM_j'$ are equivalent for every $b\in\{0,1\}$. We will call $\cM_j'$ a \emph{type-1} CM in the former case and a \emph{type-2} CM in the latter. Upon receiving a left message of the form $\left((q_0,j),(q_1,j)\right)$, the IPM $\cP$ forwarded $(q_0,q_1)$ either to $\V[f_j']\circstar\I(0)\circstar\cM_j'$ (type~1) or to $\cT_j$ (type~2), and returned the response as a left message.

    The IPM $\cT_j$ might want to interact with its right mechanism, requesting the outcome of $\RR_{\eps_j', \delta_j'}(b)$. To supply answers, $\cP$ initially interacted with $\comp(\RR_{\eps_1,\delta_1},\ldots,\RR_{\eps_k,\delta_k})$, stored its response $(r_1,\dots,r_k)$, and later used these $r_i$'s as right responses to the IPMs $\cT_j$. This reduced the $f_{\infty,\mathit{FPC}}^{\eps_1,\delta_1,\dots,\eps_k,\delta_k}$-concurrent composition of CMs to the composition of the noninteractive mechanisms $\RR_{\eps_1,\delta_1},\dots,\RR_{\eps_k,\delta_k}$.

    In this proof we make the following changes:
    \begin{enumerate}
        \item Creation queries now have the form $(\cM_j', s_j', \alpha_j', \eps_j', f_j')$ instead of $(\cM_j', s_j', \eps_j', \delta_j', f_j')$.
        \item The right mechanism of $\cP$ is  
        \[
            \V[f^{\eps_1,\dots,\eps_k}_{\mathrm{RDP},\alpha}]\circstar\I(b)\circstar\extconcomp,
        \]
        instead of $\comp(\RR_{\eps_1,\delta_1},\dots,\RR_{\eps_k,\delta_k})$.
        \item Rather than instantiating $\cT_j$ via Lemma~\ref{lem:cm-post-rr-lyu} when $q_0=q_1$ on the first query pair, $\cP$ now forwards the pair of identical creation queries $(\cM_j', s_j', \alpha_j', \eps_j', f_j')^2$ to its right mechanism and stores the index $i_j$ indicating that $\cM_j'$ corresponds to the $i_j$-th mechanism created by the right mechanism of $\cP$.
        \item Upon receiving a left message of the form $\left((q_0,j),(q_1,j)\right)$, if $\cM_j'$ is a type-2~CM, then $\cP$ forwards the message $\left((q_0,i_j),(q_1,i_j)\right)$ to its right mechanism and returns that response to the left. If $\cM_j'$ is a type-1~CM, then as before $\cP$ forwards $(q_0,q_1)$ to $\V[f_j']\circstar\I(0)\circstar\cM_j'$ and returns the resulting answer.
    \end{enumerate}

    The exact same argument as in the proof of Theorem~\ref{thm:parallel-fixed-param-comp-fpc} shows that the $f^{\eps_1,\ldots,\eps_k}_{\infty,\mathit{FPC},\mathrm{RDP},\alpha}$-concurrent composition of CMs satisfies the privacy guarantees of the $f^{\eps_1,\ldots,\eps_k}_{\mathrm{RDP},\alpha}$-concurrent composition of CMs. By Theorem~\ref{cor:comp-fixed-param-rdp}, the $f^{\eps_1,\ldots,\eps_k}_{\mathrm{RDP},\alpha}$-concurrent composition of CMs is $(\alpha, \sum_{i=1}^k \eps_i)$-RDP. Therefore, the $f^{\eps_1,\ldots,\eps_k}_{\infty,\mathit{FPC},\mathrm{RDP},\alpha}$-concurrent composition of CMs is also $(\alpha, \sum_{i=1}^k \eps_i)$-RDP, completing the proof.
\end{proof}

As discussed in Section~\ref{subsubsec:parallel-comp-im}, IMs satisfy the first-pair consistent condition. Thus, by Theorem~\ref{thm:parallel-comp-fpc-rdp}, we have: 

\begin{corollary}\label{cor:parallel-fixed-param-comp-im-rdp}
    For $\alpha >1$, $k\in\N$, and $\eps_1, \dots, \eps_k\geq 0$, the $\alpha$-RDP concurrent $k$-sparse parallel composition of \emph{IMs} with privacy parameters $\eps_1, \dots, \eps_k$ is $(\alpha, \sum_{i=1}^k\eps_i)$-RDP.
\end{corollary}

Finally, to define concurrent filter composition of RDP CMs, we modify Definition~\ref{def:filter} so that a filter $\filt$ maps a finite sequence of nonnegative reals in $[0,\infty)^*$ (instead of a sequence in $\left([0,\infty)\times [0,1]\right)^*$) to an element in $\{\top, \bot\}$. Fix $\alpha > 1$. As before, we obtain an RDP version of the verification function $f^\filt$ in Definition~\ref{def:ver-func-filter} by replacing creation queries to use RDP parameters, ensuring that each created mechanism with parameters $(\alpha_j, \eps_j)$ is $(\alpha_j, \eps_j)$-RDP, and checking whether all $\alpha_j = \alpha$. We denote this verification function by $f^{\filt}_{\mathrm{RDP}, \alpha}$.

For parameters $\alpha > 1$ and $\eps \ge 0$, we define the filter $\filt_{\mathrm{RDP}, \alpha, \eps}$ as follows. For every $t \in \N$ and every $\eps_1, \dots, \eps_t \in [0,\infty)^t$,
\[
    \filt_{\mathrm{RDP}, \alpha, \eps}(\eps_1, \dots, \eps_t)
    =
    \begin{cases}
        \top, & \text{if } \sum_{i=1}^t \eps_i \le \eps,\\[4pt]
        \bot, & \text{otherwise.}
    \end{cases}
\]
We refer to the $f^{\filt_{\mathrm{RDP}, \alpha, \eps}}$-concurrent composition of CMs as the \emph{concurrent $\filt_{\mathrm{RDP}, \alpha, \eps}$-filter composition} of CMs.

Haney et al.~\cite{haney2023concurrent} show that the concurrent $\filt_{\mathrm{RDP}, \alpha, \eps}$-filter composition of IMs is $(\alpha, \eps)$-RDP. Following the same argument as for the concurrent composition of $k$ fixed RDP CMs, this result extends to CMs by replacing each IM with the universal IM with the same privacy parameters:

\begin{corollary}\label{cor:filter-comp-rdp}
    For every $\alpha > 1$ and $\eps \ge 0$, the concurrent $\filt_{\mathrm{RDP}, \alpha, \eps}$-filter composition of CMs is $(\alpha, \eps)$-RDP against adaptive adversaries. 
\end{corollary}

\section{$f$-DP}\label{sec:f-dp}
In this section, we show concurrent composition theorems for $f$-differential privacy ($f$-DP), which generalizes $(\eps, \delta)$-DP. $f$-DP takes a statistical point of view and uses trade-off functions to measure how indistinguishable the output distributions of a mechanism are for two neighboring datasets.

\begin{remark}
    To avoid ambiguity, throughout this section we use the symbol $g$ to denote verification functions and the symbol $f$ to denote privacy functions.
\end{remark}

Consider a NIM $\cM:\calX\to\calY$ and two neighboring datasets $x_0,x_1\in \calX$. Observing the outcome of $\cM$, an analyst/adversary $\cA$ must decide whether the input dataset was $x_0$ (null hypothesis $H_0$) or $x_1$ (alternative hypothesis $H_1$). A {\em rejection rule} $\phi:\calY\to \zo$ is a function that takes the answer of $\cM$ as input and decides whether to reject the null hypothesis or not. Thus, for $b\in\zo$, $\cA$ chooses dataset $x_b$ when $\phi$ returns $b$. To quantify the indistinguishability between $\cM(x_0)$ and $\cM(x_1)$, $f$-DP evaluates the error made by the analyst in selecting the correct hypothesis. Specifically, given a rejection rule $\phi$, the \emph{type I error} $\alpha_\phi=\E[\phi(\cM(x_0))]$ is the probability of rejecting the null hypothesis $H_0$ while $H_0$ is true (i.e., guessing $x_1$ when the dataset is $x_0$). Similarly, the \emph{type II error} $\beta_\phi=1-\E[\phi(\cM(x_1))]$ is the probability of not rejecting $H_0$ while $H_1$ is true (i.e., guessing $x_0$ when the dataset is $x_1$). Although an analyst could achieve zero type II error by always rejecting the null hypothesis, when an upper bound on the type I error is set, the minimum achievable type II error reflects the difficulty of distinguishing between $\cM(x_0)$ and $\cM(x_1)$. The following function measures the optimal trade-off between the type I and type II errors.

\begin{definition}[Trade-off Function]\label{def:trade-off-func}
    Let $Y_0$ and $Y_1$ be two random variables on the same domain $\calY$. The trade-off function between $Y_0$ and $Y_1$, denoted $T(Y_0,Y_1):[0,1]\to[0,1]$, is defined as:
    $$T(Y_0,Y_1)(\alpha)=\inf_\phi\{\beta_\phi: \alpha_\phi\le \alpha\},$$
    where the infimum is taken over all possible rejection rules. Given a function $f:[0,1]\to[0,1]$, we write $T(Y_0,Y_1)\geq f$ if $T(Y_0,Y_1)(\alpha)\geq f(\alpha)$ for every $\alpha\in [0,1]$.
\end{definition}

The (meta-)function $T$ in Definition \ref{def:trade-off-func} is indeed a function that maps a pair of random variables to a trade-off function between these random variables. 

\begin{definition}[$f$-DP for NIMs~\cite{dong2022gaussian}]\label{def:f-dp-nim}
    Let $f:[0,1]\to[0,1]$ be a function. A NIM $\cM:\calX\to \calY$ is $f$-differentially private (or $f$-DP) w.r.t. a neighbor relation $\sim$ if for every two neighboring datasets $x_0,x_1\in \calX$,
    $$T(\cM(x_0), \cM(x_1))\geq f,$$
    where $T(\cM(x_0), \cM(x_1))$ denotes the trade-off function defined in Definition~\ref{def:trade-off-func}. 
\end{definition}

For $\eps\geq 0$ and $0\leq \delta\leq 1$, define the function $f_{\eps, \delta}:[0,1]\to[0,1]$ as $f_{\eps,\delta}(\alpha)=\max\{0, 1-\delta-\exp(\eps)\alpha, \exp(-\eps)(1-\delta-\alpha)\}$ for every $\alpha\in [0,1]$. Dong, Roth, and Su~\cite{dong2022gaussian} show that $f_{\eps,\delta}$-DP is equivalent to $(\eps,\delta)$-DP. They also show that if a NIM $\cN$ satisfies $f$-DP w.r.t. a neighbor relation $\sim$, where $f$ is the trade-off function of two random variables $Z_0$ and $Z_1$, then for every pair of neighboring datasets $x_0 \sim x_1$, there exists a randomized post-processing function $\textsc{Post}$ such that, for each $b\in\{0,1\}$, the distributions of $\cN(x_b)$ and $\textsc{Post}(Z_b)$ are identical.

Vadhan and Zhang~\cite{vadhan2022concurrent} prove an analogue of Lemma~\ref{lem:cm-post-rr-lyu} for $f$-DP under an additional \emph{finite communication} assumption, defined below.

\begin{definition}[Finite Communication for IMs]\label{def:finite-communication}
    An IM $\cM$ is said to have \emph{finite communication} if there exists a constant $c$ such that both its query space $Q_\cM$ and answer space $A_\cM$ have cardinality at most $c$, and $\cM$ terminates after receiving more than $c$ queries.
\end{definition}

\begin{lemma}[\cite{vadhan2022concurrent}]\label{lem:post-im-f-dp}
    Let $f$ be a privacy function, and let $\cM$ be an IM that is $f$-DP w.r.t. a neighbor relation $\sim$ and has finite communication. Then, for every pair of neighboring datasets $x_0\sim x_1$, there exists an IPM $\cP$ with a single initial state and a NIM $\cN$ taking a bit as input such that the IMs $\cP\circstar\cN(b)$ and $\cM(x_b)$ are equivalent for every $b\in\zo$, and the mechanism $\cN$ is $f$-DP w.r.t. the neighbor relation $\sim_{01}$ defined by $0 \sim_{01} 1$.
\end{lemma}

We now extend the definition of $f$-DP from IMs to CMs:

\begin{definition}[$f$-DP for CMs]\label{def:f-dp-cm}
    Let $f:[0,1]\to[0,1]$ be a function. A continual mechanism $\cM$ is said to be \emph{$f$-differentially private} (or \emph{$f$-DP}) w.r.t. a verification function $g$ against adaptive adversaries if, for every adversary $\cA$, the random variables $V^0=\View(\cA, \V[f]\circstar\I(0)\circstar\cM)$ and $V^1=\View(\cA, \V[f]\circstar\I(1)\circstar\cM)$ satisfy
    $$T(V^0, V^1)\geq f,$$
    where $T(V^0, V^1)$ denotes the trade-off function defined in Definition~\ref{def:trade-off-func}.
\end{definition}

In the DP setting, creation queries were of the form $(\cM, s, \eps, \delta, f)$, where $\cM_i$ is a CM with the initial state $s$ that is $(\eps, \delta)$-DP w.r.t. the verification function $f$. In the $f$-DP setting, creation queries are of the form $(\cM, s, f, g)$, where $\cM$ is $f$-DP w.r.t. the verification function $g$. We emphasize that the symbol $f$ in this section denotes a \emph{privacy function}, whereas $f$ in the previous DP sections denoted a \emph{verification function}.

Recall that the concurrent composition of $k$ CMs with fixed privacy parameters $\{(\eps_i, \delta_i)\}_{i=1}^k$ was defined via a verification function $f^{\eps_1, \delta_1, \dots, \eps_k, \delta_k}$ in Definition~\ref{def:ver-func-fixed-param-comp-cm}. For the $f$-DP setting, we modify this verification function to accommodate creation queries of the form $(\cM_j', s_j', f_j', g_j')$, rather than $(\cM_j', s_j', \eps_j', \delta_j', f_j')$.  
Specifically, the new verification function must ensure that each created mechanism $\cM_j'$ is $f_j'$-DP w.r.t. $g_j'$ and has \emph{finite communication}. Given privacy functions $f_1, \dots, f_k$, the verification function must verify that the multiset of the privacy functions $f_j'$ is a sub-multiset of the prefix functions $\{f_i\}_{i=1}^k$. With these modifications to Definition~\ref{def:ver-func-fixed-param-comp-cm}, 
we denote the resulting verification function by $g^{f_1, \dots, f_k}_{\mathrm{FDP}}$ and refer to the $g^{f_1, \dots, f_k}_{\mathrm{FDP}}$-concurrent composition of CMs as the \emph{concurrent composition of CMs with privacy functions $f_1, \dots, f_k$}.

An argument identical to that of Theorem~\ref{thm:comp-fixed-param-cm}, with Lemma~\ref{lem:post-im-f-dp} being used instead of Lemma~\ref{lem:post-im} implies:

\begin{corollary}\label{cor:comp-fixed-param-f-dp}
    For $k\in\N$, let $f, f_1, \dots, f_k: [0,1] \to [0,1]$ be functions. Suppose that for every $k$ NIMs $\cN_1, \dots, \cN_k$, where each $\cN_i$ is $f_i$-DP w.r.t. a neighbor relation $\sim_i$, the composition of $\cN_1, \dots, \cN_k$ is $f$-DP. Then, the concurrent composition of CMs with privacy functions $f_1, \dots, f_k$ is also $f$-DP.
\end{corollary}
\begin{proof}
    The proof is exactly the same as the proof of Theorem~\ref{thm:comp-fixed-param-cm}, with Lemma~\ref{lem:post-im-f-dp} applied instead of Lemma~\ref{lem:post-im}. For completeness, we outline how the additional finite communication requirement of Lemma~\ref{lem:post-im-f-dp} is satisfied.

    Recall that in Theorem~\ref{thm:comp-fixed-param-cm}, the $i$-th created mechanism, $\cM_i$, was $(\eps_i, \delta_i)$-DP w.r.t. a verification function $f_i$, which we denote here by $g_i$ to avoid confusion. The proof of that theorem showed that for each $b \in \{0,1\}$, the IMs $\V[f^{\eps_1, \delta_1, \dots, \eps_k, \delta_k}]\circstar\I(b)\circstar\extconcomp$ and $\V[f^{\eps_1, \delta_1, \dots, \eps_k, \delta_k}]\circstar\iecc(b)$ are equivalent, where $\iecc(b)$ was an IM internally executing the mechanism $\V[g_i]\circstar\I(b)\circstar\cM_i$ for each created mechanism $\cM_i$. The argument then replaced each instance $\V[g_i] \circstar \I(b) \circstar \cM_i$ by $\cT_i \circstar \RR_{\eps_i, \delta_i}(b)$, where $\cT_i$ is the IPM given by Lemma~\ref{lem:post-im}.

    In the $f$-DP setting, by the definition of $g^{f_1, \dots, f_k}_{\mathrm{FDP}}$, each created mechanism $\cM_i$ is required to have \emph{finite communication}. Hence, its query space $Q_{\cM_i}$ and answer space $A_{\cM_i}$ are both finite, and $\cM_i$ halts before a certain number of interactions. Thus, the mechanism $\V[g_i]\circstar\I(b)\circstar\cM_i$ also has a finite answer space and terminates before the same number of steps. Although the verifier $\V[g_i]$ may, in principle, have an infinite query space, the outer verifier $\V[g^{f_1, \dots, f_k}_{\mathrm{FDP}}]$ in $\V[g^{f_1, \dots, f_k}_{\mathrm{FDP}}]\circstar\iecc$ ensures that all messages sent to $\cM_i$ are $g_i$-valid. Consequently, the set of admissible query messages to $\V[g_i]\circstar\I(b)\circstar\cM_i$ is restricted to $Q_{\cM_i}\times Q_{\cM_i}$, which is finite. Therefore, we can assume that this mechanism has finite communication, satisfying the condition of Lemma~\ref{lem:post-im-f-dp}, and the rest of the proof follows as the proof of Theorem~\ref{thm:comp-fixed-param-cm}.
\end{proof}

Similar to the concurrent composition of $k$ CMs with fixed privacy parameters, we define, for given privacy functions $f_1, \dots, f_k$, the verification function $g^{f_1, \dots, f_k}_{\mathrm{FDP}, \infty}$ as the $f$-DP analogue of $f^{\eps_1, \delta_1, \dots, \eps_k, \delta_k}_\infty$ from Definition~\ref{def:ver-func-par-fixed-param-comp-cm}. Since $(\eps, \delta)$-DP is a special case of $f$-DP, Theorem~\ref{thm:counter-example} immediately yields a counterexample: there exists a privacy function $f_1$ and an adversary $\cA$ such that the views of $\cA$ when interacting with $\V[g^{f_1}_{\mathrm{FDP}, \infty}] \circstar \I(b) \circstar \extconcomp$ have disjoint supports for $b = 0$ and $b = 1$. Equivalently, the $g^{f_1}_{\mathrm{FDP}, \infty}$-concurrent composition of CMs is not private. 

To overcome the lower bound, as in Section~\ref{subsec:parallel-restricted-ver-func}, we can define the verification function $g^{f_1, \dots, f_k}_{\mathrm{FDP}, \infty, \mathit{FPC}}$ identical to $g^{f_1, \dots, f_k}_{\mathrm{FDP}, \infty}$ with an additional condition that checks whether the created mechanisms are first-pair consistent. A proof identical to the proof of Theorem~\ref{thm:parallel-fixed-param-comp-fpc} shows:

\begin{corollary}\label{cor:parallel-comp-fpc-f-dp}
    For $k\in\N$, let $f, f_1, \dots, f_k: [0,1] \to [0,1]$ be functions. Suppose that for every $k$ NIMs $\cN_1, \dots, \cN_k$, where each $\cN_i$ is $f_i$-DP w.r.t. a neighbor relation $\sim_i$, the composition of $\cN_1, \dots, \cN_k$ is $f$-DP. Then, the $g^{f_1, \dots, f_k}_{\mathrm{FDP}, \infty, \mathit{FPC}}$-concurrent composition of CMs is also $f$-DP.
\end{corollary}

We know that IMs satisfy the first-pair consistent condition. Define \emph{the concurrent $k$-sparse parallel composition of IMs with privacy functions} $f_1, \dots, f_k$ as in Section~\ref{subsubsec:parallel-comp-im}. We have:

\begin{corollary}\label{cor:parallel-comp-f-dp-im}
    For $k\in\N$, let $f, f_1, \dots, f_k: [0,1] \to [0,1]$ be functions. Suppose that for every $k$ NIMs $\cN_1, \dots, \cN_k$, where each $\cN_i$ is $f_i$-DP w.r.t. a neighbor relation $\sim_i$, the composition of $\cN_1, \dots, \cN_k$ is $f$-DP. Then, the concurrent $k$-sparse parallel composition of IMs with privacy functions $f_1, \dots, f_k$ is also $f$-DP.
\end{corollary}

Finally, we extend the definition of \emph{concurrent filter composition of CMs} to the $f$-DP setting. To this end, we modify Definition~\ref{def:filter} so that a filter $\filt$ maps a finite sequence of functions $f_1, \dots, f_t : [0,1] \to [0,1]$ to an element of $\{\top, \bot\}$. As before, we construct the $f$-DP version of the verification function $f^\filt$ in Definition~\ref{def:ver-func-filter} by replacing the DP parameters with privacy functions and ensuring that every created CM has finite communication. We denote this verification function by $g^\filt_{\mathrm{FDP}}$.

\begin{corollary}\label{cor:filter-comp-f-dp}
    Let $\filt$ be a filter and $f:[0,1]\to[0,1]$ a function. If the $\filt$-filter composition of NIMs is $f$-DP, then the $g^\filt_{\mathrm{FDP}}$-concurrent composition of CMs is also $f$-DP. Equivalently, the concurrent $\filt$-filter composition of CMs is $f$-DP.
\end{corollary}
\section*{Acknowledgments}
\noindent
\textsuperscript{1}Salil Vadhan was supported by NSF grant BCS-2218803, a grant from the Sloan Foundation, and a Simons Investigator Award. Work began while a Visiting Researcher at the Bocconi University Department of Computing Sciences, supported by Luca Trevisan’s ERC Project GA-834861.
\vspace{0.75em}

\noindent
\textsuperscript{2}Monika Henzinger and Roodabeh Safavi were supported by the European Research Council (ERC) under the European Union's Horizon 2020 research and innovation programme (Grant agreement No.\ 101019564), and the Austrian Science Fund (FWF) under grants DOI 10.55776/Z422, DOI 10.55776/I5982, and DOI 10.55776/P33775. For open access purposes, the author has applied a CC BY public copyright license to any author-accepted manuscript version arising from this submission.
\begin{wrapfigure}[4]{r}{4cm}
  \vspace{-5pt}
  \centering
  \includegraphics[width=4cm]{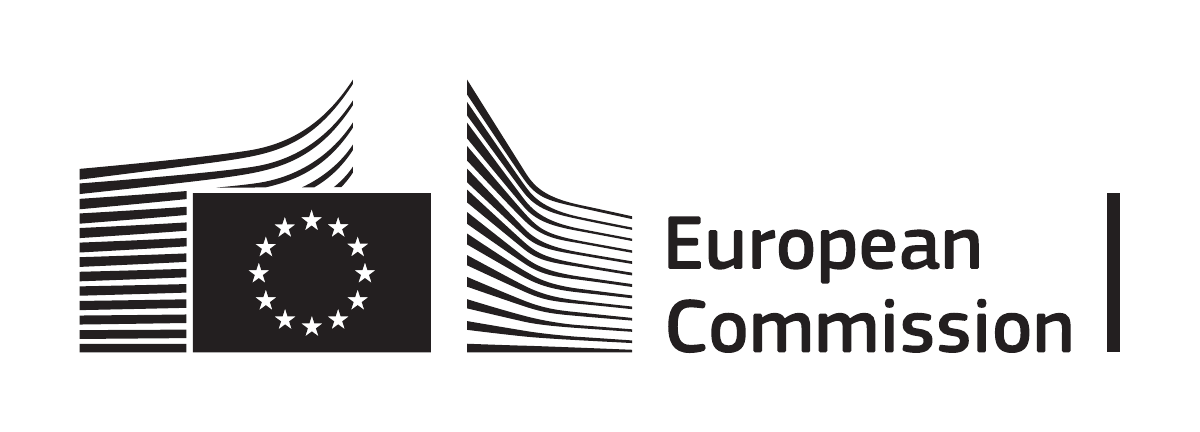}
  \vspace{-10pt}
\end{wrapfigure}
Views and opinions expressed are however those of the author(s) only and do not necessarily reflect those of the European Union or the European Research Council Executive Agency. Neither the European Union nor the granting authority can be held responsible for them.

\bibliographystyle{plain}
\bibliography{ref}
\appendix
\section{Figures}\label{app:figures}
\begin{figure}[ht] 
    \centering   
    \scalebox{0.52}{
    \begin{tikzpicture}[
    node distance=2cm and 2cm,
    db/.style={cylinder, shape border rotate=90, draw, aspect=0.25, minimum width=1cm, minimum height=0.8cm, font=\large},
    rect/.style={rectangle, draw, minimum width=2.5cm, minimum height=5.5cm, align=center, font=\large},
    arrow/.style={-{Stealth}, thick}
]


\node[rect] (im) {\LARGE IM};
\node[db, anchor=south east, xshift=-0.3cm, yshift=0.1cm] at (im.south east) (db2) {\LARGE \textcolor{blue}{$D$}};
\node[rect, left=3cm of im] (adv1) {\LARGE $\mathcal{A}$};

\draw[arrow] ([yshift=2.45cm]adv1.east) -- node[above,yshift=-3.5pt] { $q_1$} ([yshift=2.45cm]im.west);
\draw[arrow] ([yshift=1.95cm]im.west) -- node[above,yshift=-3.5pt] { $\tilde{q}_1(D)$} ([yshift=1.95cm]adv1.east);
\draw[arrow] ([yshift=1.45cm]adv1.east) -- node[above,yshift=-3.5pt] { $q_2$} ([yshift=1.45cm]im.west);
\draw[arrow]  ([yshift=0.95cm]im.west) -- node[above,yshift=-3.5pt] { $\tilde{q}_2(D)$} ([yshift=0.95cm]adv1.east);
\draw[arrow] ([yshift=0.45cm]adv1.east) -- node[above,yshift=-3.5pt] { $q_3$} ([yshift=0.45cm]im.west);
\draw[arrow] ([yshift=-0.05cm]im.west) -- node[above,yshift=-3.5pt] {$\tilde{q}_3(D)$} ([yshift=-0.05cm]adv1.east);
\draw[arrow] ([yshift=-0.55cm]adv1.east) -- node[above,yshift=-3.5pt] {$q_4$} ([yshift=-0.55cm]im.west);
\draw[arrow] ([yshift=-1.05cm]im.west) -- node[above,yshift=-3.5pt] {$\tilde{q}_4(D)$} ([yshift=-1.05cm]adv1.east);
\draw[arrow] ([yshift=-1.55cm]adv1.east) -- node[above,yshift=-3.5pt] {$q_5$} ([yshift=-1.55cm]im.west);
\draw[arrow] ([yshift=-2.05cm]im.west) -- node[above,yshift=-3.5pt] {$\tilde{q}_5(D)$} ([yshift=-2.05cm]adv1.east);
\node (dots1) at ([xshift=1.5cm,yshift=-2.4cm]adv1.east) {$\vdots$};

\node[rect, right=1cm of im] (adv3) {\LARGE $\mathcal{A}$};
\node[rect, right=3cm of adv3] (co) {\LARGE CO};
\node[db, anchor=south east, xshift=-0.3cm, yshift=0.1cm] at (co.south east) (db4) {\LARGE \textcolor{blue}{$D_t$}};

\draw[arrow] ([yshift=2.45cm]adv3.east) -- node[above,yshift=-3.5pt] { $u_1$} ([yshift=2.45cm]co.west);
\draw[arrow] ([yshift=1.95cm]co.west) -- node[above,yshift=-3.5pt] { $\tilde{q}(D_1)$} ([yshift=1.95cm]adv3.east);
\draw[arrow] ([yshift=1.45cm]adv3.east) -- node[above,yshift=-3.5pt] { $u_2$} ([yshift=1.45cm]co.west);
\draw[arrow]  ([yshift=0.95cm]co.west) -- node[above,yshift=-3.5pt] { $\tilde{q}(D_2)$} ([yshift=0.95cm]adv3.east);
\draw[arrow] ([yshift=0.45cm]adv3.east) -- node[above,yshift=-3.5pt] { $u_3$} ([yshift=0.45cm]co.west);
\draw[arrow] ([yshift=-0.05cm]co.west) -- node[above,yshift=-3.5pt] {$\tilde{q}(D_3)$} ([yshift=-0.05cm]adv3.east);
\draw[arrow] ([yshift=-0.55cm]adv3.east) -- node[above,yshift=-3.5pt] {$u_4$} ([yshift=-0.55cm]co.west);
\draw[arrow] ([yshift=-1.05cm]co.west) -- node[above,yshift=-3.5pt] {$\tilde{q}(D_4)$} ([yshift=-1.05cm]adv3.east);
\draw[arrow] ([yshift=-1.55cm]adv3.east) -- node[above,yshift=-3.5pt] {$u_5$} ([yshift=-1.55cm]co.west);
\draw[arrow] ([yshift=-2.05cm]co.west) -- node[above,yshift=-3.5pt] {$\tilde{q}(D_5)$} ([yshift=-2.05cm]adv3.east);
\node (dots1) at ([xshift=1.5cm,yshift=-2.4cm]adv3.east) {$\vdots$};

\node[rect, right=1cm of co] (adv2) {\LARGE $\mathcal{A}$};
\node[rect, right=3cm of adv2] (cm) {\LARGE CM};
\node[db, anchor=south east, xshift=-0.3cm, yshift=0.1cm] at (cm.south east) (db3) {\LARGE \textcolor{blue}{$D_t$}};

\draw[arrow] ([yshift=2.45cm]adv2.east) -- node[above,yshift=-3.5pt] { $m_1=u$} ([yshift=2.45cm]cm.west);
\draw[arrow] ([yshift=1.95cm]cm.west) -- node[above,yshift=-3.5pt] { $m_1'=\text{ack}$} ([yshift=1.95cm]adv2.east);
\draw[arrow] ([yshift=1.45cm]adv2.east) -- node[above,yshift=-3.5pt] { $m_2=u'$} ([yshift=1.45cm]cm.west);
\draw[arrow]  ([yshift=0.95cm]cm.west) -- node[above,yshift=-3.5pt] { $m_2'=\text{ack}$} ([yshift=0.95cm]adv2.east);
\draw[arrow] ([yshift=0.45cm]adv2.east) -- node[above,yshift=-3.5pt] { $m_3=q$} ([yshift=0.45cm]cm.west);
\draw[arrow] ([yshift=-0.05cm]cm.west) -- node[above,yshift=-3.5pt] {$m_3'=\tilde{q}(D_3)$} ([yshift=-0.05cm]adv2.east);
\draw[arrow] ([yshift=-0.55cm]adv2.east) -- node[above,yshift=-3.5pt] {$m_4=u''$} ([yshift=-0.55cm]cm.west);
\draw[arrow] ([yshift=-1.05cm]cm.west) -- node[above,yshift=-3.5pt] {$m_4'=\text{ack}$} ([yshift=-1.05cm]adv2.east);
\draw[arrow] ([yshift=-1.55cm]adv2.east) -- node[above,yshift=-3.5pt] {$m_5=q^*$} ([yshift=-1.55cm]cm.west);
\draw[arrow] ([yshift=-2.05cm]cm.west) -- node[above,yshift=-3.5pt] {$m_5'=\tilde{q}^*(D_5)$} ([yshift=-2.05cm]adv2.east);
\node (dots1) at ([xshift=1.5cm,yshift=-2.4cm]adv2.east) {$\vdots$};

\end{tikzpicture}
    }
    \caption{ Comparison between interactive mechanism (IMs), continual observation (CO), and continual mechanisms (CMs)} \label{fig:im-co-cm}
    \Description{Comparison between interactive mechanism (IMs), continual observation (CO), and continual mechanisms (CMs)}
\end{figure}



\section{Variants of Concurrent Composition}\label{app:variants-concomp}
Here, we elaborate on the different types of concurrent composition that we study in this paper:

    (1) \textbf{Concurrent Composition of Continual Mechanisms:} In the \emph{concurrent composition of fixed continual mechanism}, $k$ continual mechanisms $\cM_1, \dots, \cM_k$ are fixed at the beginning, and an adversary adaptively issues pairs of messages for them such that for each $\cM_i$, the sequences formed by the first messages of all message pairs and the sequence formed by the second messages of all message pairs  are neighboring. A secret bit $b \in \zo$ determines whether each mechanism receives the first or second sequence from each pair. The adversary's goal is to guess $b$ using the answers of all $k$ mechanisms. Theorem~\ref{thm:comp-fixed-cm}, which is the formal statement of Theorem~\ref{thm:extconcomp-intro} above, analyzes the privacy of this composition.

    A natural extension is the \emph{concurrent composition of continual mechanism with fixed parameters}, where instead of fixing the mechanisms themselves, $k$ privacy parameter pairs $(\eps_1, \delta_1), \dots, (\eps_k, \delta_k)$ are fixed in advance. The adversary then adaptively selects mechanisms $\cM_1, \dots, \cM_k$ over time, ensuring that, for each $1\leq i \leq k$, the mechanism $\cM_i$ is $(\eps_i, \delta_i)$-DP. The privacy guarantee for this variant is the same as Theorem~\ref{thm:comp-fixed-cm} and is stated in Theorem~\ref{thm:comp-fixed-param-cm}.

    (2) \textbf{Concurrent $k$-Sparse Parallel Composition of Continual Mechanisms:} This composition generalizes the concurrent composition of continual mechanisms with $k$ fixed parameters by permitting the creation of an \textit{arbitrary} number of continual mechanisms. 
    As before, $k$ privacy parameters $(\eps_1, \delta_1), \dots, (\eps_k, \delta_k)$ are fixed at the beginning. The adversary adaptively creates a potentially unlimited number of \emph{continual} mechanisms and issues pairs of messages to each mechanism such that for all but at most $k$ of these mechanisms, each message pair consists of identical messages. Let $\cM_1, \dots, \cM_k$ denote the exceptional mechanisms, which may receive differing message sequences. For each $1\leq i \leq k$, the mechanism $\cM_i$ must be $(\eps_i, \delta_i)$-DP, and the sequences of the first and second messages for $\cM_i$ must be neighboring. We note the choice of mechanisms $\cM_1, \dots, \cM_k$ is made adaptively over time. Specifically, the adversary decides whether to issue a pair of differing messages for a mechanism based on the previous answers of all mechanisms, including the outputs of that mechanism.
    
    As in previous compositions, a secret bit $b$ determines whether each mechanism receives the first or second message from each pair, and the adversary seeks to guess $b$ based on the outputs. 
    Let $\cM_j'$ denote the $j$-th created mechanism, and let $(\eps_j', \delta_j')$ be its privacy parameter. Theorem~\ref{thm:parallel-comp-approx}, informally stated in Theorem~\ref{thm:parallelcomp-approxDP-intro}, analyzes the privacy of the concurrent $k$-sparse parallel composition of continual mechanisms when $\sum\delta_j'$ is upper bounded by a predetermined value.
    Corollary~\ref{cor:parallel-comp-pure}, informally stated in Corollary~\ref{cor:parallelcomp-pureDP-intro}, sets this upper bound to $0$ and provides privacy guarantees for the concurrent $k$-sparse parallel composition of purely differentially private continual mechanisms.

    (3) \textbf{Concurrent $k$-Sparse Parallel Composition of Interactive Mechanisms:} In this composition, the adversary creates an unbounded number of \emph{interactive} mechanisms over time, instead of continual ones. As before, $k$ privacy parameter pairs $(\eps_1, \delta_1), \dots, (\eps_k, \delta_k)$ are fixed in advance. Each time the adversary creates an interactive mechanism, they also select a pair of neighboring datasets for it. Instead of message pairs, the adversary adaptively sends (single) queries to each mechanism. A secret bit $b$ determines whether the mechanism is run on the first or second dataset of its respective pair. The adversary then tries to infer $b$ from the responses of all mechanisms. Theorem~\ref{thm:concurrent-parall-comp-IM-intro} analyzes the privacy of this composition when all $\cM_i$ are $(\eps_0, \delta_0)$-DP. The more general case for mechanisms with possibly different $(\eps_i, \delta_i)$ parameters is stated in Theorem~\ref{thm:parallel-fixed-param-comp-fpc}.

\end{document}